\documentclass[a4paper,bibliography=totoc,headings=standardclasses,chapterprefix=false,headsepline]{scrbook}
\pdfoutput=1
\usepackage[utf8]{inputenc}

\usepackage[T1]{fontenc} 
\usepackage{textcomp} 
\usepackage{lmodern} 
\usepackage[sc,osf]{mathpazo} 
\usepackage{myT1classico} 
\DeclareMathAlphabet{\mathsf} {T1}{\sfdefault}{m}{n} 
\SetMathAlphabet{\mathsf}{bold}{T1}{\sfdefault}{b}{n} 

\usepackage{amsmath,amssymb,amsthm}
\usepackage{mathtools} 
\usepackage{cancel}
\usepackage{enumitem}
\usepackage{relsize}
\newcommand{\acronym}[1]{\texorpdfstring{\textsmaller{#1}}{#1}}
\usepackage{accents} 

\usepackage{scrlayer-scrpage}

\usepackage[ngerman,latin,british]{babel}
\usepackage{csquotes} 

\usepackage{microtype} 
\usepackage{graphicx}
\usepackage{transparent}
\usepackage{xcolor}

\usepackage{tikz}
\usetikzlibrary{decorations.markings,calc,math}

\usepackage[pdfusetitle,hidelinks,bookmarksnumbered=true]{hyperref}

\usepackage[style=ext-alphabetic,url=false,date=year,articlein=false,backref=true,maxnames=99]{biblatex}
\DeclareFieldFormat{pages}{#1} 
\DeclareFieldFormat[article]{volume}{\textbf{#1}} 
\DeclareFieldFormat[article]{journaltitle}{#1} 
\AtEveryBibitem{%
	\clearfield{number}
	\clearfield{issue}} 
\renewbibmacro*{pageref}{
	\iflistundef{pageref}{}
	{ $\hookrightarrow$ \ppspace\printlist[pageref][-\value{listtotal}]{pageref}}}

\numberwithin{equation}{section}

\newcommand{\D}{\mathrm{d}} 
\newcommand{\DD}{\mathrm{D}} 
\newcommand{\I}{\mathrm{i}} 
\DeclareMathOperator{\E}{e} 
\DeclareMathOperator{\Or}{O} 
\DeclareMathOperator{\tr}{tr} 

\newif\ifvectorsbold
\vectorsboldtrue

\ifvectorsbold
	\let\vect\boldsymbol 
	\let\ivect\vect
	\newcommand{\hatvect}[1]{\vect{\hat{#1}}}
\else
	\let\vect\vec
	\let\ivect\relax
	\let\hatvect\hat
\fi

\def\electron{%
    \fill[shading=ball,ball color=gray] (0,0) circle (.3);
}
\def\orbit(#1,#2,#3){%
  \draw[thick,
  rotate=#1,
  postaction=decorate,
  decoration={markings,pre=moveto,pre length=#3,
  mark=at position {#2} with {\electron},
}]
  (0,0) ellipse (1 and 4);
}

\newcommand{\drawatom}[1]{%
\begin{tikzpicture}
\clip(-8,-12) rectangle (8,4);
\begin{scope}[shift={(0,-11.5*(#1/100)*(#1/100))}]
\fill[shading=ball,ball color=black] (0,0) circle (.4);
\tikzmath { int \pos; \pos=mod(#1,100); }
\orbit(-38,-\pos/100,0cm)
\orbit(38,\pos/100,0cm)
\orbit(270,\pos/100,0cm)
\end{scope}
\end{tikzpicture}
}

\makeatletter
\DeclareNewLayer[
  background,
  oddpage,
  textarea,
  addhoffset=\textwidth+\marginparsep,
  addvoffset=\textheight,
  width=\marginparwidth,
  align=b,
  mode=picture,
  contents={%
    \if@mainmatter
      \putLL{\resizebox{\marginparwidth}{\marginparwidth}{\drawatom{\thepage}}}%
    \fi
  }
]{flipbook.odd}
\makeatother

\AddLayersToPageStyle{plain.scrheadings}{flipbook.odd}
\AddLayersToPageStyle{scrheadings}{flipbook.odd}

\title{Post-Newtonian Description\\ of Quantum Systems\\ in Gravitational Fields}
\author{M.\,A.\,St. Philip Klaus Schwartz}
\def\vorgelegtdate{17. Juli 2020}
\def\dispdate{18. September 2020}
\def\year{2020}
\hypersetup{pdfauthor=Philip Klaus Schwartz,pdfsubject=Dissertation}
\newif\ifgenehmigt
\genehmigttrue

\begin{document}

\frontmatter

\input{titlepage}

\chapter{Abstract}
\enlargethispage{3.5\baselineskip}

This thesis deals with the systematic treatment of 
quantum-mechanical systems situated in post-Newtonian 
gravitational fields. At first, we develop a framework of 
geometric background structures that define the notions of a 
post-Newtonian expansion and of weak gravitational fields. 
Next, we consider the description of single quantum 
particles under gravity, before continuing with a simple composite 
system. Starting from clearly spelled-out assumptions, 
our systematic approach allows to properly \emph{derive} 
the post-Newtonian coupling of quantum-mechanical systems 
to gravity based on first principles. This sets it 
apart from other, more heuristic approaches that are 
commonly employed, for example, in the description of 
quantum-optical experiments under gravitational influence.

Regarding single particles, we compare simple canonical 
quantisation of a free particle in curved spacetime to 
formal expansions of the minimally coupled Klein--Gordon 
equation, which may be motivated from the framework of 
quantum field theory in curved spacetimes. Specifically, 
we develop a general \acronym{WKB}-like post-Newtonian 
expansion of the Klein--Gordon equation to arbitrary order 
in the inverse of the velocity of light. Furthermore, for 
stationary spacetimes, we show that the Hamiltonians 
arising from expansions of the Klein--Gordon equation 
and from canonical quantisation agree up to linear order 
in particle momentum, independent of any expansion in the 
inverse of the velocity of light.

Concerning the topic of composite systems, we perform 
a fully detailed systematic derivation of the first order 
post-Newtonian quantum Hamiltonian describing the dynamics 
of an electromagnetically bound two-particle system which 
is situated in external electromagnetic and gravitational 
fields. This calculation is based on previous work by 
Sonnleitner and Barnett, which we significantly extend by 
the inclusion of a weak gravitational field as described 
by the Eddington--Robertson parametrised post-Newtonian 
metric.

In the last, independent part of the thesis, 
we prove two uniqueness results characterising the 
Newton--Wigner position observable for Poincaré-invariant 
\emph{classical} Hamiltonian systems: one is a direct 
classical analogue of the well-known quantum Newton--Wigner 
theorem, and the other clarifies the geometric interpretation 
of the Newton--Wigner position as ‘centre of spin’, as 
proposed by Fleming in 1965.

\begin{center}
	\emph{Keywords:} quantum systems under gravity, 
	post-Newtonian expansion, \\ post-Newtonian gravity, weak 
	gravity
\end{center}

\selectlanguage{ngerman}
\chapter*{Zusammenfassung}
\enlargethispage{3.5\baselineskip}

Diese Arbeit beschäftigt sich mit der systematischen 
Beschreibung quantenmechanischer Systeme in post-Newton'schen 
Gravitationsfeldern. Zunächst entwickeln wir geometrische 
Hintergrundstrukturen, welche die Konzepte einer 
post-Newton'schen Entwicklung und schwacher 
Gravitationsfelder zu definieren ermöglichen. Anschließend 
beschäftigen wir uns mit der Beschreibung einzelner 
Quantenteilchen unter Gravitation und wenden uns 
schließlich einem einfachen zusammengesetzten System 
zu. Unsere von klar formulierten Annahmen ausgehende 
systematische Vorgehensweise er\-möglicht es, die 
post-Newton'sche Kopplung quantenmechanischer Systeme an 
Gravitation im eigentlichen Sinne \emph{herzuleiten}. 
Dies unterscheidet sie von \mbox{anderen}, \mbox{heuristischeren} 
Herangehensweisen, wie sie beispielsweise oft zur 
Beschreibung quanten\-optischer Experimente unter Gravitation 
benutzt werden.

Für einzelne Teilchen vergleichen wir die einfache 
kanonische Quantisierung freier Teilchen in gekrümmten 
Raumzeiten mit formalen Entwicklungen der minimal 
gekoppelten Klein-Gordon-Gleichung, welche 
quantenfeldtheoretisch motiviert werden können. Konkret 
entwickeln wir eine allgemeine \acronym{WKB}-artige 
post-Newton'sche Entwicklung der Klein-Gordon-Gleichung 
zu beliebiger Ordnung im Inversen der Lichtgeschwindig\-keit. 
Ferner zeigen wir für stationäre Raumzeiten, dass die 
Hamilton-Operatoren, welche aus Entwicklungen der 
Klein-Gordon-Gleichung bzw. mit kanonischer Quantisierung 
hergeleitet werden, zu linearer Ordnung im Teilchenimpuls 
übereinstimmen, unabhängig von jeglicher Entwicklung im 
Inversen der Lichtgeschwindigkeit.

Wir leiten den in erster Ordnung post-Newton'schen 
Hamiltonoperator vollständig her, der die Dynamik eines 
elektromagnetisch gebundenen Zwei-Teilchen-Systems 
beschreibt, das sich in sowohl einem externen 
elektromagnetischen als auch einem Gravitationsfeld befindet. 
Diese Rechnung basiert auf einer Arbeit von Sonnleitner und 
Barnett, die wir durch die Einbeziehung der Gravitation 
maßgeblich erweitern.

Im letzten, unabhängigen Teil der Arbeit beweisen wir zwei 
Eindeutigkeitsresultate über die Newton-Wigner-Ortsobservable 
für Poincaré-invariante \emph{klassische} Hamilton'sche 
Systeme. Eines ist ein direktes klassisches Analogon des 
quantenmechanischen Newton-Wigner-Satzes; das andere gibt 
eine klare Charakterisierung der geometrischen Interpretation 
des Newton-Wigner-Orts als \glqq Spin-Zentrum\grqq, die 1965 
von Fleming vorgeschlagen wurde.

\vspace{-.07em}
\begin{center}
	\emph{Schlagworte:} Quantensysteme unter Gravitation, 
	post-Newton'sche Entwicklung, \\ post-Newton'sche 
	Gravitation, schwache Gravitation
\end{center}

\selectlanguage{latin}
\chapter*{Summarium\footnote{translatus de Philippo Sandero}}
\enlargethispage{3.5\baselineskip}

hoc opus est de descriptione sastematica systematium 
mechanicorum quanticorum in campis gravitalibus post 
Newtonum. primo recessas structuras geometricas elaborabimus, 
quibus consilia expansionis post Newtonum parvisque campis 
gravitalibus definiri potest. deinde descriptioni singulorum 
particulorum quanticorum sub gravitatione studebimus atque 
ultimo in systemate composito facili versabimur. ab 
praesumptionibus clare conceptis systematice procedenti 
post Newtonum copulationem systematium mechanicorum 
quanticorum ad gravitationem proprie \emph{dedicare} 
poterimus. qui modus procedendi ab aliis, heuristicis, 
velut ad experimenta optica quantica describendum utuntur, 
differt.

quod attinet ad singula particula, quantificationem 
canonicam facilem particulorum nullas vires 
experientum in spatiotemporibus curvatis cum 
expansionibus formalibus equationis Kleini Gordonique 
minime copulatae, quae ex ratione quanticorum camporum 
motivari possunt, comparabimus. proprie ad dicendum post 
Newtonum expansionem generalem \acronym{WKB}-bilem 
equitationis Kleini Gordonique ad quamlibet ordinem in 
inverso velocitatis lucis faciemus. quod praeterea 
attinet ad spatiotempora stationaria, operatores 
Hamiltoni de expansionibus equationis Kleini 
Gordonique aut cum quantificatione canonica dedicatos 
in ordine lineali inter se impetu particulorum consentire 
demonstrabimus. quod non obnoxium cuicumque expansioni 
in inverso velocitatis lucis est.

quod attinet ad systemata coniuncta, 
operatorem Hamiltoni in prima ordine post Newtonum radicite 
dedicemus, qui dynamiken systematis ex duobus particulis 
electromagnetice coniuncti describit, quod et in campo 
electromagnetico et in campo gravitale est. quae ratio 
in opere Sonnleitneri Barnettique posita est, quod 
gravitationem comprehendendo augebimus multo.

ultima in parte absoluta huius operis duos exitus 
perspicuitatis de loci quantitate Newtoni 
Wignerique, quod attinet ad systemata Hamiltoni 
classica invarianta secundum Poincareum, demonstrabimus. 
alius est analogon classicum directum mechanici quantici 
theorematis Newtoni Wignerique; alio locus Newtoni 
Wignerique pro \enquote{medio impetus rotationis interni} 
geometrice interpretatur, ut Flemingus proposit 
anno \acronym{MMDCCXVIII} ab urbe condita.

\begin{center}
	\vspace{-.30\baselineskip}
	\emph{proposita:} systemata quantica in gravitatione, 
	expansio post Newtonum, \\ gravitatio post Newtonum, 
	gravitatio parva
\end{center}

\selectlanguage{british}

\clearpage


\chapter{Acknowledgements}
\enlargethispage{-2\baselineskip}

First and foremost, I wish to thank my supervisor Domenico 
Giulini. Apart from guiding me through the research that 
lead up to this thesis and always being available for smaller 
as well as bigger questions of mine, Nico has deeply 
influenced my whole way of thinking about physics and mathematics. 
I am very grateful to have had a supervisor with such a 
sharp and conceptually clear personal way of understanding and 
explaining physical ideas.

Furthermore, I am grateful to Klemens Hammerer, Claus Kiefer, 
and Elmar Schrohe for their willingness to act as referees 
for my thesis and/or being part of the examination commission. 
I wish to thank Klemens as well for bringing to my attention 
the article \cite{sonnleitner18}, on which a significant part 
of my work in this thesis is built.

For taming all the probable and improbable uncertainties of 
university bureaucracy at several stages, even at very short 
notice by my side, I want to thank Birgit Ohlendorf.

My research was funded by the \acronym{DFG} through the 
Collaborative Research Centre 1227 DQ-mat, projects B08 and A05.

I want to thank my good friends Anke, Flo, and Rökki for their 
friendship, and for keeping me sane in phases of high workload.
Special thanks go to my friends and colleagues from our 
institute, who made the last few years into a very enjoyable 
time for me: Johannes, with whom I shared an office for two 
years, for many interesting discussions about 
physics and beyond; Timo, for introducing me to the best peanut 
sauce of Hannover; Schiden, for enduring many long explanations 
of mine about random interesting mathematical facts, and sign 
errors; as well as Deniz, Lennart, Yannic, Daniel, Konstantin, 
Ramona, Michi, Thomas, Oscar, Colin, and many others I probably 
forgot. Apart from being nice company, all of these people have 
also helped me a lot by discussing topics of my research.

For proofreading, I am indebted to Benjamin Haake, Florian 
Kranhold, Lennart Janshen, Michael Jung, and Michael Werner. 
All remaining typographical errors and stylistic flaws are, 
of course, my own fault. For the idea of providing a Latin 
abstract, and for translating it, I thank Philipp Sander.

I cannot adequately express in words my deep gratitude to my 
parents. They have loved and supported me unconditionally 
for the whole of my life, enabling me to pursue my interests 
and passions even at their own shortcoming. The same goes 
for my quasi-godmother Nana. My scientific 
interest was also greatly nurtured by my grandma. Although 
she sadly passed away quite some time ago, I know she 
would be proud of me for writing a doctoral thesis on, as she 
put it, `quantum physics, because I don't understand it'.

Last but not least, I thank Kaie for making me laugh and 
making me cry, for enduring me in bad mood as well as in 
good, for providing support and encouragement; that is, for 
always, unconditionally, being there for me and loving me.

\setcounter{tocdepth}{1}
\tableofcontents

\mainmatter


\chapter{Introduction}

Imagine we are given a quantum-mechanical system whose 
time evolution in the absence of gravity is known in 
terms of the ordinary time-dependent Schrödinger equation. 
In other words: we know the system's Hamiltonian if all 
gravitational interactions are neglected. We now ask: which 
principles do we use in order to deduce the system's 
interaction with a given external gravitational field? 
Note that by `gravitational field' we understand all the 
ten independent components $g_{\mu\nu}$ of the spacetime 
metric, subject to Einstein's field equations of general 
relativity --~or, more generally, to the equations of some 
other metric theory of gravity~-- and not just the scalar 
component $\phi$ representing the Newtonian potential.

In fact, for Newtonian gravity there is no problem at all 
in describing its coupling to ordinary quantum mechanics: 
we may simply include a background Newtonian gravitational 
potential $\phi$ into the Schrödinger equation describing 
a `non-relativistic'\footnote
	{As a matter of principle, we try to avoid the common 
	but misleading adjective `non-relativistic' to distinguish 
	Galilei-invariant dynamical laws from `relativistic' ones, 
	by which one then means those obeying Poincaré invariance. 
	It is not the validity of the relativity principle that 
	distinguishes both cases, but rather the way which that 
	principle is implemented in. Nevertheless, since we cannot 
	entirely escape traditionally established nomenclature, we 
	will occasionally use the term `non-relativistic' in the 
	sense just explained and think of it as always being put 
	between inverted commas.}
particle of mass $m$ and zero spin, giving
\begin{equation}
	\I\hbar\partial_t \psi = \left(-\frac{\hbar^2}{2m}\Delta + m \phi\right) \psi .
\end{equation}
This equation has extensively been tested in the gravitational 
field of the earth, beginning with neutron interferometry 
in the classic Colella--Overhauser--Werner experiment 
\cite{COW75} and leading up to atom interferometers of the 
Kasevich--Chu type, accomplishing, e.g., highly precise 
measurements of the gravitational acceleration $g$ on the 
earth \cite{farah14}. We ask what kind of `post-Newtonian 
corrections' to this equation arise from general relativity 
or other metric theories of gravity, considering additional 
terms involving the Newtonian potential $\phi$ as well as 
new terms involving all metric components.

The behaviour of quantum systems in general gravitational 
fields is naturally of fundamental conceptual theoretic 
interest. However, it is also of immediate practical 
importance, relating to recent experimental developments 
in quantum optics and matter-wave interferometry: 
these have now reached a degree of precision that covers 
`relativistic corrections' which were hitherto not 
considered in such settings. In particular, this includes 
couplings between `internal' and `centre of mass' degrees 
of freedom of composite systems without a Newtonian 
analogue~-- as for example induced by post-Newtonian 
gravitational fields. The most famous example for a 
possible implication of such couplings is probably the 
controversially discussed topic of gravitationally induced 
quantum dephasing \cite{zych11,pikovski15,bonder15,pang16}. 
Other experimentally inclined topics for which the 
gravity--quantum matter coupling beyond Newton is relevant 
include, for example, atom interferometric gravitational 
wave detection \cite{gao18}, quantum tests of the classical 
equivalence principle \cite{schlippert14}, or proposals 
of quantum formulations of the equivalence principle 
\cite{zych18} and tests thereof \cite{rosi17}.
\enlargethispage{\baselineskip}

Clearly, such experiments require proper `relativistic' 
treatments for their theoretical descriptions, that may 
be trusted as describing the situation in a correct way. 
However, the descriptions one finds in the literature are 
often restricted to the more or less \emph{ad hoc} addition 
of `relativistic effects' known from classical physics, 
such as velocity-dependent masses, second-order Doppler 
shifts, or redshifted energies and time dilations due to 
relative velocities and/or gravitational potentials; see, 
e.g., \cite{dimopoulos08,zych11,pikovski15,roura18,
giese19,loriani19,zych19}. Such approaches are 
conceptually dangerous for a number of reasons: they 
neither guarantee completeness and independence of the 
`relativistic effects', nor do they need to apply in 
non-classical situations where quantum properties dominate 
the dynamics. Namely, as is common in atom interferometry, 
these treatments make use in an essential way of 
semi-classical notions like `wordline' and `redshift', 
which have no immediate meaning in quantum theory unless 
the state of the system is severely restricted in an 
\emph{a priori} fashion: the overall pure state of the 
system has to be assumed to separate into the tensor 
product of a pure state for the centre of mass degrees 
of freedom with a pure state for the relative degrees of 
freedom; and furthermore, the state for the centre of 
mass has to be of semiclassical nature, so as to determine 
a worldline for which the notion of proper time can be 
defined.\footnote
	{We recall that the path integral in ordinary quantum 
	mechanics generally receives contributions from continuous 
	but nowhere differentiable paths. Only in very special 
	situations is the dominant contribution given by the action 
	along a smooth classical path, such that one may define an 
	arc length, i.e.\ a proper time.}
It may well be that these \emph{a priori} restrictions 
can be justified in specific applications within quantum 
optics and atom interferometry. However, we wish to promote 
the view that the theoretical problem of describing the 
coupling between quantum-mechanical systems and post-Newtonian 
gravity should be solved \emph{independently} of such 
restrictions, in a systematic and well-defined way. Such a 
proper systematic \emph{derivation} of the coupling will 
also make sure that all relevant `relativistic corrections' 
to the Newtonian description are present, and that none is 
included multiple times.

In answering the question of what such a systematic 
coupling procedure could look like, we have to 
address a conceptual difficulty that does not arise 
for classical systems. Namely, for classical matter 
obeying Poincaré-invariant dynamical laws, there is a 
systematic, almost algorithmic procedure one can employ 
in order to couple it to metric theories of gravity: the 
usual `minimal coupling scheme'. We recall that, in a 
nutshell, this scheme consists in a two-step process 
\cite{misner73}: first, write down the matter's dynamical 
law in a Poincaré-invariant fashion in Minkowski spacetime; 
second, replace the flat Minkowski metric $\eta$ by the 
potentially curved Lorentzian metric $g$ of spacetime, 
and the partial derivatives with respect to the affine 
inertial coordinates of Minkowski spacetime (i.e.\ 
the covariant derivatives with respect to $\eta$) by 
Levi--Civita covariant derivatives with respect to 
$g$. This gives a unique way of coupling classical 
matter fields to metric theories of gravity, up to the 
well-known issue of curvature ambiguities (arising from 
the non-commutativity of covariant derivatives in the 
curved case) and the possibility of non-minimal coupling 
(i.e.\ explicit coupling to the curvature tensor).

The minimal coupling scheme is rooted in Einstein's 
equivalence principle, whose essence is that `gravity' 
can be fully encoded in the metric geometry of spacetime, 
which is \emph{common to all matter components}. We 
stress that this is the important point, encoding the 
\emph{universality} of gravitational interaction: any 
matter component, be it some elementary particle with 
or without mass, spin, electric charge, or other features, 
or be it a macroscopic body, like a football or a planet, 
will couple to gravity in a way that only depends on one 
and the same geometry of spacetime; compare \cite{thorne73} 
and \cite{will93}. Note that this does in no way imply 
that all bodies `fall' in the same way: for a realistic 
body, which is spinning and/or possesses mass multipoles 
of higher order than the single monopole of an idealised 
`test particle', any approximate `central worldline' will 
depend on the characteristics of the body and deviate from 
that of a test particle (i.e.\ a geodesic). However, as 
long as all these deviations find their explanations in 
couplings to the spacetime geometry, no violation of the 
equivalence principle should be concluded. This remark also 
applies in connection with attempts to formulate the 
equivalence principle in quantum mechanics: simple quantum 
translations of some notion of `universality of free fall' 
--~as the Newtonian one proposed in \cite{zych18}~-- should 
not be seen as capturing any core statement of the 
equivalence principle; to the contrary, they even bear the 
danger of falsely concluding violations. Furthermore, such 
formulations depend on notions of `worldlines', and thus 
are based on \emph{a priori} assumptions concerning the 
state of the matter. We are convinced that any possible 
generally valid implementation of the equivalence principle 
into quantum mechanics should not make such assumptions. 
An extensive discussion of these important conceptual 
issues may be found in our article \cite{schwartz19:AiG}.

We now return to the more concrete question of systematic 
coupling procedures of quantum mechanics to gravitational 
fields. The above-mentioned conceptual problem which we 
face here is that the minimal coupling scheme simply 
cannot be applied in that case: ordinary quantum 
mechanics is Galilei-invariant, and so even the first step 
of the minimal coupling procedure cannot be implemented. 
As is well-known, enforcing Poincaré symmetry upon quantum 
mechanics eventually leads to the framework of 
Poincaré-invariant quantum field theory, often called 
`Relativistic Quantum Field Theory' (\acronym{RQFT}), 
whose mathematical structure and physical interpretation 
is far more complex than that of ordinary `non-relativistic' 
quantum mechanics. In particular, \acronym{RQFT} does not 
have a form similar to a usual, classical field theory on 
Minkowski spacetime~-- i.e.\ also \acronym{RQFT} cannot be 
coupled to metric gravity by a direct application of the 
minimal coupling scheme. Instead, the framework of quantum 
field theory in curved spacetimes (\acronym{QFTCS}) applies 
minimal coupling at the \emph{classical} level, and 
then employs methods to quantise the minimally coupled 
classical field theories \cite{baer09,wald94}.

So we are lead to accept the fact that it does not seem 
to be possible to couple an `already quantised' theory to 
gravity, and thus to turn to \acronym{QFTCS} as the best 
available solution for the systematic description of gravity--quantum 
matter coupling. Does that mean we would have to employ the 
whole machinery of \acronym{QFTCS} in order to just answer 
simple questions concerning matter--gravity interactions 
that go beyond the simplest couplings to the Newtonian 
potential? We think that the answer is no, at least as 
long as we are merely interested in leading order 
`relativistic corrections' below the threshold of 
quantum-field-theoretic pair production, and as long 
as the spacetime geometry is at least approximately 
stationary, such that there is a consistent field-theoretic 
concept of particles. At the same time, we think that the 
alternative to full \acronym{QFTCS} should not consist of 
\emph{ad hoc} procedures guided by more or less well 
founded `physical intuition'. Rather we should look for 
general and systematic methods that allow to derive 
the full coupling, and arguably qualify as a proper 
post-Newtonian approximation. This thesis aims to provide 
a positive contribution to this end.

\section{Plan of this thesis}

In chapter \ref{chap:geometric_structures}, we will set up 
the conceptual framework for our systematic post-Newtonian 
expansions in the following chapters: we introduce a set of 
geometric background structures that enable us to define 
the notions of weak gravitational fields and post-Newtonian 
expansions.

Based on this framework, chapter \ref{chap:PNSE} will deal 
with the systematic description of \emph{single} quantum 
particles under gravity. We introduce a simple method of 
canonical quantisation of a free particle in a post-Newtonian 
spacetime, and aim to compare its results to methods 
which are more firmly rooted in first principles. 
Therefore, motivated from \acronym{QFTCS}, we develop two 
different kinds of formal post-Newtonian expansions of the 
minimally coupled Klein--Gordon equation, and compare their 
results to those from the canonical quantisation method. 
This will lead to the conclusion that at the lowest 
relevant post-Newtonian orders, simple canonical 
quantisation may safely be employed.

Chapter \ref{chap:atom_in_gravity} will continue the 
investigation with the study of a simple \emph{composite} 
quantum system in post-Newtonian gravity. We consider a 
simple `atomic' system consisting of two electromagnetically 
bound bosonic particles, situated in an external 
electromagnetic field as well as an external gravitational 
field described by the Eddington--Robertson parametrised 
post-Newtonian metric. We give a fully detailed systematic 
derivation of the first order post-Newtonian quantum 
Hamiltonian describing the dynamics of the atomic system in 
this situation.

The last proper chapter \ref{chap:Newton--Wigner} is entirely 
independent of the rest of the thesis: it is concerned with 
the investigation of the special-relativistic localisation 
problem for classical (i.e.\ non-quantum) systems, in 
particular with characterisations of the Newton--Wigner 
position observable for such systems. Even though this 
topic is almost completely disconnected from the description 
of quantum systems in post-Newtonian gravity, it arose 
in a natural way from the investigations in chapter 
\ref{chap:atom_in_gravity}. For this reason, and due to the 
particular conceptional and mathematical beauty I (the 
author) see in the results obtained in this chapter, I decided 
to include it into this thesis.

We end with a few concluding remarks in chapter \ref{chap:conclusion}.

\section*{Publication list}

This thesis is based on the following articles, as 
indicated in the beginning of the chapters:

\begin{itemize}[align=left]
	\item[\cite{schwartz19:PNSE}] \fullcite{schwartz19:PNSE}
	\item[\cite{schwartz19:AiG}] \fullcite{schwartz19:AiG}
	\item[\cite{Schwartz.Giulini:2020}] \fullcite{Schwartz.Giulini:2020}
\end{itemize}


\chapter{Geometric structures for post-Newtonian expansions}
\label{chap:geometric_structures}

In an arbitrary general-relativistic\footnote
	{Or described by any other metric theory of gravity.}
spacetime, the concept of a `post-Newtonian \mbox{expansion}' 
does not exist \emph{per se}: to make sense of it, we 
need to introduce certain background structures that give 
meaning to notions like `weak gravitational fields' and 
`slow velocities' of objects in the spacetime. This 
chapter will be devoted to the introduction of such 
structures and the description of our post-Newtonian 
expansion framework, which will be used in the subsequent 
chapters. We also use this chapter to introduce some 
further notations and conventions.

This chapter is partly based on the sections introducing 
the corresponding concepts in \cite{schwartz19:AiG}, and 
also incorporates material from \cite{schwartz19:PNSE}.

\section{General conventions}

We use the `mostly plus' $(-{++}+)$ signature convention 
for the spacetime metric and stick, as indicated, to four 
dimensions. However, in many places our work has a 
straightforward generalisation to higher dimensions. 
The velocity of light will be denoted by $c$, and 
\emph{not} set equal to $1$.

When talking about Minkowski spacetime, we will view it 
as an affine space or, even more often, as an abstract 
differentiable manifold endowed with a Lorentzian metric, and 
not identify it with a vector space, unless otherwise 
stated.\enlargethispage*{2\baselineskip}

\section{Background structures}
\label{sec:background_structures}

As soon as gravity is geometrised in a metric sense, it 
does not make sense to speak of the `absence' of 
gravitational fields\footnote
	{This is not necessarily true in all geometric theories 
	of gravitation. For example, in teleparallel gravity 
	theories, inertial and gravitational effects can be 
	naturally separated \cite{pereira14}.},
and therefore also not of their `weakness' -- this can 
only be spoken of with respect to some background metric 
to compare the physical metric to. This background 
metric then defines the concept of `absence' of gravity.

In order to perform a Newtonian limit and to analyse the 
behaviour of physical systems and theories near this 
Newtonian limit -- that is, to perform a post-Newtonian 
expansion -- we also need some means of decomposing 
spacetime into `space' and `time'. The general idea is 
that such a decomposition of the background spacetime 
can be accomplished by considering a `time evolution' 
vector field, i.e.\ a vector field that is, with respect 
to the background metric, timelike, of constant Lorentzian 
length, and hypersurface orthogonal. We can then consider 
the integral curves of this vector field as `time', and 
the leaves of the orthogonal distribution as `space'.

Since we want the geodesic structure of the background 
spacetime, and its de\-composition into space and time, 
to be compatible with Newtonian concepts, we will take 
as the background spacetime four-dimensional Minkowski 
spacetime $(M,\eta)$ and as `time evolution' vector field 
a timelike geodesic vector field $u$ on $(M,\eta)$. Here 
$\eta$ denotes the Minkowski metric. For reasons of 
physical dimensionality, we assume $u$ to have Minkowski 
square $\eta(u,u) = -c^2$. We also fix, once and for all, 
an orientation and a time orientation on Minkowski 
spacetime, and assume $u$ to be future-directed. Sometimes, 
we will interpret $u$ as the four-velocity vector 
field of a family of inertial observers in background 
Minkowski spacetime.

That the gravitational field be weak now means that the 
physical spacetime metric on $M$, which we denote by $g$, 
deviate only little from the background Minkowski metric 
$\eta$. The notion of `deviating only little' 
will be made more precise in the following 
section. As described above, we now use $\eta$ and the 
preferred timelike vector field $u$ to decompose spacetime 
into time (integral curves of $u$) and space (hyperplanes 
$\eta$\babelhyphen{nobreak}orthogonal to $u$). We endow `space' 
with a flat Riemannian metric $\delta$, the restriction 
of $\eta$ to the hyperplanes, such that it just becomes 
ordinary flat Euclidean space. Interpreted as a tensor on 
four-dimensional spacetime (which annihilates the time 
direction $u$ and may therefore also be viewed as a purely 
`spatial' object), $\delta$ can be expressed in geometric, 
coordinate-free language as
\begin{equation}
	\delta := \eta + c^{-2} \, u^\flat \otimes u^\flat \; ,
\end{equation}
where $u^\flat := \eta(u, \cdot)$ denotes the one-form 
corresponding to $u$ via the metric. The time evolution 
vector field $u$ also allows us to define a notion 
of small / `slow' velocities~-- namely spatial velocities, 
as seen from an observer moving along $u$, being small 
compared to $c$.

We are free to use the `flat' structure of spacetime and 
space introduced by the background structures to perform 
all our computations. However, once results are 
established, we have to keep in mind that physical 
distances and times are measured with the physical metric 
$g$, not the auxiliary metric $\eta$. We will see that in 
some cases it is precisely such a re-interpretation in 
terms of the physical metric that lends the results good 
physical meaning.

For later use, we introduce the `physical spatial metric' 
$^{(3)}g$, which is the restriction of the physical 
spacetime metric $g$ to three-dimensional `space', i.e.\ to 
the orthogonal complement of the preferred vector field $u$. 
The inverse of this physical spatial metric will be denoted 
by $^{(3)}g^{-1}$.

Let us stress here that all the structures introduced 
and all the conditions of `\mbox{weakness}' and `slowness' 
mentioned are entirely independent of coordinates that 
we may choose. That is not to 
say that there may not be preferred coordinates which 
are particularly \mbox{adapted} to the given background 
structure. Indeed, such adapted coordinates \mbox{obviously} 
exist, namely positively oriented inertial coordinates 
$(x^0,x^1,x^2,x^3)$ in Minkowski spacetime $(M,\eta)$, 
with respect to some arbitrarily chosen origin, 
such that $x^0 = ct$, $u = \partial/\partial t$, and 
$\eta = \eta_{\mu\nu} \, \D x^\mu \otimes \D x^\nu$ with 
$(\eta_{\mu\nu}) = \mathrm{diag}(-1,1,1,1)$. Unless 
otherwise stated, we will always work in such coordinates 
adapted to the background structures when dealing with 
post-Newtonian expansions.

\section{Further geometric notation and conventions}
\label{sec:geometric_notation_conventions}

In our calculations, vectors and tensors will be 
represented by their components with respect to the 
chosen coordinate system $(x^\mu) = (ct, x^a)$. We let 
Greek indices run from $0$ to $3$ and Latin indices from 
$1$ to $3$, and we shall use the Einstein summation 
convention for like indices at different levels (one 
up- and one downstairs). Indices are lowered and raised 
by the physical spacetime metric $g_{\mu\nu}$ and its 
inverse $g^{\mu\nu}$, respectively. The Minkowski metric 
takes its usual diagonal form, as stated above. The 
spatial metric $\delta$ induced by the background 
structures has the usual Euclidean form with components 
$(\delta_{ab}) = \mathrm{diag}(1,1,1)$, and its inverse 
has components $(\delta^{ab}) = \mathrm{diag}(1,1,1)$.

We will often employ a `three-vector' notation, where the 
three-tuple of spatial components of some geometric object 
will be denoted by
\ifvectorsbold
	a boldface letter:
\else
	an arrow accent:
\fi
for example,
$\vect v = (v^1, v^2, v^3)$ is the `vector' of spatial 
components of some tangent vector $v$ on $M$, or 
$\vect A = (A_1, A_2, A_3)$ the `vector' of spatial 
components of some one-form $A$. When using this 
notation, a dot between two such `vectors' will denote 
the \emph{component-wise} `Euclidean scalar product', i.e.
\begin{equation}
	\vect v \cdot \vect w := \delta_{ab} v^a w^b = \sum_{a=1}^3 v^a w^a
\end{equation}
or
\begin{equation}
	\vect v \cdot \vect A := v^a A_a = \sum_{a=1}^3 v^a A_a \; .
\end{equation}
Note that the latter does not depend on $\delta$ in the 
formula, but nevertheless relies on the $3+1$ split 
induced by the background structures. Similarly, a cross 
multiplication symbol will denote the \emph{component-wise} 
vector product, i.e.
\begin{equation}
	(\vect v \times \vect w)^a := \delta^{an} \, {^{(3)}\varepsilon_{nbc}} v^b w^c
\end{equation}
where $^{(3)}\varepsilon_{abc}$ is the usual 
three-dimensional totally antisymmetric symbol. 
Geometrically, $^{(3)}\varepsilon_{abc}$ can be 
understood as the components of the spatial volume form 
induced by the Euclidean metric $\delta$. We will lower 
and raise the indices of $^{(3)}\varepsilon$ by 
$\delta_{ab}$ and $\delta^{ab}$ respectively, 
i.e.\ $^{(3)}\varepsilon^a_{\phantom{a}bc} := \delta^{an} \, {^{(3)}\varepsilon_{nbc}}$ 
etc., such that we can write 
$(\vect v \times \vect w)^a = {^{(3)}\varepsilon^a_{\phantom{a}bc}} v^b w^c$.

A boldface nabla symbol $\vect\nabla$ denotes the 
three-tuple of \emph{partial} derivatives
\begin{equation} \label{eq:Nabla}
	\vect \nabla = (\partial_1, \partial_2, \partial_3),
\end{equation}
which can be geometrically understood as the component 
representation of the spatial covariant derivatives 
with respect to the flat Euclidean metric. It will 
be used to express component-wise vector calculus 
operations in the usual short-hand notation, 
for example writing
\begin{equation}
	(\vect \nabla \times \vect A)^a = {^{(3)}\varepsilon^{abc}} \partial_b A_c
\end{equation}
for the component-wise curl of $\vect A$.

In view of the structures introduced, we stress again 
that all the operations reported here and used in the 
sequel make good geometric sense. They do depend on 
the geometric structures that we made explicit above, 
i.e. on the background metric $\eta$ and the time evolution 
vector field $u$, but they do not depend on the 
coordinates or frames that one uses in order to express 
the geometric objects (including the background 
structures) in terms of their real-valued components.
\enlargethispage{-\baselineskip}

\section{Formal expansions in \texorpdfstring{$c^{-1}$}{1/c}}
\label{sec:formal_expansions}

In order to perform a post-Newtonian expansion, we need 
some means to keep track of `how far away' from the 
Newtonian limit some term in a calculation is. A 
convenient way for doing so is to expand all relevant 
quantities as formal power series in $c^{-1}$, i.e.\ in 
the inverse of the velocity of light. A term of order 
$c^0$ then corresponds to the Newtonian limit of the 
considered quantity, and the higher-order terms give 
higher and higher orders of post-Newtonian `corrections'. 
Even though it might at first sight seem somewhat 
peculiar to perform an expansion in a dimensionful 
quantity, there is nothing to worry about when using this 
as a method to just formally keep track of post-Newtonian 
effects, since no questions of convergence ever arise. 
Note that the Newtonian limit of a quantity corresponds 
to formally taking the limit $c \to \infty$ in the power 
series. A quantity $X$ being of order (at least) $k$ in 
the formal $c^{-1}$-expansion will be denoted by
\begin{equation}
	X = \Or(c^{-k}).
\end{equation}
Let us again stress that this does not entail any 
analytic statement at all; it is just a \mbox{notation} for 
orders in formal power series. To put it differently, we 
view a post-Newtonian theory as a (formal) deformation 
of its `Newtonian limit', \mbox{implementing} the deformation of 
Galilei to Poincaré symmetry well-known at the level of 
Lie \mbox{algebras}~\cite{inonu53}.

Sometimes, we will need to consider quantities 
incorporating terms of \emph{negative} order in $c^{-1}$ 
(i.e.\ of positive order in $c$). However, we will always 
encounter but finitely many\footnote
	{In fact, only up to order $c^4$.}
negative-order terms, meaning that we are considering formal 
Laurent series in the expansion parameter $c^{-1}$. Note 
that no Newtonian limit exists for a quantity with 
non-vanishing such negative-order terms.
\enlargethispage{-\baselineskip}

In the Newtonian limit, coordinate time\footnote
	{Even though we call it `coordinate time' here, $t$ can of 
	course be characterised in a coordinate-free way as the 
	evolution parameter of integral curves of the background 
	time evolution vector field $u$.} 
$t$ shall be identified with Newtonian absolute time. 
Therefore, we have to treat $t$ as 
being of order $c^0$ in our formal expansion, instead 
of the timelike coordinate $x^0 = ct$ with dimension 
of length: were we to take $x^0$ to be of order $c^0$, 
then $t = c^{-1} x^0$ would vanish in the Newtonian 
limit. However, to the spatial coordinates $(x^a)$ we 
assign order $c^0$. This necessity of treating the time 
direction differently is, of course, well-known: it arises 
whenever one wants to obtain well-defined Newtonian limits 
of (locally) Poincaré-relativistic theories, for example in 
the context of Newton--Cartan theory \cite{ehlers81,ehlers19}.

Considering the background Minkowski metric
\begin{equation}
	\eta = \eta_{\mu\nu} \, \D x^\mu \otimes \D x^\nu = - c^2 \D t^2 + \D \vect x^2 ,
\end{equation}
we see that, due to our treating differently the time 
coordinate, it consists of terms of different order in 
$c^{-1}$: a temporal part of order $c^2$, and a spatial 
part of order $c^0$. This analogously goes for the 
inverse Minkowski metric
\begin{equation}
	\eta^{-1} = - c^{-2} \frac{\partial}{\partial t} \otimes \frac{\partial}{\partial t} + \delta^{ab} \frac{\partial}{\partial x^a} \otimes \frac{\partial}{\partial x^b} \; .
\end{equation}

We now turn to the description of the formal 
$c^{-1}$-expansion of the \emph{physical} spacetime 
metric $g$, which is to make precise the notion of $g$ 
deviating only little from the Minkowski background 
$\eta$. For the computations in chapter \ref{chap:PNSE}, 
it turns out that it is notationally easiest to label 
the coefficients in the expansion of the components of 
the \emph{inverse} metric, instead of the metric itself:
we expand the components of the inverse metric as formal 
power series
\begin{equation} \label{eq:exp_metric}
	g^{\mu\nu} = \eta^{\mu\nu} + \sum_{k=1}^\infty c^{-k} g^{\mu\nu}_{(k)} \; ,
\end{equation}
the lowest-order term being given by the components of 
the inverse Minkowski metric. Note that the coefficients 
in \eqref{eq:exp_metric} refer to the coordinates 
$(x^\mu) = (ct, x^a)$. Thus, when considering the inverse 
metric proper (and not its components), we obtain
\begin{align}
	g^{-1} &= g^{\mu\nu} \frac{\partial}{\partial x^\mu} \otimes \frac{\partial}{\partial x^\nu} \nonumber\\
	&= \eta^{-1} + \sum_{k=1}^\infty c^{-k} \left[ c^{-2} g^{00}_{(k)} \frac{\partial}{\partial t} \otimes \frac{\partial}{\partial t} 
		+ c^{-1} g^{0a}_{(k)} \frac{\partial}{\partial t} \vee \frac{\partial}{\partial x^a}
		+ g^{ab}_{(k)} \frac{\partial}{\partial x^a} \otimes \frac{\partial}{\partial x^b} \right]:
\end{align}
Coefficients carrying the same notational order label 
`$(k)$' appear in different orders of the formal expansion 
of the proper geometric object $g^{-1}$. For the sake of 
notational convenience, we also define
\begin{align} \label{eq:exp_metric_single_term}
	g^{-1}_{(k)} := g^{\mu\nu}_{(k)} \frac{\partial}{\partial x^\mu} \otimes \frac{\partial}{\partial x^\nu} \; ,
\end{align}
to which the same observation applies.

In chapter \ref{chap:atom_in_gravity}, we will 
discuss electromagnetic quantities. In that context, 
we will treat the electromagnetic four-potential form 
$A$ and the four-current density $j$ as being of formal 
expansion order $c^0$ when considered as tensor (density) 
fields. Their components with respect to our adapted 
coordinate system are then of the orders
\begin{equation}
	A_0 = \Or(c^{-1}), \, A_a = \Or(c^0), \, j^0 = \Or(c^1), \, j^a = \Or(c^0),
\end{equation}
factors of $c$ arising in them from $x^0 = ct$ involving 
a factor of $c$. This implies that the electric potential 
$\phi_\text{el.} = - c A_{0}$ and the charge density 
$\rho = \frac{1}{c} j^0$ are again quantities\footnote
	{In fact, they are -- apart from the conventional minus sign in 
	$\phi_\text{el.}$ -- simply the $t$ components of the 
	fields with respect to the coordinates $(t, x^a)$.}
of order $c^0$. In particular, for the (non-vanishing) 
components of the electromagnetic field tensor $F = \D A$, 
we have $F_{a0} = \Or(c^{-1})$ and $F_{ab} = \Or(c^0)$.

To ensure consistency in the treatment of expansion orders 
when dealing with electromagnetism, we will write 
equations in terms of the vacuum permittivity 
$\varepsilon_0$ only, to which we assign the formal order
$\varepsilon_0 = \Or(c^0)$, and avoid usage of the 
vacuum permeability $\mu_0 = 1/(\varepsilon_0 c^2)$ 
altogether.

\section{The Eddington--Robertson parametrised post-Newtonian metric}
\label{sec:ER_PPN_metric}

One of the easiest and most important physically 
relevant post-Newtonian metrics is the 
Eddington--Robertson parametrised post-Newtonian metric, 
whose components are given by
\begin{equation} \label{eq:ER_PPN_metric}
	(g_{\mu\nu}) = \begin{pmatrix}
		-1 - 2 \frac{\phi}{c^2} - 2 \beta \frac{\phi^2}{c^4} + \Or(c^{-6}) & \Or(c^{-5})\\
		\Or(c^{-5}) & (1 - 2 \gamma \frac{\phi}{c^2})\mathbb{1} + \Or(c^{-4})
	\end{pmatrix},
\end{equation}
where $\phi$ is a scalar function on spacetime that 
may be seen as the analogue of the Newtonian gravitational 
potential in this approximation scheme.

The metric also contains two dimensionless parameters 
$\beta$ and $\gamma$, the so-called `Eddington--Robertson 
parameters'. These account for possible deviations from 
\mbox{general} relativity, which corresponds to the 
values $\beta = \gamma = 1$. In that case, the metric 
\eqref{eq:ER_PPN_metric} solves the Einstein field 
equations of general relativity approximately in a 
$c^{-1}$-expansion for a static source, with $\phi$ being 
the Newtonian gravitational potential of the source. The 
metrics for different values of these parameters are then 
considered to correspond to so-called `test theories' 
against which the predictions of general relativity can 
be tested.\looseness-1

In fact, the Eddington--Robertson \acronym{PPN} metric 
(\acronym{PPN} = `parametrised post-Newtonian') is 
just the simplest of a much bigger family of 
\acronym{PPN} metrics, encompassing a large range of 
lowest-order post-Newtonian effects of metric theories 
of gravity and thus offering a large set of theories to 
test general relativity against. For an extensive 
discussion of the parametrised post-Newtonian formalism 
and its applications in tests of gravitational theory, we 
recommend the monograph \cite{will93}.

The explicit inclusion of $\beta$ and $\gamma$ 
allows us to track the consequences of post-Newtonian 
corrections in the spatial and the temporal part of the 
metric separately. It also opens the possibility to 
apply our results to potential future quantum tests of 
general relativity itself, which are, however, outside 
the scope of this thesis.

Note that even though in its true post-Newtonian origin 
the function $\phi$ appearing in the Eddington--Robertson 
\acronym{PPN} metric is time-independent, we will allow for 
it to depend on time for the sake of higher generality.

The components of the inverse metric to $g$ are easily 
obtained as
\begin{equation} \label{eq:ER_PPN_metric_inv}
	(g^{\mu\nu}) = \begin{pmatrix}
		-1 + 2 \frac{\phi}{c^2} + (2\beta - 4) \frac{\phi^2}{c^4} + \Or(c^{-6}) & \Or(c^{-5})\\
		\Or(c^{-5}) & (1 + 2 \gamma \frac{\phi}{c^2})\mathbb{1} + \Or(c^{-4})
	\end{pmatrix}.
\end{equation}


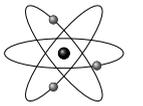
\newcommand{\KGe}{Klein--Gordon equation}

\chapter{Post-Newtonian corrections to Schrödinger equations in gravitational fields}
\label{chap:PNSE}

In this chapter, we deal with systematic methods to 
couple single, free quantum particles to post-Newtonian 
gravitational fields. More specifically, we extend a 
\mbox{\acronym{WKB}-like} \mbox{post-Newtonian} expansion of the 
minimally coupled \KGe\ after Kiefer and Singh 
\cite{kiefer91}, Lämmerzahl \cite{laemmerzahl95}, and 
Giulini and Großardt \cite{giulini12} to arbitrary order 
in $c^{-1}$, leading to Schrödinger equations describing 
a free quantum particle in a general gravitational field 
in post-Newtonian expansion. We will compare the results 
of this approach to canonical quantisation of a free 
particle in curved spacetime, following Wajima 
\emph{et al.}\ \cite{wajima97}.

Furthermore, using a more `formal', operator-algebraic 
approach, expansions of the \KGe\ and the canonical 
quantisation method are shown to lead to the same results 
for terms in the Hamiltonian up to linear order in 
particle momentum, when the particle is described with 
respect to a stationary time evolution vector field in a 
stationary spacetime. For this, no expansion in the 
inverse of the velocity of light has to be employed. This 
result means in particular that the lowest-order coupling 
to gravitomagnetism is described in the same way by both 
methods.

The material in this chapter has been published in 
\cite{schwartz19:PNSE}.

\section{Introduction}

In the existing literature, one finds two different main 
approaches to the problem of post-Newtonian `correction 
terms' for the Schrödinger equation describing a free 
quantum particle in a curved spacetime. The first, 
described, e.g., by Wajima \mbox{\emph{et al.}}\ \cite{wajima97}, 
starts from a classical description of the particle and 
applies canonical quantisation rules adapted to the 
situation (in a somewhat \emph{ad hoc} fashion) to \mbox{derive} 
a quantum-mechanical Hamiltonian. By an expansion in 
powers of $c^{-1}$ (at the stage of the classical 
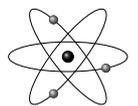
\mbox{Hamiltonian}), one finds the desired correction terms. 
Other \mbox{intimately} related \mbox{methods} use path integral 
quantisation on the classical system, as, e.g., the 
\mbox{semi-classical} \mbox{calculation} by Dimopoulos \emph{et al.}\ 
\cite{dimopoulos08}. As discussed 
in the \mbox{introduction}, such a semi-classical path integral 
perspective is the most widely used method for the 
description of gravitational coupling in quantum optics.

The second, fundamentally different approach takes a 
field-theoretic perspective and derives the Schrödinger 
equation as an equation for the positive frequency 
solutions of the minimally coupled classical \KGe. 
This is accomplished by Kiefer and Singh \cite{kiefer91}, 
Lämmerzahl \cite{laemmerzahl95}, and Giulini and Großardt 
\cite{giulini12} by making a \acronym{WKB}-like ansatz 
for the Klein--Gordon field, thereby formally expanding 
the \KGe\ in powers of $c^{-1}$, in the end viewing the 
Klein--Gordon theory as a formal deformation of the 
Schrödinger theory, as explained before in section~%
\ref{sec:formal_expansions}. This second method seems 
to be more firmly rooted in first principles than the 
canonical quantisation method, since it can at least 
heuristically be motivated from quantum field theory in 
curved spacetimes (see section~\ref{sec:KG_heur}). In a 
similar vein, one can apply such expansion methods to 
the Dirac equation, leading to a proper treatment of 
fermionic particles.

Although the two methods for obtaining post-Newtonian 
Schrödinger equations described above are very different in 
spirit, they lead to comparable results in lowest orders. 
To make possible a general comparison beyond the explicit 
examples considered in the existing literature\footnote
	{Wajima \emph{et al.}\ \cite{wajima97} considered a 
	first-order post-Newtonian metric for a point-like 
	rotating source, Lämmerzahl \cite{laemmerzahl95} used 
	the first-order Eddington--Robertson \acronym{PPN} 
	metric.},
we will apply the methods to as general a metric as 
possible. In section~\ref{sec:canon_quant}, we will 
give a brief overview over the canonical quantisation 
method (and extend it to the case of time-dependent 
metrics). After a heuristic quantum-field-theoretic 
motivation for considering the classical \KGe\ in the 
description of single quantum particles in section~%
\ref{sec:KG_heur}, section~\ref{sec:WKB} will develop 
the \acronym{WKB}-like formal expansion of the \KGe\ to 
arbitrary order in $c^{-1}$ in a general metric given as 
a formal power series in $c^{-1}$, significantly extending 
existing explicit examples to the general case. This leads 
to some simple comparisons of the resulting Hamiltonian 
with the one coming from canonical quantisation.

In section~\ref{sec:mom_exp}, we consider a formal 
expansion of the \KGe\ in powers of momentum operators 
leading to a Schrödinger form of the equation. This yields 
a general statement about agreement between the canonical 
and the Klein--Gordon methods for terms in the Hamiltonian 
up to linear order in momentum in the case of a stationary 
spacetime, without any necessity of an expansion in 
powers of $c^{-1}$.

A similar general \acronym{WKB}-like post-Newtonian formal 
expansion of the \KGe\ to obtain a Schrödinger equation 
was already considered by Tagirov in \cite{tagirov90} and 
a series of follow-up papers \cite{tagirov92,tagirov96}, 
as summarised in \cite{tagirov99}; but unlike our 
approach, these works did not expand the metric, thus 
not allowing to directly apply the results to metrics 
given as a power series in $c^{-1}$. Tagirov also compared 
his \acronym{WKB}-like approach to methods of canonical 
quantisation \cite{tagirov03}, but did this only for the 
case of static metrics.

Since we are concerned mostly with conceptual questions, 
we will generally not be mathematically very rigorous in 
this chapter, and in particular not mention domains of 
definition of operators.

\section{Canonical quantisation of a free particle}
\label{sec:canon_quant}

In the following, we will describe the canonical 
quantisation approach that was used by Wajima 
\emph{et al.}\ \cite{wajima97} to derive a Hamiltonian 
for a quantum particle in the post-Newtonian gravitational 
field of a point-like rotating source. We will allow 
metrics as general as possible, and focus on the conceptual 
issues of the procedure when adapted to our geometric 
framework from chapter \ref{chap:geometric_structures}. 
We will also extend the procedure such that we are able to 
define a quantum theory in the case of a time-dependent 
metric.

The classical action for a `relativistic' point particle 
of mass $m$ in curved spacetime with metric $g$ is
\begin{equation} \label{eq:class_action}
	S = -mc \int \D \lambda \, \sqrt{-g(x'(\lambda),x'(\lambda))} = -mc \int \D \lambda \, \sqrt{-g_{\mu\nu}x'^\mu x'^\nu} \; ,
\end{equation}
where $x(\lambda)$ is the arbitrarily parametrised 
worldline of the particle. Parametrising the worldline 
by coordinate time $t = x^0/c$, i.e.\ `background time' 
measured along the background time evolution vector field 
$u$ as introduced in section \ref{sec:background_structures}, 
the classical Hamiltonian for $t$-evolution can be computed 
to be
\begin{equation} \label{eq:class_Ham}
	H = \frac{1}{\sqrt{-g^{00}}} c \left[m^2 c^2 + \left(g^{ab} - \frac{1}{g^{00}} g^{0a} g^{0b}\right) p_a p_b\right]^{1/2} \kern-.8em+ \frac{c}{g^{00}} g^{0a} p_a
\end{equation}
when expressed in terms of the components of the (inverse) 
spacetime metric, where $p_a$ are the momenta conjugate 
to $x^a$. Full details of this calculation can be found 
in appendix \ref{app:details_class_Ham}.

Note that this Hamiltonian formalism makes use of the 
decomposition of spacetime into space and time as induced 
by the background structures. In the following, we will 
denote the spacelike leaf of `space' at background time 
$t$, as given by the background structures, by 
$\Sigma_t \subset M$. The `spaces' corresponding to 
different values of $t$ may naturally be identified along 
the flow of the background time evolution vector field, 
which in our adapted coordinates is just given by 
identifying points with the same spatial coordinates, i.e.\ 
$(c t_1, x^a) \mapsto (c t_2, x^a)$. The quotient space, 
which may be viewed as `abstract' Euclidean three-space 
proper, will be denoted by $\Sigma$, and there is a natural 
embedding $\Sigma \overset{\cong}{\to} \Sigma_t \subset M$ 
for each $t$. Of course, all of this depends on the 
background structures, and thus will the quantum theory 
we are about to construct\footnote
	{In fact, the constructions of this section can also be 
	applied in a slightly different setting. We could assume 
	the spacetime to be globally hyperbolic and perform a 
	$3+1$ decomposition \cite{giulini14}: we foliate spacetime 
	$M$ into three-dimensional spacelike Cauchy surfaces $\Sigma_t$ 
	which are images of an `abstract' Cauchy surface $\Sigma$ 
	under a family of embeddings $\mathcal E_t \colon \Sigma \to M$, 
	parametrised by a `foliation parameter' $t \in \mathbb R$, 
	and introduce spacetime coordinates such that $x^a$ are 
	coordinates on $\Sigma$ and $x^0 = ct$. In this setting, 
	the embeddings $\mathcal E_t$ defining the $3+1$ 
	decomposition would constitute 	the `background structure' 
	on which the quantum theory will depend.}.

Now, we want to `canonically quantise' the classical 
Hamiltonian \eqref{eq:class_Ham}. To this end, we expand 
the square root in \eqref{eq:class_Ham} to the desired 
order in $c^{-1}$ (or in momenta, see section 
\ref{sec:mom_exp}), and afterwards replace the classical 
momentum and position variables by corresponding operators, 
satisfying the canonical commutation relations. Of course, 
for doing so we have to choose an operator ordering scheme 
for symmetrising products of momenta and (functions of) 
position. We thus obtain a quantised Hamiltonian $\hat H$, 
acting on the Hilbert space on which the position and 
momentum operators are defined, and can postulate a 
Schrödinger equation in the usual form
\begin{equation}
	\I\hbar\partial_t \psi = \hat H \psi.
\end{equation}
Let us stress once more that this Hamiltonian will depend 
not only on the background structures $\eta, u$ which 
define the post-Newtonian approximation, but also on the 
choice of operator ordering scheme, which we leave open 
in order to keep the discussion as general as possible.
\enlargethispage{-\baselineskip}

Note that, according to the Stone--von Neumann theorem, 
the Hilbert space on which the position and momentum 
operators act and the form they take are \mbox{essentially} 
uniquely determined (up to unitary equivalence) by 
demanding the canonical commutation relations\footnote
	{Of course, as is well-known, the uniqueness statement is, 
	due to the unboundedness of the operators, only strictly 
	true when considering the `exponentiated' version of the 
	canonical commutation relations, i.e.\ the Weyl relations. 
	This essentially amounts to a regularity condition, which 
	we shall also implicitly assume on physical grounds.}.
Thus, the quantum theory is completely specified by the 
choice of ordering scheme, without any further choice 
concerning a possible explicit form of the Hilbert space. 
Nevertheless, we will now discuss explicit realisations 
of the Hilbert space and the position and momentum 
operators, in order to gain a more direct geometric 
interpretation thereof. This will also become important 
when \mbox{comparing} canonical quantisation to formal expansions 
of the \KGe\ in the following sections.

Since the position variables in the classical Hamiltonian 
\eqref{eq:class_Ham} are the spatial \mbox{coordinates} $x^a$ on 
three-dimensional `space', we want the quantum position 
\mbox{operators} to directly correspond to these. That is, we 
want to define the Hilbert space as some space of 
square-integrable `wavefunctions' of the $x^a$, such that 
we can take as position operators simply the operators of 
multiplication with the coordinates, thus obtaining a 
direct interpretation of the `wavefunctions' in the Hilbert 
space as `position probability amplitude distributions'. 
The question of explicit realisation of the Hilbert space 
thus becomes a question of choice of a scalar product on 
(some subspace of) the space of functions of the $x^a$.

To be more precise, we do not just need a single Hilbert 
space: to any time $t$ we want to associate a wavefunction 
$\psi(t)$ giving rise to a position probability 
distribution \emph{on the spatial leaf $\Sigma_t$ 
corresponding to $t$}, so we need to consider an individual 
Hilbert space for each spatial leaf. But since we want to 
relate these wavefunctions by a Schrödinger equation, we 
have to somehow identify the Hilbert spaces corresponding 
to different times.
\enlargethispage{-\baselineskip}

A natural, geometric choice of scalar product on the space 
of functions on $\Sigma_t$ is the $\mathrm{L}^2$-scalar 
product with respect to the induced metric measure (compare 
\cite{wajima97}), i.e.
\begin{equation} \label{eq:canon_scalar_prod_induced}
	\langle \psi, \varphi \rangle_{\Sigma_t} := \int \D^3\ivect x \, \overline\psi \varphi \sqrt{^{(3)}g|_{\Sigma_t}} \; ,
\end{equation}
where here and in the following, we use the short-hand 
notation ${^{(3)}g} = \det(g_{ab})$ for the determinant 
of the matrix of coordinate components of the spatial 
metric, when no confusion with the spatial metric $^{(3)}g$ 
proper can arise. Consider first the case that the spatial 
metric components $g_{ab}$ be independent of $t$, i.e.\ 
that the induced geometry be `the same' for all spatial 
leaves (implicitly identifying each $\Sigma_t$ with 
`abstract space' $\Sigma$ via the natural embedding 
$\Sigma \overset{\cong}{\to} \Sigma_t$). Then the scalar 
product \eqref{eq:canon_scalar_prod_induced} is independent 
of $t$, such that the Hilbert spaces corresponding to the 
different spatial slices are canonically identified by 
simply identifying the wavefunctions (again identifying 
$\Sigma_t \cong \Sigma$). We can then define the momentum 
operator as
\begin{equation} \label{eq:canon_mom_op_induced}
	\hat p_a := -\I\hbar \, {^{(3)}g^{-1/4}} \, \partial_a ({^{(3)}g^{1/4}}\cdot),
\end{equation}
which is symmetric with respect to the scalar product 
and fulfils the canonical commutation relation 
$[x^a, \hat p_b] = \I \hbar \delta^a_b$, 
and carry out canonical quantisation as described above.

If we allow for the $g_{ab}$ to depend on $t$, the scalar 
product \eqref{eq:canon_scalar_prod_induced} depends on 
$t$ and thus the canonical map 
$\mathrm{L}^2(\Sigma, \langle\cdot, \cdot\rangle_{\Sigma_t}) 
\ni \psi \mapsto \psi \in 
\mathrm{L}^2(\Sigma, \langle\cdot, \cdot\rangle_{\Sigma_s})$ 
no longer is an isomorphism of Hilbert spaces. I.e.\ the 
natural identification from above does not work, spoiling 
the program of canonical quantisation with this concrete 
realisation / geometric interpretation of the Hilbert 
spaces. A natural solution to this problem is to instead 
consider the time-independent `flat' $\mathrm L^2$-scalar 
product
\begin{equation} \label{eq:scalar_prod_flat}
	\langle \psi_\mathrm{f}, \varphi_\mathrm{f} \rangle_\mathrm{f} := \int \D^3\ivect x \, \overline{\psi_\mathrm{f}} \, \varphi_\mathrm{f}
\end{equation}
together with the `flat' momentum operator 
$\bar p_a := -\I\hbar \partial_a$. Using these, we obtain 
a `geometric realisation' of our canonical quantisation 
Hilbert space also in the case of time-dependent $g_{ab}$. 
At first sight, this scalar product could seem less 
`geometric' than \eqref{eq:canon_scalar_prod_induced}, 
but it can be seen to have as much invariant meaning as 
the latter by realising that, geometrically speaking, the 
`flat' wavefunctions $\psi_\mathrm{f}, \varphi_\mathrm{f}$ 
be scalar \emph{densities} (of weight $1/2$) on $\Sigma$ 
instead of scalar functions. Since this choice of `flat' 
scalar product can be applied to more general situations, 
and it eases the comparison to usual Galilei-invariant 
Schrödinger theory and to the Klein--Gordon expansion 
methods to be discussed in the following, we will adopt 
it from now on, i.e.\ `canonically quantise' the expanded 
classical Hamiltonian by replacing the classical momentum 
by the flat momentum operator (applying our chosen 
ordering scheme).

As explained above, the two choices of explicit realisation 
of the Hilbert space that we described for the case of 
time-independent $g_{ab}$ have to be unitarily equivalent 
by the Stone--von Neumann theorem. The unitary operator 
implementing this equivalence can be directly read off 
from the definitions of the two scalar products, and is 
given by $\psi \mapsto \psi_\mathrm{f} = {^{(3)}g^{1/4}} \, \psi$.

\section[Formal expansions of the \KGe: heuristic motivation from quantum field theory in stationary spacetimes]{Formal expansions of the \KGe: heuristic motivation from quantum field theory\newline in stationary spacetimes}
\sectionmark{Formal expansions of the \KGe: heuristic motivation from \acronym{QFTCS}}
\label{sec:KG_heur}

In the following, we will consider formal expansions of 
the classical, minimally coupled \KGe\ for a particle of 
mass $m > 0$,
\begin{equation} \label{eq:KG}
	\left(\Box - \frac{m^2 c^2}{\hbar^2}\right)\Psi_\mathrm{KG} = 0,
\end{equation}
leading to a Schrödinger equation with post-Newtonian 
corrections. In section \ref{sec:WKB}, we shall deal with 
a \acronym{WKB}-inspired formal expansion in $c^{-1}$, 
while in section \ref{sec:mom_exp}, we will draw a 
comparison to canonical quantisation based on an expansion 
in spatial momentum. To lay a conceptual foundation for 
these investigations, we will in this section give a 
heuristic motivation for consideration of the classical 
\KGe\ from quantum field theory in curved spacetimes.

Instead of \eqref{eq:KG} one could also consider the more 
general case of a possibly \emph{non}-minimally coupled 
\KGe, i.e.\ including some curvature term. This is 
customary in modern literature on quantum field theory 
in curved spacetime, where an additional term 
$-\xi R \Psi_\mathrm{KG}$ is included in the equation, 
$R$ being the scalar curvature of the spacetime 
\cite[eq. (5.57)]{baer09}. In particular, for the choice 
of $\xi = \frac{1}{6}$ (`conformal coupling'), the equation 
becomes conformally invariant in the massless case $m=0$, 
and also in the massive case there are some arguments 
favouring the conformally coupled \KGe, in particular in 
de Sitter spacetime \cite{tagirov73}. Nevertheless, we 
will for the sake of simplicity stick with the minimally 
coupled equation in this thesis, leaving non-minimal 
coupling for possible later investigations.

Now, we turn to the advertised motivation of consideration 
of the classical \KGe\ on a heuristic level. Namely, 
the quantum field theory construction for the free 
Klein--Gordon field on a globally hyperbolic stationary 
spacetime proceeds as follows (see, e.g., 
\cite[section 4.3]{wald94}).

We consider the \KGe\ \eqref{eq:KG} on a general globally 
hyperbolic stationary spacetime, and the Klein--Gordon 
inner product, which for two solutions 
$\Psi_\mathrm{KG}, \Phi_\mathrm{KG}$ of \eqref{eq:KG} 
is given by
\begin{align} \label{eq:KG_ip_general}
	\langle\Psi_\mathrm{KG},\Phi_\mathrm{KG}\rangle_\mathrm{KG}
		&= \I\hbar c \int_\Sigma \D^3\ivect x \, \sqrt{^{(3)}g} \, n^\nu \big[ \overline{\Psi_\mathrm{KG}} \big(\nabla_\nu \Phi_\mathrm{KG}\big) - \big(\nabla_\nu \overline{\Psi_\mathrm{KG}}\big) \Phi_\mathrm{KG} \big] \nonumber\\
	&= \I\hbar c \int_\Sigma \D^3\ivect x \, \sqrt{^{(3)}g} \, n^\nu \big[\overline{\Psi_\mathrm{KG}} \big(\partial_\nu \Phi_\mathrm{KG}\big) - \big(\partial_\nu \overline{\Psi_\mathrm{KG}}\big) \Phi_\mathrm{KG} \big] ,
\end{align}
where $\Sigma$ is a spacelike Cauchy surface, $^{(3)}g$ 
is the determinant of the induced metric on $\Sigma$, 
and $n$ is the future-directed unit normal vector field 
of $\Sigma$. In the second line, which is valid in a 
coordinate basis, we used that the covariant derivative 
of a scalar function is just the ordinary exterior 
derivative, i.e.\ given by a partial derivative in the 
case of a coordinate basis. Using the \KGe\ and Gauß' 
theorem, \eqref{eq:KG_ip_general} can be shown to be 
independent of the choice of $\Sigma$ under the assumption 
that the fields satisfy suitable boundary conditions.

The Hilbert space of the quantum field theory is now the 
bosonic Fock space over the `one-particle' Hilbert space 
constructed, loosely speaking, as the completion of the 
space of classical solutions of the \KGe\ with `positive 
frequency' (with respect to the stationarity Killing field) 
with the Klein--Gordon inner product.

To be more precise, the construction of the `one-particle' 
Hilbert space is a little more involved, since it is not 
\emph{a priori} clear what is meant by `positive frequency 
solutions': at first, the space of classical solutions of 
the \KGe\ is completed in a certain inner product to obtain 
an `intermediate' Hilbert space on which the generator of 
time translations (with respect to the stationarity Killing 
field) can be shown to be a self-adjoint operator; the 
positive spectral subspace of this operator is then 
completed in the Klein--Gordon inner product to give the 
Hilbert space of one-particle states. For details on the 
construction, see \cite[section 4.3]{wald94} and the 
references cited therein.

So the one-particle sector of the free Klein--Gordon 
quantum field theory in \mbox{globally} hyperbolic stationary 
spacetime is described by an appropriate notion of 
positive \mbox{frequency} solutions of the classical \KGe, 
using the Klein--Gordon inner product. Note that in this 
representation, in which the Klein--Gordon inner product 
takes its usual form, the `naive position operator' 
(multiplying with \mbox{coordinate} position) is the well-known 
Newton--Wigner position when we are considering Minkowski 
spacetime.

At this point, the quantum-field-theoretic motivation 
of our Klein--Gordon expansion methods becomes merely 
heuristic: since in the following we will not solve the 
\KGe\ exactly, but consider formal expansions of it 
(either in powers of $c^{-1}$ or in powers of spatial 
momentum), it will not be possible to exactly determine 
the space of positive frequency solutions according to 
the procedure described above; instead, we will merely 
choose an oscillating phase factor such as to guarantee 
the solution to have positive instead of negative 
frequency in lowest order in the expansion (see 
\eqref{eq:WKB_pos_freq}). 
If analysed more rigorously, it could turn out that for 
an asymptotic solution to be of positive frequency in 
some stricter sense, additional restrictions on the 
solution have to be made, possibly altering the function 
space under consideration. I.e.\ in principle, this could 
lead to the Hamiltonian we will obtain being altered 
when considering a rigorous analytic post-Newtonian 
expansion of quantum field theory in curved spacetime, 
instead of just a formal power series expansion.

In the non-stationary case, there is no canonical notion 
of particles and thus, strictly speaking, the whole 
question about the behaviour of single quantum particles 
does not make sense. Nevertheless, for an observer moving 
on an orbit which is approximately Killing, the classical 
Klein--Gordon theory can, on a heuristic level, still be 
expected to lead to approximately correct predictions 
regarding this observer's observations.

Even if this motivation is just a heuristic, the 
\acronym{WKB}-like approach of expanding the \KGe\ in 
powers of $c^{-1}$ will allow us to view the classical 
Klein--Gordon theory as a formal deformation of the 
`non-relativistic' Schrödinger theory, and makes the 
sense in which that happens formally precise, the same 
happening for the momentum expansion.

\section{\acronym{WKB}-like expansion of the \KGe}
\label{sec:WKB}

Now, we will consider \acronym{WKB}-like formal expansions 
in $c^{-1}$ of the \KGe\ \eqref{eq:KG}, as first introduced 
by Kiefer and Singh in \cite{kiefer91} for Minkowski 
spacetime, and later considered by 
Lämmerzahl in \cite{laemmerzahl95} for the simple 
Eddington--Robertson \acronym{PPN} metric, and by Giulini 
and Großardt in \cite{giulini12} for general spherically 
symmetric metrics.

After developing the expansion of the \KGe\ to arbitrary 
order in $c^{-1}$, we will explain the transformation to 
a `flat' $\mathrm L^2$-scalar product for comparison to 
canonical quantisation, and finally consider the metric 
of the Eddington--Robertson \acronym{PPN} test theory as 
a simple explicit example.

\subsection{General derivation}

We assume the post-Newtonian physical spacetime metric 
to be given by a formal power series in $c^{-1}$ as in 
\eqref{eq:exp_metric}. Let us remind ourselves that 
we will work in our coordinate system that is adapted 
to the background structures defining the notion of 
post-Newtonian expansion.

In coordinates, the d'Alembert operator in a general 
Lorentzian metric, as acting on scalar functions, is given by
\begin{align} \label{eq:box_op}
	\Box f &= \nabla^\mu \nabla_\mu f \nonumber\\
	&= \frac{1}{\sqrt{-g}} 
		\partial_\mu(\sqrt{-g} \, g^{\mu\nu} \partial_\nu f) \nonumber\\
	&= \frac{1}{\sqrt{-g}} (\partial_\mu\sqrt{-g}) g^{\mu\nu} \partial_\nu f 
		+ \partial_\mu(g^{\mu\nu}) \partial_\nu f 
		+ g^{\mu\nu} \partial_\mu\partial_\nu f,
\end{align}
where we use the short-hand notation $g = \det(g_{\mu\nu})$ 
for the determinant of the matrix of coordinate components 
of the metric, when no confusion with the metric proper 
can arise. The second and third term in this expression 
can easily be expanded in $c^{-1}$ by inserting the 
expansion \eqref{eq:exp_metric} of the components of the 
inverse metric and using $x^0 = ct$: the third term is
\begin{align} \label{eq:box_expanded_term3}
	g^{\mu\nu} \partial_\mu \partial_\nu 
	&= -c^{-2} \partial_t^2 + \Delta 
		+ \sum_{k=1}^\infty c^{-k} g^{00}_{(k)} c^{-2} \partial_t^2 
		+ \sum_{k=1}^\infty c^{-k} 2 g^{0a}_{(k)} c^{-1} \partial_t \partial_a 
		+ \sum_{k=1}^\infty c^{-k} g^{ab}_{(k)} \partial_a \partial_b \nonumber\\
	&= -c^{-2} \partial_t^2 + \Delta 
		+ \sum_{k=3}^\infty c^{-k} g^{00}_{(k-2)} \partial_t^2 
		+ \sum_{k=2}^\infty c^{-k} 2 g^{0a}_{(k-1)} \partial_t \partial_a 
		+ \sum_{k=1}^\infty c^{-k} g^{ab}_{(k)} \partial_a \partial_b \; ,
\end{align}
where $\Delta = \delta^{ab} \partial_a \partial_b$ denotes 
the `flat' Euclidean Laplacian on three-dimensional space, 
as induced by the background structures. Similarly, the 
second term evaluates to
\begin{align} \label{eq:box_expanded_term2}
	(\partial_\mu g^{\mu\nu})\partial_\nu 
	&= \sum_{k=3}^\infty c^{-k} (\partial_t g^{00}_{(k-2)}) \partial_t 
		+ \sum_{k=2}^\infty c^{-k} (\partial_t g^{0a}_{(k-1)}) \partial_a \nonumber\\
		&\quad + \sum_{k=2}^\infty c^{-k} (\partial_a g^{0a}_{(k-1)}) \partial_t 
		+ \sum_{k=1}^\infty c^{-k} (\partial_a g^{ab}_{(k)}) \partial_b \; .
\end{align}

Since the remaining first term of \eqref{eq:box_op} 
involves the expression
\begin{equation} \label{eq:box_expanded_term1_part}
	\frac{1}{\sqrt{-g}} \partial_\mu\sqrt{-g} 
	= \frac{1}{2g} \partial_\mu g 
	= \frac{1}{2} g^{\rho\sigma} \partial_\mu g_{\rho\sigma} 
	= -\frac{1}{2} g_{\rho\sigma} \partial_\mu g^{\rho\sigma} ,
\end{equation}
we need an expression for the $c^{-1}$-expansion of the 
components of the \emph{metric}, not just the inverse 
metric. Rewriting the expansion of the inverse metric as
\begin{equation}
	g^{\mu\nu} = \left[ \left(\mathbb{1} + \sum_{k=1}^\infty c^{-k} g^{-1}_{(k)} \eta\right) \eta^{-1} \right]^{\mu\nu} ,
\end{equation}
where we used the objects $g^{-1}_{(k)}$ introduced in 
\eqref{eq:exp_metric_single_term}, we see that a formal 
Neumann series can be used to invert the power series. 
This gives the coefficients of the metric as
\begin{equation}
	g_{\mu\nu} = \left\{ \eta \left[\mathbb{1} + \sum_{n=1}^\infty \left(-\sum_{k=1}^\infty c^{-k} g^{-1}_{(k)} \eta\right)^n \right] \right\}_{\mu\nu} .
\end{equation}
Iterating the Cauchy product formula, we have
\enlargethispage{-2\baselineskip}
\begin{equation}
	\left(-\sum_{k=1}^\infty c^{-k} g^{-1}_{(k)} \eta\right)^n 
	= (-1)^n \sum_{k=1}^\infty c^{-k} \sum_{\substack{i_1+\cdots+i_n = k \\ 1 \le i_1,\ldots,i_n \le k}} g^{-1}_{(i_1)}\eta \cdots g^{-1}_{(i_n)} \eta .
\end{equation}
Using this and introducing the notation
\begin{equation}
	g^{-1}_{(k,n)} := \sum_{\substack{i_1+\cdots+i_n = k \\ 1 \le i_1,\ldots,i_n \le k}} g^{-1}_{(i_1)}\eta g^{-1}_{(i_2)}\eta \cdots g^{-1}_{(i_n)} \; ,
\end{equation}
we can write the metric as
\begin{equation}
	g_{\mu\nu} = \eta_{\mu\nu} + \sum_{k=1}^\infty c^{-k} \sum_{n=1}^\infty (-1)^n (\eta g^{-1}_{(k,n)} \eta)_{\mu\nu} \; .
\end{equation}
Thus, returning to \eqref{eq:box_expanded_term1_part} 
we obtain, using the Cauchy product formula again,
\begin{align}
	g_{\rho\sigma}\partial_\mu g^{\rho\sigma} 
		&= \left(\eta_{\rho\sigma} + \sum_{k=1}^\infty c^{-k} \sum_{n=1}^\infty (-1)^n (\eta g^{-1}_{(k,n)} \eta)_{\rho\sigma}\right) 
		\sum_{m=1}^\infty c^{-m} \partial_\mu g^{\rho\sigma}_{(m)} \nonumber\\
	&= \sum_{k=1}^\infty c^{-k} \partial_\mu \tr(\eta g^{-1}_{(k)}) 
		+ \sum_{k=2}^\infty c^{-k} \sum_{l + m = k} \; \sum_{n=1}^\infty (-1)^n (\eta g^{-1}_{(l,n)} \eta)_{\rho\sigma} \, \partial_\mu g^{\rho\sigma}_{(m)} \; ,
\end{align}
where in the sum $\sum_{l+m = k}$\,, the summation variables 
$l$ and $m$ take values $\ge 1$, which we notationally 
suppress here and in the following. Using
\begin{align}
	\sum_{l + m = k} (\eta g^{-1}_{(l,n)} \eta)_{\rho\sigma} \, \partial_\mu g^{\rho\sigma}_{(m)} 
		&= \sum_{l + m = k} \; \sum_{i_1+\cdots+i_n = l} (\eta g^{-1}_{(i_1)} \cdots g^{-1}_{(i_n)} \eta)_{\rho\sigma} \, \partial_\mu g^{\rho\sigma}_{(m)} \nonumber\\
	&= \sum_{i_1+\cdots+i_n + m = k} \tr(\eta g^{-1}_{(i_1)} \cdots g^{-1}_{(i_n)} \eta \, \partial_\mu g^{-1}_{(m)}) \nonumber\\
	&= \frac{1}{n+1} \, \partial_\mu \sum_{i_1+\cdots+i_n + m = k} \tr(\eta g^{-1}_{(i_1)} \cdots g^{-1}_{(i_n)} \eta g^{-1}_{(m)}) \nonumber\\
	&= \frac{1}{n+1} \, \partial_\mu \tr(\eta g^{-1}_{(k,n+1)}) 
\end{align}
and the facts that $g^{-1}_{(k,n)} = 0$ for $n > k$ and 
$g^{-1}_{(k,1)} = g^{-1}_{(k)}$, we can rewrite this as
\begin{align}
	g_{\rho\sigma}\partial_\mu g^{\rho\sigma} 
		&= \sum_{k=1}^\infty c^{-k} \partial_\mu \tr(\eta g^{-1}_{(k)}) 
		+ \sum_{k=2}^\infty c^{-k} \sum_{n=2}^\infty (-1)^{n-1} \frac{1}{n} \, \partial_\mu \tr(\eta g^{-1}_{(k,n)}) \nonumber\\
	&= \sum_{k=1}^\infty c^{-k} \partial_\mu \tr(\eta g^{-1}_{(k)}) 
		+ \sum_{k=1}^\infty c^{-k} \sum_{n=2}^\infty (-1)^{n-1} \frac{1}{n} \, \partial_\mu \tr(\eta g^{-1}_{(k,n)}) \nonumber\\
	&= \sum_{k=1}^\infty c^{-k} \sum_{n=1}^\infty (-1)^{n-1} \frac{1}{n} \, \partial_\mu \tr(\eta g^{-1}_{(k,n)}) .
\end{align}
Thus, we finally obtain the expansion
\enlargethispage{-2\baselineskip}
\begin{align} \label{eq:box_expanded_term1}
	&\frac{1}{\sqrt{-g}} (\partial_\mu\sqrt{-g}) g^{\mu\nu} \partial_\nu f \nonumber\\
	&\qquad= -\frac{1}{2} (g_{\rho\sigma} \, \partial_\mu g^{\rho\sigma}) g^{\mu\nu} \partial_\nu f \nonumber\\
	&\qquad= \frac{1}{2} \sum_{k=1}^\infty c^{-k} \sum_{n=1}^\infty (-1)^n \frac{1}{n} [\partial_\mu \tr(\eta g^{-1}_{(k,n)})] \left(\eta^{\mu\nu} + \sum_{m=1}^\infty c^{-m} g^{\mu\nu}_{(m)} \right) \partial_\nu f \nonumber\\
	&\qquad= \frac{1}{2} \sum_{k=1}^\infty c^{-k} \sum_{n=1}^\infty (-1)^n \frac{1}{n} [\partial_\mu \tr(\eta g^{-1}_{(k,n)})] \, \eta^{\mu\nu} \partial_\nu f \nonumber\\
		&\qquad\quad + \frac{1}{2} \sum_{k=2}^\infty c^{-k} \sum_{l+m = k} \; \sum_{n=1}^\infty (-1)^n \frac{1}{n} [\partial_\mu \tr(\eta g^{-1}_{(l,n)})] \, g^{\mu\nu}_{(m)} \partial_\nu f
\end{align}
for the first term in the d'Alembert operator \eqref{eq:box_op}.

Inserting \eqref{eq:box_expanded_term3}, 
\eqref{eq:box_expanded_term2}, and 
\eqref{eq:box_expanded_term1} into \eqref{eq:box_op} and 
sorting the sums by order of $c^{-1}$, the full expansion 
of the d'Alembert operator reads
\begin{align} \label{eq:box_expanded}
	\Box f &= \frac{1}{2} \sum_{k=4}^\infty c^{-k} \sum_{l+m = k-2} \; \sum_{n=1}^\infty (-1)^n \frac{1}{n} g^{00}_{(m)} [\partial_t \tr(\eta g^{-1}_{(l,n)})] \, \partial_t f \nonumber\\
	&\quad + \frac{1}{2} \sum_{k=3}^\infty c^{-k} \sum_{l+m = k-1} \; \sum_{n=1}^\infty (-1)^n \frac{1}{n} g^{0a}_{(m)} \left([\partial_t \tr(\eta g^{-1}_{(l,n)})] \, \partial_a f + [\partial_a \tr(\eta g^{-1}_{(l,n)})] \, \partial_t f\right) \nonumber\\
	&\quad - \frac{1}{2} \sum_{k=3}^\infty c^{-k} \sum_{n=1}^\infty (-1)^n \frac{1}{n} [\partial_t \tr(\eta g^{-1}_{(k-2,n)})] \, \partial_t f \nonumber\\
	&\quad + \sum_{k=3}^\infty c^{-k} (\partial_t g^{00}_{(k-2)}) \, \partial_t f 
		+ \sum_{k=3}^\infty c^{-k} g^{00}_{(k-2)} \partial_t^2 f \nonumber\\
	&\quad + \frac{1}{2} \sum_{k=2}^\infty c^{-k} \sum_{l+m = k} \; \sum_{n=1}^\infty (-1)^n \frac{1}{n} g^{ab}_{(m)} [\partial_a \tr(\eta g^{-1}_{(l,n)})] \, \partial_b f \nonumber\\
	&\quad + \sum_{k=2}^\infty c^{-k} \left((\partial_t g^{0a}_{(k-1)}) \, \partial_a f + (\partial_a g^{0a}_{(k-1)}) \, \partial_t f\right) 
		+ \sum_{k=2}^\infty c^{-k} \, 2 g^{0a}_{(k-1)} \partial_t\partial_a f 
		- c^{-2}\partial_t^2 f \nonumber\\
	&\quad + \frac{1}{2} \sum_{k=1}^\infty c^{-k} \sum_{n=1}^\infty (-1)^n \frac{1}{n} [\partial_a \tr(\eta g^{-1}_{(k,n)})] \delta^{ab} \partial_b f \nonumber\\
	&\quad + \sum_{k=1}^\infty c^{-k} (\partial_a g^{ab}_{(k)}) \partial_b f 
		+ \sum_{k=1}^\infty c^{-k} g^{ab}_{(k)} \partial_a \partial_b f 
		+ \Delta f .
\end{align}
\enlargethispage{-\baselineskip}

Now, we make the \acronym{WKB}-like ansatz
\begin{equation} \label{eq:WKB_ansatz}
	\Psi_\mathrm{KG} = \exp\left(\frac{\I c^2}{\hbar}S\right) \psi, \; \psi = \sum_{k=0}^\infty c^{-k} a_k
\end{equation}
for the Klein--Gordon field (compare \cite{giulini12}), 
where $S$ is a \emph{real} function; i.e.\ we separate 
off a phase factor and expand the remainder as a power 
series in $c^{-1}$. All the functions $S, a_k$ are 
assumed to be independent of the expansion parameter 
$c^{-1}$. The derivatives of the field are
\begin{equation}
	\partial_\mu \Psi_\mathrm{KG} = \frac{\I c^2}{\hbar} (\partial_\mu S) \Psi_\mathrm{KG} + \exp(\ldots) \partial_\mu\psi
\end{equation}
and
\begin{align}
	\partial_\mu\partial_\nu \Psi_\mathrm{KG} &= \exp\left(\frac{\I c^2}{\hbar}S\right) \bigg({- \frac{c^4}{\hbar^2}} (\partial_\mu S)(\partial_\nu S) \psi + \frac{\I c^2}{\hbar} \big[(\partial_\mu\partial_\nu S) \psi \nonumber\\
		&\qquad+ (\partial_\mu S) \partial_\nu\psi + (\partial_\nu S) \partial_\mu\psi\big] + \partial_\mu\partial_\nu \psi\bigg).
\end{align}

Using these and the expansion \eqref{eq:box_expanded} 
of the d'Alembert operator, we can now analyse the \KGe\ 
\eqref{eq:KG} order by order in $c^{-1}$. At the lowest 
occurring order $c^4$, we get
\begin{equation}
	-\exp\left(\frac{\I c^2}{\hbar}S\right) \frac{1}{\hbar^2} \delta^{ab} (\partial_a S) (\partial_b S) a_0 = 0,
\end{equation}
which is equivalent\footnote
	{For nontrivial solutions, i.e.\ $a_0 \ne 0$.}
to $\partial_a S = 0$. So $S$ is a function of (coordinate) 
time only. Using this, the \KGe\ has no term of order $c^3$.

At $c^2$, we get
\begin{equation}
	\exp\left(\frac{\I c^2}{\hbar}S\right) \left(\frac{1}{\hbar^2} (\partial_t S)^2 - \frac{m^2}{\hbar^2}\right) a_0 = 0,
\end{equation}
equivalent to $\partial_t S = \pm m$. Since we are 
interested in positive-frequency solutions of the \KGe, 
we choose $\partial_t S = -m$, leading to
\begin{equation} \label{eq:WKB_pos_freq}
	S = -mt
\end{equation}
(an additional constant term would lead to an irrelevant 
global phase).
\enlargethispage{-\baselineskip}

The $c^1$ coefficient leads to the equation
\begin{equation}
	-\exp\left(\frac{\I c^2}{\hbar}S\right) g^{00}_{(1)} \frac{m^2}{\hbar^2} a_0 = 0,
\end{equation}
equivalent to
\begin{equation} \label{eq:WKB_g00_lowest_order_vanishes}
	g^{00}_{(1)} = 0.
\end{equation}
Thus the requirement that the \KGe\ have solutions which 
are formal power series of the form \eqref{eq:WKB_ansatz} 
imposes restrictions on the components of the metric. 
In the following, we will freely use the vanishing of 
$g^{00}_{(1)}$.

\pagebreak
Using \eqref{eq:WKB_pos_freq} and 
\eqref{eq:WKB_g00_lowest_order_vanishes}, the positive 
frequency \KGe\ for our \acronym{WKB}-like solutions 
is equivalent to the following equation for $\psi$:
\begin{align} \label{eq:WKB_KG_psi}
	0 &= \sum_{k=5}^\infty c^{-k} \frac{1}{2} \sum_{l+m = k-2} \; \sum_{n=1}^\infty (-1)^n \frac{1}{n} g^{00}_{(m)} [\partial_t \tr(\eta g^{-1}_{(l,n)})] \, \partial_t \psi \nonumber\\
	&\quad + \sum_{k=4}^\infty c^{-k} (\partial_t g^{00}_{(k-2)}) \, \partial_t \psi 
		+ \sum_{k=4}^\infty c^{-k} g^{00}_{(k-2)} \partial_t^2 \psi \nonumber\\
	&\quad - \sum_{k=3}^\infty c^{-k} \frac{\I m}{2\hbar} \sum_{l+m = k} \; \sum_{n=1}^\infty (-1)^n \frac{1}{n} g^{00}_{(m)} [\partial_t \tr(\eta g^{-1}_{(l,n)})] \psi \nonumber\\
	&\quad + \sum_{k=3}^\infty c^{-k} \frac{1}{2} \sum_{l+m = k-1} \; \sum_{n=1}^\infty (-1)^n \frac{1}{n} g^{0a}_{(m)} \left( [\partial_t \tr(\eta g^{-1}_{(l,n)})] \, \partial_a \psi + [\partial_a \tr(\eta g^{-1}_{(l,n)})] \, \partial_t \psi\right) \nonumber\\
	&\quad - \sum_{k=3}^\infty c^{-k} \frac{1}{2} \sum_{n=1}^\infty (-1)^n \frac{1}{n} [\partial_t \tr(\eta g^{-1}_{(k-2,n)})] \, \partial_t \psi \nonumber\\
	&\quad - \sum_{k=2}^\infty c^{-k} \frac{\I m}{\hbar} (\partial_t g^{00}_{(k)}) \psi 
		- \sum_{k=2}^\infty c^{-k} \frac{2\I m}{\hbar} g^{00}_{(k)} \partial_t \psi \nonumber\\
	&\quad + \sum_{k=2}^\infty c^{-k} \frac{1}{2} \sum_{l+m = k} \; \sum_{n=1}^\infty (-1)^n \frac{1}{n} g^{ab}_{(m)} [\partial_a \tr(\eta g^{-1}_{(l,n)})] \, \partial_b \psi \nonumber\\
	&\quad + \sum_{k=2}^\infty c^{-k} \left((\partial_t g^{0a}_{(k-1)}) \, \partial_a \psi + (\partial_a g^{0a}_{(k-1)}) \, \partial_t \psi\right) 
		+ \sum_{k=2}^\infty c^{-k} \, 2 g^{0a}_{(k-1)} \partial_t \partial_a \psi - c^{-2} \partial_t^2 \psi \nonumber\\
	&\quad - \sum_{k=1}^\infty c^{-k} \frac{\I m}{2\hbar} \sum_{l+m = k+1} \; \sum_{n=1}^\infty (-1)^n \frac{1}{n} g^{0a}_{(m)} [\partial_a \tr(\eta g^{-1}_{(l,n)})] \psi \nonumber\\
	&\quad + \sum_{k=1}^\infty c^{-k} \frac{\I m}{2\hbar} \sum_{n=1}^\infty (-1)^n \frac{1}{n} [\partial_t \tr(\eta g^{-1}_{(k,n)})] \psi \nonumber\\
	&\quad + \sum_{k=1}^\infty c^{-k} \frac{1}{2} \sum_{n=1}^\infty (-1)^n \frac{1}{n} [\partial_a \tr(\eta g^{-1}_{(k,n)})] \delta^{ab} \partial_b \psi 
		+ \sum_{k=1}^\infty c^{-k} (\partial_a g^{ab}_{(k)}) \partial_b \psi \nonumber\\
	&\quad + \sum_{k=1}^\infty c^{-k} g^{ab}_{(k)} \partial_a \partial_b \psi \nonumber\\
	&\quad - \sum_{k=0}^\infty c^{-k} \frac{m^2}{\hbar^2} g^{00}_{(k+2)} \psi 
		- \sum_{k=0}^\infty c^{-k} \frac{\I m}{\hbar} (\partial_a g^{0a}_{(k+1)}) \psi 
		- \sum_{k=0}^\infty c^{-k} \frac{2\I m}{\hbar} g^{0a}_{(k+1)} \partial_a \psi \nonumber\\
	&\quad + \frac{2\I m}{\hbar} \partial_t \psi 
		+ \Delta \psi
\end{align}
\vfill\pagebreak

Inserting the expansion $\psi = \sum_{k=0}^\infty c^{-k} a_k$ 
and using the Cauchy product formula, this is equivalent to
\begin{align} \label{eq:WKB_KG}
	0 &= \sum_{k=5}^\infty c^{-k} \frac{1}{2} \sum_{l+m + \tilde k = k-2} \; \sum_{n=1}^\infty (-1)^n \frac{1}{n} g^{00}_{(m)} [\partial_t \tr(\eta g^{-1}_{(l,n)})] \, \partial_t a_{\tilde k} \nonumber\\
	&\quad + \sum_{k=4}^\infty c^{-k} \sum_{l + \tilde k = k-2} (\partial_t g^{00}_{(l)}) \, \partial_t a_{\tilde k} 
		+ \sum_{k=4}^\infty c^{-k} \sum_{l + \tilde k = k-2} g^{00}_{(l)} \partial_t^2 a_{\tilde k} \nonumber\\
	&\quad - \sum_{k=3}^\infty c^{-k} \frac{\I m}{2\hbar} \sum_{l+m + \tilde k = k} \; \sum_{n=1}^\infty (-1)^n \frac{1}{n} g^{00}_{(m)} [\partial_t \tr(\eta g^{-1}_{(l,n)})] a_{\tilde k} \nonumber\\
	&\quad + \sum_{k=3}^\infty c^{-k} \frac{1}{2} \sum_{l+m + \tilde k = k-1} \; \sum_{n=1}^\infty (-1)^n \frac{1}{n} g^{0a}_{(m)} \left([\partial_t \tr(\eta g^{-1}_{(l,n)})] \, \partial_a a_{\tilde k} + [\partial_a \tr(\eta g^{-1}_{(l,n)})] \, \partial_t a_{\tilde k}\right) \nonumber\\
	&\quad - \sum_{k=3}^\infty c^{-k} \frac{1}{2} \sum_{l + \tilde k = k-2} \; \sum_{n=1}^\infty (-1)^n \frac{1}{n} [\partial_t \tr(\eta g^{-1}_{(l,n)})] \, \partial_t a_{\tilde k} \nonumber\\
	&\quad - \sum_{k=2}^\infty c^{-k} \frac{\I m}{\hbar} \sum_{l + \tilde k = k} (\partial_t g^{00}_{(l)}) a_{\tilde k} 
		- \sum_{k=2}^\infty c^{-k} \frac{2\I m}{\hbar} \sum_{l + \tilde k = k} g^{00}_{(l)} \partial_t a_{\tilde k} \nonumber\\
	&\quad + \sum_{k=2}^\infty c^{-k} \frac{1}{2} \sum_{l+m + \tilde k = k} \; \sum_{n=1}^\infty (-1)^n \frac{1}{n} g^{ab}_{(m)} [\partial_a \tr(\eta g^{-1}_{(l,n)})] \, \partial_b a_{\tilde k} \nonumber\\
	&\quad + \sum_{k=2}^\infty c^{-k} \sum_{l + \tilde k = k-1} \left((\partial_t g^{0a}_{(l)}) \, \partial_a a_{\tilde k} + (\partial_a g^{0a}_{(l)}) \, \partial_t a_{\tilde k}\right) 
		+ \sum_{k=2}^\infty c^{-k} \, 2 \sum_{l + \tilde k = k-1} g^{0a}_{(l)} \partial_t \partial_a a_{\tilde k} \nonumber\\
	&\quad - \sum_{k=2}^\infty c^{-k} \partial_t^2 a_{k-2} 
		- \sum_{k=1}^\infty c^{-k} \frac{\I m}{2\hbar} \sum_{l+m + \tilde k = k+1} \; \sum_{n=1}^\infty (-1)^n \frac{1}{n} g^{0a}_{(m)} [\partial_a \tr(\eta g^{-1}_{(l,n)})] a_{\tilde k} \nonumber\\
	&\quad + \sum_{k=1}^\infty c^{-k} \frac{\I m}{2\hbar} \sum_{l + \tilde k = k} \; \sum_{n=1}^\infty (-1)^n \frac{1}{n} [\partial_t \tr(\eta g^{-1}_{(l,n)})] a_{\tilde k} \nonumber\\
	&\quad + \sum_{k=1}^\infty c^{-k} \frac{1}{2} \sum_{l + \tilde k = k} \; \sum_{n=1}^\infty (-1)^n \frac{1}{n} [\partial_a \tr(\eta g^{-1}_{(l,n)})] \delta^{ab} \partial_b a_{\tilde k} 
		+ \sum_{k=1}^\infty c^{-k} \sum_{l + \tilde k = k} (\partial_a g^{ab}_{(l)}) \partial_b a_{\tilde k} \nonumber\\
	&\quad + \sum_{k=1}^\infty c^{-k} \sum_{l + \tilde k = k} g^{ab}_{(l)} \partial_a \partial_b a_{\tilde k} 
		- \sum_{k=0}^\infty c^{-k} \frac{m^2}{\hbar^2} \sum_{l + \tilde k = k+2} g^{00}_{(l)} a_{\tilde k} \nonumber\\
	&\quad - \sum_{k=0}^\infty c^{-k} \frac{\I m}{\hbar} \sum_{l + \tilde k = k+1} (\partial_a g^{0a}_{(l)}) a_{\tilde k} 
		- \sum_{k=0}^\infty c^{-k} \frac{2\I m}{\hbar} \sum_{l + \tilde k = k+1} g^{0a}_{(l)} \partial_a a_{\tilde k} \nonumber\\
	&\quad + \sum_{k=0}^\infty c^{-k} \frac{2\I m}{\hbar} \partial_t a_k 
		+ \sum_{k=0}^\infty c^{-k} \Delta a_k \; ,
\end{align}
where in sums like $\sum_{l+m+\tilde k = k}$\,, $l$ and $m$ 
are $\ge 1$ as before, but $\tilde k$ is $\ge 0$.

Using the fully expanded \eqref{eq:WKB_KG}, we can obtain 
equations for the $a_k$, order by order, which can then 
be combined into a Schrödinger equation for $\psi$: at 
order $c^0$, we have
\begin{equation}
	0 = \left(-\frac{m^2}{\hbar^2} g^{00}_{(2)} - \frac{\I m}{\hbar} (\partial_a g^{0a}_{(1)}) - \frac{2\I m}{\hbar}g^{0a}_{(1)} \partial_a + \frac{2\I m}{\hbar} \partial_t + \Delta\right) a_0 \; ,
\end{equation}
i.e.\ the Schrödinger equation
\begin{equation} \label{eq:Schroedinger_WKB_exp_0}
	\I\hbar \partial_t a_0 = \left(-\frac{\hbar^2}{2m} \Delta + \frac{\I\hbar}{2} (\partial_a g^{0a}_{(1)}) + \I\hbar g^{0a}_{(1)} \partial_a + \frac{m}{2} g^{00}_{(2)}\right) a_0 \; .
\end{equation}
By the relation $\psi = a_0 + \Or(c^{-1})$, this also 
gives a Schrödinger equation for $\psi$ in $0^\mathrm{th}$ 
order in $c^{-1}$.

At order $c^{-1}$, \eqref{eq:WKB_KG} yields the following 
Schrödinger-like equation for $a_1$ with correction terms 
involving $a_0$:
\begin{align} \label{eq:Schroedinger_WKB_exp_1}
	\I\hbar \partial_t a_1 &= \left(-\frac{\hbar^2}{2m} \Delta + \frac{\I\hbar}{2} (\partial_a g^{0a}_{(1)}) + \I\hbar g^{0a}_{(1)} \partial_a + \frac{m}{2} g^{00}_{(2)}\right) a_1 \nonumber\\
		&\quad +\bigg(-\frac{\I \hbar}{4} g^{0a}_{(1)} [\partial_a \tr(\eta g^{-1}_{(1)})] + \frac{\I \hbar}{4} [\partial_t \tr(\eta g^{-1}_{(1)})] + \frac{\hbar^2}{4m} [\partial_a \tr(\eta g^{-1}_{(1)})] \delta^{ab} \partial_b \nonumber\\
			&\qquad - \frac{\hbar^2}{2m} (\partial_a g^{ab}_{(1)}) \partial_b - \frac{\hbar^2}{2m} g^{ab}_{(1)} \partial_a \partial_b + \frac{m}{2} g^{00}_{(3)} + \frac{\I \hbar}{2} (\partial_a g^{0a}_{(2)}) + \I \hbar g^{0a}_{(2)} \partial_a\bigg) a_0
\end{align}
Using $\psi = a_0 + c^{-1} a_1 + \Or(c^{-2})$, we can 
combine \eqref{eq:Schroedinger_WKB_exp_1} with 
\eqref{eq:Schroedinger_WKB_exp_0} into a Schrödinger 
equation for $\psi$ up to order $c^{-1}$:
\begin{align} \label{eq:Schroedinger_WKB}
	\I\hbar \partial_t \psi &= \Bigg[-\frac{\hbar^2}{2m} \Delta + \frac{\I\hbar}{2} (\partial_a g^{0a}_{(1)}) + \I\hbar g^{0a}_{(1)} \partial_a + \frac{m}{2} g^{00}_{(2)} + c^{-1} \bigg(-\frac{\I \hbar}{4} g^{0a}_{(1)} [\partial_a \tr(\eta g^{-1}_{(1)})] \nonumber\\
		&\qquad + \frac{\I \hbar}{4} [\partial_t \tr(\eta g^{-1}_{(1)})] + \frac{\hbar^2}{4m} [\partial_a \tr(\eta g^{-1}_{(1)})] \delta^{ab} \partial_b - \frac{\hbar^2}{2m} (\partial_a g^{ab}_{(1)}) \partial_b - \frac{\hbar^2}{2m} g^{ab}_{(1)} \partial_a \partial_b  \nonumber\\
		&\qquad + \frac{m}{2} g^{00}_{(3)} + \frac{\I \hbar}{2} (\partial_a g^{0a}_{(2)}) + \I \hbar g^{0a}_{(2)} \partial_a\bigg) + \Or(c^{-2})\Bigg] \psi =: H \psi
\end{align}

Continuing this process of evaluating \eqref{eq:WKB_KG}, 
we can, in principle, get Schrödinger equations for 
$\psi$ to arbitrary order in $c^{-1}$, i.e.\ obtain the 
Hamiltonian in the Schrödinger form of the positive 
frequency \KGe\ to arbitrary order in $c^{-1}$.
\enlargethispage{-\baselineskip}

However, when considering higher orders, a difficulty 
arises: the Schrödinger-like equations for $a_k$ begin 
to involve time derivatives of the lower order functions 
$a_l$, so we have to re-use the derived equations for the 
$a_l$ in order to get a true Schrödinger equation for 
$\psi$ (with a purely `spatial' Hamiltonian, i.e.\ not 
involving any time derivatives)~-- i.e.\ the process 
becomes recursive. As far as concrete calculations up to 
some finite order are concerned, this is merely a 
computational obstacle; but for a general analysis of the 
expansion method this poses a bigger problem, since no 
general closed form can be easily obtained. This motivated 
the study of the \KGe\ as a quadratic equation for the 
time derivative operator, leading to the `momentum 
expansion' method described in section \ref{sec:mom_exp}.
\enlargethispage{-\baselineskip}

\subsection{Transformation to `flat' scalar product and comparison with canonical quantisation}
\label{sec:WKB_trafo_flat}

To transform the Hamiltonian obtained in 
\eqref{eq:Schroedinger_WKB} from the representation of 
the Hilbert space with the Klein--Gordon inner product 
\eqref{eq:KG_ip_general} to the `flat' scalar product 
\eqref{eq:scalar_prod_flat} in order to compare it to 
the result from canonical quantisation, we note that 
for two positive frequency solutions 
$\Psi_\mathrm{KG} = \exp(-\I mc^2 t/\hbar) \psi$ and 
$\Phi_\mathrm{KG} = \exp(-\I mc^2 t/\hbar) \varphi$, 
the Klein--Gordon inner product is given by
\begin{align} \label{eq:KG_ip}
	\langle\Psi_\mathrm{KG},\Phi_\mathrm{KG}\rangle_\mathrm{KG} 
		&= \I\hbar c \int \D^3\ivect x \, \sqrt{^{(3)}g} \, g^{0\nu} [(\partial_\nu \overline{\Psi_\mathrm{KG}}) \Phi_\mathrm{KG} - \overline{\Psi_\mathrm{KG}} (\partial_\nu \Phi_\mathrm{KG})] \frac{1}{\sqrt{-g^{00}}} \nonumber\\
	&= \int \D^3\ivect x \, \sqrt{^{(3)}g} \, \Bigg(\sqrt{-g^{00}} \left[2mc^2 \overline{\psi}\varphi + \overline{(H \psi)}\varphi + \overline\psi (H \varphi)\right]
		\nonumber\\&\qquad + \I\hbar c \frac{g^{0a}}{\sqrt{-g^{00}}} \left[\overline{(\partial_a \psi)}\varphi - \overline\psi(\partial_a \varphi)\right] \Bigg) ,
\end{align}
where we used our adapted coordinates and chose 
$\Sigma = \{t = \mathrm{const.}\}$ in the general form 
\eqref{eq:KG_ip_general} of the Klein--Gordon inner product.

Using $\sqrt{-g^{00}} = 1 + \Or(c^{-2})$, 
$g^{0a} (-g^{00})^{-1/2} = \Or(c^{-1})$, and 
$H = \Or(c^{0})$, we get
\begin{equation}
	\frac{1}{2mc^2} \langle \Psi_\mathrm{KG}, \Phi_\mathrm{KG} \rangle_\mathrm{KG} = \int \D^3\ivect x \, \sqrt{^{(3)}g} \, [\overline{\psi}\varphi + \Or(c^{-2})] .
\end{equation}
For this to equal the `flat' scalar product 
$\int \D^3\ivect x \, \overline{\psi_\mathrm{f}} \, \varphi_\mathrm{f}$, 
we see that the `flat wavefunction' has to have the form 
$\psi_\mathrm{f} = {^{(3)}g^{1/4}} \, \psi + \Or(c^{-2})$ 
and therefore evolves according to the Schrödinger equation 
$\I\hbar \partial_t \psi_\mathrm{f} = H_\mathrm{f} \, \psi_\mathrm{f}$ 
with the `flat Hamiltonian'
\begin{equation}
	H_\mathrm{f} = \I\hbar\left(\partial_t {^{(3)}g^{1/4}}\right) \, {^{(3)}g^{-1/4}} 
		+ {^{(3)}g^{1/4}} \, H \left({^{(3)}g^{-1/4}} \cdot\right) + \Or(c^{-2}).
\end{equation}

Using\footnote
	{The metric determinant satisfies 
	$g^{-1} = -1 - c^{-1} \tr(\eta g^{-1}_{(1)}) + \Or(c^{-2})$. 
	Using the well-known identity ${^{(3)}g} = g^{00} g$ for a 
	$3+1$ decomposed metric, this gives ${^{(3)}g} = g^{00} g 
	= [-1 + \Or(c^{-2})] [-1 + c^{-1} \tr(\eta g^{-1}_{(1)}) + \Or(c^{-2})] 
	= 1 - c^{-1} \tr(\eta g^{-1}_{(1)}) + \Or(c^{-2})$.}
${^{(3)}g^{1/4}} = 1 - c^{-1}\frac{1}{4} \tr(\eta g^{-1}_{(1)}) + \Or(c^{-2})$ 
and noting that conjugation with a multiplication operator 
leaves multiplication operators invariant, we obtain
\begin{align} \label{eq:Schroedinger_WKB_flat}
	H_\mathrm{f} 
		&= -\I\hbar c^{-1} \frac{1}{4} [\partial_t \tr(\eta g^{-1}_{(1)})] 
		+ H - c^{-1}\frac{\hbar^2}{8m} [\Delta, \tr(\eta g^{-1}_{(1)})] \nonumber\\
		&\quad + c^{-1} \frac{\I\hbar}{4} g^{0a}_{(1)} [\partial_a, \tr(\eta g^{-1}_{(1)})] 
		+ \Or(c^{-2}) \nonumber\\
	&= -\I\hbar c^{-1} \frac{1}{4} [\partial_t \tr(\eta g^{-1}_{(1)})] 
		+ H - c^{-1}\frac{\hbar^2}{8m} \left([\Delta \tr(\eta g^{-1}_{(1)})] 
		+ 2[\partial_a \tr(\eta g^{-1}_{(1)})] \delta^{ab} \partial_b\right) \nonumber\\
		&\quad + c^{-1} \frac{\I\hbar}{4} g^{0a}_{(1)} [\partial_a \tr(\eta g^{-1}_{(1)})] 
		+ \Or(c^{-2}) \nonumber\\
	&= -\frac{\hbar^2}{2m} \Delta 
		+ \frac{\I\hbar}{2} (\partial_a g^{0a}_{(1)}) 
		+ \I\hbar g^{0a}_{(1)} \partial_a + \frac{m}{2} g^{00}_{(2)} 
		+ c^{-1} \bigg(-\frac{\hbar^2}{2m} (\partial_a g^{ab}_{(1)}) \, \partial_b 
		- \frac{\hbar^2}{2m} g^{ab}_{(1)} \partial_a \partial_b \nonumber\\
		&\quad + \frac{m}{2} g^{00}_{(3)} 
		+ \frac{\I \hbar}{2} (\partial_a g^{0a}_{(2)}) 
		+ \I \hbar g^{0a}_{(2)} \partial_a 
		- \frac{\hbar^2}{8m} [\Delta \tr(\eta g^{-1}_{(1)})]\bigg) 
		+ \Or(c^{-2}) \nonumber\\
	&= -\frac{\hbar^2}{2m} \Delta 
		- \frac{1}{2} \left\{g^{0a}_{(1)}, -\I\hbar\partial_a\right\} 
		+ \frac{m}{2} g^{00}_{(2)} 
		+ c^{-1} \bigg(\frac{1}{2m} (-\I\hbar) \partial_a \Big(g^{ab}_{(1)} (-\I\hbar)\partial_b \cdot \Big) \nonumber\\
		&\quad + \frac{m}{2} g^{00}_{(3)} 
		- \frac{1}{2} \left\{g^{0a}_{(2)}, -\I\hbar\partial_a\right\} 
		- \frac{\hbar^2}{8m} [\Delta \tr(\eta g^{-1}_{(1)})]\bigg) 
		+ \Or(c^{-2}),
\end{align}
where $\{A,B\} = AB + BA$ denotes the anticommutator. This 
is the Hamiltonian \mbox{appearing} in the `flat' Schrödinger 
form of the positive frequency Klein--Gordon \mbox{equation} up to order $c^{-1}$, 
obtained by the \acronym{WKB}-like approximation in a 
general metric.
\enlargethispage{-2\baselineskip}

For comparison of this result with the canonical 
quantisation scheme, we have to subtract the rest energy 
$mc^2$ from the classical Hamiltonian of equation 
\eqref{eq:class_Ham}, corresponding to the phase factor 
separated off the Klein--Gordon field, and expand it 
in $c^{-1}$, yielding
\begin{align}
	H_\mathrm{class} 
		&= \frac{1}{\sqrt{-g^{00}}} c \left[m^2 c^2 + \left(g^{ab} - \frac{1}{g^{00}} g^{0a} g^{0b}\right) p_a p_b\right]^{1/2} 
		\kern-.8em- mc^2 + \frac{c}{g^{00}} g^{0a} p_a \nonumber\\
	&= \frac{m}{2} g^{00}_{(2)} + \frac{\vect p^2}{2m} 
		- g^{0a}_{(1)} p_a + c^{-1} \left(\frac{m}{2} g^{00}_{(3)} 
		+ g^{ab}_{(1)} \frac{p_a p_b}{2m} - g^{0a}_{(2)} p_a\right) 
		+ \Or(c^{-2}).
\end{align}
Comparing this with \eqref{eq:Schroedinger_WKB_flat}, 
we see that by `canonical quantisation' of this classical 
Hamiltonian using the rule `$p_i \to -\I\hbar\partial_i$', 
we can reproduce, using a specific ordering scheme, all 
terms appearing in the \acronym{WKB} expansion, apart from 
$-\frac{\hbar^2}{8mc} [\Delta \tr(\eta g^{-1}_{(1)})]$. For 
this last term to arise by naive canonical quantisation, 
consisting only of symmetrising according to some ordering 
scheme and replacing momenta by operators, in the classical 
Hamiltonian there would have to be a term proportional to 
$\frac{\vect p^2}{mc} \tr(\eta g^{-1}_{(1)}) 
= \frac{\vect p^2}{mc} \delta_{ab} g^{ab}_{(1)}$, which is 
not the case.
\enlargethispage{-2\baselineskip}

As the most simple non-trivial example, for the 
`Newtonian' metric with line element
\begin{equation}
	\D s^2 = -\left(1+2\frac{\phi}{c^2}\right)c^2\D t^2 + \D \vect x^2 + \Or(c^{-2}),
\end{equation}
the inverse metric has components
\begin{equation}
	(g^{\mu\nu}) = \begin{pmatrix}
		-1 + 2 \frac{\phi}{c^2} + \Or(c^{-4}) & \Or(c^{-3})\\
		\Or(c^{-3}) & \mathbb{1} + \Or(c^{-2})
	\end{pmatrix},
\end{equation}
leading to the quantum Hamiltonian 
$H = -\frac{\hbar^2}{2m}\Delta + m \phi + \Or(c^{-2})$ 
in both schemes, i.e.\ just the standard Hamiltonian with 
Newtonian potential.

The occurrence of an extra term in a geometrically 
motivated quantum theory which one cannot arrive at by 
naive canonical quantisation is reminiscent of the 
occurrence of a `quantum-mechanical potential' term in 
the Hamiltonian found by DeWitt in his 1952 treatment of 
quantum motion in a curved space \cite{dewitt52}: by 
demanding the (free part of the) Hamiltonian to be given by 
$H^\mathrm{DeWitt} = -\frac{\hbar^2}{2m} {^{(3)} \hspace{-.2em} \Delta_\mathrm{LB}}$ 
in terms of the spatial Laplace--Beltrami operator 
$^{(3)}\hspace{-.2em}\Delta_\mathrm{LB}$ (induced by the 
\emph{physical} spatial metric ${^{(3)}g}$, not the 
background flat one), it turns out to have the form 
$H^\mathrm{DeWitt} = \frac{1}{2m}\hat p_a \, {^{(3)}g^{ab}} \, \hat p_b + \hbar^2 Q$ 
of a sum of a naively canonically quantised kinetic term\footnote
	{Note that DeWitt uses the `geometric' scalar product 
	\eqref{eq:canon_scalar_prod_induced}, not the `flat' one.}
and the quantum-mechanical potential\footnote
	{Using the form 
	\begin{align}
		-\hbar^2 \, {^{(3)}\hspace{-.2em} \Delta_\mathrm{LB}} 
			&= -\hbar^2 \frac{1}{\sqrt{^{(3)}g}} \partial_a \left( \sqrt{^{(3)}g} \, {^{(3)}g^{ab}} \partial_b \, \cdot \right) \nonumber\\
		&= {^{(3)}g^{-1/4}} \, \hat p_a \, {^{(3)}g^{1/2}} \, {^{(3)}g^{ab}} \, \hat p_b \, {^{(3)}g^{-1/4}}
	\end{align}
	of the Laplace--Beltrami operator in terms of the momentum 
	operator \eqref{eq:canon_mom_op_induced}, it can be 
	expressed as 
	\begin{equation}
		-\hbar^2 \, {^{(3)}\hspace{-.2em} \Delta_\mathrm{LB}} 
			= \hat p_a \, {^{(3)}g^{ab}} \, \hat p_b 
			- {^{(3)}g^{-1/4}} [\hat p_a, {^{(3)}g^{ab}} [\hat p_b, {^{(3)}g^{1/4}}]],
	\end{equation}
	giving the above expression for the quantum-mechanical potential.}
$\hbar^2 Q = \frac{\hbar^2}{2m} {^{(3)}g^{-1/4}} \partial_a ({^{(3)}g^{ab}} \partial_b {^{(3)}g^{1/4}})$.

In fact, for our metric \eqref{eq:exp_metric}, in lowest 
order in $c^{-1}$ the quantum-mechanical potential is 
given by $\hbar^2 Q = -\frac{\hbar^2}{8mc} \Delta (\delta_{ab} g^{ab}_{(1)}) 
+ \Or(c^{-2}) = -\frac{\hbar^2}{8mc}[\Delta \tr(\eta g^{-1}_{(1)})] 
+ \Or(c^{-2})$, thus reproducing the additional term 
arising in the \acronym{WKB} method. This apparent 
connection of our \acronym{WKB}-like expansion to the 
three-dimensional `spatial' geometry seems interesting, 
but further investigation in this direction goes beyond 
the scope of this thesis, since in this post-Newtonian 
context, the explicit comparison to the Newtonian limit 
--~which also includes flat space~-- is the specific 
subject of interest.

Note that one could argue that DeWitt's Hamiltonian 
\emph{can} be arrived at by canonical quantisation in 
some sense, since the Laplace--Beltrami operator can be 
written as $- \hbar^2 \, {^{(3)}\hspace{-.2em} \Delta_\mathrm{LB}} 
= {^{(3)}g^{-1/4}} \, \hat p_a \, {^{(3)}g^{1/2}} \, {^{(3)}g^{ab}} \, \hat p_b \, {^{(3)}g^{-1/4}}$ 
in terms of the momentum operator \eqref{eq:canon_mom_op_induced} 
corresponding to the `geometric' scalar product 
\eqref{eq:canon_scalar_prod_induced} which was used 
by DeWitt. However, such a `clever rewriting' of the 
Newtonian kinetic term in the classical Hamiltonian as 
$\frac{1}{2m} {^{(3)}g^{ab}} p_a p_b 
= \frac{1}{2m} {^{(3)}g^{-1/4}} \, p_a \, {^{(3)}g^{1/2}} \, {^{(3)}g^{ab}} \, p_b \, {^{(3)}g^{-1/4}}$ 
before replacing momenta by operators involves more than 
just choosing some symmetrised operator ordering, and 
thus is not part of what we called `canonical 
quantisation' above.

\subsection{The Eddington--Robertson \acronym{PPN} metric as an explicit example}
\label{sec:WKB_ER_PPN}

We now will apply the \acronym{WKB}-like expansion method 
to the Eddington--Robertson parametrised post-Newtonian 
metric as given by \eqref{eq:ER_PPN_metric}, 
\eqref{eq:ER_PPN_metric_inv}.

Inserting the metric components, the equations arising for 
the coefficient functions $a_0, a_1$ from \eqref{eq:WKB_KG} 
at orders $c^0, c^{-1}$ are simply the Schrödinger equations
\begin{equation} \label{eq:Schroedinger_WKB_ER_exp_0}
	\I\hbar \partial_t a_i = \left(-\frac{\hbar^2}{2m} \Delta + m\phi\right) a_i, \quad i = 0,1.
\end{equation}
At orders $c^{-2}, c^{-3}$, we get --~again for $i = 0,1$~--
\begin{align}
	0 &= \bigg[- \frac{\I m}{\hbar} (\partial_t g^{00}_{(2)}) 
			- \frac{2\I m}{\hbar} g^{00}_{(2)} \partial_t - \partial_t^2 
			+ \frac{\I m}{2\hbar} \bigg( -[\partial_t \tr(\eta g^{-1}_{(2)})] 
				+ \frac{1}{2} [\partial_t \tr(\eta \underbrace{g^{-1}_{(2,2)}}_{\mathclap{= g^{-1}_{(1)} \eta g^{-1}_{(1)} = 0}} )] \bigg) \nonumber\\
			&\qquad + \frac{1}{2} \left( -[\partial_a \tr(\eta g^{-1}_{(2)})] \delta^{ab} \partial_b 
				+ \frac{1}{2} [\partial_a \tr(\eta g^{-1}_{(2,2)})] \delta^{ab} \partial_b \right) \nonumber\\
			&\qquad + (\partial_a g^{ab}_{(2)}) \, \partial_b 
			+ g^{ab}_{(2)} \partial_a \partial_b 
			- \frac{m^2}{\hbar^2} g^{00}_{(4)} \bigg] a_i 
		+ \left(-\frac{m^2}{\hbar^2} g^{00}_{(2)} 
			+ \frac{2\I m}{\hbar} \partial_t + \Delta\right) a_{i+2} \nonumber\displaybreak[0]\\
	&= \bigg(- \frac{4\I m}{\hbar} \phi \partial_t - \partial_t^2 
			- \frac{\I m}{\hbar} (3\gamma + 1) (\partial_t \phi) 
			- (\gamma - 1) (\partial_a \phi) \delta^{ab} \partial_b \nonumber\\
			&\qquad + 2\gamma \phi \Delta 
			- \frac{m^2}{\hbar^2} (2\beta - 4) \phi^2 \bigg) a_i 
		+ \left(-\frac{2 m^2}{\hbar^2} \phi 
			+ \frac{2\I m}{\hbar} \partial_t + \Delta\right) a_{i+2} \; ,
\end{align}
or equivalently the Schrödinger-like equations
\begin{align}
	\I\hbar \partial_t a_{i+2} 
		&= \left(-\frac{\hbar^2}{2m} \Delta + m\phi\right) a_{i+2} 
		+ \bigg(2\I\hbar \phi \partial_t 
			+ \frac{\hbar^2}{2m} \partial_t^2 
			+ \frac{\I\hbar}{2} (3\gamma + 1) (\partial_t \phi) \nonumber\\
		&\qquad + \frac{\hbar^2}{2m} (\gamma - 1) (\partial_a \phi) \delta^{ab} \partial_b 
			- \frac{\hbar^2}{m} \gamma \phi \Delta 
			+ \frac{m}{2} (2\beta - 4) \phi^2 \bigg) a_i
\end{align}
for $a_2, a_3$. Using the Schrödinger equation 
\eqref{eq:Schroedinger_WKB_ER_exp_0} for $a_0, a_1$, we have
\begin{align}
	\frac{\hbar^2}{2m} \partial_t^2 a_i 
		&= - \frac{\I\hbar}{2m} \partial_t \left(-\frac{\hbar^2}{2m} \Delta + m\phi\right) a_i 
		= -\frac{\I\hbar}{2} (\partial_t \phi) a_i 
			- \frac{1}{2m} \left(-\frac{\hbar^2}{2m} \Delta + m\phi\right) \I\hbar \partial_t a_i \nonumber\\
	&= -\frac{\I\hbar}{2} (\partial_t \phi) a_i 
		- \frac{1}{2m} \left(-\frac{\hbar^2}{2m} \Delta + m\phi\right)^2 a_i \nonumber\\
	&= -\frac{\I\hbar}{2} (\partial_t \phi) a_i 
		-\frac{\hbar^4}{8m^3} \Delta \Delta a_i 
		+ \frac{\hbar^2}{4m} \Delta (\phi a_i) 
		+ \frac{\hbar^2}{4m} \phi \Delta a_i 
		- \frac{m}{2} \phi^2 a_i \nonumber\\
	&= -\frac{\I\hbar}{2} (\partial_t \phi) a_i 
		-\frac{\hbar^4}{8m^3} \Delta \Delta a_i 
		+ \frac{\hbar^2}{4m} (\Delta\phi) a_i 
		+ \frac{\hbar^2}{2m} (\partial_a \phi) \delta^{ab} \partial_b a_i 
		+ \frac{\hbar^2}{2m} \phi \Delta a_i 
		- \frac{m}{2} \phi^2 a_i \; ,
\end{align}
and thus the equation for $a_2, a_3$ becomes
\begin{align} \label{eq:Schroedinger_WKB_ER_exp_2}
	\I\hbar \partial_t a_{i+2} 
	&= \left(-\frac{\hbar^2}{2m} \Delta + m\phi\right) a_{i+2}
	+ \bigg(-\frac{\hbar^4}{8m^3} \Delta \Delta 
		+ \frac{\hbar^2}{4m} (\Delta\phi) 
		+ \frac{3 \I\hbar}{2} \gamma (\partial_t \phi) \nonumber\\
		&\qquad + \frac{\hbar^2}{2m} \gamma (\partial_a \phi) \delta^{ab} \partial_b 
		- \frac{\hbar^2}{2m} (2\gamma + 1) \phi \Delta 
		+ \frac{m}{2} (2\beta - 1) \phi^2 \bigg) a_i \; .
\end{align}
At higher orders, the coefficients in the expanded \KGe\ 
\eqref{eq:WKB_KG} are undetermined, since the metric 
components are undetermined.

Combining the equations \eqref{eq:Schroedinger_WKB_ER_exp_0} 
for $a_0, a_1$ and \eqref{eq:Schroedinger_WKB_ER_exp_2} for 
$a_2, a_3$, the Hamiltonian in the Schrödinger equation 
$\I\hbar \partial_t \psi = H \psi$ for the `wavefunction' 
(i.e.\ phase-shifted positive-frequency Klein--Gordon 
field) $\psi$ reads
\begin{align}
	H &= -\frac{\hbar^2}{2m} \Delta + m\phi
	+ \frac{1}{c^2}\bigg(-\frac{\hbar^4}{8m^3} \Delta \Delta 
		+ \frac{\hbar^2}{4m} (\Delta\phi) 
		+ \frac{3 \I\hbar}{2} \gamma (\partial_t \phi) \nonumber\\
		&\qquad + \frac{\hbar^2}{2m} \gamma (\partial_a \phi) \delta^{ab} \partial_b 
		- \frac{\hbar^2}{2m} (2\gamma + 1) \phi \Delta 
		+ \frac{m}{2} (2\beta - 1) \phi^2 \bigg) + \Or(c^{-4}),
\end{align}
reproducing, up to notational differences and the fact that 
we did not consider coupling to an electromagnetic field, 
the result of Lämmerzahl \cite[eq. (8)]{laemmerzahl95}.

To transform to the flat scalar product, we note that in 
our metric and using this Hamiltonian, the Klein--Gordon 
inner product \eqref{eq:KG_ip} is given by
\begin{equation} \label{eq:ER_PPN_KG_ip}
	\frac{1}{2mc^2} \langle \Psi_\mathrm{KG}, \Phi_\mathrm{KG} \rangle_\mathrm{KG} 
	= \int \D^3\ivect x \, \sqrt{^{(3)}g} \, \left(\overline{\psi}\varphi - \frac{\hbar^2}{2m^2c^2} \overline{\psi} \Delta\varphi + \Or(c^{-4}) \right) .
\end{equation}
Note that in the brackets, we did not need to expand 
any further since the factor $\sqrt{^{(3)}g}$ is 
only determined up to $\Or(c^{-4})$ by the metric 
\eqref{eq:ER_PPN_metric}. For the expression 
\eqref{eq:ER_PPN_KG_ip} to equal the flat 
scalar product 
$\int \D^3\ivect x \, \overline{\psi_\mathrm{f}} \, \varphi_\mathrm{f}$, 
the flat wavefunction has to have the form $\psi_\mathrm{f} 
= \left(1 - \frac{\hbar^2}{2m^2c^2}\Delta\right)^{1/2} {^{(3)}g^{1/4}} \, \psi + \Or(c^{-4})$ 
(note that $\frac{1}{c^2} \Delta$ commutes with $^{(3)}g$ 
up to higher-order terms), resulting in the flat Hamiltonian
\begin{align}
	H_\mathrm{f} &= \I\hbar \left(\partial_t {^{(3)}g^{1/4}}\right) {^{(3)}g^{-1/4}} \nonumber\\
		&\quad+ \left(1 - \frac{\hbar^2}{2m^2c^2}\Delta\right)^{1/2} \kern-.5em {^{(3)}g^{1/4}} \, H \, {^{(3)}g^{-1/4}} \left(1 - \frac{\hbar^2}{2m^2c^2}\Delta\right)^{-1/2} \kern-1em+ \Or(c^{-4}).
\end{align}
Using ${^{(3)}g^{1/4} = 1 - \frac{3}{2} \gamma \frac{\phi}{c^2} + \Or(c^{-4})}$ 
and $\left(1 - \frac{\hbar^2}{2m^2c^2}\Delta\right)^{1/2} 
= 1 - \frac{\hbar^2}{4m^2c^2}\Delta + \Or(c^{-4})$, 
this yields
\begin{align}
	H_\mathrm{f} 
	&= -\I\hbar\left(\partial_t \frac{3}{2} \gamma \frac{\phi}{c^2} \right) 
		+ H + \left[-\frac{3}{2}\gamma\frac{\phi}{c^2}, -\frac{\hbar^2}{2m} \Delta\right] 
		+ \left[- \frac{\hbar^2}{4m^2c^2}\Delta, m\phi\right] + \Or(c^{-4}) \nonumber\\
	&= -\frac{3\I\hbar}{2c^2} \gamma (\partial_t \phi) 
		+ H - \frac{\hbar^2}{4mc^2} (3\gamma + 1) [\Delta, \phi] 
		+ \Or(c^{-4}) \nonumber\\
	&= -\frac{3\I\hbar}{2c^2} \gamma (\partial_t \phi) 
		+ H - \frac{\hbar^2}{4mc^2} (3\gamma + 1) ((\Delta \phi) 
		+ 2 (\partial_a \phi) \delta^{ab} \partial_b) 
		+ \Or(c^{-4}) \nonumber\\
	&= -\frac{\hbar^2}{2m} \Delta + m\phi
		+ \frac{1}{c^2} \bigg(-\frac{\hbar^4}{8m^3} \Delta \Delta 
			- \frac{3\hbar^2}{4m} \gamma (\Delta\phi) \nonumber\\
			&\qquad - \frac{\hbar^2}{2m} (2\gamma + 1) (\partial_a \phi) \delta^{ab} \partial_b 
			- \frac{\hbar^2}{2m} (2\gamma + 1) \phi \Delta 
			+ \frac{m}{2} (2\beta - 1) \phi^2 \bigg) + \Or(c^{-4}),
\end{align}
reproducing the flat Hamiltonian of Lämmerzahl 
\cite[eq. (16)]{laemmerzahl95}.

In comparison, the classical Hamiltonian (minus the 
rest energy) expands to
\begin{align}
	H_\mathrm{class} 
		&= \frac{1}{\sqrt{-g^{00}}} c \left[m^2 c^2 + \left(g^{ab} - \frac{1}{g^{00}} g^{0a} g^{0b}\right) p_a p_b\right]^{1/2} 
		\kern-.8em- mc^2 + \frac{c}{g^{00}} g^{0a} p_b \nonumber\\
	&= \frac{\vect p^2}{2m} + m\phi 
		+ c^{-2} \left(-\frac{(\vect p^2)^2}{8m^3} 
			+ \frac{m \phi^2}{2}(2\beta - 1) 
			+ \frac{\phi}{2m}(2\gamma + 1) \vect p^2\right) 
			+ \Or(c^{-4}).
\end{align}
By canonical quantisation of this, we cannot reproduce 
the Hamiltonian obtained from the \acronym{WKB} expansion 
in the case of a general $\gamma$, but just for some 
special choices of~$\gamma$, depending on the ordering 
scheme: for example, in the anticommutator ordering scheme, 
we would quantise the classical function $\phi \vect p^2$ 
as
\begin{equation}
	\frac{1}{2} \{-\hbar^2\Delta, \phi\} 
	= -\frac{\hbar^2}{2} (\Delta\phi) 
		- \hbar^2 (\partial_a \phi) \delta^{ab} \partial_b 
		- \hbar^2 \phi \Delta \; ,
\end{equation}
reproducing the \acronym{WKB} Hamiltonian in the case 
of $\gamma = 1$; but when quantising it as 
$-\hbar^2 \delta^{ab} \partial_a (\phi \partial_b \, \cdot) 
= -\hbar^2 (\partial_a \phi) \delta^{ab} \partial_b 
	- \hbar^2 \phi \Delta$, 
this would lead to agreement with the \acronym{WKB} 
Hamiltonian for $\gamma = 0$. Note however that this 
difference concerns a term proportional to $\Delta \phi$, 
the Laplacian of the Newtonian potential. By the Newtonian 
gravitational field equation, this term is (in lowest 
order) proportional to the mass density generating the 
gravitational field. Thus it is irrelevant in physical 
situations concerning the outside of the generating matter 
distribution, for example in quantum-optical experiments 
in the gravitational field of the earth taking place 
outside of the earth. Nevertheless, this example shows 
that the way in which \acronym{PPN} parameters enter a 
quantum description delicately depends on the quantisation 
method.

\section{General comparison of the two methods by momentum expansion}
\label{sec:mom_exp}

We will now describe a method by which general statements 
about similarities and differences between the two 
approaches explained above can be made in the case of 
stationary spacetimes, without any post-Newtonian expansion 
in $c^{-1}$. Instead, we consider `potential' terms and 
terms linear, quadratic, \ldots\ in momentum, i.e.\ we 
perform a (formal) expansion in momenta. Of course, this 
also amounts to somewhat of a post-Newtonian expansion~-- 
although just relating to the particle momentum/velocity, 
not the gravitational field \emph{per se}.

\subsection{The \KGe\ as a quadratic equation for the Hamiltonian}

We assume a \emph{stationary} physical spacetime such that 
the background time evolution vector field\footnote
	{In fact, for the `momentum expansion' to be developed in 
	the following we do not need to expand the physical metric 
	in any way, and thus we do not need a background metric to 
	define a notion of `absence of gravity'. Nevertheless, we 
	need a notion of `space'~-- but this could also be given 
	by something else than the orthogonal complement of the 
	stationarity field with respect to a background metric. 
	In any case, our approach based on a background metric 
	leads to a decomposition as needed in an easy and 
	well-defined geometric way.}
$u = \partial_t$ 
is (a constant multiple of) the stationarity Killing field, 
i.e. $\partial_t g_{\mu\nu} = 0$. The coordinate expression 
for the d'Alembert operator on functions is thus
\begin{align}
	\Box f &= \frac{1}{\sqrt{-g}} \partial_\mu(\sqrt{-g} g^{\mu\nu} \partial_\nu f) \nonumber\\
	&= \frac{1}{\sqrt{-g}} (\partial_\mu\sqrt{-g}) g^{\mu\nu} \partial_\nu f 
		+ (\partial_\mu g^{\mu\nu}) \partial_\nu f 
		+ g^{\mu\nu} \partial_\mu\partial_\nu f \nonumber\\
	&= \frac{1}{2g} (\partial_a g) g^{a\nu} \partial_\nu f 
		+ (\partial_a g^{a\nu}) \partial_\nu f 
		+ g^{\mu\nu} \partial_\mu\partial_\nu f.
\end{align}
Hence, the minimally coupled \KGe\ reads
\begin{align}
	0 &= \left(\Box - \frac{m^2c^2}{\hbar^2}\right)\Psi \nonumber\\
	&= \frac{1}{c} \frac{1}{2g} (\partial_a g) g^{0a} \partial_t \Psi 
		+ \frac{1}{2g} (\partial_a g) g^{ab} \partial_b \Psi 
		+ \frac{1}{c} (\partial_a g^{0a}) \, \partial_t \Psi 
		+ (\partial_a g^{ab}) \partial_b \Psi \nonumber\\
		&\quad + \frac{1}{c^2} g^{00} \partial_t^2 \Psi 
		+ \frac{2}{c} g^{0a} \partial_a \partial_t \Psi 
		+ g^{ab} \partial_a \partial_b \Psi 
		- \frac{m^2c^2}{\hbar^2}\Psi.
\end{align}
This means that the space of solutions of the \KGe\ is 
the kernel of $\mathcal P(\I\hbar \partial_t)$, where for 
an operator $A$ acting on the functions on the spacetime, 
$\mathcal P(A)$ is the following operator:
\begin{align}
	\mathcal P(A) = 
		&-\frac{\I}{\hbar c} \frac{1}{2g} (\partial_a g) g^{0a} A 
		+ \frac{1}{2g} (\partial_a g) g^{ab} \partial_b 
		- \frac{\I}{\hbar c} (\partial_a g^{0a}) A 
		+ (\partial_a g^{ab}) \partial_b \nonumber\\
		&- \frac{1}{\hbar^2 c^2} g^{00} A^2 
		- \frac{2\I}{\hbar c} g^{0a} \partial_a \circ A 
		+ g^{ab} \partial_a \partial_b 
		- \frac{m^2c^2}{\hbar^2}
\end{align}
\enlargethispage{-\baselineskip}

Thus, wanting to write the \KGe\ in the form of a 
Schrödinger equation $\I\hbar\partial_t \Psi = H \Psi$ 
(and thus restricting to the solutions of the \KGe\ for 
which this is possible), we see that this can be achieved 
by demanding the Hamiltonian $H$ to be a solution of the 
quadratic operator equation
\begin{equation} \label{eq:KG_operator_expression}
	0 = \mathcal P(H)
\end{equation}
and be composed only of spatial derivative operators 
and coefficients of the metric, not involving any time 
derivatives: stationarity of the metric then implies 
$[\partial_t, H] = 0$, such that the Schrödinger equation 
yields $(\I\hbar\partial_t)^2 \Psi = \I\hbar\partial_t H \Psi 
= H \I\hbar\partial_t \Psi = H^2 \Psi$, leading to 
$\mathcal P(\I\hbar\partial_t) \Psi = \mathcal P(H) \Psi = 0$ 
by \eqref{eq:KG_operator_expression}; i.e.\ every solution 
of the Schrödinger equation is also a solution of the \KGe.

In the following, we will solve equation 
\eqref{eq:KG_operator_expression} by expanding $H$ as a 
formal power series in spatial derivative operators, i.e.\ 
momentum operators. The two possible solutions we will 
obtain for $H$ correspond to positive and negative 
frequency solutions of the \KGe, respectively.
\enlargethispage{-\baselineskip}

\subsection{Momentum expansion and first-order solution}

We expand $H$ as $H = H_{(0)} + H_{(1)} + \Or(\partial_a^2)$, 
where $H_{(k)}$ includes all terms involving $k$ spatial 
derivative operators. Using this notation, the lowest 
order term of \eqref{eq:KG_operator_expression}, involving 
no spatial derivatives, reads
\begin{equation}
	0 = -\frac{1}{\hbar^2 c^2} g^{00} H_{(0)}^2 - \frac{m^2c^2}{\hbar^2},
\end{equation}
giving
\begin{equation} \label{eq:KG_operator_H0}
	H_{(0)} = \frac{mc^2}{\sqrt{-g^{00}}}
\end{equation}
where we choose the positive square root since we are 
interested in positive frequency solutions of the \KGe.

At order $\partial_a^1$, equation \eqref{eq:KG_operator_expression} 
gives
\begin{align}
	0 = &-\frac{\I}{\hbar c} \frac{1}{2g} (\partial_a g) g^{0a} H_{(0)} 
		- \frac{\I}{\hbar c} (\partial_a g^{0a}) H_{(0)} \nonumber\\
	&- \frac{1}{\hbar^2 c^2} g^{00} (2 H_{(0)} H_{(1)} 
		+ [H_{(1)}, H_{(0)}]) 
		- \frac{2\I}{\hbar c} g^{0a} \partial_a\circ H_{(0)} \; .
\end{align}
Writing $H_{(1)} = H_{(1,M)} + H_{(N,C)}^a \partial_a$ 
where $H_{(1,M)}$ is a \emph{m}ultiplication operator 
(involving one spatial differentiation of some function) 
and $H_{(N,C)}^a$ are \emph{c}oefficient functions 
\emph{n}ot involving any differentiations, we have 
$[H_{(1)}, H_{(0)}] = [H_{(1,C)}^a \partial_a, H_{(0)}] 
= H_{(1,C)}^a (\partial_a H_{(0)})$. Thus, the equation 
reads
\begin{align} \label{eq:KG_operator_expansion_order_1}
	0 = &-\frac{\I}{\hbar c} \frac{1}{2g} (\partial_a g) g^{0a} H_{(0)} 
		- \frac{\I}{\hbar c} (\partial_a g^{0a}) H_{(0)} 
		- \frac{2g^{00}}{\hbar^2 c^2} H_{(0)} H_{(1)} \nonumber\\
		&- \frac{g^{00}}{\hbar^2 c^2} H_{(1,C)}^a (\partial_a H_{(0)}) 
		- \frac{2\I}{\hbar c} g^{0a} (\partial_a H_{(0)}) 
		- \frac{2\I}{\hbar c} g^{0a} H_{(0)} \partial_a \; .
\end{align}
The right-hand side now has two different components: 
a multiplication operator and an operator differentiating 
the function it acts upon. We demand that these components 
vanish independently. The `differentiating part' of 
\eqref{eq:KG_operator_expansion_order_1} is
\begin{equation}
	0 = - \frac{2g^{00}}{\hbar^2 c^2} H_{(0)} H_{(1,C)}^a \partial_a 
		- \frac{2\I}{\hbar c} g^{0a} H_{(0)} \partial_a \; ,
\end{equation}
or equivalently
\begin{equation} \label{eq:KG_operator_H1C}
	H_{(1,C)}^a = - \I \hbar c \frac{g^{0a}}{g^{00}} \; .
\end{equation}
Using this, the multiplication operator part of 
\eqref{eq:KG_operator_expansion_order_1} reads
\begin{equation}
	0 = -\frac{\I}{\hbar c} \frac{1}{2g} (\partial_a g) g^{0a} H_{(0)} 
	- \frac{\I}{\hbar c} (\partial_a g^{0a}) H_{(0)} 
	- \frac{2g^{00}}{\hbar^2 c^2} H_{(0)} H_{(1,M)} 
	- \frac{\I}{\hbar c} g^{0a} (\partial_a H_{(0)}),
\end{equation}
giving
\begin{equation} \label{eq:KG_operator_H1M}
	H_{(1,M)} = 
		- \frac{\I\hbar c}{4g^{00} g} (\partial_a g) g^{0a} 
		- \frac{\I\hbar c}{2g^{00}} (\partial_a g^{0a}) 
		- \frac{\I\hbar c}{2g^{00}} g^{0a} \frac{1}{H_{(0)}} (\partial_a H_{(0)}).
\end{equation}
Since $\frac{1}{H_{(0)}} (\partial_a H_{(0)}) 
= \sqrt{-g^{00}} \, \partial_a \frac{1}{\sqrt{-g^{00}}} 
= \frac{g^{00}}{2} \, \partial_a \frac{1}{g^{00}}$, equations 
\eqref{eq:KG_operator_H0}, \eqref{eq:KG_operator_H1C} and 
\eqref{eq:KG_operator_H1M} together yield the result
\begin{equation} \label{eq:KG_operator_H_complete}
	H = \frac{mc^2}{\sqrt{-g^{00}}} 
		- \frac{\I\hbar c}{4g^{00} g} (\partial_a g) g^{0a} 
		- \frac{\I\hbar c}{2g^{00}} (\partial_a g^{0a}) 
		- \frac{\I\hbar c}{4} g^{0a} \left(\partial_a \frac{1}{g^{00}}\right) 
		- \I \hbar c \frac{g^{0a}}{g^{00}} \partial_a 
		+ \Or(\partial_a^2)
\end{equation}
for the Hamiltonian in the Schrödinger form
\begin{equation}
	\I\hbar \partial_t \Psi = H \Psi
\end{equation}
of the positive frequency \KGe, at first order in momenta.

\subsection{Transformation to `flat' scalar product and comparison with canonical quantisation}

To transform this Hamiltonian to the `flat' scalar product, 
we note that for two positive frequency solutions $\Psi$ 
and $\Phi$, the Klein--Gordon inner product is given by
\begin{align}
	\langle\Psi,\Phi\rangle_\mathrm{KG} 
		&= \I\hbar c \int \D^3\ivect x \, \sqrt{^{(3)}g} \, g^{0\nu} [(\partial_\nu \overline\Psi)\Phi - \overline\Psi(\partial_\nu \Phi)] \frac{1}{\sqrt{-g^{00}}} \nonumber\\
	&= \int \D^3\ivect x \, \sqrt{^{(3)}g} \, \Bigg(\sqrt{-g^{00}} \left[\overline{(H \Psi)}\Phi 
		+ \overline\Psi(H \Phi)\right] \nonumber\\
		&\qquad+ \I\hbar c \frac{g^{0a}}{\sqrt{-g^{00}}} \left[ \overline{(\partial_a \Psi)} \Phi - \overline\Psi (\partial_a \Phi) \right] \Bigg) \nonumber\\
	\text{(using \eqref{eq:KG_operator_H_complete})}
		\quad &= \int \D^3\ivect x \, \sqrt{^{(3)}g} \, 2mc^2 \, \overline\Psi \Phi
		+ \Or(\partial_a^2).
\end{align}
For this to equal the `flat' scalar product 
$\int \D^3\ivect x \, \overline{\Psi_\mathrm{f}} \, \Phi_\mathrm{f}$, 
we see that the `flat wavefunction' has to have the form 
$\Psi_\mathrm{f} = \sqrt{2mc^2} \, {^{(3)}g^{1/4}} \, \Psi 
+ \Or(\partial_a^2)$, and therefore evolves according to 
the Schrödinger equation $\I\hbar \partial_t \Psi_\mathrm{f} 
= H_\mathrm{f} \Psi_\mathrm{f}$ with the `flat Hamiltonian'
\begin{equation}
	H_\mathrm{f} = {^{(3)}g^{1/4}} \, H \, \left({^{(3)}g^{-1/4}} \cdot\right) 
		+ \Or(\partial_a^2).
\end{equation}
For calculating $H_\mathrm{f}$ from $H$, we note that 
conjugating with a multiplication operator leaves 
multiplication operators invariant and that
\begin{align}
	{^{(3)}g^{1/4}} \, \partial_a \left({^{(3)}g^{-1/4}} \, \cdot\right) 
		&= \partial_a - \frac{1}{4} \left[\partial_a \ln\left({^{(3)}g}\right)\right] \nonumber\\
	&= \partial_a - \frac{1}{4} \left[\partial_a \ln\left(g^{00} g\right)\right] \nonumber\\
	&= \partial_a - \frac{1}{4} \, \frac{1}{g} (\partial_a g) 
		- \frac{1}{4} \, \frac{1}{g^{00}} \left(\partial_a g^{00}\right),
\end{align}
yielding the final result
\begin{align}
	H_\mathrm{f} &= \frac{mc^2}{\sqrt{-g^{00}}} 
		- \frac{\I\hbar c}{4g^{00} g} (\partial_a g) g^{0a} 
		- \frac{\I\hbar c}{2g^{00}} \left(\partial_a g^{0a}\right) 
		- \frac{\I\hbar c}{4} g^{0a} \left(\partial_a \frac{1}{g^{00}}\right) \nonumber\\
		&\quad - \I \hbar c \frac{g^{0a}}{g^{00}} \left(\partial_a 
			- \frac{1}{4 g} (\partial_a g) 
			- \frac{1}{4 g^{00}} \left(\partial_a g^{00}\right)\right) 
		+ \Or(\partial_a^2) \nonumber\\
	&= \frac{mc^2}{\sqrt{-g^{00}}} 
		- \frac{\I\hbar c}{2} \left(\partial_a \frac{g^{0a}}{g^{00}}\right) 
		- \I \hbar c \frac{g^{0a}}{g^{00}} \partial_a 
		+ \Or(\partial_a^2) \nonumber\\
	&= \frac{mc^2}{\sqrt{-g^{00}}} 
		+ c \frac{1}{2} \left\{ \frac{g^{0a}}{g^{00}}, -\I\hbar \partial_a \right\} 
		+ \Or(\partial_a^2).
\end{align}
\enlargethispage{-\baselineskip}

Looking at the momentum expansion of the classical 
Hamiltonian
\begin{align}
	H_\mathrm{class} &= \frac{1}{\sqrt{-g^{00}}} c \left[ m^2 c^2 
		+ \left(g^{ab} - \frac{1}{g^{00}} g^{0a} g^{0b}\right) p_a p_b \right]^{1/2} 
		\kern-.8em+ \frac{c}{g^{00}} g^{0a} p_a \nonumber\\
	&= \frac{mc^2}{\sqrt{-g^{00}}} 
		+ \frac{c}{g^{00}} g^{0a} p_a + \Or(p_a^2),
\end{align}
we see that `canonical quantisation' of this Hamiltonian 
will lead to the same `potential term' and to the same 
term linear in momentum as did the Klein--Gordon equation, 
regardless of the adopted ordering scheme. The reason for 
this is that for terms of linear order in momentum, any 
ordering scheme leads to `anticommutator quantisation', 
as is easily shown:

Any general canonically quantised, arbitrarily symmetrised 
operator of linear order in momentum is the sum of 
terms of the form $\hat A = \frac{1}{2} (f \bar p_a h 
+ h \bar p_a f)$, where $f, h$ are real-valued functions 
of position (here identified with the corresponding 
self-adjoint multiplication operators). The classical phase 
space function corresponding to $\hat A$ is $A = \frac{1}{2} 
(f p_a h + h p_a f) = f h p_a$. Rewriting $\hat A$ as
\begin{align}
	\hat A &= \frac{1}{2} (f \bar p_a h + h \bar p_a f) 
		= \frac{1}{2} \big(\bar p_a f h + [f, \bar p_a] h 
		+ h f \bar p_a + h [\bar p_a, f] \big) \nonumber\\
	&= \frac{1}{2} \big(\bar p_a f h + (\I\hbar \partial_a f) h 
		+ h f \bar p_a - h (\I\hbar \partial_a f) \big) 
		= \frac{1}{2} (\bar p_a f h + h f \bar p_a) \nonumber\\
	&= \frac{1}{2} \{f h, \bar p_a\},
\end{align}
we thus see that it arises from $A$ by `anticommutator 
quantisation', as desired.

We thus have shown that in stationary post-Newtonian 
spacetimes, the Hamiltonians obtained by naive canonical 
quantisation of free particle motion and by formally 
expanding the \KGe\ agree to linear order in momentum. In 
particular, this means that the lowest-order coupling to 
gravitomagnetic fields agrees in both methods.


\newcommand{\ul}{\underline} 

\newcommand{\SBtag}[1]{\tag{\cite{sonnleitner18}.#1}}
\newcommand{\SBtagc}[1]{\tag{\cite{sonnleitner18}.#1$\star$}}
\newcommand{\corr}{($\star$)}
\newcommand{\red}[1]{\textcolor{red}{#1}}
\newcommand{\redbin}[1]{\mathbin{\textcolor{red}{#1}}}

\chapter{Post-Newtonian Hamiltonian description of an atom in a\texorpdfstring{\newline}{} weak gravitational field}
\chaptermark{Post-Newtonian Hamiltonian description of an atom in a weak gravitational field}
\label{chap:atom_in_gravity}

In this chapter, we extend the systematic calculation 
of an `approximately relativistic', i.e.\ first order 
post-Newtonian, Hamiltonian for centre of mass and internal 
dynamics of an electromagnetically bound two-particle 
system by Sonnleitner and Barnett \cite{sonnleitner18} to 
the case including a weak post-Newtonian gravitational 
background field, described by the Eddington--Robertson 
\acronym{PPN} metric. Starting from a properly relativistic 
description of the situation, this approach allows to 
systematically \emph{derive} the coupling of the model 
system to gravity, instead of `guessing' it by means of 
classical notions of `relativistic effects'.\looseness-1
\enlargethispage{\baselineskip}

This chapter is based on material that has been published 
in \cite{schwartz19:AiG}. However, here we significantly 
extend the published results by dropping the approximating 
assumption of constant gravitational potential over the 
extent of the system. We also clarify a small inconsistency 
that was present in the treatment of the non-gravitational 
case by Sonnleitner and Barnett in \cite{sonnleitner18}, 
as well as in our published article \cite{schwartz19:AiG}.

\section{Introduction}

Motivated by inconsistencies in the usual approximative 
Galilei-relativistic description of quantum-optical 
interactions of atoms with light, which by an \emph{ad hoc} 
semi-classical argumentation are easily seen to be possibly 
resolved in a special-relativistic description, Sonnleitner 
and Barnett have developed in \cite{sonnleitner18} a fully 
systematic derivation of an `approximately relativistic' 
Hamiltonian describing a simple atom in an external 
electromagnetic field. It is the purpose of this chapter 
to extend this so as to also include gravity approximately, 
more precisely a post-Newtonian gravitational field as 
described by the Eddington--Robertson \acronym{PPN} metric. 
As discussed in the introduction, such a generalisation 
is, apart from its conceptual value, of immediate interest 
for describing and devising quantum-optical experiments 
in gravitational fields, e.g. in atom interferometry.

The greatest value of Sonnleitner and Barnett's basing 
their whole calculation in \cite{sonnleitner18} on a 
properly relativistic treatment of the situation (an atom 
interacting with an external electromagnetic field) can 
be seen in allowing a systematic derivation of a complete 
description without any \emph{ad hoc} assumptions. In the 
end, the first-order post-Newtonian Hamiltonian they 
obtained could \emph{then} be used to interpret aspects 
of the situation in terms of classical `relativistic 
corrections'. For example, the `centre of mass' part of 
the final Hamiltonian has the form of a single-particle 
kinetic Hamiltonian, where the rôle of the rest mass of 
this particle is played by the total mass-energy of the 
atom, i.e.\ the sum of the rest masses of the constituent 
particles and the internal atomic energy divided by 
$c^2$. Thus, the computation in \cite{sonnleitner18} 
explicitly shows that this physically intuitive picture 
of a `composite particle', suggested by mass--energy 
equivalence, can, in fact, be derived in a controlled 
and systematic approximation scheme, rather than merely 
made plausible from semi-intuitive physical considerations.

As will be shown by our calculations, a similar 
interpretation is possible for the situation including 
external gravitational fields: when expressing the final 
Hamiltonian using the physical spacetime metric, an 
intuitive `composite point particle' picture including 
the `mass defect' due to mass--energy equivalence will 
again be available for the centre of mass dynamics. This 
lends justification based on detailed calculations within 
systematic approximation schemes to some of the naiver 
approaches that are based on \emph{a priori} assumptions 
concerning the gravity--matter coupling.

In section \ref{sec:situation_CS}, we set up the background 
for our calculations: after describing the physical 
system under consideration, we will give a somewhat 
detailed exposition of the method of computation in 
\cite{sonnleitner18}, in which we will also address an 
inconsistency of the original approach. Then we will 
discuss how our geometric post-Newtonian expansion 
framework introduced in chapter \ref{chap:geometric_structures} 
allows us to develop our gravitational calculation in 
parallel to that from \cite{sonnleitner18}.

In the following, we will compute in detail the
`gravitational corrections' to the calculation by 
Sonnleitner and Barnett \cite{sonnleitner18} arising 
from the presence of the gravitational field. Section 
\ref{sec:coupling_GF_particles} will deal with the coupling 
of the gravitational field to the kinetic terms of the 
particles only, ignoring couplings of the gravitational 
to the electromagnetic field.\looseness-1

In section \ref{sec:coupling_GF_EMF}, we will then compute 
the Lagrangian of the electromagnetic field in the presence 
of the gravitational field. This allows us to compute the 
total Hamiltonian describing the atomic system in section 
\ref{sec:total_Hamiltonian_comp}, by repeating the calculation 
from section \ref{sec:coupling_GF_particles} while including 
the `gravitational corrections' to electromagnetism as 
obtained in section~\ref{sec:coupling_GF_EMF}. The 
resulting Hamiltonian will then be interpreted in terms of 
the physical spacetime metric and compared to earlier 
results in the remainder of section \ref{sec:total_Hamiltonian}.

In sections \ref{sec:coupling_GF_particles} and 
\ref{sec:total_Hamiltonian_comp}, we will very closely follow the 
calculation from and presentation in \cite{sonnleitner18}. 
For the reader's convenience, we have reproduced all the 
relevant formulae from \cite{sonnleitner18} that are used 
in our calculation in section \ref{sec:work_sonnleitner_barnett}, 
in which we describe Sonnleitner and Barnett's work. We 
use the original numbering, prepended with 
`\cite{sonnleitner18}.', so for example 
\eqref{eq:Hamiltonian_com_cross_orig} refers to equation 
(25f) of \cite{sonnleitner18}. As some of the equations 
from \cite{sonnleitner18} contain minor errors (mostly 
sign errors), we here give corrected versions. The 
corresponding equation numbers are marked with a star, 
e.g.\ \eqref{eq:Hamiltonian_class_orig}.

A calculation using methods very similar to those of 
\cite{sonnleitner18} including external gravitational 
fields was performed by Marzlin already in 1995 
\cite{marzlin95}\footnote
	{I am grateful to Alexander Friedrich for pointing out 
	this reference to me.};
but unlike Sonnleitner and Barnett in \cite{sonnleitner18} 
or our calculation in the following, Marzlin did not 
perform a full first-order post-Newtonian expansion and 
instead focused on the electric dipole coupling only.

\section{A composite system in external electromagnetic and gravitational fields}
\label{sec:situation_CS}

We consider a simple system consisting of two 
particles without spin, with respective electric charges 
$e_1, e_2$, masses $m_1, m_2$, and spatial positions 
$\vect r_1, \vect r_2$. For simplicity we assume the 
charges to be equal and opposite, i.e.\ $e_2 = -e_1 =: e$. 
In what follows, we will take into account their mutual 
electromagnetic interaction, but neglect their mutual 
gravitational interaction. This two-particle system, which 
we will sometimes refer to as `atom', will be placed in 
an external electromagnetic field, which we will take into 
account, as well as an external gravitational field, which 
we will also take into account. It is our inclusion of the 
latter that extends the previous study \cite{sonnleitner18}.

\subsection{External electromagnetic fields -- the work of Sonnleitner and Barnett}
\label{sec:work_sonnleitner_barnett}

In \cite{sonnleitner18}, Sonnleitner and Barnett describe a 
systematic method to obtain an `approximately relativistic' 
quantum Hamiltonian for a system as described above 
interacting with an external electromagnetic field, where 
`approximately relativistic' refers to the inclusion of 
lowest order post-Newtonian correction terms, i.e.\ of 
order $c^{-2}$. Their work was motivated by their own 
observation \cite{sonnleitner17,sonnleitner18a} that the 
electromagnetic interaction of a decaying atom, which in 
\acronym{QED} follows an intrinsically special-relativistic 
symmetry (i.e.\ Poincaré invariance), will give rise to 
unnaturally looking friction-like terms that seem to 
contradict the relativity principle (which, of course, they 
don't) if interpreted in a `non-relativistic' (i.e.\ 
Galilei-invariant) setting of ordinary quantum mechanics. 
Their correct conclusion in \cite{sonnleitner18} was that 
this confusion can be altogether avoided by replacing this 
`hotchpotch' (their wording, see last line on p. 042106-9 
of \cite{sonnleitner18}) of symmetry concepts by a 
systematic post-Newtonian derivation starting from a 
common, manifestly Poincaré-symmetric description.

As our development will closely follow theirs, we will now 
describe the strategy of \cite{sonnleitner18} in some 
detail. In the course of this, we will also reproduce all 
formulae from \cite{sonnleitner18} that will be used in the 
remainder of this chapter. We use the original numbering, 
prepended with `\cite{sonnleitner18}.'. For formulae 
containing errors in \cite{sonnleitner18} (mostly sign 
errors), we give here a corrected version; the corrections 
are highlighted in \red{red} and the number is marked 
with a star. In addition to that, there is a conceptual 
inconsistency in the treatment in \cite{sonnleitner18} that 
we will address below. This will necessitate some 
further (rather small) amendments to the equations, which 
will be marked in the same way as the other errors.
\enlargethispage{-\baselineskip}

Sonnleitner and Barnett start with the classical 
Poincaré-invariant Lagrangian function describing two 
particles interacting with electromagnetic potentials\footnote
	{In the absence of gravity, as this is the situation considered in \cite{sonnleitner18}.}:
\begin{align*} \label{eq:Lagrangian_class_start_orig}
	L = &- \sum_{i=1,2} m_i c^2 \sqrt{1 - \dot{\vect r}_i^2 / c^2} 
		+ \int \D^3\ivect x \, (\vect j \cdot \vect A_\text{tot.} - \rho \phi_\text{el.,tot.}) \\
	&+ \frac{\varepsilon_0}{2} \int \D^3\ivect x \, [(\partial_t \vect A_\text{tot.} + \vect \nabla \phi_\text{el.,tot.})^2 
		- c^2 (\vect\nabla \times \vect A_\text{tot.})^2]. \SBtag{4}
\end{align*}
Note that we have changed the variable name of the total 
electric potential to $\phi_\text{el.,tot.}$ so as to avoid 
confusion with the Newtonian gravitational potential 
$\phi$ from the Eddington--Robertson \acronym{PPN} metric. 
$\vect j$ denotes the electric current density of the 
particles, and $\rho$ the charge density.

Sonnleitner and Barnett then split the electromagnetic 
potentials into `internal' (i.e.\ generated by the 
particles) and `external' parts, employ the Coulomb gauge, 
and solve the Maxwell equations for the internal part in 
lowest order, expressing the solutions in terms of the 
particles' positions and velocities (see the solutions in 
\eqref{eq:el_pot_orig} and \eqref{eq:mag_pot_orig} at the 
end of this section). The total vector potential, which is 
a transverse field (in the Helmholtz decomposition) due to 
the gauge condition, is split as $\vect A_\text{tot.}^\perp 
= \vect A^\perp + \vect{\mathcal A}^\perp$ 
where $\vect A^\perp$ denotes the external and 
$\vect{\mathcal A}^\perp$ the internal part. Due to the 
absence of external electric charges and the gauge 
condition, the external electric potential vanishes, such 
that $\phi_\text{el.,tot.} = \phi_\text{el.}$ is purely 
internal.

The idea is now to insert the solutions for the internal 
potentials into the Lagrangian \eqref{eq:Lagrangian_class_start_orig} 
and expand the kinetic terms for the particles, so as to 
obtain a post-Newtonian Lagrangian on which to base the 
further derivation. However, at this stage an inconsistency 
is introduced into the framework of \cite{sonnleitner18}, 
which we are now going to explain. Sonnleitner and Barnett 
want to keep the external vector potential $\vect A^\perp$ 
as a dynamical variable; as such, its equations of motion 
have to be the vacuum Maxwell equations (i.e.\ without any 
source term), while it still has to enter the equations 
of motion of the particles themselves. This is indeed the 
case for the Lagrangian which arises from directly inserting 
the internal potentials as obtained by solving the 
Maxwell equations: variation of the action given by 
this Lagrangian leads to Euler--Lagrange equations 
just as desired. This Lagrangian, however, contains 
second-order time derivatives of the particle positions, 
spoiling the application of conventional Hamiltonian 
formalism. This problem does not show up when following 
Sonnleitner and Barnett, since the problematic terms 
are related to formally diverging backreaction terms 
and are therefore disregarded from the Lagrangian in 
\cite{sonnleitner18}. However, \emph{this} last neglection 
is problematic if one keeps the external vector potential 
$\vect A^\perp$ as dynamical: the just-eliminated terms 
would have been the ones ensuring the vacuum Maxwell 
equations as equations of motion for the external 
potential~-- without them, the Lagrangian gives, again, 
the \emph{sourced} Maxwell equations for the external 
potential, and the formalism becomes inconsistent. This 
inconsistency was not addressed in \cite{sonnleitner18}, 
and we were also not aware of it at the time of publication 
of our article \cite{schwartz19:AiG}.
\enlargethispage{-\baselineskip}

However, as it turns out, there is a very easy way to 
remedy this problem: we proceed almost exactly like 
Sonnleitner and Barnett did, the only difference being that 
we remove the external vector potential $\vect A^\perp$ 
from its role as dynamical degree of freedom, treating it 
as a given external field instead (satisfying the vacuum 
Maxwell equations). This way we can ensure the consistency 
of the equations of motion while still performing the 
internal--external field split\footnote
	{By employing some form of perturbation theory on a given 
	non-zero classical background, as is sometimes used in 
	quantum optics, it is probably possible to 
	render the split into internal and external fields 
	consistent while still keeping \emph{some} electromagnetic~/ 
	photonic degrees of freedom as dynamical variables. 
	However, I (the author) am not well enough acquainted 
	with such techniques --~I 
	myself being, more or less, a classical relativist~-- and 
	thus restrict to those parts of the argumentation which I 
	am confident of. If such a perturbation-theoretic treatment 
	is indeed possible, it should be easily applicable to the 
	results we will derive below.}.
The one point in Sonnleitner and Barnett's derivation where 
one might be questioning if it still works without the 
electromagnetic field being a dynamical variable, namely 
the \acronym{PZW} transformation, will turn out to still be 
applicable just fine when reinterpreted in the right way, 
as we will explain below. Note that although the external 
field is eliminated as a dynamical variable, when Legendre 
transforming the Lagrangian in order to go over to the 
Hamiltonian formalism, we are going to add a term 
corresponding to the external field to the resulting 
Hamiltonian, such as to stay as close as possible to the 
original work of \cite{sonnleitner18}, and to obtain the 
correct value for the energy, including the external field 
energy\footnote
	{And to make our results as easily amenable as possible to 
	a potential perturbation-theoretic treatment~/ interpretation 
	as alluded to in the previous footnote.}.

Inserting the internal potential solutions and expanding 
the kinetic terms for the particles to order $c^{-2}$ 
(disregarding the rest energy term), as well as neglecting 
electromagnetic terms of order $\Or(c^{-4})$ and dropping 
terms related to formally diverging backreaction terms, 
one arrives at the post-Newtonian Lagrangian
\begin{align*}
	L\red{(\vect r_1, \dot{\vect r}_1, \vect r_2, \dot{\vect r}_2)} 
	&= L_\text{Darwin}(\vect r_1, \dot{\vect r}_1, \vect r_2, \dot{\vect r}_2) 
		+ \frac{\varepsilon_0}{2} \int \D^3\ivect x \, [(\partial_t \vect A^\perp)^2
			\\&\qquad - c^2 (\vect\nabla \times \vect A^\perp)^2] 
		+ \int \D^3\ivect x \, \vect j \cdot \vect A^\perp 
	\; , \SBtagc{8} \label{eq:Lagrangian_class_postNewt_orig} \\
	L_\text{Darwin}(\vect r_1, \dot{\vect r}_1, \vect r_2, \dot{\vect r}_2) 
	&= \frac{m_1 \dot{\vect r}_1^2}{2} 
		+ \frac{m_1 \dot{\vect r}_1^4}{8c^2} 
		+ \frac{m_2 \dot{\vect r}_2^2}{2} 
		+ \frac{m_2 \dot{\vect r}_2^4}{8c^2}
		\\&\quad - \frac{1}{4\pi\varepsilon_0} \, \frac{e_1 e_2}{r} \left(1 
			- \frac{\dot{\vect r}_1 \cdot \dot{\vect r}_2}{2c^2}\right) 
		+ \frac{e_1 e_2}{4\pi\varepsilon_0} \, \frac{(\dot{\vect r}_1 \cdot \vect r) (\dot{\vect r}_2 \cdot \vect r)}{2 r^3 c^2} 
	\; , \SBtag{9} \label{eq:Lagrangian_class_Darwin_orig}
\end{align*}
where $\vect r = \vect r_1 - \vect r_2$ and 
$r = |\vect r|$. Note that here, as explained above, 
$\vect A^\perp$ is treated as a given external field that 
appears in the Lagrangian, not a dynamical variable. 
$L_\text{Darwin}$ is the famous Darwin Lagrangian 
\cite{darwin20}, involving `correction terms' to the 
Coulomb potential arising from the internal atomic motion.

This classical Lagrangian is then Legendre transformed to 
obtain a classical Hamiltonian. As explained above, in 
order to get the correct value for the energy, including 
the external field energy, we add a term as one would 
obtain when Legrendre transforming also with respect to 
the external field, even though it is not a dynamical 
variable. We also use the notation $\vect\Pi^\perp 
= \varepsilon_0 \partial_t \vect A^\perp$ for the 
`would-be canonical momentum' conjugate to the external 
field, but have to keep in mind that it is a fixed field, 
not a real momentum conjugate to any configuration 
variable. As would be the case for a `true' electromagnetic 
canonical momentum, $-\vect\Pi^\perp / \varepsilon_0 
= -\partial_t \vect A^\perp = \vect E^\perp$ is, physically 
speaking, the external electric field.

Keeping these caveats in mind, the classical Hamiltonian 
reads
\begin{align*} \label{eq:Hamiltonian_class_orig}
	H &= \frac{\bar{\vect p}_1^2}{2 m_1} 
		\redbin{-} \frac{\bar{\vect p}_1^4}{8 m_1^3 c^2} 
		+ \frac{\bar{\vect p}_2^2}{2 m_2} 
		\redbin{-} \frac{\bar{\vect p}_2^4}{8 m_2^3 c^2} 
		+ \frac{1}{4\pi\varepsilon_0} \, \frac{e_1 e_2}{r} 
			\left(1 - \frac{\bar{\vect p}_1 \cdot \bar{\vect p}_2}{2 m_1 m_2 c^2}\right)
		\\&\quad - \frac{e_1 e_2}{4\pi\varepsilon_0} \, \frac{(\bar{\vect p}_1 \cdot \vect r)(\bar{\vect p}_2 \cdot \vect r)}{2 r^3 c^2 \red{m_1 m_2}} 
		+ \frac{\varepsilon_0}{2} \int \D^3\ivect x \, [(\vect\Pi^\perp / \varepsilon_0)^2 
			+ c^2(\vect\nabla \times \vect A^\perp)^2] 
		, \SBtagc{12}
\end{align*}
where $\bar{\vect p}_i = \vect p_i 
\redbin{-} e_i \vect A^\perp(\vect r_i)$ \corr.

This classical Hamiltonian is now canonically quantised 
to obtain a quantum Hamiltonian in what Sonnleitner and 
Barnett call the `minimal coupling form'. They then 
perform a Power--Zienau--Woolley (\acronym{PZW}) unitary 
transformation \cite{power59, woolley71, babiker83} 
together with a multipolar expansion of the external field 
in order to transform the Hamiltonian into a so-called 
`multipolar form'. The details of this, including the 
neccessary amendments due to $\vect A^\perp$ no longer 
being a dynamical field, are as follows.

The \acronym{PZW} transformation operator is
\begin{equation*} \SBtagc{14} \label{eq:PZW}
	U = \mathrm{e}^{-\I\Lambda} 
	= \exp\left[\redbin{-} \frac{\I}{\hbar} \int\D^3\ivect x \, 
		\vect{\mathcal P}(\vect x, t) \cdot \vect A^\perp(\vect x, t)\right] ,
\end{equation*}
where $\vect{\mathcal P}$ is the polarisation field
\begin{equation*} \SBtag{15}
	\vect{\mathcal P}(\vect x, t) = \sum_{i=1,2} e_i [\vect r_i(t) - \vect R(t)] 
		\int_0^1 \D\lambda \, \delta \big(\vect x - \vect R(t) 
			- \lambda [\vect r_i(t) - \vect R(t)] \big).
\end{equation*}
The transformation amounts to the following change of 
canonical momenta:
\begin{equation*} \SBtag{19a}
	\vect p_i \to U \vect p_i U^\dagger = \vect p_i + \hbar \vect\nabla_{\vect r_i} \Lambda
\end{equation*}
Since we treat the external field as non-dynamical, none 
of the variables corresponding to it change under the 
transformation. However, to reflect the change that 
\emph{would} happen if $\vect A^\perp$ were still a 
dynamical field\footnote
	{Again with the intent of staying as close as possible to 
	the original work \cite{sonnleitner18}, and to allow a 
	possible perturbation-theoretic reinterpretation.},
we introduce the notation $\tilde{\vect\Pi}^\perp 
:= \vect\Pi^\perp - \vect{\mathcal P}^\perp$ for the 
`would-be canonical field momentum' after the \acronym{PZW} 
transformation, amounting to the change 
\begin{equation*} \SBtagc{19b}
	\vect\Pi^\perp(\vect x) \to \red{\tilde{\textcolor{black}{\vect\Pi}}}^\perp(\vect x) \redbin{+} \vect{\mathcal P}^\perp(\vect x)
\end{equation*}
in the Hamiltonian. Physically, in line with the usual 
interpretation for the canonical field momentum after 
a \acronym{PZW} transformation \cite{babiker83}, 
$-\tilde{\vect\Pi}^\perp = -\vect\Pi^\perp 
+ \vect{\mathcal P}^\perp(\vect x) = \vect D^\perp$ 
is the electric displacement field. Note 
that in \cite{sonnleitner18}, the somewhat misleading 
notation $\vect E^\perp$ is used for the quantity 
`$-(\text{external field momentum after \acronym{PZW} 
trafo})^\perp / \varepsilon_0$', as if it corresponded to 
an electric field proper.

In electric dipole approximation, i.e.\ expanding to 
first order in $\bar{\vect r}_i := \vect r_i - \vect R$, 
and using $\sum_{j=1,2} e_j = 0$, one finds (see 
\cite{sonnleitner18} for details)
\begin{equation*} \SBtagc{21} \label{eq:SB_21}
	\hbar \vect\nabla_{\vect r_{1,2}} \Lambda 
	\simeq \red{e_{1,2}} [\vect A^\perp(\vect R) 
			+ (\bar{\vect r}_{1,2} \cdot \vect\nabla) \vect A^\perp(\vect R)]
		+ \frac{e_1 \vect r_1 + e_2 \vect r_2}{2} 
			\times [\vect\nabla \times \vect A^\perp(\vect R)]. 
\end{equation*}
Thus, under the \acronym{PZW} transformation and 
the dipole approximation the momenta transform as 
$\vect p_i \redbin{-} e_i \vect A(\vect r_i) 
\to \vect p_i + \vect d \times \vect B(\vect R) / 2$ \corr, 
where $\vect d = e_1 \vect r_1 + e_2 \vect r_2$ is the 
electric dipole moment.

Terms of the form
\begin{equation*} \SBtag{22} \label{eq:SB_negl_interaction}
	\frac{\vect p_i \cdot [\vect d \times \vect B(\vect R)]}{m_i m_j c^2} 
	\propto \frac{|\vect p_i|}{m_i c} \, \frac{|\vect d \cdot \vect E(\vect R)|}{m_j c^2}
\end{equation*}
are neglected, since the atom--light interaction energy 
is assumed much smaller than the internal atomic energy, 
which is in turn much smaller than the rest energies of 
the particles. The multipolar Hamiltonian in electric 
dipole approximation is then
\begin{align*} \label{eq:Hamiltonian_mult_orig}
	H_\text{[mult]} 
	&\simeq \frac{[\vect p_1 + \frac{1}{2} \vect d \times \vect B(\vect R)]^2}{2 m_1} 
		+ \frac{[\vect p_2 + \frac{1}{2} \vect d \times \vect B(\vect R)]^2}{2 m_2}
		\\&\quad - \frac{e^2}{4\pi\varepsilon_0 r} 
		+ \frac{\varepsilon_0}{2} \int \D^3\ivect x \, [ (\red{\tilde{\textcolor{black}{\vect\Pi}}}^\perp 
				\redbin{+} \vect{\mathcal P}_d^\perp)^2 / \varepsilon_0^2 
			+ c^2 \vect B^2]
		\\&\quad \redbin{-} \frac{\vect p_1^4}{8 m_1^3 c^2} 
		\redbin{-} \frac{\vect p_2^4}{8 m_2^3 c^2} 
		+ \frac{e^2}{16\pi\varepsilon_0 c^2 m_1 m_2}
			\\&\quad \times \left[ \vect p_1 \cdot \frac{1}{r} \vect p_2 
				+ (\vect p_1 \cdot \vect r) \frac{1}{r^3} (\vect r \cdot \vect p_2) 
				+ (1 \leftrightarrow 2) \right]
	, \SBtagc{23}
\end{align*}
where $\vect{\mathcal P}_d = \redbin{+} \vect d \, 
\delta(\vect x - \vect R)$ \corr\ is the polarisation 
in electric dipole approximation.

Then, introducing Newtonian centre of mass and relative 
coordinates $\vect R, \vect r$, and the corresponding 
canonical momenta $\vect P, \vect p_{\vect r}$, Sonnleitner 
and Barnett arrive at what they call the centre of mass 
Hamiltonian:
\begin{subequations}
	\makeatletter
	\def\@currentlabel{\cite{sonnleitner18}.25$\star$}
	\makeatother
	\label{eq:Hamiltonian_com_orig}
\begin{align*}
	H_\text{[com]} &= H_\text{C} + H_\text{A} + H_\text{AL} 
		+ H_\text{L} + H_\text{X} \SBtag{25a} \\
	H_\text{C} &= \frac{\vect P^2}{2M} \left[1 
		- \frac{\vect P^2}{4M^2 c^2} 
		- \frac{1}{M c^2} \left(\frac{\vect p_{\vect r}^{\red{2}}}{2\mu} 
			- \frac{e^2}{4\pi\varepsilon_0 r}\right)\right]
	\SBtagc{25b} \\
	H_\text{A} &= \frac{\vect p_{\vect r}^2}{2\mu} \left(1 
			- \frac{m_1^3 + m_2^3}{M^3} \,
			\frac{\vect p_{\vect r}^2}{4 \mu^2 c^2}\right) 
		- \frac{e^2}{4\pi\varepsilon_0}
		\\&\quad \times \left[\frac{1}{r} 
			+ \frac{1}{2\mu M c^2} \left( \vect p_{\vect r} \cdot \frac{1}{r} \vect p_{\vect r} 
				+ \vect p_{\vect r} \cdot \vect r \frac{1}{r^3} \vect r \cdot \vect p_{\vect r} \right) \right]
	\SBtag{25c} \\
	H_\text{AL} &= - \vect d \cdot \red{\frac{\vect D^\perp (\vect R)}{\varepsilon_0}} 
		+ \frac{1}{2M} \{\vect P \cdot [\vect d \times \vect B(\vect R)] + \text{H.c.}\}
		\\&\quad - \frac{m_1 - m_2}{\red{4} m_1 m_2} \{\vect p_{\vect r} \cdot [\vect d \times \vect B(\vect R)] + \text{H.c.}\}
		\\&\quad + \frac{1}{8\mu} (\vect d \times \vect B(\vect R))^2 
		+ \frac{1}{2\varepsilon_0} \int\D^3\ivect x \, {\vect{\mathcal P}_d^\perp}^2 (\vect x, t)
	\SBtagc{25d} \label{eq:Hamiltonian_com_AL_orig} \\
	H_\text{L} &= \frac{\varepsilon_0}{2} \int\D^3\ivect x \, 
		[\red{(\vect D^\perp / \varepsilon_0)}^2 + c^2 \vect B^2]
	\SBtagc{25e} \\
	H_\text{X} &= - \frac{(\vect P \cdot \vect p_{\vect r})^2}{2 M^2 \mu c^2} 
		+ \frac{e^2}{4\pi\varepsilon_0 r} \, \frac{(\vect P\cdot \vect r / r)^2}{2 M^2 c^2}
		\\&\quad + \frac{m_1 - m_2}{2\mu M^2 c^2} \bigg\{ (\vect P \cdot \vect p_{\vect r}) \vect p_{\vect r}^2 / \mu 
			- \frac{e^2}{8\pi\varepsilon_0}
				\\&\quad \times \left[\frac{1}{r} \vect P \cdot \vect p_{\vect r} 
				+ \frac{1}{r^3} (\vect P \cdot \vect r) (\vect r \cdot \vect p_{\vect r}) + \text{H.c.}\right] \bigg\}
	\SBtag{25f} \label{eq:Hamiltonian_com_cross_orig}
\end{align*}
\end{subequations}
Note that the Hamiltonian has been expressed in a form in 
which the external field enters in terms of the magnetic 
field $\vect B = \vect\nabla \times \vect A^\perp$ and the 
electric displacement field $\vect D^\perp = - \tilde{\vect\Pi}^\perp 
= - \varepsilon_0 \partial_t \vect A^\perp + \vect{\mathcal P}^\perp$ 
(which was, as mentioned above, a little misleadingly called 
$\varepsilon_0 \vect E^\perp$ in \cite{sonnleitner18}). The 
Hamiltonian is split into terms that may be interpreted as 
describing the central motion of the atom ($H_\text{C}$), 
the internal atomic motion ($H_\text{A}$), the interaction 
between the atom and the external (`light') field 
($H_\text{AL}$), and a term giving the external 
electromagnetic field energy ($H_\text{L}$), as well as 
`cross terms' ($H_\text{X}$) coupling the relative degrees 
of freedom to the central momentum $\vect P$.

In order to eliminate this cross-term coupling, Sonnleitner 
and Barnett perform a final canonical transformation to 
new coordinates $\vect Q, \vect q$ and momenta $\vect P, 
\vect p$, which leaves the Hamiltonian unchanged up to 
terms of order $c^{-4}$ except for elimination of the 
cross terms and the replacements $(\vect R, \vect r, 
\vect p_{\vect r}) \to (\vect Q, \vect q, \vect p)$. This 
canonical transformation reads as follows:
\begin{subequations}
	\makeatletter
	\def\@currentlabel{\cite{sonnleitner18}.26}
	\makeatother
	\label{eq:coords_decoup_orig}
\begin{align*}
	\vect R &= \vect Q + \frac{m_1 - m_2}{2 M^2 c^2} 
			\left[ \left(\frac{\vect p^2}{2\mu} \vect q + \text{H.c.}\right) 
		- \frac{e^2}{4\pi \varepsilon_0 q} \vect q \right]
		\\&\quad - \frac{1}{4 M^2 c^2} [(\vect q \cdot \vect P) \vect p 
			+ (\vect P \cdot \vect p) \vect q + \text{H.c.}] 
	\SBtag{26a} \\
	\vect r &= \vect q + \frac{m_1 - m_2}{2 \mu M^2 c^2} [(\vect q \cdot \vect P) \vect p + \text{H.c.}] - \frac{\vect q \cdot \vect P}{2 M^2 c^2} \vect P \SBtag{26b} \\
	\vect p_{\vect r} &= \vect p + \frac{\vect p \cdot \vect P}{2 M^2 c^2} \vect P - \frac{m_1 - m_2}{2 M^2 c^2}
		\\&\quad\times \left[ \frac{\vect p^2}{\mu} \vect P - \frac{e^2}{4\pi \varepsilon_0} \left(\frac{1}{q} \vect P - \frac{1}{q^3} (\vect P \cdot \vect q) \vect q\right) \right] \SBtag{26c}
\end{align*}
\end{subequations}

Finally, the internal electromagnetic potentials to our 
order of approximation (thus in particular neglecting 
retardation), as obtained by solving the internal Maxwell 
equations, are as follows:
\begin{align*}
	\phi_\text{el.,ng}(\vect x,t) &= \frac{1}{4\pi\varepsilon_0} \int\D^3\ivect x' \, \frac{\rho(\vect x',t)}{|\vect x - \vect x'|} \SBtag{A1} \label{eq:el_pot_orig} \\
	\vect{\mathcal A}^\perp_\text{ng}(\vect x,t) &\simeq \frac{1}{4\pi \varepsilon_0 c^2} \int\D^3\ivect x' \, \frac{\vect j(\vect x',t)}{|\vect x - \vect x'|} + \frac{1}{(4\pi)^2 \varepsilon_0 c^2} \int\D^3\ivect x'
		\\&\quad \times \int\D^3\ivect x'' \, \frac{\vect x - \vect x'}{|\vect x - \vect x'|^3} \, \frac{\vect j(\vect x'',t) \cdot (\vect x' - \vect x'')}{|\vect x' - \vect x''|^3}\\
		&= \frac{1}{8\pi \varepsilon_0 c^2} \sum_{i=1,2} e_i \left\{ \frac{\dot{\vect r}_i}{|\vect x - \vect r_i|} + \frac{(\vect x - \vect r_i) [\dot{\vect r}_i \cdot (\vect x - \vect r_i)]}{|\vect x - \vect r_i|^3} \right\} \SBtag{A3} \label{eq:mag_pot_orig}
\end{align*}
Here we have changed the variable names of the potentials 
to conform to our notation --~in particular we added the 
suffix `ng', standing for `non-gravitational'~-- and 
expressed the magnetic potential in terms of $\varepsilon_0$ 
instead of $\mu_0 = 1/(\varepsilon_0 c^2)$.

\subsection{Including weak external gravitational fields}
\label{sec:including_gravity}

As already stated above, our contribution in this 
chapter will consist in generalising the calculation 
of \cite{sonnleitner18} to the case of the atom 
being situated in a weak external gravitational 
field in addition to the electromagnetic field already 
considered in  \cite{sonnleitner18}. Our aim is to 
likewise obtain an `approximately relativistic', 
i.e.\ first-order post-Newtonian, Hamiltonian 
describing this situation. The gravitational field will 
be described by the Eddington--Robertson \acronym{PPN} 
metric as introduced in section \ref{sec:ER_PPN_metric}.

Our post-Newtonian expansion scheme as laid out in chapter 
\ref{chap:geometric_structures}, based on the introduction 
of geometric background structures that give meaning 
to `weak' gravitational fields and `slow' velocities 
in the setting of a non-flat spacetime, provides the 
conceptual and computational basis which will allow us 
to implement the post-Newtonian expansion employed in 
\cite{sonnleitner18} also in the gravitational case. 
This enables us to develop our calculation in great 
parallel with that of \cite{sonnleitner18}: we use the 
`flat' background structure to perform our computations, 
the benefit being the aimed-for direct comparison with 
\cite{sonnleitner18}. In the course of our derivation, 
`gravitational correction terms' to the non-gravitational 
formulae will show up. However, as already alluded to in 
section \ref{sec:background_structures}, it often is of 
great physical value to re-express the obtained results in 
terms of the physical metric $g$ instead of the background 
metric $\eta$. For example, the results will contain 
geometric operations, like scalar products, which may be 
taken using either of the metric structures provided by 
the formalism. What may at first appear as a more or less 
complicated gravitational correction to the flat space 
result will often, in fact, turn out to be a simple and 
straightforward transcription of the latter into the proper 
physical metric, as one might have anticipated from some 
more or less naive working-version of the equivalence 
principle. Interpretational issues like this are well-known 
in the literature on gravitational couplings of quantum
systems; see, e.g., \cite{marzlin95,laemmerzahl95}. For 
us, too, they will once more turn out to be relevant 
in connection with the total Hamiltonian in section 
\ref{sec:total_Hamiltonian}. We will derive and interpret 
the relevant gravitational terms relative to the background 
structures $(\eta,u)$ in order to keep the analogy with the 
computation in \cite{sonnleitner18}, but then we shall 
re-interpret the results in terms of the proper physical 
metric $g$ in order to reveal their naturalness.

Since we are interested in a lowest-order post-Newtonian 
description, we will work up to (and including) terms of 
order $c^{-2}$ and neglect higher order terms. In fact, 
corrections of higher order cannot be treated in a simple 
Hamiltonian formalism as employed here, without explicitly 
including the internal electromagnetic field degrees 
of freedom as dynamical variables: elimination of the 
internal field variables by solving Maxwell's equations 
will introduce retardation effects at higher orders, thus 
leading to an action that is non-local in time, spoiling 
the application of conventional Hamiltonian formalism.

\section{Coupling the gravitational field to the particles}
\label{sec:coupling_GF_particles}

In this section we will work out the influence of the 
gravitational field when coupled to the kinetic terms 
of the particles only, ignoring its couplings to the 
electromagnetic field. The latter will be the subject 
of the following sections.

Starting from the Lagrangian for our atom in the 
\emph{absence} of gravity and adding the `gravitational 
corrections' to the kinetic terms of the particles, we 
will then repeat the calculation of \cite{sonnleitner18} to 
obtain a quantum Hamiltonian in centre of mass coordinates.

\subsection{The classical Hamiltonian}

For a single free point particle with mass $m$ and 
position $\vect x$, the classical kinetic Lagrangian 
(parametrising the worldline by coordinate time) in our 
metric \eqref{eq:ER_PPN_metric} reads
\begin{align}\label{eq:Lagrangian_point_grav}
	L_\text{point} 
		&= -mc^2 \sqrt{-g_{\mu\nu} \dot x^\mu \dot x^\nu/c^2} \nonumber\\
	&= \frac{m \dot{\vect x}^2}{2} \left(1 + \frac{\dot{\vect x}^2}{4 c^2}\right) 
		- m c^2 - m\phi \left(1 + (2\beta-1) \frac{\phi}{2c^2}\right) 
		- \frac{2\gamma + 1}{2} \, \frac{m\phi}{c^2} \dot{\vect x}^2 
		+ \Or(c^{-4}).
\end{align}
Now considering our two-particle system, the kinetic 
terms for the particles in gravity are given as the sum 
of two terms as in \eqref{eq:Lagrangian_point_grav}. 
These lowest-order `gravitationally corrected' kinetic 
terms we include into the classical Lagrangian from 
\eqref{eq:Lagrangian_class_start_orig}\footnote
	{We remind the reader that all the equations from 
	\cite{sonnleitner18} that we refer to explicitly are 
	reproduced in section \ref{sec:work_sonnleitner_barnett}.},
which described two particles interacting with an 
electromagnetic field in the \emph{absence} of gravity.

Eliminating the internal electromagnetic fields literally 
as in the non-gravitational case, we arrive at the 
post-Newtonian classical Lagrangian
\begin{align}
	L_\text{new} &= L - m_1 \phi(\vect r_1) - m_2 \phi(\vect r_2) 
		- \frac{2\gamma + 1}{2} \, \frac{m_1\phi(\vect r_1)}{c^2} \, \dot{\vect r}_1^2 
		- \frac{2\gamma + 1}{2} \, \frac{m_2\phi(\vect r_2)}{c^2} \dot{\vect r}_2^2 \nonumber
		\\&\quad - (2\beta-1) \frac{m_1 \phi(\vect r_1)^2}{2c^2} 
		- (2\beta-1) \frac{m_2 \phi(\vect r_2)^2}{2c^2}
\end{align}
describing our electromagnetically bound two-particle 
system in the given external electromagnetic field. Here 
$L$ is the final classical Lagrangian from 
\eqref{eq:Lagrangian_class_postNewt_orig} and 
\eqref{eq:Lagrangian_class_Darwin_orig}. Note that, as 
explained in section \ref{sec:work_sonnleitner_barnett}, 
for reasons of consistency, we view the external vector 
potential as a given background field, not as a dynamical 
variable.

Legendre transforming this Lagrangian with respect to the 
particle velocities $\dot{\vect r}_i$ and adding a term 
as one would obtain when Legrendre transforming also with 
respect to the external electromagnetic vector potential 
(see section \ref{sec:work_sonnleitner_barnett} before 
\eqref{eq:Hamiltonian_class_orig}), we obtain the total 
classical Hamiltonian
\begin{align} \label{eq:Hamiltonian_class_no_em}
	H_\text{new} &= H + m_1 \phi(\vect r_1) + m_2 \phi(\vect r_2) 
		+ \frac{2\gamma+1}{2 m_1 c^2} \phi(\vect r_1) \bar{\vect p}_1^2 
		+ \frac{2\gamma+1}{2 m_2 c^2} \phi(\vect r_2) \bar{\vect p}_2^2 \nonumber
		\\&\quad+ (2\beta-1) \frac{m_1 \phi(\vect r_1)^2}{2c^2} 
		+ (2\beta-1) \frac{m_2 \phi(\vect r_2)^2}{2c^2}.
\end{align}
Here $H$ is the classical Hamiltonian from 
\eqref{eq:Hamiltonian_class_orig} and $\bar{\vect p}_i 
= \vect p_i - e_i \vect A^\perp(\vect r_i)$ is the 
kinetic momentum. Note that we dropped all terms that 
go beyond our order of approximation.
\enlargethispage{\baselineskip}

\subsection{Canonical quantisation and \acronym{PZW} transformation to a multipolar Hamiltonian}

Now, we canonically quantise this Hamiltonian and perform 
the \acronym{PZW} transformation and electric dipole 
approximation used in \cite{sonnleitner18} to arrive 
at the `multipolar' Hamiltonian from 
\eqref{eq:Hamiltonian_mult_orig}. 
Neglecting terms of the form 
$\frac{\vect p_i \cdot [\vect d \times \vect B(\vect R)]}{m_i m_j c^2}$ 
as in \eqref{eq:SB_negl_interaction}, in our gravitational 
correction terms from \eqref{eq:Hamiltonian_class_no_em} 
these transformations amount just to the replacement 
$\bar{\vect p}_i \to \vect p_i$ (compare section 
\eqref{sec:work_sonnleitner_barnett} from \eqref{eq:PZW} 
to \eqref{eq:SB_21}). Hence the multipolar Hamiltonian 
including the gravitational correction terms is
\begin{align} \label{eq:Hamiltonian_mult_no_em}
	H_\text{[mult],new} 
		&= H_\text{[mult]} + m_1 \phi(\vect r_1) + m_2 \phi(\vect r_2) 
		+ \frac{2\gamma+1}{2 m_1 c^2} \vect p_1 \cdot \phi(\vect r_1) \vect p_1 
		+ \frac{2\gamma+1}{2 m_2 c^2} \vect p_2 \cdot \phi(\vect r_2) \vect p_2 \nonumber
		\\&\quad+ (2\beta-1) \frac{m_1 \phi(\vect r_1)^2}{2c^2} 
		+ (2\beta-1) \frac{m_2 \phi(\vect r_2)^2}{2c^2},
\end{align}
where $H_\text{[mult]}$ is the multipolar Hamiltonian 
from \eqref{eq:Hamiltonian_mult_orig}.

Now that we are on the quantum level, we had to choose 
a symmetrised operator ordering for the $\vect p^2\phi$ 
terms. We chose an ordering of the `obvious' form 
$\vect p \cdot \phi \vect p$. As we have seen in section 
\ref{sec:WKB_ER_PPN}, this operator ordering also results 
from the description of single quantum particles in 
an Eddington--Robertson \acronym{PPN} metric by our 
\acronym{WKB}-like expansion of the Klein--Gordon equation, 
if we neglect terms proportional to $\Delta\phi$ (which 
vanish outside the matter generating the Newtonian 
potential $\phi$, and thus are irrelevant in physical 
situations concerning the outside of the generating matter 
distribution).

\subsection{Introduction of centre of mass variables}

We now want to express the correction terms in (Newtonian) 
centre of mass and relative variables,
\begin{align}
	\vect R &= \frac{m_1 \vect r_1 + m_2 \vect r_2}{M} \; , 
	& \vect r &= \vect r_1 - \vect r_2 \; ,\\
	\vect P &= \vect p_1 + \vect p_2 \; , 
	& \vect p_{1,2} &= \frac{m_{1,2}}{M} \vect P \pm \vect p_{\vect r} \; ,
\end{align}
where $M = m_1 + m_2$. To this end, we expand the 
gravitational potential $\phi$ around the centre of 
mass position $\vect R$ \emph{in linear order}. In this 
approximation, we have 
$m_1 \phi(\vect r_1) + m_2 \phi(\vect r_2) = M \phi(\vect R)$ and 
$m_1 \phi(\vect r_1)^2 + m_2 \phi(\vect r_2)^2 = M \phi(\vect R)^2$. 
Furthermore using
\begin{align}
	\vect p_{1,2} \cdot \phi(\vect r_{1,2}) \vect p_{1,2} 
	&= \left(\frac{m_{1,2}}{M} \vect P \pm \vect p_{\vect r}\right) \cdot 
		\phi(\vect r_{1,2}) 
		\left(\frac{m_{1,2}}{M} \vect P \pm \vect p_{\vect r}\right) \nonumber \\
	&= \frac{m_{1,2}^2}{M^2} \vect P \cdot \phi(\vect r_{1,2}) \vect P 
		\pm \frac{m_{1,2}}{M} (\vect P \cdot \phi(\vect r_{1,2}) \vect p_{\vect r} + \text{H.c.}) 
		+ \vect p_{\vect r} \cdot \phi(\vect r_{1,2}) \vect p_{\vect r}
\end{align}
and the relations $\phi(\vect r_1) - \phi(\vect r_2) 
= \vect r \cdot \vect\nabla\phi(\vect R)$ as well as
\begin{align}
	\frac{1}{m_1} \phi(\vect r_1) + \frac{1}{m_2} \phi(\vect r_2) 
	&= \left(\frac{1}{m_1} + \frac{1}{m_2}\right) \phi(\vect R) 
		+ \frac{1}{M} \left(\frac{m_2}{m_1} 
			- \frac{m_1}{m_2}\right) 
			\vect r \cdot \vect\nabla\phi(\vect R) \nonumber \\
	&= \frac{1}{\mu} \phi(\vect R) 
		- \frac{m_1 - m_2}{m_1 m_2} \vect r \cdot \vect\nabla\phi(\vect R)
\end{align}
where $\mu = \frac{m_1 m_2}{M}$ is the system's reduced 
mass, we arrive at the centre of mass Hamiltonian
\begin{align} \label{eq:Hamiltonian_com_no_em}
	H_\text{[com],new} &= H_\text{[com]} + M \phi(\vect R) 
		+ (2\beta-1) \frac{M \phi(\vect R)^2}{2 c^2} 
		+ \frac{2\gamma+1}{2Mc^2} \vect P \cdot \phi(\vect R) \vect P \nonumber
		\\&\quad+ \frac{2\gamma+1}{2\mu c^2} \vect p_{\vect r}^2 \phi(\vect R) 
		+ \frac{2\gamma+1}{2Mc^2} \left[\vect P \cdot (\vect r \cdot \vect\nabla\phi(\vect R)) \vect p_{\vect r} + \text{H.c.}\right] \nonumber
		\\&\quad- \frac{2\gamma+1}{2c^2} \, \frac{m_1 - m_2}{m_1 m_2} \vect p_{\vect r} \cdot (\vect r \cdot \vect\nabla\phi(\vect R)) \vect p_{\vect r} \; ,
\end{align}
where $H_\text{[com]}$ is the centre of mass Hamiltonian 
from \eqref{eq:Hamiltonian_com_orig}.
\enlargethispage{\baselineskip}

This can, similarly to \cite{sonnleitner18}, be brought 
into the form
\begin{equation}
	H_\text{[com],new} = H_\text{C,new} + H_\text{A,new} 
		+ H_\text{AL} + H_\text{L} + H_\text{X} + H_\text{deriv,new},
\end{equation}
where
\begin{equation}
	H_\text{C,new} = H_\text{C} 
		+ \frac{2\gamma+1}{2Mc^2} \vect P \cdot \phi(\vect R) \vect P 
		+ \left(M + \frac{\vect p_{\vect r}^2}{2\mu c^2}\right) \phi(\vect R) 
		+ (2\beta-1) \frac{M \phi(\vect R)^2}{2 c^2}
\end{equation}
describes the dynamics of the centre of mass and
\begin{equation} \label{eq:Hamiltonian_com_atom_no_em}
	H_\text{A,new} = H_\text{A} 
		+ 2\gamma\frac{\phi(\vect R)}{c^2} \frac{\vect p_{\vect r}^2}{2\mu} 
		- \frac{2\gamma+1}{2c^2} \, \frac{m_1 - m_2}{m_1 m_2} \vect p_{\vect r} 
			\cdot (\vect r \cdot \vect\nabla\phi(\vect R)) \vect p_{\vect r}
\end{equation}
describes the internal dynamics of the atom, both 
modified in comparison to \cite{sonnleitner18}. Here, 
we have included the term 
$2\gamma\frac{\phi(\vect R)}{c^2} \frac{\vect p_{\vect r}^2}{2\mu}$ 
into $H_\text{A,new}$ since it can be combined with 
$\frac{\vect p_{\vect r}^2}{2\mu}$ from $H_\text{A}$ into
\begin{equation} \label{eq:internal_kin_metr}
	\frac{\vect p_{\vect r}^2}{2\mu} \left(1 + 2\gamma\frac{\phi(\vect R)}{c^2}\right) 
	= \frac{{^{(3)}g^{-1}_{\vect R}} (\vect p_{\vect r}, \vect p_{\vect r})}{2\mu},
\end{equation}
giving the geometrically correctly expressed Newtonian 
internal kinetic energy, using the metric square of the 
internal momentum. Here $^{(3)}g_{\vect R}^{-1}$ denotes 
the inverse of the physical spatial metric at position 
$\vect R$, as explained in section \ref{sec:background_structures}.

The terms $H_\text{AL}$, $H_\text{L}$, and $H_\text{X}$ 
containing, respectively, the atom-light interaction terms, 
the external electromagnetic field energy, and the `cross 
terms' are not changed compared to \cite{sonnleitner18}. 
The new final summand
\begin{equation}
	H_\text{deriv,new} 
	= \frac{2\gamma+1}{2Mc^2} 
		[\vect P \cdot (\vect r 
			\cdot \vect\nabla\phi(\vect R)) \vect p_{\vect r} + \text{H.c.}]
\end{equation}
is an additional central--internal `cross term' involving 
the derivative $\vect\nabla\phi$ of the gravitational 
potential.

Note that if we assumed that the gravitational potential 
$\phi$ vary slowly over the extension of the atom, we could 
neglect the terms $\vect r \cdot \vect\nabla\phi(\vect R)$. 
However, such terms might turn out interesting for 
experimental applications employing large superpositions.

\section{Coupling the gravitational to the electromagnetic field}
\label{sec:coupling_GF_EMF}

Having determined the gravitational field's coupling to 
the particles in the previous section, we now turn to its 
coupling to the electromagnetic field, whose Lagrangian 
in the presence of gravity we will compute in this section. 
In the following section \ref{sec:total_Hamiltonian} we 
will then combine all couplings into a single Hamiltonian.

\subsection{Solution of the gravitationally modified Maxwell equations}

The electromagnetic part of the total action of our system, 
including interaction with matter, is
\begin{equation} \label{eq:em_action}
	S_\text{em} = \int\D t \, \D^3\ivect x \, \sqrt{-g} 
		\left(-\frac{\varepsilon_0 c^2}{4} F_{\text{tot.}\mu\nu} F_\text{tot.}^{\mu\nu} + J^\mu A_{\text{tot.}\mu}\right),
\end{equation}
where $g$ denotes the determinant of the matrix $(g_{\mu\nu})$ 
of metric components, $J = J^\mu \partial_\mu$ is 
the four-current `density' vector field, $A_\text{tot.} 
= A_{\text{tot.}\mu} \D x^\mu$ is the total (i.e. 
not decomposed into internal and external parts) 
electromagnetic four-potential form, and $\D A_\text{tot.} 
= F_\text{tot.} = F_{\text{tot.}\mu\nu} \D x^\mu \otimes\D x^\nu 
= (\partial_\mu A_{\text{tot.}\nu} - \partial_\nu A_{\text{tot.}\mu}) 
\D x^\mu \otimes\D x^\nu$ is the electromagnetic 
field tensor. This is the standard action describing 
electromagnetism in a gravitational field, which is 
obtained by minimally coupling the special-relativistic 
action for electromagnetism \cite{jackson98} to a general 
spacetime metric \cite{misner73,hawking73}.

Note that $J^\mu$ are the components of a proper vector 
field and not of a density; their relation to the 
four-current \emph{density} with components $j^\mu$, in 
terms of which the interaction part of the action takes the 
form $\int\D t \, \D^3\ivect x \, j^\mu A_{\text{tot.}\mu}$, is 
given by
\begin{equation}
	J^\mu = \frac{1}{\sqrt{-g}} j^\mu.
\end{equation}
The current density of our system of two particles is given 
by\footnote
	{For a single particle of charge $q$ on an arbitrarily 
	parametrised timelike worldline $r^\mu(\lambda)$, the 
	current density is given by
	\[j^\mu(x) = q c\int\D\lambda \, \frac{\D r^\mu}{\D\lambda} 
		\delta^{(4)}\bigl(x - r(\lambda)\bigr).\]
	Parametrising by coordinate time and considering two 
	particles, we arrive at the above expression.}
\begin{equation} \label{eq:current_dens}
	j^\mu(t,\vect x) = \sum_{i=1}^2 e_i 
		\delta^{(3)}(\vect x - \vect r_i(t)) \dot r_i^\mu(t),
\end{equation}
where the dot denotes differentiation with respect to 
coordinate time $t$. The charge density is
\begin{equation} \label{eq:charge_dens}
	\rho = \frac{1}{c} j^0.
\end{equation}
Similarly, the electric potential is
\begin{equation}
	\phi_\text{el.,tot.} = - c A_{\text{tot.}0}.
\end{equation}

The Maxwell equations obtained by varying the action 
with respect to $A_{\text{tot.}\mu}$ take the form
\begin{equation}
	\nabla_\mu F_\text{tot.}^{\mu\nu} 
	= - \frac{1}{\varepsilon_0 c^2} J^\nu
\end{equation}
in terms of the current vector field, or
\begin{equation} \label{eq:Maxwell_general_cov}
	\nabla_\mu F_\text{tot.}^{\mu\nu} 
	= - \frac{1}{\varepsilon_0 c^2} \, \frac{1}{\sqrt{-g}} j^\nu
\end{equation}
in terms of the current density. It will be useful to 
consider the form
\begin{equation} \label{eq:Maxwell_general}
	\nabla^\mu F_{\text{tot.}\mu\nu} 
	= - \frac{1}{\varepsilon_0 c^2} \, \frac{1}{\sqrt{-g}} j_\nu
\end{equation}
instead.

We employ the `background Coulomb gauge' condition
\begin{equation} \label{eq:gauge}
	0 = \vect\nabla \cdot \vect A_\text{tot.} 
	= \delta^{ab} \partial_a A_{\text{tot.}b},
\end{equation}
implying in particular $\delta^{ab} \partial_a F_{\text{tot.}b\mu} 
= \Delta A_{\text{tot.}\mu}$ where $\Delta = \delta^{ab} \partial_a \partial_b$ 
denotes the `flat' Euclidean Laplacian defined 
by the background structures. In terms of the 
Helmholtz decomposition, the gauge condition 
means $\vect A_\text{tot.}^\parallel = 0$, i.e.\ 
$\vect A_\text{tot.} = \vect A_\text{tot.}^\perp$.

\subsubsection{Divergence of the field strength tensor}

Using the Christoffel symbols of the Eddington--Robertson 
\acronym{PPN} metric, which are computed in full detail in 
appendix \ref{app:Christoffel}, we can now calculate the 
components of the divergence of the field strength tensor 
$F_\text{tot.}$. For the calculations, we remind the reader 
that the components of the field tensor are of the orders 
$F_{\text{tot.}a0} = \Or(c^{-1})$ and $F_{\text{tot.}ab} = \Or(c^0)$, 
as explained in section \ref{sec:formal_expansions}. The 
$0$ component of the divergence now is as follows:
\begin{align}
	\nabla^\mu F_{\text{tot.}\mu0} 
	&= g^{\mu\rho} (\partial_\rho F_{\text{tot.}\mu0} 
		- \Gamma^\sigma_{\rho\mu} F_{\text{tot.}\sigma0} 
		- \Gamma^\sigma_{\rho0} F_{\text{tot.}\mu\sigma}) \nonumber\\
	&= {\underbrace{g^{a\rho} \partial_\rho F_{\text{tot.}a0}}
			_{\mathrlap{= g^{ab} \partial_b F_{\text{tot.}a0} + \Or(c^{-7})}}}
		- g^{\mu\rho} \Gamma^a_{\rho\mu} F_{\text{tot.}a0}
		- \underbrace{g^{\mu\rho} \Gamma^\sigma_{\rho0} F_{\text{tot.}\mu\sigma}}
			_{\mathrlap{= g^{00} \Gamma^a_{00} F_{\text{tot.}0a} 
				+ g^{ab} \Gamma^0_{b0} F_{\text{tot.}a0} 
				+ g^{ab} \Gamma^c_{b0} F_{\text{tot.}ac} 
				+ \Or(c^{-7})}} \nonumber\\
	&= \left(1 + 2 \gamma \frac{\phi}{c^2}\right) \delta^{ab} \partial_b F_{\text{tot.}a0}
		- (\gamma-1) \delta^{ab} \frac{\partial_b \phi}{c^2} F_{\text{tot.}a0}
		+ \delta^{ab} \frac{\partial_b \phi}{c^2} F_{\text{tot.}0a} 
		- \delta^{ab} \frac{\partial_b \phi}{c^2} F_{\text{tot.}a0} \nonumber\\
		&\quad + \cancelto{0}{\delta^{ab} \gamma \delta^c_b \frac{\partial_t\phi}{c^3} F_{\text{tot.}ac}}
		+ \Or(c^{-5}) \nonumber\\
	&= \left(1 + 2 \gamma \frac{\phi}{c^2}\right) \Delta A_{\text{tot.}0} 
		- (\gamma+1) \delta^{ab} \frac{\partial_b \phi}{c^2} (\partial_a A_{\text{tot.}0} 
			- \partial_0 A^\perp_{\text{tot.}a}) 
		+ \Or(c^{-5}) \nonumber\\
	&= -\frac{1}{c} \left(1 + 2 \gamma \frac{\phi}{c^2}\right) \Delta \phi_\text{el.,tot.} 
		+ (\gamma+1) \delta^{ab} \frac{\partial_b \phi}{c^3} (\partial_a \phi_\text{el.,tot.} 
			+ \partial_t A^\perp_{\text{tot.}a}) 
		+ \Or(c^{-5})
\end{align}
Employing `three-vector' notation as introduced in section 
\ref{sec:geometric_notation_conventions}, this is equivalent to
\begin{equation}
	c \nabla^\mu F_{\text{tot.}\mu0} 
	= - \left(1 + 2 \gamma \frac{\phi}{c^2}\right) \Delta \phi_\text{el.,tot.} 
		+ (\gamma+1) \frac{\vect\nabla \phi}{c^2} \cdot (\vect\nabla \phi_\text{el.,tot.} 
		+ \partial_t \vect A^\perp_\text{tot.}) + \Or(c^{-4}),
\end{equation}
or (multiplying by $(1 - 2 \gamma \frac{\phi}{c^2})$) to
\begin{equation} \label{eq:field_divergence_0}
	\Delta \phi_\text{el.,tot.} 
	= - \left(1 - 2 \gamma \frac{\phi}{c^2}\right) c \nabla^\mu F_{\text{tot.}\mu0} 
		+ (\gamma+1) \frac{\vect\nabla \phi}{c^2} \cdot (\vect\nabla \phi_\text{el.,tot.} 
		+ \partial_t \vect A^\perp_\text{tot.}) 
		+ \Or(c^{-4}).
\end{equation}

For the spatial components, we obtain
\begin{align}
	\nabla^\mu F_{\text{tot.}\mu a} 
	&= g^{\mu\rho} (\partial_\rho F_{\text{tot.}\mu a} 
		- \Gamma^\sigma_{\rho\mu} F_{\text{tot.}\sigma a} 
		- \Gamma^\sigma_{\rho a} F_{\text{tot.}\mu\sigma}) \nonumber\\
	&= g^{00} \partial_0 F_{\text{tot.}0a} 
		+ g^{bc} \partial_c F_{ba}
		- \underbrace{g^{\mu\rho} \Gamma^0_{\rho\mu}}
			_{\mathclap{= \Or(c^{-3})}} 
			F_{\text{tot.}0a} 
		- \underbrace{g^{\mu\rho} \Gamma^b_{\rho\mu}}
			_{\mathrlap{= (\gamma-1) \delta^{bc} \frac{\partial_c \phi}{c^2} 
				+ \Or(c^{-4})}} 
			F_{\text{tot.}ba} \nonumber\\
		&\quad - \underbrace{g^{\mu\rho} \Gamma^\sigma_{\rho a} F_{\text{tot.}\mu\sigma}}
			_{\mathrlap{= g^{00} \Gamma^b_{0a} F_{\text{tot.}0b} 
				+ g^{bc} \Gamma^\sigma_{ca} F_{\text{tot.}b\sigma} 
				+ \Or(c^{-7})}} 
		+ \Or(c^{-6}) \nonumber\displaybreak[0]\\
	&= - \partial_0 F_{\text{tot.}0a} 
		+ \left(1 + 2 \gamma \frac{\phi}{c^2}\right) \delta^{bc} \partial_c F_{\text{tot.}ba}
		- (\gamma-1) \delta^{bc} \frac{\partial_c \phi}{c^2} F_{\text{tot.}ba} \nonumber\\
		&\quad - \underbrace{\delta^{bc} \Gamma^d_{ca} F_{\text{tot.}bd}}
			_{\mathrlap{= - 2 \gamma \delta^{bc} \frac{\partial_c \phi}{c^2} F_{\text{tot.}ba} 
				+ \Or(c^{-4})}}
		+ \Or(c^{-4}) \nonumber\\
	&= - \frac{1}{c^2} \partial_t (\partial_t A^\perp_{\text{tot.}a} 
			+ \partial_a \phi_\text{el.,tot.}) 
		+ \left(1 + 2 \gamma \frac{\phi}{c^2}\right) \Delta A^\perp_{\text{tot.}a} \nonumber\\
		&\quad + (\gamma+1) \delta^{bc} \frac{\partial_c \phi}{c^2} 
			(\partial_b A^\perp_{\text{tot.}a} - \partial_a A^\perp_{\text{tot.}b})
		+ \Or(c^{-4}).
\end{align}
Multiplying by $(1 - 2 \gamma \frac{\phi}{c^2})$, this 
is equivalent to
\begin{align} \label{eq:field_divergence_spatial}
	(\Delta - c^{-2} \partial_t^2) A^\perp_{\text{tot.}a} 
	&= \left(1 - 2 \gamma \frac{\phi}{c^2}\right) \nabla^\mu F_{\text{tot.}\mu a} 
		+ \frac{1}{c^2} \partial_a \partial_t \phi_\text{el.,tot.} \nonumber\\
		&\quad - (\gamma+1) \delta^{bc} \frac{\partial_c \phi}{c^2} 
			(\partial_b A^\perp_{\text{tot.}a} - \partial_a A^\perp_{\text{tot.}b}) 
		+ \Or(c^{-4}).
\end{align}

\subsubsection{The source terms and the explicit form of the Maxwell equations}

We now consider the right hand side of the Maxwell 
equations \eqref{eq:Maxwell_general}, i.e.\ the source term 
$-\frac{1}{\varepsilon_0 c^2} \, \frac{1}{\sqrt{-g}} j_\nu$. 
Using
\begin{equation}
	\frac{1}{\sqrt{-g}} = 1 + (3\gamma-1) \frac{\phi}{c^2} + \Or(c^{-4})
\end{equation}
and the metric coefficients, we can easily express the 
source term in terms of the charge and current densities: 
the $0$ component is
\begin{align} \label{eq:Maxwell_source_0}
	-\frac{1}{\varepsilon_0 c^2} \, \frac{1}{\sqrt{-g}} j_0 
	&= - \frac{1}{\varepsilon_0 c^2} \, \frac{1}{\sqrt{-g}} (g_{00} j^0 
		+ \underbrace{g_{0a} j^a}
			_{\mathclap{= \Or(c^{-5})}}) \nonumber\\
	&= \frac{1}{\varepsilon_0 c} \, \frac{1}{\sqrt{-g}} (-g_{00} \rho 
		+ \Or(c^{-6})) \nonumber\\
	&= \frac{1}{\varepsilon_0 c} \left(1 + (3\gamma+1) \frac{\phi}{c^2}\right) \rho 
		+ \Or(c^{-5}),
\end{align}
and the spatial components are
\begin{align} \label{eq:Maxwell_source_spatial}
	-\frac{1}{\varepsilon_0 c^2} \frac{1}{\sqrt{-g}} j_a 
	&= - \frac{1}{\varepsilon_0 c^2} \, \frac{1}{\sqrt{-g}} (g_{ab} j^b 
		+ \underbrace{g_{a0} j^0}
			_{\mathclap{= \Or(c^{-4})}}) \nonumber\\
	&= -\frac{1}{\varepsilon_0 c^2} \left(1 + (\gamma-1) \frac{\phi}{c^2}\right) 
		\delta_{ab} j^b + \Or(c^{-6}).
\end{align}

Using the source terms \eqref{eq:Maxwell_source_0}, 
\eqref{eq:Maxwell_source_spatial} and the re-arranged 
field strength divergences \eqref{eq:field_divergence_0}, 
\eqref{eq:field_divergence_spatial}, the Maxwell equations 
\eqref{eq:Maxwell_general} are equivalent to the following 
equations:
\begin{subequations} \label{eq:Maxwell_total}
\begin{align}
	\Delta \phi_\text{el.,tot.} 
	&= - \frac{1}{\varepsilon_0} \left(1 + (\gamma+1) \frac{\phi}{c^2}\right) \rho \nonumber\\
		&\quad + (\gamma+1) \frac{\vect\nabla \phi}{c^2} 
			\cdot (\vect\nabla \phi_\text{el.,tot.} 
		+ \partial_t \vect A^\perp_\text{tot.}) 
		+ \Or(c^{-4}) \\
	(\Delta - c^{-2} \partial_t^2) A^\perp_{\text{tot.}a} 
	&= -\frac{1}{\varepsilon_0 c^2} \delta_{ab} j^b 
		+ \frac{1}{c^2} \partial_a \partial_t \phi_\text{el.,tot.} \nonumber\\
		&\quad - (\gamma+1) \delta^{bc} \frac{\partial_c \phi}{c^2} 
			(\partial_b A^\perp_{\text{tot.}a} - \partial_a A^\perp_{\text{tot.}b}) 
		+ \Or(c^{-4})
\end{align}
\end{subequations}

Now, as done in \cite{sonnleitner18}, we split the total 
potentials $(\phi_\text{el.,tot.}, \vect A^\perp_\text{tot.})$ 
into internal and external parts, both satisfying the gauge 
condition, where the internal potentials $(\phi_\text{el.}, 
\vect{\mathcal A}^\perp)$ satisfy the Maxwell equations 
with the internal charge and current densities as sources, 
and the external potentials $(\phi_\text{el.,ext.}, 
\vect A^\perp)$ the vacuum Maxwell equations. Note that 
the internal electric potential does not carry a subscript 
`int.' or similar, as opposed to the external one. 
Similarly, we write $F_{\text{tot.}\mu\nu} 
= \mathcal F_{\mu\nu} + F_{\mu\nu}$, where $\mathcal F 
= \D\mathcal A$ is the internal and $F = \D A$ is the 
external field tensor (employing the obvious notation 
$\mathcal A_0 = - \frac{1}{c} \phi_\text{el.}$, $A_0 
= - \frac{1}{c} \phi_\text{el.,ext.}$).

\subsubsection{Solution of the internal Maxwell equations}

From \eqref{eq:Maxwell_total}, the Maxwell equations 
for the internal potentials are as follows:
\begin{subequations} \label{eq:Maxwell_int}
\begin{align}
	\label{eq:Maxwell_int_Poisson}
	\Delta \phi_\text{el.} 
	&= - \frac{1}{\varepsilon_0} \left(1 + (\gamma+1) \frac{\phi}{c^2}\right) \rho 
		+ (\gamma+1) \frac{\vect\nabla \phi}{c^2} 
			\cdot (\vect\nabla \phi_\text{el.} + \partial_t \vect{\mathcal A}^\perp) 
		+ \Or(c^{-4}) \\
	\label{eq:Maxwell_int_wave}
	(\Delta - c^{-2} \partial_t^2) \mathcal A^\perp_a 
	&= -\frac{1}{\varepsilon_0 c^2} \delta_{ab} j^b 
		+ \frac{1}{c^2} \partial_a \partial_t \phi_\text{el.} \nonumber\\
		&\quad - (\gamma+1) \delta^{bc} \frac{\partial_c \phi}{c^2} 
			(\partial_b \mathcal A^\perp_a - \partial_a \mathcal A^\perp_b) 
		+ \Or(c^{-4})
\end{align}
\end{subequations}

We will now solve \eqref{eq:Maxwell_int} perturbatively in 
a formal expansion in $c^{-2}$. Expanding the potentials 
as $\phi_\text{el.} = \phi_\text{el.}^{(0)} 
+ c^{-2} \phi_\text{el.}^{(2)} + \Or(c^{-4})$ and 
$\vect{\mathcal A}^\perp = \vect{\mathcal A}^{\perp(0)} 
+ c^{-2} \vect{\mathcal A}^{\perp(2)} + \Or(c^{-4})$, the 
lowest orders of the Poisson equation for $\phi_\text{el.}$ 
read
\begin{subequations}
\begin{align}
	\label{eq:Poisson_order_0}
	\Delta \phi_\text{el.}^{(0)} &= - \frac{1}{\varepsilon_0} \rho, \\
	\label{eq:Poisson_order_2}
	\Delta \phi_\text{el.}^{(2)} 
	&= - \frac{1}{\varepsilon_0} (\gamma+1) \phi \rho 
		+ (\gamma+1) \vect\nabla \phi \cdot (\vect\nabla \phi_\text{el.}^{(0)} 
		+ \partial_t \vect{\mathcal A}^{\perp(0)}),
\end{align}
\end{subequations}
and the lowest orders of the wave equation for 
$\vect{\mathcal A}^\perp$ are
\begin{subequations}
\begin{align}
	\label{eq:wave_order_0}
	(\Delta - c^{-2} \partial_t^2) \mathcal A_a^{\perp(0)} &= 0, \\
	\label{eq:wave_order_2}
	(\Delta - c^{-2} \partial_t^2) \mathcal A_a^{\perp(2)} 
	&= -\frac{1}{\varepsilon_0} \delta_{ab} j^b 
		+ \partial_a \partial_t \phi_\text{el.}^{(0)} 
		- (\gamma+1) \delta^{bc} \frac{\partial_c \phi}{c^2} 
			(\partial_b \mathcal A_a^{\perp(0)} 
		- \partial_a \mathcal A_b^{\perp(0)}).
\end{align}
\end{subequations}
Being the usual, `non-gravitational' Poisson equation, 
\eqref{eq:Poisson_order_0} gives
\begin{equation}
	\phi_\text{el.}^{(0)} = \phi_\text{el.,ng} \; ,
\end{equation}
where $\phi_\text{el.,ng}$ is the internal electric 
potential solution in the absence of gravity as given 
by \eqref{eq:el_pot_orig}.

For the wave equation \eqref{eq:Maxwell_int_wave} we 
are interested in purely retarded solutions without any 
additional radiative terms, since the internal potentials 
shall correspond to just `what is generated by the 
particles'. Therefore, \eqref{eq:wave_order_0} directly 
implies $\vect{\mathcal A}^{\perp(0)} = 0$.

Thus, \eqref{eq:wave_order_2} reduces to the 
`non-gravitational' wave equation for the potential 
$\vect{\mathcal A}^\perp_\text{ng}$, but applied to 
$c^{-2} \vect{\mathcal A}^{\perp(2)}$, implying 
$c^{-2} \vect{\mathcal A}^{\perp(2)} 
= \vect{\mathcal A}^\perp_\text{ng}$, where 
$\vect{\mathcal A}^\perp_\text{ng}$ is the 
non-gravitational retarded solution, expanded to 
lowest non-vanishing order in $c^{-1}$, as given by 
\eqref{eq:mag_pot_orig}. Hence we have
\begin{equation} \label{eq:internal_vect_pot_result}
	\vect{\mathcal A}^\perp 
	= \vect{\mathcal A}^\perp_\text{ng} + \Or(c^{-4}) 
	= \Or(c^{-2}).
\end{equation}

Finally, solving \eqref{eq:Poisson_order_2} directly gives
\begin{align}
	\phi_\text{el.}^{(2)}(\vect x,t) 
	&= \frac{\gamma+1}{4\pi \varepsilon_0} 
		\int\D^3\ivect x' \frac{\phi(\vect x',t) \rho(\vect x',t)}{|\vect x - \vect x'|} 
	- \frac{\gamma+1}{4\pi} \int\D^3\ivect x' \frac{1}{|\vect x - \vect x'|} 
		\left(\vect\nabla \phi \cdot \vect\nabla \phi_\text{el.}^{(0)} \right)(\vect x',t).
\end{align}

For later convenience, we will now compute the interaction 
integral $-\frac{1}{2} \int \D^3\ivect x \, \rho \phi_\text{el.}$. 
We suppress time dependence in the notation. Using the 
explicit form of the charge density, $\rho(\vect x) 
= e_1 \delta^{(3)}(\vect x - \vect r_1) 
+ e_2 \delta^{(3)}(\vect x - \vect r_2)$, and dropping 
infinite self-interaction terms, we obtain
\begin{align}
	-\frac{1}{2} \int\D^3\ivect x \, \rho \phi_\text{el.} 
	&= -\frac{e_1 e_2}{4\pi \varepsilon_0 r} \left(1 + (\gamma+1) 
			\frac{\phi(\vect r_1) + \phi(\vect r_2)}{2c^2}\right) \nonumber\\
		&\quad + \underbrace{\frac{\gamma+1}{8\pi c^2} \int\D^3\ivect x 
			\left( \frac{e_1}{|\vect x - \vect r_1|} 
				+ \frac{e_2}{|\vect x - \vect r_2|} \right) 
			\vect\nabla \phi \cdot \vect\nabla \phi_\text{el.}^{(0)}}
			_{= \frac{\varepsilon_0(\gamma+1)}{2c^2} \int\D^3\ivect x \, 
				\phi_\text{el.}^{(0)} \vect\nabla \phi 
				\cdot \vect\nabla \phi_\text{el.}^{(0)}}
		+ \Or(c^{-4}),
\end{align}
where we used the explicit form of the lowest-order 
potential $\phi_\text{el.}^{(0)}$. For the last integral, 
partial integration gives
\begin{align}
	\int\D^3\ivect x \, \phi_\text{el.}^{(0)} \vect\nabla \phi 
			\cdot \vect\nabla \phi_\text{el.}^{(0)} 
	&= - \int\D^3\ivect x \, \phi_\text{el.}^{(0)} \vect\nabla \cdot 
			\left(\phi_\text{el.}^{(0)} \vect\nabla \phi\right) \nonumber\\
	&= - \int\D^3\ivect x \, \phi_\text{el.}^{(0)} \vect\nabla \phi_\text{el.}^{(0)} 
			\cdot \vect\nabla \phi 
		- \int\D^3\ivect x \left(\phi_\text{el.}^{(0)}\right)^2 \Delta\phi,
\end{align}
implying
\begin{equation} \label{eq:interaction_internal_electric_Delta_phi}
	\int\D^3\ivect x \, \phi_\text{el.}^{(0)} \vect\nabla \phi 
			\cdot \vect\nabla \phi_\text{el.}^{(0)} 
	= -\frac{1}{2} \int\D^3\ivect x \, \left(\phi_\text{el.}^{(0)}\right)^2 \Delta\phi.
\end{equation}
In the following, we will neglect this term: due to the 
Newtonian field equation, $\Delta\phi$ is non-vanishing 
only inside the matter generating the gravitational 
potential, and $\phi_\text{el.}^{(0)}$ is negligibly small 
there for an atom situated outside of this matter (e.g.\ 
in a quantum-optical experiment outside of the earth).
Thus, the relevant part of the above interaction integral 
is just
\begin{equation} \label{eq:interaction_internal_electric}
	-\frac{1}{2} \int\D^3\ivect x \, \rho \phi_\text{el.} 
	= -\frac{e_1 e_2}{4\pi \varepsilon_0 r} \left(1 
			+ (\gamma+1) \frac{\phi(\vect r_1) + \phi(\vect r_2)}{2c^2}\right) 
		+ \Or(c^{-4}).
\end{equation}
\enlargethispage{-2\baselineskip}

\subsubsection{The external Maxwell equations}

We will now consider the Maxwell equations for the external 
potentials. Since we assume the absence of external 
charges, the Poisson equation for $\phi_\text{el.,ext.}$ 
reads as follows:
\begin{equation}
	\Delta \phi_\text{el.,ext.} 
	= (\gamma+1) \frac{\vect\nabla \phi}{c^2} 
			\cdot (\vect\nabla \phi_\text{el.,ext.} + \partial_t \vect A^\perp) 
		+ \Or(c^{-4})
\end{equation}
Solving this equation perturbatively as for the internal 
potentials, we obtain the solution
\begin{equation}
	\phi_\text{el.,ext.}(\vect x,t) 
	= -\frac{\gamma+1}{4\pi c^2} \int\D^3\ivect x' 
		\frac{1}{|\vect x - \vect x'|} (\vect\nabla\phi 
			\cdot \partial_t \vect A^\perp)(\vect x',t) 
		+ \Or(c^{-4}),
\end{equation}
expressed solely in terms of the external vector potential. 
In fact, we will not need this explicit form of the 
potential, but just the expansion order
\begin{equation} \label{eq:external_el_pot_order}
	\phi_\text{el.,ext.} = \Or(c^{-2}).
\end{equation}
Now considering the wave equation for the vector potential 
$\vect A^\perp$, which due to the absence of external 
currents and the above result on $\phi_\text{el.,ext.}$ 
reads
\begin{align}
	(\Delta - c^{-2} \partial_t^2) A^\perp_a 
	&= \frac{1}{c^2} \partial_a \partial_t \phi_\text{el.,ext.} 
		- (\gamma+1) \delta^{bc} \frac{\partial_c \phi}{c^2} 
			(\partial_b A^\perp_a - \partial_a A^\perp_b) 
		+ \Or(c^{-4}) \nonumber\\
	&= -(\gamma+1) \delta^{bc} \frac{\partial_c \phi}{c^2} 
			(\partial_b A^\perp_a - \partial_a A^\perp_b) 
		+ \Or(c^{-4}),
\end{align}
and employing a further expansion $\vect A^\perp 
= \vect A^{\perp(0)} + c^{-2} \vect A^{\perp(2)} 
+ \Or(c^{-4})$, we obtain in lowest order
\begin{equation}
	(\Delta - c^{-2} \partial_t^2) \vect A^{\perp(0)} = 0.
\end{equation}
Differently to the internal case, we now allow for 
radiative solutions\footnote
	{At the end of the day, the idea is to put an atom 
	into a laser beam.},
thus \emph{not} getting $\vect A^{\perp(0)} = 0$. However, 
we can conclude that $\partial_a \vect A^\perp 
= \Or(c^{-1} \partial_t \vect A^\perp)$. Treating 
$\partial_t \vect A^\perp$, which corresponds (up to 
a gravitational correction factor of order unity) to 
the external electric field, as being of order $c^0$, 
we thus have
\begin{equation} \label{eq:external_vect_pot_spatial_deriv_order}
	\partial_a \vect A^\perp = \Or(c^{-1}).
\end{equation}
\enlargethispage{-2\baselineskip}

\subsection{Computation of the electromagnetic Lagrangian}

We will now compute the electromagnetic Lagrangian
\begin{equation} \label{eq:L_em}
	L_\text{em} = \int\D^3\ivect x \, \left(
		-\frac{\varepsilon_0 c^2}{4} \sqrt{-g} \, 
			F_{\text{tot.}\mu\nu} F_\text{tot.}^{\mu\nu} 
		+ j^\mu A_{\text{tot.}\mu}\right),
\end{equation}
which follows from the action \eqref{eq:em_action}.

The internal kinetic Maxwell term is
\begin{align}
	&-\frac{\varepsilon_0 c^2}{4} \int\D^3\ivect x \, \sqrt{-g} \, 
			\mathcal F_{\mu\nu} \mathcal F^{\mu\nu} \nonumber \\
	&\quad= -\frac{\varepsilon_0 c^2}{2} \int\D^3\ivect x \, \sqrt{-g} \, 
			\partial_\mu \mathcal A_\nu \mathcal F^{\mu\nu} \nonumber \\
	\text{(P.I.)} &\quad= \frac{\varepsilon_0 c^2}{2} \int\D^3\ivect x \, \sqrt{-g} \, 
			\mathcal A_\nu \nabla_\mu \mathcal F^{\mu\nu} 
		- \frac{\varepsilon_0 c^2}{2} \int\D^3\ivect x \, 
			\partial_0(\sqrt{-g} \mathcal A_\nu \mathcal F^{0\nu}).
\end{align}
The first integral on the right-hand side is equal to 
$-\frac{1}{2} \int\D^3\ivect x \, \mathcal A_\nu j^\nu$ 
by the internal part of the general Maxwell equations 
\eqref{eq:Maxwell_general_cov}, and for the second 
integral we obtain
\begin{align}
	&\hspace{-3em} -\frac{\varepsilon_0 c^2}{2} \int\D^3\ivect x \, \partial_0(\sqrt{-g} \, 
		\mathcal A_\nu \mathcal F^{0\nu}) \nonumber\\
	&= - \frac{\varepsilon_0 c^2}{2} \int\D^3\ivect x \, \partial_0(\sqrt{-g} \, 
		\mathcal A_a \mathcal F^{0a}) \nonumber\\
	&\quad \text{(using $\mathcal F^{0a} = g^{00} g^{ab} \mathcal F_{0b} + \Or(c^{-5})$)} \nonumber\\
	&= - \frac{\varepsilon_0 c^2}{2} \int\D^3\ivect x \, \partial_0 \left(\sqrt{-g} \, 
			g^{00} g^{ab} \mathcal A_a (\partial_0 \mathcal A_b 
				- \partial_b \mathcal A_0) \right) 
		+ \Or(c^{-4}) \nonumber\\
	&= - \frac{\varepsilon_0}{2} \int\D^3\ivect x \, \partial_t \left(\sqrt{-g} \, 
			g^{00} g^{ab} \mathcal A_a (\partial_t \mathcal A_b 
				+ \partial_b \phi_\text{el.}) \right) 
		+ \Or(c^{-4}) \nonumber\\
	\text{(P.I.)} \quad 
	&= - \frac{\varepsilon_0}{2} \int\D^3\ivect x \, \partial_t \left(\sqrt{-g} \, 
			g^{00} g^{ab} \mathcal A_a \partial_t \mathcal A_b \right) \nonumber\\
		&\quad + \frac{\varepsilon_0}{2} \int\D^3\ivect x \, \partial_t 
			\Big( \underbrace{\partial_b \left(\sqrt{-g} \, g^{00} g^{ab}\right)}
					_{= \Or(c^{-2})} 
			\mathcal A_a \phi_\text{el.} \Big) \nonumber\\
		&\quad + \frac{\varepsilon_0}{2} \int\D^3\ivect x \, \partial_t \Big(\sqrt{-g} \, 
			g^{00} \underbrace{g^{ab} \partial_b \mathcal A_a}
				_{\mathclap{= (1 + 2\gamma\frac{\phi}{c^2}) \delta^{ab} 
						\partial_b \mathcal A^\perp_a + \Or(c^{-4}) 
					= \Or(c^{-4})}} 
			\phi_\text{el.} \Big)
		+ \Or(c^{-4}) \nonumber\\
	&= \Or(c^{-4}),
\end{align}
where in the partial integration step we used the gauge 
condition \eqref{eq:gauge} and that $\mathcal A_a$ is of 
order $c^{-2}$ according to \eqref{eq:internal_vect_pot_result}. 
Thus, the `purely internal' contribution of 
electromagnetism to the Lagrangian, including the explicit 
coupling term of the internal potential to the current, is
\begin{align} \label{eq:L_em_int}
	L_\text{em,int.} 
	&= \int\D^3\ivect x \, \left(-\frac{\varepsilon_0 c^2}{4} \sqrt{-g} \, 
				\mathcal F_{\mu\nu} \mathcal F^{\mu\nu} 
			+ j^\mu \mathcal A_\mu\right) \nonumber\\
	&= \frac{1}{2} \int\D^3\ivect x \, j^\mu \mathcal A_\mu 
		+ \Or(c^{-4}) \nonumber\\
	&= \frac{1}{2} \int\D^3\ivect x \, (\vect j \cdot \vect{\mathcal A}^\perp 
			- \rho \phi_\text{el.}) 
		+ \Or(c^{-4}).
\end{align}

To compute the purely external and mixed external-internal 
contributions to the electromagnetic Lagrangian, we first 
explicitly compute the kinetic Maxwell term in terms 
of the potentials. Inserting the explicit form of the 
\acronym{PPN} metric, we obtain
\begin{align} \label{eq:kin_Maxwell_total}
	-\frac{\varepsilon_0 c^2}{4} \sqrt{-g} \, 
			F_{\text{tot.}\mu\nu} F_\text{tot.}^{\mu\nu}
	&= \frac{\varepsilon_0}{2} \sqrt{-g} \, \bigg[-g^{00} g^{ab} (\partial_t A_{\text{tot.}a} 
				+ \partial_a \phi_{\text{el.,tot.}}) 
			(\partial_t A_{\text{tot.}b} 
				+ \partial_b \phi_{\text{el.,tot.}}) \nonumber\\
		&\qquad- c^2 (g^{ab} g^{cd} - g^{ad} g^{cb}) \, 
			\partial_a A_{\text{tot.}c} \, \partial_b A_{\text{tot.}d}\bigg] 
		+ \Or(c^{-4}) \nonumber\displaybreak[0]\\
	&= \frac{\varepsilon_0}{2} \bigg[ \left(1 - (\gamma+1) \frac{\phi}{c^2}\right) 
				(\partial_t \vect A_\text{tot.} 
					+ \vect\nabla \phi_{\text{el.,tot.}})^2 \nonumber\\
			&\qquad- c^2 \left(1 + (\gamma+1) \frac{\phi}{c^2}\right) 
				(\vect\nabla \times \vect A_\text{tot.})^2 \bigg] 
		+ \Or(c^{-4}).
\end{align}
Note that according to \eqref{eq:internal_vect_pot_result} 
and \eqref{eq:external_vect_pot_spatial_deriv_order} we 
have $\vect\nabla \times \vect A_\text{tot.} = \Or(c^{-1})$, 
such that the second term in the square brackets does 
indeed include terms up to (and including) order $c^{-2}$, 
such that the total given expansion order makes sense. 
We also recall that, as introduced in section 
\ref{sec:geometric_notation_conventions}, 
$\vect\nabla \times \vect A_\text{tot.}$ denotes the 
`component-wise curl' of $\vect A_\text{tot.}$, which 
is a well-defined geometric operation (i.\,e.\ independent 
of coordinates) once we have introduced the background 
structures.

The internal-internal term of \eqref{eq:kin_Maxwell_total} 
was considered above in \eqref{eq:L_em_int}. The purely 
external term gives
\begin{align} \label{eq:L_em_ext}
	L_\text{em,ext.} 
	&= -\frac{\varepsilon_0 c^2}{4} \int\D^3\ivect x \, \sqrt{-g} \, 
			F_{\mu\nu} F^{\mu\nu} \nonumber\\
	&= \frac{\varepsilon_0}{2} \int\D^3\ivect x \, 
			\bigg[ \left(1 - (\gamma+1) \frac{\phi}{c^2}\right) 
					\Big((\partial_t \vect A^\perp)^2 
						+ 2 \partial_t \vect A^\perp 
							\cdot \underbrace{\vect\nabla \phi_\text{el.,ext.}}
							_{\overset{\eqref{eq:external_el_pot_order}}{=} \Or(c^{-2})} \nonumber\\
						&\qquad + \underbrace{(\vect\nabla \phi_\text{el.,ext.})^2}
							_{\overset{\eqref{eq:external_el_pot_order}}{=} \Or(c^{-4})} \Big) 
				- c^2 \left(1 + (\gamma+1) \frac{\phi}{c^2}\right) 
					(\vect\nabla \times \vect A^\perp)^2 \bigg] 
		+ \Or(c^{-4}) \nonumber\\
	&\quad \text{(using P.I., \eqref{eq:gauge})} \nonumber\\
	&= \frac{\varepsilon_0}{2} \int\D^3\ivect x \, 
			\bigg[ \left(1 - (\gamma+1) \frac{\phi}{c^2}\right) 
					(\partial_t \vect A^\perp)^2 \nonumber\\
				&\quad- c^2 \left(1 + (\gamma+1) \frac{\phi}{c^2}\right) 
					(\vect\nabla \times \vect A^\perp)^2 \bigg]
		+ \Or(c^{-4}).
\end{align}

For the external-internal mixed term plus the interaction of the external potential with the current, we obtain
\begin{align} \label{eq:L_em_ext_int}
	L_\text{em,ext.-int.} 
	&= \int\D^3\ivect x \, j^\mu A_\mu 
		- \frac{\varepsilon_0 c^2}{2} \int\D^3\ivect x \, \sqrt{-g} \, 
			\mathcal F_{\mu\nu} F^{\mu\nu} \nonumber\\
	&= \int\D^3\ivect x \, (\vect j \cdot \vect A^\perp 
			- \rho \phi_\text{el.,ext.}) \nonumber\\
		&\quad+ \varepsilon_0 \int\D^3\ivect x \, \bigg[ 
			\left(1 - (\gamma+1) \frac{\phi}{c^2}\right) 
				(\partial_t \vect{\mathcal A}^\perp + \vect\nabla \phi_\text{el.}) 
				\cdot (\partial_t \vect A^\perp + \vect\nabla \phi_\text{el.,ext.}) \nonumber\\
			&\qquad- c^2 \left(1 + (\gamma+1) \frac{\phi}{c^2}\right) 
				(\vect\nabla \times \vect{\mathcal A}^\perp)
				\cdot (\vect\nabla \times \vect A^\perp) \bigg]
		+ \Or(c^{-4}) \nonumber\displaybreak[0]\\
	&\quad \text{(using \eqref{eq:internal_vect_pot_result}, \eqref{eq:external_el_pot_order})} \nonumber\\
	&= \int\D^3\ivect x \, (\vect j \cdot \vect A^\perp - \rho \phi_\text{el.,ext.}) \nonumber\\
		&\quad+ \varepsilon_0 \int\D^3\ivect x \, \bigg[ (\partial_t \vect{\mathcal A}^\perp) 
				\cdot (\partial_t \vect A^\perp) \nonumber\\
			&\qquad- c^2 \left(1 + (\gamma+1) \frac{\phi}{c^2}\right) 
				(\vect\nabla \times \vect{\mathcal A}^\perp) 
				\cdot (\vect\nabla \times \vect A^\perp) \nonumber\\
			&\qquad+ \left(1 - (\gamma+1) \frac{\phi}{c^2}\right) 
				\vect\nabla \phi_\text{el.} \cdot \partial_t \vect A^\perp 
			+ \vect\nabla \phi_\text{el.} \cdot \vect\nabla \phi_\text{el.,ext.} \bigg]
		+ \Or(c^{-4}) \nonumber\\
	&\quad \text{(using P.I., \eqref{eq:gauge}, \eqref{eq:Maxwell_int_Poisson})} \nonumber\\
	&= \int\D^3\ivect x \, \vect j \cdot \vect A^\perp 
		+ \varepsilon_0 \int\D^3\ivect x \, \bigg[ (\partial_t \vect{\mathcal A}^\perp) 
				\cdot (\partial_t \vect A^\perp) \nonumber\\
			&\qquad- c^2 \left(1 + (\gamma+1) \frac{\phi}{c^2}\right) 
				(\vect\nabla \times \vect{\mathcal A}^\perp) 
				\cdot (\vect\nabla \times \vect A^\perp) \bigg] \nonumber\\
		&\qquad+ \varepsilon_0 \int\D^3\ivect x \, (\gamma+1) \, \phi_\text{el.}^{(0)} 
			\frac{\vect\nabla\phi}{c^2} \cdot \partial_t \vect A^\perp
		+ \Or(c^{-4}).
\end{align}
Following appendix B of \cite{sonnleitner18}, we will 
neglect the second integral in this expression since 
it is related to formally diverging backreaction terms.

Adding the Lagrangians \eqref{eq:L_em_int}, 
\eqref{eq:L_em_ext} and \eqref{eq:L_em_ext_int}, the total 
post-Newtonian electromagnetic Lagrangian (with the 
above-mentioned neglections following \cite{sonnleitner18}) 
reads
\begin{align} \label{eq:Lagrangian_em_grav}
	L_\text{em} 
	&= \frac{1}{2} \int\D^3\ivect x \, (\vect j \cdot \vect{\mathcal A}^\perp 
			- \rho \phi_\text{el.}) 
		+ \int\D^3\ivect x \, \vect j \cdot \vect A^\perp \nonumber\\
		&\quad + \frac{\varepsilon_0}{2} \int\D^3\ivect x \, 
			\bigg[ \left(1 - (\gamma+1) \frac{\phi}{c^2}\right) 
				(\partial_t \vect A^\perp)^2 
			- c^2 \left(1 + (\gamma+1) \frac{\phi}{c^2}\right) 
				(\vect\nabla \times \vect A^\perp)^2 \bigg] \nonumber\\
		&\quad+ \varepsilon_0 \int\D^3\ivect x \, (\gamma+1) \, \phi_\text{el.}^{(0)} 
			\frac{\vect\nabla\phi}{c^2} \cdot \partial_t \vect A^\perp
		+ \Or(c^{-4}).
\end{align}
We remind the reader that, as for the non-gravitational 
calculation discussed in section~\ref{sec:work_sonnleitner_barnett},
$\vect A^\perp$ is treated as a given external field that 
appears in the Lagrangian, not a dynamical variable. 
Inserting the internal magnetic potential 
\eqref{eq:internal_vect_pot_result} and using the electric 
interaction integral \eqref{eq:interaction_internal_electric} 
computed above, for the internal term we obtain (dropping 
infinite self-interaction terms)
\begin{align} \label{eq:Lagrangian_em_grav_internal_explicit}
	\frac{1}{2} \int\D^3\ivect x \, (\vect j \cdot \vect{\mathcal A}^\perp 
			- \rho \phi_\text{el.}) 
	&= -\frac{e_1 e_2}{4\pi \varepsilon_0 r} \left(1 + (\gamma+1) 
			\frac{\phi(\vect r_1) + \phi(\vect r_2)}{2c^2}\right) \nonumber\\
		&\quad+ \frac{e_1 e_2}{8\pi \varepsilon_0 c^2} 
			\left[ \frac{\dot{\vect r}_1 \cdot \dot{\vect r}_2}{r} 
				+ \frac{(\dot{\vect r}_1 \cdot \vect r) (\dot{\vect r}_2 \cdot \vect r)}{r^3} \right]
		+ \Or(c^{-4}).
\end{align}

\section{The total Hamiltonian including all interactions}
\label{sec:total_Hamiltonian}

In this section we collect all previous findings and 
combine them into the total Hamiltonian that characterises 
the dynamics of our two-particle system that is now also 
exposed to a non-trivial gravitational field. We will see 
that the Hamiltonian suffers various `corrections' as 
compared to the gravity-free case, and that these terms 
acquire an intuitive interpretation if re-expressed in 
terms of the physical spacetime metric $g$.

\subsection{Computation of the Hamiltonian} 
\label{sec:total_Hamiltonian_comp}

We will now compute the total Hamiltonian describing the 
atom in external electromagnetic and gravitational fields 
by repeating the calculation from section 
\ref{sec:coupling_GF_particles} while including the 
`gravitational corrections' to electromagnetism obtained 
in section \ref{sec:coupling_GF_EMF}.

Comparing the gravitationally corrected electromagnetic 
Lagrangian as given by \eqref{eq:Lagrangian_em_grav}, 
\eqref{eq:Lagrangian_em_grav_internal_explicit} to the 
one without gravitational field ($\phi = 0$), we see 
that (at our order of approximation) the differences 
consist of new prefactors involving $\phi$ in the external 
electromagnetic term and the internal Coulomb interaction 
term, as well as an additional term involving the derivative 
$\vect\nabla\phi$ of the gravitational potential (last line of 
\eqref{eq:Lagrangian_em_grav_internal_explicit}). Thus, 
when calculating the Hamiltonian, we have to take care 
of these changes compared to the discussion of section 
\ref{sec:coupling_GF_particles}.

\subsubsection{The classical Hamiltonian}

As explained in section \ref{sec:work_sonnleitner_barnett}, 
although the external field is not treated as a dynamical 
variable, when Legendre transforming the Lagrangian 
in order to compute a Hamiltonian we will add a term 
corresponding to the energy of the external field. We 
also use the notation $\vect\Pi^\perp$ for the `would-be 
canonical momentum' conjugate to the external field (i.e.\ 
the canonical momentum if $\vect A^\perp$ were a dynamical 
variable). For our Lagrangian, this `would-be canonical 
momentum' is
\begin{align} \label{eq:field_momentum_with_gravity}
	\vect\Pi^\perp 
	&= \frac{\delta L_\text{em}}{\delta(\partial_t \vect A^\perp)} \nonumber\\
	&= \varepsilon_0 \left(1 - (\gamma+1) \frac{\phi}{c^2}\right) \partial_t \vect A^\perp 
		+ \varepsilon_0 (\gamma+1) 
			\left(\phi_\text{el.}^{(0)} \frac{\vect\nabla\phi}{c^2}\right)^{\kern-.3em\perp} 
		+ \Or(c^{-4}).
\end{align}
Inverting this, we get
\begin{equation}
	\partial_t \vect A^\perp 
	= \left(1 + (\gamma+1) \frac{\phi}{c^2}\right) \frac{\vect\Pi^\perp}{\varepsilon_0} 
		- (\gamma+1) \left(\phi_\text{el.}^{(0)} \frac{\vect\nabla\phi}{c^2}\right)^{\kern-.3em\perp} 
		+ \Or(c^{-4}).
\end{equation}
Expressing the first part of the external electromagnetic 
Lagrangian \eqref{eq:L_em_ext} in terms of this, we have
\begin{align}
	&\frac{\varepsilon_0}{2} \int\D^3\ivect x \, 
			\left(1 - (\gamma+1) \frac{\phi}{c^2}\right) 
			(\partial_t \vect A^\perp)^2 \nonumber\\
	&\quad= \frac{\varepsilon_0}{2} \int\D^3\ivect x \, 
			\left( \frac{\vect\Pi^\perp}{\varepsilon_0} 
				- (\gamma+1) \, \phi_\text{el.}^{(0)} 
					\frac{\vect\nabla\phi}{c^2} \right) 
			\cdot (\partial_t \vect A^\perp) 
		+ \Or(c^{-4}) \nonumber\\
	&\quad= \frac{\varepsilon_0}{2} \int\D^3\ivect x \, 
			\left(1 + (\gamma+1) \frac{\phi}{c^2}\right) 
			(\vect\Pi^\perp/\varepsilon_0)^2 
		- \int\D^3\ivect x \, (\gamma+1) \, \phi_\text{el.}^{(0)} 
			\frac{\vect\nabla\phi}{c^2} \cdot \vect\Pi^\perp 
		+ \Or(c^{-4}).
\end{align}
Furthermore, we have
\begin{align}
	\int\D^3\ivect x \, \vect\Pi^\perp \cdot \partial_t \vect A^\perp 
	&= \varepsilon_0 \int\D^3\ivect x \, \left(1 + (\gamma+1) \frac{\phi}{c^2}\right) 
			(\vect\Pi^\perp/\varepsilon_0)^2 \nonumber\\
		&\quad - \int\D^3\ivect x \, (\gamma+1) \, \phi_\text{el.}^{(0)} 
			\frac{\vect\nabla\phi}{c^2} \cdot \vect\Pi^\perp 
		+ \Or(c^{-4}).
\end{align}
Thus, the Hamiltonian for the external electromagnetic 
field and the external-internal interaction is
\begin{align}
	H_\text{em,ext.,ext.-int.} 
	&= \int\D^3\ivect x \, \vect\Pi^\perp \cdot \partial_t \vect A^\perp 
		- L_\text{em,ext.} - L_\text{em,ext.-int.} \nonumber\\
	&= \int\D^3\ivect x \, \vect\Pi^\perp \cdot \partial_t \vect A^\perp 
		- L_\text{em,ext.} \nonumber\\
		&\quad - \int\D^3\ivect x \, \vect j \cdot \vect A^\perp 
		- \varepsilon_0 \int\D^3\ivect x \, (\gamma+1) \, \phi_\text{el.}^{(0)} 
			\frac{\vect\nabla\phi}{c^2} 
			\cdot \underbrace{\partial_t \vect A^\perp}
				_{\mathrlap{= \vect\Pi^\perp/\varepsilon_0 + \Or(c^{-2})}}
		+ \Or(c^{-4}) \nonumber\\
	&= \frac{\varepsilon_0}{2} \int\D^3\ivect x \, 
			\left(1 + (\gamma+1) \frac{\phi}{c^2}\right) 
			\left[ (\vect\Pi^\perp/\varepsilon_0)^2 
				+ c^2 (\vect\nabla \times \vect A^\perp)^2 \right] \nonumber\\
		&\quad- \int\D^3\ivect x \, \vect j \cdot \vect A^\perp 
		- \int\D^3\ivect x \, (\gamma+1) \, \phi_\text{el.}^{(0)} 
			\frac{\vect\nabla\phi}{c^2} \cdot \vect\Pi^\perp
		+ \Or(c^{-4}).
\end{align}

When including the gravitational corrections to 
electromagnetism, the final total classical Hamiltonian 
thus will differ from the one without these corrections, 
as given by \eqref{eq:Hamiltonian_class_no_em} and 
\eqref{eq:Hamiltonian_class_orig}, in the following points:
\begin{itemize}
	\item The external `field energy' $\displaystyle 
		\frac{\varepsilon_0}{2} \int\D^3\ivect x \, 
		\left[ (\vect\Pi^\perp/\varepsilon_0)^2 
		+ c^2 (\vect\nabla \times \vect A^\perp)^2 \right]$ 
		gains a prefactor $\left(1 + (\gamma+1) \frac{\phi}{c^2}\right)$,

	\item the Coulomb term gains a prefactor 
		$\left(1 + (\gamma+1) \frac{\phi(\vect r_1) 
		+ \phi(\vect r_2)}{2c^2}\right)$, and

	\item there is an additional term
		\begin{equation}
			- \int\D^3\ivect x \, (\gamma+1) \, \phi_\text{el.}^{(0)} 
			\frac{\vect\nabla\phi}{c^2} \cdot \vect\Pi^\perp.
		\end{equation}
\end{itemize}

\subsubsection{Canonical quantisation, \acronym{PZW} transformation, and introduction of centre of mass coordinates}

We can now canonically quantise this classical Hamiltonian 
and perform the \acronym{PZW} transformation precisely 
as in the case without the gravitational corrections to 
electromagnetism~-- we just have to see how the correction 
terms transform. The resulting final multipolar Hamiltonian 
differs from the one without these corrections, as given by 
\eqref{eq:Hamiltonian_mult_no_em} and 
\eqref{eq:Hamiltonian_mult_orig}, in the following points:
\begin{itemize}
	\item The transformed external `field energy' $\displaystyle 
		\frac{\varepsilon_0}{2} \int\D^3\ivect x \, 
		\left[ \frac{(\vect{\tilde\Pi}^\perp 
			+ \vect{\mathcal P}_d^\perp)^2}{\varepsilon_0^2} 
		+ c^2 (\vect\nabla \times \vect A^\perp)^2 \right]$ 
		gains a prefactor $\left(1 + (\gamma+1) \frac{\phi}{c^2}\right)$,

	\item the Coulomb term gains a prefactor 
		$\left(1 + (\gamma+1) \frac{\phi(\vect r_1) 
		+ \phi(\vect r_2)}{2c^2}\right)$, and

	\item there are additional terms
		\begin{equation}
			- \int\D^3\ivect x \, (\gamma+1) \, \phi_\text{el.}^{(0)} 
			\frac{\vect\nabla \phi}{c^2} 
			\cdot (\vect{\tilde\Pi}^\perp + \vect{\mathcal P}_d^\perp).
		\end{equation}
\end{itemize}
For the Coulomb term, expanding $\phi$ to linear order, 
we have
\begin{align}
	\phi(\vect r_1) + \phi(\vect r_2) 
	&= 2 \phi(\vect R) 
		+ (\vect r_1 - \vect R + \vect r_2 - \vect R) 
			\cdot \vect\nabla\phi(\vect R) \nonumber\\
	&= 2 \phi(\vect R) 
		+ \frac{m_2 - m_1}{M} \vect r \cdot \vect\nabla\phi(\vect R).
\end{align}
Using this, we can rewrite the corrected Coulomb term as
\begin{align}
	- \left(1 + (\gamma+1) \frac{\phi(\vect r_1) 
			+ \phi(\vect r_2)}{2c^2}\right) 
			\frac{e^2}{4\pi \varepsilon_0 r} 
	&= -\frac{e^2}{4\pi \varepsilon_0 r} 
			\left(1 + (\gamma+1) \frac{\phi(\vect R)}{c^2}\right) \nonumber\\
		&\quad- \frac{\gamma+1}{c^2} \, \frac{e^2}{8\pi \varepsilon_0 r} \, 
			\frac{m_2 - m_1}{M} \vect r \cdot \vect\nabla \phi(\vect R)
\end{align}
in terms of centre of mass and relative coordinates.

\subsubsection{The total Hamiltonian}

Putting everything together, we arrive at the total 
Hamiltonian describing our simple atomic system in 
external electromagnetic and post-Newtonian gravitational 
fields. Here it is, in its full glory:
\begin{subequations}
	\label{eq:Hamiltonian_com_final}
\begin{align}
	H_\text{[com],final} &= H_\text{C,final} + H_\text{A,final} 
		+ H_\text{AL,final} + H_\text{L,final} + H_\text{X} + H_\text{deriv,new} 
		+ \Or(c^{-4})\\
	H_\text{C,final} 
	&= \frac{\vect P^2}{2M} 
			\left[1 - \frac{1}{M c^2} \left(\frac{\vect p_{\vect r}^2}{2\mu} 
				- \frac{e^2}{4\pi\varepsilon_0 r}\right)\right] 
		+ \left[M + \frac{1}{c^2} \left(\frac{\vect p_{\vect r}^2}{2\mu} 
			- \frac{e^2}{4\pi \varepsilon_0 r}\right) \right] \phi(\vect R) \nonumber\\
		&\quad - \frac{\vect P^4}{8M^3 c^2} 
		+ \frac{2\gamma+1}{2Mc^2} \vect P \cdot \phi(\vect R) \vect P 
		+ (2\beta-1) \frac{M \phi(\vect R)^2}{2 c^2} 
	\label{eq:Hamiltonian_com_C_final} \\
	H_\text{A,final} 
	&= \left(1 + 2\gamma\frac{\phi(\vect R)}{c^2}\right) 
			\frac{\vect p_{\vect r}^2}{2\mu} 
		- \left(1 + \gamma\frac{\phi(\vect R)}{c^2}\right) 
			\frac{e^2}{4\pi\varepsilon_0 r} \nonumber\\
		&\quad - \frac{m_1^3 + m_2^3}{M^3} \, \frac{\vect p_{\vect r}^4}{8 \mu^3 c^2} 
		- \frac{e^2}{4\pi\varepsilon_0} \, \frac{1}{2\mu M c^2} 
			\left( \vect p_{\vect r} \cdot \frac{1}{r} \vect p_{\vect r} 
			+ \vect p_{\vect r} \cdot \vect r \frac{1}{r^3} \vect r 
				\cdot \vect p_{\vect r} \right) \nonumber\\
		&\quad - \frac{2\gamma+1}{2c^2} \, \frac{m_1 - m_2}{m_1 m_2} 
			\vect p_{\vect r} \cdot (\vect r \cdot \vect\nabla\phi(\vect R)) 
			\vect p_{\vect r}
		- \frac{\gamma+1}{c^2} \, \frac{e^2}{8\pi \varepsilon_0 r} \, 
			\frac{m_2 - m_1}{M} \vect r \cdot \vect\nabla \phi(\vect R) 
		\label{eq:Hamiltonian_com_A_final} \\
	H_\text{AL,final} 
	&= \left(1 + (\gamma+1) \frac{\phi(\vect R)}{c^2}\right) 
			\frac{\vect{\tilde\Pi}^\perp(\vect R)}{\varepsilon_0} 
			\cdot \vect d 
		+ \frac{1}{2M} \{\vect P \cdot [\vect d \times \vect B(\vect R)] 
			+ \text{H.c.}\} \nonumber\\
		&\quad - \frac{m_1 - m_2}{4 m_1 m_2} \{\vect p_{\vect r} 
			\cdot [\vect d \times \vect B(\vect R)] + \text{H.c.}\} \nonumber\\
		&\quad + \frac{1}{8\mu} (\vect d \times \vect B(\vect R))^2 
		+ \frac{1}{2\varepsilon_0} \int\D^3\ivect x \, 
			\left(1 + (\gamma+1) \frac{\phi}{c^2}\right) 
			{\vect{\mathcal P}_d^\perp}^2 (\vect x, t) \nonumber\\
		&\quad - \int\D^3\ivect x \, (\gamma+1) \, \phi_\text{el.}^{(0)} 
			\frac{\vect\nabla \phi}{c^2} \cdot (\vect{\tilde\Pi}^\perp 
				+ \vect{\mathcal P}_d^\perp) 
	\label{eq:Hamiltonian_com_AL_final} \\
	H_\text{L,final} 
	&= \frac{\varepsilon_0}{2} \int\D^3\ivect x \, 
		\left(1 + (\gamma+1) \frac{\phi}{c^2}\right) 
		\left[ (\vect{\tilde\Pi}^\perp/\varepsilon_0)^2 
			+ c^2 (\vect\nabla \times \vect A^\perp)^2 \right] 
	\label{eq:Hamiltonian_com_L_final} \\
	H_\text{X} 
	&= - \frac{(\vect P \cdot \vect p_{\vect r})^2}{2 M^2 \mu c^2} 
		+ \frac{e^2}{4\pi\varepsilon_0 r} \, 
			\frac{(\vect P\cdot \vect r / r)^2}{2 M^2 c^2} \nonumber\\
		&\quad + \frac{m_1 - m_2}{2\mu M^2 c^2} 
			\bigg\{ (\vect P \cdot \vect p_{\vect r}) \vect p_{\vect r}^2 / \mu 
				- \frac{e^2}{8\pi\varepsilon_0}
					\left[\frac{1}{r} \vect P \cdot \vect p_{\vect r} 
					+ \frac{1}{r^3} (\vect P \cdot \vect r) (\vect r \cdot \vect p_{\vect r}) 
					+ \text{H.c.}\right] \bigg\} \\
	H_\text{deriv,new} 
	&= \frac{2\gamma+1}{2Mc^2} 
		[\vect P \cdot (\vect r 
			\cdot \vect\nabla\phi(\vect R)) \vect p_{\vect r} + \text{H.c.}]
\end{align}
\end{subequations}
Here, we have included the term $-\gamma \frac{1}{c^2} 
\frac{e^2}{4\pi \varepsilon_0 r} \phi(\vect R)$ (from the 
`corrections' to the Coulomb term) into $H_\text{A,final}$ 
(instead of into $H_\text{C,final}$) since it can be 
combined with the original Coulomb term from $H_\text{A}$ 
into
\begin{equation} \label{eq:Coulomb_metr}
	-\frac{e^2}{4\pi \varepsilon_0 r} 
		\left(1 + \gamma\frac{\phi(\vect R)}{c^2}\right) 
	= -\frac{e^2}{4\pi \varepsilon_0 \sqrt{{^{(3)}g_{\vect R}} (\vect r, \vect r)}} \; ,
\end{equation}
i.e.\ a Coulomb term expressed with the correct, metric relative distance.

To correctly interpret the atom--light interaction 
Hamiltonian \eqref{eq:Hamiltonian_com_AL_final}, one has 
to keep in mind that the field variables $\tilde{\vect\Pi}^\perp 
= \vect\Pi^\perp - \vect{\mathcal P}^\perp$ and $\vect B 
= \vect\nabla \times \vect A^\perp$ appearing in it do not 
refer to an orthonormal frame in the physical spacetime 
metric $g$ in the presence of gravitational fields, but 
are related to components of the electromagnetic field 
tensor in the \emph{coordinate} frame~-- or, more 
geometrically speaking, to an inertial frame with respect 
to the background Minkowski metric. This issue will be 
discussed in more detail in section \ref{sec:EM_ONB}.

Since the cross terms $H_\text{X}$ are the same as in 
\cite{sonnleitner18}, we could now introduce new canonical 
variables $\vect Q,\vect q,\vect p$ literally as in 
\eqref{eq:coords_decoup_orig} to eliminate these cross 
terms. Since the gravitational correction terms are of 
order $\Or(c^{-2})$, for them this canonical transformation 
would just amount to the replacements $\vect R \to \vect Q, 
\vect r \to \vect q, \vect p_{\vect r} \to \vect p$ at our 
order of approximation. Since it will not alter the 
following discussion, we will not perform this coordinate 
change in order to avoid adding an extra layer of 
potentially confusing notation.

\subsection{The system as a composite point particle}

We now take another look at the central and the internal 
Hamiltonian \eqref{eq:Hamiltonian_com_C_final}, 
\eqref{eq:Hamiltonian_com_A_final}, where we rewrite the 
latter in the form
\begin{align} \label{eq:Hamiltonian_com_A_final_geom}
	H_\text{A,final} 
		&= \frac{{^{(3)}g^{-1}_{\vect R}} (\vect p_{\vect r}, \vect p_{\vect r})}{2\mu} 
			- \frac{e^2}{4\pi\varepsilon_0 \sqrt{{^{(3)}g_{\vect R}} (\vect r, \vect r)}} \nonumber\\
			&\quad - \frac{m_1^3 + m_2^3}{M^3} \, \frac{\vect p_{\vect r}^4}{8 \mu^3 c^2} 
			- \frac{e^2}{4\pi\varepsilon_0} \, \frac{1}{2\mu M c^2} 
				\left( \vect p_{\vect r} \cdot \frac{1}{r} \vect p_{\vect r} 
				+ \vect p_{\vect r} \cdot \vect r \frac{1}{r^3} \vect r 
					\cdot \vect p_{\vect r} \right) \nonumber\\
			&\quad - \frac{2\gamma+1}{2c^2} \, \frac{m_1 - m_2}{m_1 m_2} 
				\vect p_{\vect r} \cdot (\vect r \cdot \vect\nabla\phi(\vect R)) 
				\vect p_{\vect r}
			- \frac{\gamma+1}{c^2} \, \frac{e^2}{8\pi \varepsilon_0 r} \, 
				\frac{m_2 - m_1}{M} \vect r \cdot \vect\nabla \phi(\vect R)
\end{align}
by combining the gravitational correction terms which do 
not involve the potential derivative $\vect\nabla\phi$ 
into metrically defined kinetic energy and Coulomb terms 
as in \eqref{eq:internal_kin_metr}, \eqref{eq:Coulomb_metr}.

Now comparing the central Hamiltonian $H_\text{C,final}$ 
\eqref{eq:Hamiltonian_com_C_final} to the Hamiltonian of 
a single point particle of mass $m$ in the \acronym{PPN} 
metric,
\begin{equation} \label{eq:Hamiltonian_point}
	H_\text{point}(\vect P, \vect R; m) 
	= \frac{\vect P^2}{2m} + m \phi(\vect R) 
		- \frac{\vect P^4}{8m^3 c^2} 
		+ \frac{2\gamma+1}{2mc^2} \vect P \cdot \phi(\vect R) \vect P 
		+ (2\beta-1) \frac{m \phi(\vect R)^2}{2 c^2} \; ,
\end{equation}
we see that the central Hamiltonian has, up to (and 
including) $\Or(c^{-2})$, exactly this form, with the 
mass $m$ replaced by $M + \frac{H_\text{A,final}}{c^2}$,
\begin{equation} \label{eq:Hamiltonian_composite_point}
	H_\text{C,final} 
	= H_\text{point}\bigg(\vect P, \vect R; 
		M + \frac{H_\text{A,final}}{c^2}\bigg),
\end{equation}
as could be naively expected from mass--energy equivalence. 
Thus, starting from first principles, we have shown that 
\emph{the system behaves as a `composite point particle' 
whose (inertial as well as gravitational) mass is comprised 
of the rest masses of the constituent particles as well as 
the internal energy.}
\enlargethispage{\baselineskip}

Note that this conclusion depends on the identification 
of terms as being `kinetic' and `interaction' energies, 
which in turn depends on the metric structure in their 
expressions. Had we not rewritten the internal kinetic 
energy \eqref{eq:internal_kin_metr} and the Coulomb 
interaction \eqref{eq:Coulomb_metr} in terms of the 
physical metric $g$, but included only the corresponding 
terms $\frac{\vect p_{\vect r}^2}{2\mu}$ and 
$-\frac{e^2}{4\pi\varepsilon_0 r}$ with respect to 
the background metric into the internal Hamiltonian, 
the above conclusion could not have resulted. Rather, 
having additional terms in the central Hamiltonian, 
to obtain it in a `composite particle' point of view 
we would have had to replace the inertial mass in 
\eqref{eq:Hamiltonian_point} by $M + H_\text{A}/c^2$ 
and the gravitational mass by $M + (2\gamma+1) 
\frac{\vect p_{\vect r}^2}{2\mu c^2} 
- (\gamma+1) \frac{e^2}{4\pi \varepsilon_0 r c^2} 
= M + \frac{H_\text{A}}{c^2} + \gamma 
\left(2 \frac{\vect p_{\vect r}^2}{2\mu c^2} 
- \frac{e^2}{4\pi \varepsilon_0 r c^2}\right)$, 
which one could have erroneously interpreted as a violation 
of some naive form of the weak equivalence principle. But, 
clearly, such a conclusion would be premature, for it is 
based on the identification of terms --~like inertial and 
gravitational mass~-- that is itself ambiguous. That 
ambiguity is here seen as a dependence on the background 
structure, which is used to define distances of positions 
and squares of momenta. Once these quantities are measured 
with the physical metric $g$, ambiguities and apparent 
conflicts with naive expectations disappear. That point 
has also been made in \cite{zych19}.

The quantities $\vec p'^2$ and $r'$ entering the 
Hamiltonian in \cite[eq. (18)]{zych19}, which are, in the 
language of \cite{zych19}, the square of the internal 
momentum and the distance `in the \acronym{CM} rest frame', 
are nothing but the geometric expressions 
${^{(3)}g^{-1}_{\vect R}} (\vect p_{\vect r}, \vect p_{\vect r})$ 
and $\sqrt{{^{(3)}g_{\vect R}} (\vect r, \vect r)}$ from 
above, measured using the physical metric of space. The 
internal Hamiltonian \eqref{eq:Hamiltonian_com_A_final_geom} 
thus consists of kinetic and Coulomb interaction energies 
in terms of the physical geometry, in agreement with the 
expressions from \cite{zych19}, as well as the expected 
special-relativistic and `Darwin' corrections, and terms 
involving the gravitational potential's derivative\footnote
	{Since all these corrections are themselves of order $1/c^2$, 
	the deviations of the physical from the flat metric do not 
	enter here.}.

\subsection{The electromagnetic expressions in terms of components with respect to orthonormal frames}
\label{sec:EM_ONB}

The expressions derived above in \eqref{eq:Hamiltonian_com_final} 
include components of the external electromagnetic field 
with respect to coordinates which, albeit not chosen 
arbitrarily, have no direct metric significance. We 
recall that we used coordinates that are adapted to the 
background structure $(\eta,u)$, in the sense that 
$u = \partial/\partial t$ and $\eta = \eta_{\mu\nu} \, 
\D x^\mu \otimes \D x^\nu$ with $(\eta_{\mu\nu}) 
= \mathrm{diag}(-1,1,1,1)$. The corresponding local 
reference frames $(\partial_\mu = \partial/\partial x^\mu) 
_{\mu = 0,1,2,3}$ are orthonormal with respect to the 
background metric $\eta$, but not with respect to the 
physical metric $g$.
\enlargethispage{\baselineskip}

In this section we will re-express our findings in 
terms of components with respect to $g$-orthonormal 
frames, which we will call the `physical 
components', as opposed to the `coordinate components' 
used so far. We stress that, despite this terminology, 
there is nothing wrong or `unphysical' with representing 
fields in terms of components of non-orthonormal bases, 
as long as the metric properties are spelled out at the 
same time. Yet it is clearly convenient to be able to 
read off metric properties, which bear direct metric 
significance, from the expressions involving the components 
alone, without at the same time having to recall the values 
of the metric components as well.

\subsubsection{Electromagnetic quantities in non-orthonormal and orthonormal frames}

At first, we will discuss the meaning of several 
`electromagnetic quantities' in non-orthonormal and 
orthonormal frames in general, before specialising to the 
case of our \acronym{PPN} metric and considering the terms 
in our Hamiltonian. Suppose that we are given our usual 
`background' coordinate system $(x^0,x^a)$, that is 
possibly non-orthonormal with respect to the physical 
spacetime metric $g$, as well as a time- and space-oriented 
orthonormal frame / `tetrad' $(\E_{\ul{\mu}})_{\ul{\mu} = 0,1,2,3}$ 
for the physical metric, where from now on underlined 
indices refer to components with respect to the tetrad. We 
write the tetrad fields in terms of the coordinate basis 
fields, and vice versa, as
\begin{equation}
	\E_{\ul{\mu}} = \E_{\ul{\mu}}^\nu \partial_\nu \; , \quad
	\partial_\mu = \E^{\ul{\nu}}_\mu \E_{\ul{\nu}} \; ,
\end{equation}
where the matrices of coefficients $(\E_{\ul{\mu}}^\nu)$ 
and $(\E^{\ul{\nu}}_\mu)$ are inverses of each other. Since 
the tetrad is orthonormal, i.e. the tetrad components 
$g_{\ul{\mu}\ul{\nu}}$ of the metric are numerically equal 
to the components $\eta_{\mu\nu}$ of the Minkowski metric 
in a Lorentz frame, we can express the coordinate 
components of the metric as
\begin{equation}
	g_{\mu\nu} = \E^{\ul{\rho}}_\mu \E^{\ul{\sigma}}_\nu g_{\ul{\rho}\ul{\sigma}}
	= \sum_{\rho, \sigma = 0}^3 \E^{\ul{\rho}}_\mu \E^{\ul{\sigma}}_\nu \eta_{\rho\sigma} \; .
\end{equation}

We make the further assumption that at any point, time and 
space as defined by both bases be the same, i.e.\ that
\begin{equation}
	\mathrm{span}\{e_0\} = \mathrm{span}\{\partial_0\}, \quad
	\mathrm{span}\{\E_1, \E_2, \E_3\} = \mathrm{span}\{\partial_1, \partial_2, \partial_3\},
\end{equation}
i.e.\ that the coefficients $\E_{\ul{0}}^a, \E_{\ul{a}}^0, 
\E^{\ul{0}}_a, \E^{\ul{a}}_0, g_{0a}, g^{0a}$ all vanish 
(which will in our later application to the electromagnetic 
Hamiltonian be satisfied to our order of expansion). This 
also implies that $\sqrt{-g_{00}} = \E^{\ul{0}}_0 = 1/\E_{\ul{0}}^0$.

Now we can consider how electromagnetic quantities are 
related to the components of the field tensor in the 
two different frames. The electric and magnetic field 
components with respect to the tetrad (which we call, as 
introduced above, the `physical' field components) are 
\begin{equation} \label{eq:EM_fields_tetrad}
	E_{\text{phys.}\ul{a}} 
	= c F_{\ul{a} \ul{0}} \; , \quad
	B_\text{phys.}^{\ul{a}} 
	= {^{(3)}\hat\varepsilon^{\ul{a}\ul{b}\ul{c}}} F_{\ul{b}\ul{c}}
\end{equation}
in terms of the tetrad components of the field tensor, 
where $^{(3)}\hat\varepsilon^{\ul{a}\ul{b}\ul{c}}$ is the 
three-dimensional totally antisymmetric symbol. Note that 
although written in component form, these formulae have 
an invariant geometric meaning that depends only on being 
given the time direction $\mathrm{span}\{\partial_0\} 
= \mathrm{span}\{\E_0\}$ and the (physical) metric $g$: the 
electric field is simply the spatial\footnote
	{Here `spatial' means geometric objects (i.e.\ tensor 
	densities) defined on the `spatial' submanifolds 
	$x^0 = \mathrm{const.}$, which integrate the spatial 
	distribution $\mathrm{span}\{\E_1, \E_2, \E_3\} 
	= \mathrm{span}\{\partial_1, \partial_2, \partial_3\}$.}
one-form obtained by inserting the unit future-pointing 
time direction vector $\E_0$ into the two-form $F$ (in its 
second argument), projecting onto space, and multiplying 
with $c$. Considering the magnetic field, we recall the 
well-known fact that we can view 
$^{(3)}\hat\varepsilon^{\ul{a}\ul{b}\ul{c}}$ as 
the components of a spatial tensor density 
$^{(3)}\hat\varepsilon$ of weight\footnote
	{A more well-known fact is probably that the \emph{covariant} 
	totally antisymmetric symbol is a tensor density of weight 
	$-1$; that the contravariant symbol may be considered a 
	density of weight $+1$ one sees in exactly analogous fashion.}
$+1$ which is defined by demanding that its components in 
\emph{any} positively oriented frame are given by the 
totally antisymmetric symbol. Thus, the magnetic field is a 
contraction of this tensor density $^{(3)}\hat\varepsilon$ 
and the (spatially projected) field tensor (which is a 
proper tensor, i.e.\ a density of weight $0$), i.e.\ itself 
a spatial vector density of weight $+1$.

This means that what one might call the `coordinate 
components' of the magnetic field, namely the expressions
\begin{equation}
	B_\text{coord.}^a 
	 = {^{(3)}\hat\varepsilon^{abc}} F_{bc}
\end{equation}
where ${^{(3)}\hat\varepsilon^{abc}}$ is the totally 
antisymmetric symbol, are in fact the components with 
respect to the coordinate frame of \emph{the same} 
geometric object as for the `physical components', namely 
the above tensor density. Thus, the components are related 
by the usual transformation formula for tensor densities, 
i.e.\
\begin{equation}
	B_\text{phys.}^{\ul{a}} 
	= \det(\E_{\ul{c}}^d) \cdot \E^{\ul{a}}_b B_\text{coord.}^b 
	= \frac{1}{\sqrt{^{(3)}g}} \E^{\ul{a}}_b B_\text{coord.}^b \; ,
\end{equation}
where $^{(3)}g$ denotes the determinant of the matrix of 
coordinate components of the spatial metric. Note that 
due to the numerical identity $^{(3)}\hat\varepsilon^{abc} 
= {^{(3)}\varepsilon^{abc}}$, where $^{(3)}\varepsilon^{abc}$ 
are the `index-raised' components of the background spatial 
volume form as introduced in section \ref{sec:geometric_notation_conventions}, 
the components $B_\text{coord.}^a$ are, in fact, 
numerically equal to the components of the `three-vector' 
$\vect B = \vect\nabla \times \vect A^\perp$ used in the 
previous sections (although we treated $\vect B$ as a 
different geometric object there, namely as a spatial 
vector field instead of a spatial vector field density).
The interpretation of the magnetic field as a vector 
density also goes nicely with an intuitive point of view, 
namely that of the field representing the spatial density 
of magnetic field lines.

Expressing the electric field's tetrad components in terms 
of the coordinate components of the field tensor, we 
directly obtain
\begin{equation}
	E_{\text{phys.}\ul{a}} 
	= \E_{\ul{a}}^b \E_{\ul{0}}^0 c F_{b0}
	= \frac{1}{\sqrt{-g_{00}}} \E_{\ul{a}}^b c F_{b0} \; .
\end{equation}
One could interpret $\frac{1}{\sqrt{-g_{00}}} c F_{b0} 
=: E_{\text{coord.}b}$ as the `coordinate components' of 
the electric field (interpreted as a spatial one-form as 
discussed above); however, we will not make much use of 
this notation.\footnote
	{In index-free notation, we can express the electric and 
	magnetic fields as follows. The electric field is 
	\begin{equation*}
		E = - c \kern.1em \iota_{\E_{\ul{0}}} F
	\end{equation*}
	where $\iota$ denotes the interior product of a vector 
	field and a differential form, i.e.\ insertion of the 
	vector field into the \emph{first} argument of the form. 
	Note that $E$ is spatial due to the antisymmetry of $F$.

	\newcommand{\voldens}{{^{(3)}\rlap{\raisebox{-.35ex}{$\kern-.06em\widetilde{\phantom{\mathrm{vol}}}$}}\mathrm{vol}}}
	The magnetic field can be expressed as
	\begin{equation*}
		B = \voldens \cdot \left[ {^{(3)}\tilde{*}} 
			\left( \mathrm{Pr}^\perp(F) \right) \right]^{\tilde\sharp} .
	\end{equation*}
	The objects occurring in this formula are the following: 
	$\mathrm{Pr}^\perp$ denotes the orthogonal projection map 
	onto three-space $\mathrm{span}\{\E_{\ul{0}}\}^\perp$, 
	extended to arbitrary tensors. The operator $^{(3)}\tilde{*}$ 
	is the spatial Hodge star with respect to the physical 
	spatial metric $^{(3)}g$. A superscript $\tilde\sharp$ 
	denotes the natural isomorphism from spatial one-forms to 
	spatial vector fields induced by the physical spatial 
	metric, i.e.\ ${^{(3)}g}(\alpha^{\tilde\sharp}, \cdot) 
	= \alpha$. Finally, $\voldens$ denotes the spatial `volume 
	density', i.e. the spatial scalar density whose value in 
	any frame is the $^{(3)}g$-volume of the parallelepiped 
	spanned by the frame's vectors. (The value of $\voldens$ is 
	given by the square root of the determinant of the matrix 
	of $^{(3)}g$'s components in the respective frame.)}

Next, we will consider electric dipole moments and 
(electric) polarisation. Imagining an ideal dipole in the 
usual way as arising in a limit process from two separated 
opposite point charges getting closer and closer, with 
their respective charges growing accordingly, the resulting 
dipole moment is to be a proper spatial vector (and not a 
density of non-zero weight): its magnitude is an invariant 
(i.\,e.\ frame-independent) quantity, namely the product of 
charge and distance of the two particles, held constant in 
the limit process, and its direction is the limit of the 
direction `from one particle to the other'. Therefore, when 
a dipole moment has coordinate components $d_\text{coord.}^a$, 
its tetrad components are simply
\begin{equation}
	d_\text{phys.}^{\ul{a}} = \E^{\ul{a}}_b d_\text{coord.}^b \; .
\end{equation}
Now, since a polarisation is simply a density of dipole 
moment per spatial volume, the natural geometric 
perspective is that polarisation is a spatial vector 
density field. Thus a polarisation field with coordinate 
components $\mathcal P_\text{coord.}^a$ has tetrad 
components
\begin{equation}
	\mathcal P_\text{phys.}^{\ul{a}} 
	= \det(\E_{\ul{c}}^d) \cdot \E^{\ul{a}}_b \mathcal P_\text{coord.}^b 
	= \frac{1}{\sqrt{^{(3)}g}} \E^{\ul{a}}_b \mathcal P_\text{coord.}^b \; .
\end{equation}

Finally, we turn to the electric displacement field. In 
vacuum, it has the interpretation of electric flux density, 
with the only contribution to its value coming from 
the electric field times $\varepsilon_0$. However, the 
displacement is to be a vector \emph{density}, so one has 
to identify the electric field spatial one-form with the 
corresponding vector field via the metric (`index raising') 
and consider the density that is metrically associated to 
this. In a medium, the displacement field is the sum of 
this vacuum displacement and the polarisation, i.e.\ we have
\begin{equation}
	D_\text{coord.}^a 
	= \varepsilon_0 \, {^{(3)}g^{ab}} E_{\text{coord.}b} 
			\cdot \sqrt{^{(3)}g} 
		+ \mathcal P_\text{coord.}^a \; .
\end{equation}
In tetrad components, it takes the form
\begin{equation}
	D_\text{phys.}^{\ul{a}} 
	= \varepsilon_0 \, {^{(3)}g^{\ul{a}\ul{b}}} E_{\text{phys.}\ul{b}} 
		+ \mathcal P_\text{phys.}^{\ul{a}} \; ,
\end{equation}
where $^{(3)}g^{\ul{a}\ul{b}}$ are the tetrad components of 
the physical spatial metric, which are numerically equal to 
the Kronecker delta $\delta^{ab}$.

\subsubsection{Application to the Hamiltonian}

We will now rewrite the parts of the total Hamiltonian 
\eqref{eq:Hamiltonian_com_final} in which the external 
electromagnetic field appears in terms of electromagnetic 
quantities in an orthonormal tetrad, as discussed above.
Due to the form of the Eddington--Robertson \acronym{PPN} 
metric, in order to obtain a tetrad, to our order of 
approximation we simply need to divide each of the 
coordinate basis vectors $\partial_\mu$ by the square root 
of the modulus of $g_{\mu\mu}$ (no summation):
\begin{subequations} \label{eq:tetrad}
\begin{align}
	\E_{\ul{0}} = \frac{1}{\sqrt{-g_{00}}} \partial_0 
	&= \left(1 - \frac{\phi}{c^2}\right) \partial_0 \\
	\E_{\ul{a}} = \frac{1}{\sqrt{g_{aa}}} \partial_a 
	&= \left(1 + \gamma \frac{\phi}{c^2}\right) \partial_a
\end{align}
\end{subequations}

Inserting this explicit form of the tetrad, the relevant 
equations from above relating the tetrad components of 
electromagnetic quantities to their coordinate components 
attain the following numerical forms:
\begin{subequations} \label{eq:EM_tetrad_explicit}
\begin{align}
	B_\text{phys.}^{\ul{a}} 
	&= \left(1 + 2 \gamma \frac{\phi}{c^2}\right) B_\text{coord.}^a 
	\label{eq:B_tetrad_explicit} \\
	d_\text{phys.}^{\ul{a}}
	&= \left(1 - \gamma \frac{\phi}{c^2}\right) d_\text{coord.}^a \\
	\mathcal P_\text{phys.}^{\ul{a}} 
	&= \left(1 + 2 \gamma \frac{\phi}{c^2}\right) \mathcal P_\text{coord.}^a \\
	D_\text{phys.}^{\ul{a}} 
	&= \varepsilon_0 \left(1 + (\gamma - 1) \frac{\phi}{c^2}\right) 
			c F_{a0} 
		+ \left(1 + 2 \gamma \frac{\phi}{c^2}\right) \mathcal P_\text{coord.}^a
	\label{eq:displacement_tetrad_explicit}
\end{align}
\end{subequations}
Comparing \eqref{eq:displacement_tetrad_explicit} to 
the form \eqref{eq:field_momentum_with_gravity} of the 
`would-be canonical field momentum' $\vect\Pi^\perp$ and 
its relation $\tilde{\vect\Pi}^\perp = \vect\Pi^\perp 
- \vect{\mathcal P}^\perp$ to the `would-be field momentum' 
after the \acronym{PZW} transformation, we can relate 
the coordinate components of the latter to the tetrad 
components of the displacement field by
\begin{equation} \label{eq:displacement_tetrad_momentum}
	D_\text{phys.}^{\perp \, \ul{a}} 
	= - \left(1 + 2 \gamma \frac{\phi}{c^2}\right) \tilde\Pi^{\perp \, a} 
		+ \varepsilon_0 (\gamma+1) \bigg( \phi_\text{el.}^{(0)} 
			\frac{\vect\nabla\phi}{c^2} \bigg)^{\kern-.3em\perp \, a} \; .
\end{equation}
Note that up to the additional second term arising from 
the additional gravitational coupling in the Lagrangian, 
this means that the canonical momentum is just minus 
the displacement (when interpreted as a spatial vector 
density), as in the non-gravitational case after a 
\acronym{PZW} transformation.

Now, using the relations in \eqref{eq:EM_tetrad_explicit} 
and \eqref{eq:displacement_tetrad_momentum}, we can 
express all the interaction terms from 
\eqref{eq:Hamiltonian_com_AL_final} in terms of `physical', 
i.e.\ tetrad, components of the external electromagnetic 
quantities. For example, the electric dipole interaction 
term $\left(1 + (\gamma+1) \frac{\phi(\vect R)}{c^2}\right) 
\frac{\vect{\tilde\Pi}^\perp (\vect R)}{\varepsilon_0} 
\cdot \vect d$ in the Hamiltonian takes the form
\begin{align}
	&\hspace{-1em}
	-\left(1 + \frac{\phi(\vect R)}{c^2}\right) \sum_a
			\frac{D_\text{phys.}^{\perp \, \ul{a}}(\vect R)}{\varepsilon_0} \, 
			d_\text{phys.}^{\ul{a}}
		+ (\gamma+1) \sum_a \bigg( \phi_\text{el.}^{(0)} 
			\frac{\vect\nabla\phi}{c^2} \bigg)^{\kern-.3em\perp \, a}
			\kern-.7em(\vect R) \, d_\text{phys.}^{\ul{a}} \nonumber\\
	&= -\sqrt{-g_{00}(\vect R)} \, \sum_a
			\frac{D_\text{phys.}^{\perp \, \ul{a}}(\vect R)}{\varepsilon_0} \, 
			d_\text{phys.}^{\ul{a}}
		+ (\gamma+1) \sum_a \bigg( \phi_\text{el.}^{(0)} 
			\frac{\vect\nabla\phi}{c^2} \bigg)^{\kern-.3em\perp \, a}
			\kern-.7em(\vect R) \, d_\text{phys.}^{\ul{a}}
\end{align}
when expressed in terms of physical components. The 
`gravitational time dilation' factor $\sqrt{-g_{00}}$ 
in this expression could now also be absorbed by referring 
the time evolution to the proper time of the observer 
situated at $\vect R$ instead of coordinate time, 
leading to a dipole coupling of the usual form 
`$- \frac{\vect D}{\varepsilon_0} \cdot \vect d$\,' 
\cite{marzlin95,laemmerzahl95} (up to the additional term 
originating from the additional $\partial_t \vect A^\perp$ 
coupling in $L_\text{em}$ \eqref{eq:Lagrangian_em_grav}).

Similarly, all the other interaction terms from the 
Hamiltonian \eqref{eq:Hamiltonian_com_AL_final} can be 
rewritten in terms of tetrad components. The only 
difficulty arises when considering the Röntgen term, i.e.\ 
the second term in the interaction Hamiltonian, since it 
involves the momentum $\vect P$, and the similar third 
term: if the components $P_a$ were just the components 
of a classical one-form field, there would be no problem 
in computing its tetrad components as
\begin{equation} \label{eq:momentum_phys}
	P_{\text{phys.}\ul{a}} 
	= \E^b_{\ul{a}} P_b 
	= \left(1 - \gamma \frac{\phi(\vect R)}{c^2}\right) P_a \; .
\end{equation}
However, the $P_a$ are operators that don't commute with 
the centre of mass position $\vect R$, such that in the 
application of \eqref{eq:momentum_phys} one has to deal 
with with operator ordering ambiguities (which is, of 
course, a well-known issue regarding curvilinear coordinate 
transformations in quantum mechanics). Of course, to avoid 
dealing with these ambiguities, one can stay with the 
coordinate components of the momentum and rewrite only the 
\emph{other} quantities in terms of tetrad components, 
arriving at 
\begin{equation}
	\frac{1}{2M} \{\vect P \cdot [\vect d \times \vect B(\vect R)] + \text{H.c.}\} 
	= \frac{1}{2M} \left\{ \sum_a P_a 
			\left(1 - \gamma \frac{\phi(\vect R)}{c^2}\right) 
			{^{(3)}\tilde\varepsilon_{\ul{a}\ul{b}\ul{c}}} \, 
			d_\text{phys.}^{\ul{b}} B_\text{phys.}^{\ul{b}}(\vect R) 
		+ \text{H.c.}\right\}
\end{equation}
for the Röntgen term, where $^{(3)}\tilde\varepsilon$ 
denotes the spatial volume form induced by the physical 
metric (with tetrad components given by the antisymmetric 
symbol). When doing so, to give a well-defined geometric 
meaning to the resulting expression on the right-hand side, 
one has to keep in mind that the components $P_a$ of the 
momentum refer to the coordinate basis and the components 
$d_\text{phys.}^{\ul{b}}, B_\text{phys.}^{\ul{b}}$ of the 
dipole moment and the magnetic field refer to the tetrad.
\enlargethispage{\baselineskip}

Rewriting all possible terms in the atom--light interaction 
Hamiltonian in terms of tetrad components, glossing over 
the just-described ordering ambiguities, we arrive at
\begin{align} \label{eq:Hamiltonian_com_AL_final_tetrad}
	H_\text{AL,final} 
	&= -\sqrt{-g_{00}(\vect R)} \,
			\frac{\vect D_\text{phys.}^\perp(\vect R)}{\varepsilon_0} 
			\cdot \vect d_\text{phys.} 
		+ \frac{1}{2M} \{\vect P_\text{phys.} 
				\cdot [\vect d_\text{phys.} 
				\times \vect B_\text{phys.}(\vect R)] 
			+ \text{H.c.}\}  \nonumber\\
		&\quad - \frac{m_1 - m_2}{4 m_1 m_2} 
			\{\vect p_{\vect r \; \text{phys.}} 
				\cdot [\vect d_\text{phys.} 
				\times \vect B_\text{phys.}(\vect R)] 
			+ \text{H.c.}\} \nonumber\\
		&\quad + \frac{1}{8\mu} \left(1 - 
				2 \gamma \frac{\phi(\vect R)}{c^2}\right) 
			(\vect d_\text{phys.} \times \vect B_\text{phys.}(\vect R))^2 
		+ \frac{1}{2\varepsilon_0} \int\D^3\ivect x \, \sqrt{-g} \, 
			{\vect{\mathcal P}_{d \; \text{phys.}}^\perp}^{\kern-1.5em 2} 
			\kern1em (\vect x, t) \nonumber\\
		&\quad + \int\D^3\ivect x \, (\gamma+1) \, \phi_\text{el.}^{(0)} 
			\frac{\vect\nabla \phi}{c^2} \cdot \vect D_\text{phys.}^\perp
		\; .
\end{align}
Here we employed `three-vector' notation also for 
three-tuples of tetrad components, i.e.\ a `dot product' 
$\vect X_\text{phys.} \cdot \vect Y_\text{phys.} 
:= \sum_{\ul{a}} X_\text{phys.}^{\ul{a}} Y_\text{phys.}^{\ul{a}}$ 
is a scalar product with respect to the \emph{physical} 
spatial metric, and a `cross product' 
$(\vect Y_\text{phys.} \times \vect Z_\text{phys.})_{\ul{a}} 
= {^{(3)}\tilde\varepsilon_{\ul{a}\ul{b}\ul{c}}} 
Y_\text{phys.}^{\ul{b}} \vect Z_\text{phys.}^{\ul{c}}$ 
is also defined by the spatial volume form 
$^{(3)}\tilde\varepsilon$ induced by the \emph{physical} 
spatial metric.

To the best of our knowledge, the atom--light interaction terms in the 
presence of gravity obtained in \eqref{eq:Hamiltonian_com_AL_final} 
and discussed above are new, save for the electric dipole 
coupling which was already discussed in \cite{marzlin95,laemmerzahl95}.
\enlargethispage{\baselineskip}

Finally, expressing the external field energy 
\eqref{eq:Hamiltonian_com_L_final} in terms of tetrad 
components, i.e.\ inserting \eqref{eq:B_tetrad_explicit} 
and \eqref{eq:displacement_tetrad_momentum}, we obtain
\begin{align}
	H_\text{L,final} 
	&= \frac{\varepsilon_0}{2} \int\D^3\ivect x \, 
			\left(1 + (1 - 3\gamma) \frac{\phi}{c^2}\right) 
			\left[ (\vect D_\text{phys.}^\perp/\varepsilon_0)^2 
				+ c^2 \vect B_\text{phys.}^2 \right] \nonumber\\
		&\quad - \int\D^3\ivect x \, (\gamma+1) \, \phi_\text{el.}^{(0)} 
			\frac{\vect\nabla \phi}{c^2} \cdot \vect D_\text{phys.}^\perp
		\nonumber \displaybreak[0]\\
	&= \frac{\varepsilon_0}{2} \int\D^3\ivect x \, \sqrt{-g}
			\left[ (\vect D_\text{phys.}^\perp/\varepsilon_0)^2 
				+ c^2 \vect B_\text{phys.}^2 \right] \nonumber\\
		&\quad - \int\D^3\ivect x \, (\gamma+1) \, \phi_\text{el.}^{(0)} 
			\frac{\vect\nabla \phi}{c^2} \cdot \vect D_\text{phys.}^\perp 
		\; .
\end{align}
Up to the second integral, which cancels with the last term 
from \eqref{eq:Hamiltonian_com_AL_final_tetrad}, this is 
the standard result of the flat-spacetime electromagnetic 
field energy \cite{jackson98} minimally coupled to gravity 
\cite{misner73}, as was to be expected.\footnote
	{This result would have been immediate if we did the whole 
	calculation in terms of tetrad components instead of 
	coordinate components, as would have some steps in the 
	calculation of the electromagnetic Lagrangian. However, as 
	stressed in section \ref{sec:including_gravity}, the 
	approach based on the background structures with adapted 
	coordinates enabled us to provide a direct comparison with 
	the original calculation of \cite{sonnleitner18}.}

As we have seen in the previous section for internal 
energies and in this section for electromagnetic quantities, 
several terms in the final post-Newtonian Hamiltonian 
\eqref{eq:Hamiltonian_com_final} obtain a natural 
interpretation when expressed in terms of quantities with 
direct metric significance, i.e.\ in terms of components 
with respect to an orthonormal tetrad frame. Note, however, 
that such a tetrad \eqref{eq:tetrad} depends on the metric 
$g$, i.e.\ on part of the physical field configuration. 
This entails that, when comparing physical situations in 
\emph{different} gravitational fields, i.e.\ with different 
physical metrics $g$, it is not at all conceptually obvious 
how to relate predictions made for the two situations to 
each other: even though the \emph{Hamiltonian} looks 
the same in both cases when expressed in terms of tetrad 
components, it may be the case that the quantum-mechanical 
state vector take different forms when expressed in terms 
of metric quantities in the two situations, due to some 
specific nature of its preparation procedure (which might, 
for example, depend on spacetime curvature in some way).

Thus, for a proper interpretation of calculational 
predictions for experimental situations, one has (in 
principle) to describe the \emph{whole} experimental 
situation, including all preparation and measurement 
procedures, in terms of operationally defined quantities, 
and express all predicted results in terms of these 
operational quantities. This is the only way to ensure 
true coordinate- and frame-independence of predictions.


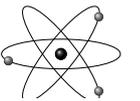
\newtheorem{thm}{Theorem}[section]
\newenvironment{thmqed}{\pushQED{\qed}\begin{thm}}{\popQED\end{thm}} 
\newtheorem{lem}[thm]{Lemma}
\newenvironment{lemqed}{\pushQED{\qed}\begin{lem}}{\popQED\end{lem}}
\newtheorem{prop}[thm]{Proposition}

\theoremstyle{definition}
\newtheorem{defn}[thm]{Definition}
\newtheorem{exmp}[thm]{Example}

\chapter{Classical perspectives on the Newton--Wigner position observable}
\label{chap:Newton--Wigner}

This chapter, which is thematically entirely independent from 
the rest of the thesis, deals with the Newton--Wigner position 
observable for Poincaré-invariant \emph{classical} systems. 
To explain at least the little connection to the rest of the 
thesis that there is, let me (the author) briefly describe how 
my interest in the topics of this chapter arose.
Sonnleitner and Barnett in \cite{sonnleitner18}, as 
well as myself in my calculations based on theirs as 
documented in chapter \ref{chap:atom_in_gravity}, employed 
\emph{Newtonian} centre of mass coordinates in the 
description of a (locally) Poincaré-symmetric composite 
system, simply for the sake of computational simplicity. 
This led me to the old question of what kind of central 
positions one could --~or perhaps should?~-- use for such 
descriptions. Thus, I was led to extending my knowledge 
of special-relativistic localisation and position 
observables, in particular with the beautiful geometric 
`hyperplane observable' perspective of Fleming 
\cite{Fleming:1965a}. In the course of this, I wondered if 
and how one could understand Fleming's `centre of spin' 
interpretation of the Newton--Wigner observable in a more 
geometric way, and also if the quantum Newton--Wigner 
theorem has a classical analogue (which it `should' have, 
morally speaking). This chapter is the outcome of those 
considerations. We will prove an existence and uniqueness 
theorem for elementary systems that parallels the 
well-known Newton--Wigner theorem in the quantum 
context, and also discuss and justify Fleming's geometric 
interpretation of the Newton--Wigner position as 
`centre of spin'.

Other than in the previous chapters, here we will be fully 
mathematically rigorous, and also adopt a more mathematical 
style of presentation.
The material in this chapter is also contained in the 
preprint \cite{Schwartz.Giulini:2020}, which is under 
consideration for publication as of the writing of this 
thesis.\enlargethispage*{2\baselineskip}

\section{Introduction}

Even though we shall in this chapter exclusively deal 
with \emph{classical} (i.e.\ non-quantum) aspects of 
the Newton--Wigner position observable, we wish to 
start with a brief discussion of its historic 
origin, which is based in the early history of 
relativistic quantum field theory (\acronym{RQFT}). After 
that we will briefly remark on its classical importance 
and give an outline of the investigation that is to follow.
A more detailed overview of the history of the localisation 
problem in special-relativistic quantum theory may be found 
in our preprint \cite{Schwartz.Giulini:2020}.

As is well-known, the Newtonian concepts of spatial 
position of elementary, i.e. \mbox{indecomposable}, systems and 
of centre of mass of composite systems satisfy 
the \mbox{expected} covariance properties under spatial 
translations and rotations, and \mbox{readily} translate to 
ordinary, Galilei-invariant quantum mechanics. There, 
concepts like `\mbox{position} operators' and the associated 
projection operators for positions within any measurable 
subset of space can be defined, again fulfilling the 
expected transformation rules under spatial motions.

However, serious difficulties with naive localisation 
concepts arise in attempts to combine quantum mechanics 
with special relativity, connected with the fact that 
negative-energy modes are necessarily introduced if a 
`naive position operator' (like multiplying a naive 
`wave function' with the position coordinate) is applied 
to a positive-energy state. However, in 1949, Newton and 
Wigner showed that it was nevertheless possible to define 
localised states in a special-relativistic quantum context
\cite{Newton.Wigner:1949}: their method was to 
write down axioms for what it meant that a 
system is `localised in space at a given time' 
and then investigate existence as well as uniqueness 
for corresponding position operators. It turned 
out that existence and uniqueness are indeed given 
for elementary systems (with fields being elements of 
irreducible representations of the Poincaré group), 
except for massless fields of higher helicity. 
A more rigorous derivation was later given by Wightman 
\cite{Wightman:1962}.

In 1965, Fleming gave a geometric discussion of 
special-relativistic position observables
\cite{Fleming:1965a} that highlighted the 
group-theoretic properties (regarding the group 
of spacetime automorphisms) underlying several 
constructions and thereby clarified many of the 
sometimes controversial issues regarding 
`covariance'. Fleming focussed on three position 
observables which he called `centre of inertia', 
`centre of mass', and the Newton--Wigner position 
observable, for which he, at the very end of his 
paper and almost in passing, suggested the name 
`centre of spin'. We shall give a more detailed 
geometric justification for that name in this chapter.

It should be emphasised that the Newton--Wigner 
notion of localisation still suffers from the 
acausal spreading of localisation domains that 
is typical of fields satisfying special-relativistic 
wave equations, an observation made many times in the 
literature in one form or another; see, e.g., \cite{Segal.Goodman:1965,Hegerfeldt:1974,Ruijsenaars:1981}: 
if a system is Newton--Wigner localised 
at a point in space at a time $t$, it is not strictly 
localised anymore in any bounded region of space at 
any time later than $t$ 
\cite{Newton.Wigner:1949,Wightman.Schweber:1955}. 
Conceptual issues of that sort, and related ones 
concerning, in particular, the relation between 
Newton--Wigner localisation and the Reeh-Schlieder 
theorem in \acronym{RQFT} have been discussed many times 
in the literature even up to the more recent past; see, 
e.g., \cite{Fleming.Butterfield:1999} and 
\cite{Fleming:2000,Halvorson:2001}. 
For us, however, these quantum field theoretic issues 
are not the point of interest.

Clearly, due to its historical development, 
most discussions of Newton--Wigner localisation 
put their emphasis on its relevance for \acronym{RQFT}: 
the study of deeply relativistic \emph{classical} systems 
was simply not considered relevant at the time 
when special-relativistic localisation was 
first investigated. However, that has clearly 
changed with the advent of modern relativistic 
astrophysics. For example, modern analytical studies 
of close compact binary-star systems also make use 
of various definitions of `centre of mass' in an 
attempt to separate the `overall' from the `internal' 
motion as far as possible. In that respect, it 
turns out that modern treatments of gravitationally 
interacting two-body systems within the theoretical 
framework of Hamiltonian general relativity show 
a clear preference for the Newton--Wigner position 
\cite{Steinhoff:2011,Schaefer.Jaranowski:2018}, 
emphasising once more its distinguished role, 
now in a purely classical context. A concise account 
of the various definitions of `centres' that 
have been used in the context of general 
relativity is given in \cite{Costa.EtAl:2018}, which 
also contains most of the original references 
in its bibliography. In our opinion, all this 
provides sufficient motivation for further attempts 
to work out the characteristic properties of 
Newton--Wigner localisation in the \emph{classical} realm.

The plan of our investigation is as follows. 
After setting up our notation and conventions in 
section \ref{sec:notation}, where we also 
introduce some mathematical background, we 
prove a few results in 
section \ref{sec:NW_centre_of_spin} which are intended to 
explain in what sense the Newton--Wigner position 
is indeed a `centre of spin' and in what sense 
it is uniquely so (theorem \ref{thm:NW_SSC_centre_of_spin}). 
We continue in section \ref{sec:NW_theorem} 
with the statement and proof of a classical 
analogue of the Newton--Wigner theorem, according to 
which the Newton--Wigner position is the 
unique observable satisfying a set of axioms. 
The result is presented in theorem \ref{thm:NW} 
and in a slightly different formulation in theorem \ref{thm:NW2}. 
They say that for a classical elementary 
Poincaré-invariant system with timelike four-momentum 
(as classified by Arens \cite{Arens:1971a,Arens:1971b}), 
there is a unique observable transforming 
`as a position should' under translations, rotations, 
and time reversal, having Poisson-commuting 
components, and satisfying a regularity condition 
(being $C^1$ on all of phase space). This observable 
is the Newton--Wigner position.

\section{Notation and conventions}
\label{sec:notation}

This section is meant to list our notation and 
conventions in the general sense, by also providing 
some background material on the geometric and 
group-theoretic setting onto which the following 
two sections are based.

\subsection{Minkowski spacetime and the Poincaré group}

As before, we use the `mostly plus' $(-{++}+)$ signature 
convention for the spacetime metric and stick to four 
spacetime dimensions. This is not to say 
that our analysis cannot be generalised to other 
dimensions. In fact, as will become clear as we 
proceed, many of our statements have an obvious 
generalisation to higher dimensions. On the 
other hand, as will also become 
clear, there are a few constructions which would 
definitely look different in other dimensions, 
like, e.g., the use of the Pauli--Lubański `vector' 
in section \ref{sec:PauliLubanski}, which 
becomes an $(n-3)$-form in $n$ dimensions, 
or the classification of elementary systems.

In this chapter, we will view Minkowski spacetime as 
an affine space $M$, and the corresponding vector 
space of `difference vectors' will be denoted by $V$. 
The Minkowski metric will be denoted by 
$\eta \colon V\times V \to \mathbb R$. 
The isomorphism of $V$ with its dual space 
$V^*$ induced by $\eta$ (`index lowering') will be denoted by a 
superscript `flat' symbol $\flat$, i.e.\ for a 
vector $v \in V$ the corresponding one-form is 
$v^\flat := \eta(v, \cdot) \in V^*$. The inverse 
isomorphism (`index raising') will be denoted by a superscript 
sharp symbol $\sharp$. Note that under 
a Lorentz transformation $\Lambda$, $v\in V$ transforms 
under the defining representation, 
$(\Lambda, v) \mapsto \Lambda v$, whereas its image 
$v^\flat\in V^*$ under the $\eta$-induced isomorphism 
transforms under the inverse transposed,
$(\Lambda, v^\flat) \mapsto (\Lambda^{-1})^\top v^\flat = v^\flat \circ \Lambda^{-1}$.

We fix an orientation and a time orientation on $M$. 
The (homogeneous) Lorentz group, i.e.\ the group of 
linear isometries of $(V,\eta)$, will be denoted 
by $\mathcal L := \mathsf{O}(V,\eta)$. The Poincaré 
group, i.e.\ the group of affine isometries of 
$(M,\eta)$, will be denoted by $\mathcal P$. 
The proper orthochronous Lorentz and Poincaré 
groups (i.e.\ the connected components of the identity) 
will be denoted by $\mathcal L_+^\uparrow$ and 
$\mathcal P_+^\uparrow$, respectively\footnote
	{Note that speaking of just orthochronous or proper 
	Lorentz / Poincaré transformations does not make 
	invariant sense without specifying a time direction.}.

We employ standard index notation for Minkowski 
spacetime, using lowercase Greek letters for 
spacetime indices. When working with respect to bases, 
we will, unless otherwise stated, assume them to be 
positively oriented and orthonormal, and we will 
use $0$ for the timelike and lowercase Latin letters 
for spatial indices. We will adhere to standard practice 
in physics where lowering and raising of indices are 
done while keeping the same kernel symbol; i.e.\ for a 
vector $v \in V$ with components $v^\mu$, the components 
of the corresponding one-form $v^\flat \in V^*$ will be 
denoted simply by $v_\mu$. For the sake of notational 
clarity, we will sometimes denote the Minkowski 
inner product of two vectors $u, v \in V$ simply by
\begin{equation}
	u \cdot v := \eta(u,v) = u_\mu v^\mu.
\end{equation}

We fix, once and for all, a reference point / 
\emph{origin} $o \in M$ in (affine) Minkowski 
spacetime, allowing us to identify $M$ with its 
corresponding vector space $V$ (identifying the 
reference point $o \in M$ with the zero vector 
$0 \in V$, i.e.\ via $M \ni x \mapsto (x-o) \in V$), 
which we will do most of the time. Using the 
reference point $o \in M$, the Poincaré group splits 
as a semidirect product
\begin{equation}
	\mathcal P = \mathcal L \ltimes V
\end{equation}
where the Lorentz group factor in this decomposition 
arises as the stabiliser of the reference point~-- 
i.e.\ a Poincaré transformation is considered a 
homogeneous Lorentz transformation if and only if 
it leaves $o$ invariant. Thus, a homogeneous 
Lorentz transformation $\Lambda \in \mathcal L$ 
acts on a point $x \in M \equiv V$ as 
$(\Lambda x)^\mu = \Lambda^\mu_{\hphantom{\mu}\nu} x^\nu$, 
and a Poincaré transformation $(\Lambda,a) \in \mathcal P$ 
acts as $((\Lambda,a) \cdot x)^\mu = \Lambda^\mu_{\hphantom{\mu}\nu} x^\nu + a^\mu$.

We will sometimes make use of the set of spacelike 
hyperplanes in (affine) Minkowski spacetime $M$, 
which we will denote by
\begin{equation}
	\label{eq:def_SpHp}
	\mathsf{SpHP} := \{\Sigma \subset M : \Sigma \; \text{spacelike hyperplane}\}.
\end{equation}
Since the image of a spacelike hyperplane under 
a Poincaré transformation is again a spacelike 
hyperplane, there is a natural action of the Poincaré 
group on $\mathsf{SpHP}$, which we will denote by 
$((\Lambda, a), \Sigma) \mapsto (\Lambda, a) \cdot \Sigma$ 
and spell out in more detail in \eqref{eq:action_on_SpHP} 
below.\enlargethispage{-\baselineskip}

\subsection{The Poincaré algebra}

When considering the Lie algebra $\mathfrak p$ 
of the Poincaré group (or symplectic representations 
thereof), we will denote the generators of translations 
by $P_\mu$ such that $a^\mu P_\mu$ is the 
`infinitesimal transformation' corresponding 
to the translation by $a\in V$, and the generators 
of homogeneous Lorentz transformations (with respect 
to the chosen origin $o$) by $J_{\mu\nu}$, 
such that $- \frac{1}{2} \omega^{\mu\nu} J_{\mu\nu}$ 
is the `infinitesimal transformation' corresponding 
to the Lorentz transformation $\exp(\omega)\in\mathcal
L_+^\uparrow \subset \mathsf{GL}(V)$ for 
$\omega \in \mathfrak l = \mathrm{Lie}(\mathcal L) \subset \mathrm{End}(V)$.

Since we are using the $(-{++}+)$ signature 
convention, the minus sign in the \mbox{expression} 
$- \frac{1}{2} \omega^{\mu\nu} J_{\mu\nu}$ is 
necessary in order that $J_{ab}$ generate rotations 
in the $\E_a$--$\E_b$ plane from $\E_a$ towards $\E_b$, 
which is the convention we want to adopt. A detailed 
\mbox{discussion} of these issues regarding sign conventions 
for the generators of special orthogonal groups 
can be found in appendix \ref{app:sign_convention_so}. 
Moreover, if $u \in V$ is a future-directed unit 
timelike vector, then $c P_\mu u^\mu$ (i.e.\ $c P_0$ 
in the Lorentz frame defined by $u=\E_0$), which is 
\emph{minus} the energy in the frame defined by $u$, 
is the generator of active time translations in the 
direction of $u$. Therefore, with our conventions, 
for the case of causal four-momentum $P \in V$ the energy 
(with respect to future-directed time directions) 
is positive if and only if $P$ is future-directed.

With our conventions, the commutation relations 
for the Poincaré generators are as follows:
\begin{subequations} 
\label{eq:Poinc_gen}
\begin{align}
\label{eq:Poinc_gen-a}
	[P_\mu, P_\nu] &= 0 \\
\label{eq:Poinc_gen-b}
	[J_{\mu\nu}, P_\rho] &= \eta_{\mu\rho} P_\nu - \eta_{\nu\rho} P_\mu \\
\label{eq:Poinc_gen-c}
	[J_{\mu\nu}, J_{\rho\sigma}] &= \eta_{\mu\rho} J_{\nu\sigma} + \text{(antisymm.)} \nonumber\\
		&= \Big(\eta_{\mu\rho} J_{\nu\sigma} - (\mu \leftrightarrow \nu)\Big) - \Big(\rho \leftrightarrow \sigma\Big)
\end{align}
\end{subequations}
As indicated, the abbreviation `antisymm.' 
stands for the additional three terms that one obtains 
by first antisymmetrising (without a factor of $1/2$) 
in the first pair of indices on the left hand 
side, here $(\mu\nu)$, and then the ensuing 
combination once more in the second set of 
indices, here $(\rho\sigma)$, again without a 
factor $1/2$.

\subsection{Symplectic geometry}

We employ the following sign conventions for symplectic
geometry (as used by \mbox{Abraham} and Marsden in 
\cite{Abraham.Marsden:1978}, but different to those 
of Arnold in \cite{Arnold:1989}). Let $(\Gamma, \omega)$ 
be a symplectic manifold. For a smooth function 
$f \in C^\infty(\Gamma)$, we define the Hamiltonian 
vector field $X_f\in \mathfrak X(\Gamma)$ ($\mathfrak X$ 
denoting the space of smooth vector fields) corresponding 
to $f$ by
\begin{equation} \label{eq:def_Ham_VF}
	\iota_{X_f} \omega := \omega(X_f, \cdot) = \D f,
\end{equation}
where $\iota$ denotes the interior product 
between vector fields and differential forms.
The Poisson bracket of two smooth functions 
$f,g \in C^\infty(\Gamma)$ is then defined as
\begin{equation}
	\{f, g\} := \omega(X_f, X_g) = \D f(X_g) = \iota_{X_g} \D f.
\end{equation}
These conventions give the usual coordinate forms 
of the Hamiltonian flow equations and the Poisson 
bracket if the symplectic form $\omega$ takes the 
coordinate form (sign-opposite to that in \cite{Arnold:1989})
\begin{equation}
	\omega = \D q^a \wedge \D p_a\,.
\end{equation}
It is important to note that $C^\infty(\Gamma)$ 
as well as $\mathfrak X(\Gamma)$ are (infinite dimensional) 
Lie \mbox{algebras} with respect to the Poisson bracket 
and the commutator respectively, and that, with 
respect to these Lie structures, the map 
$C^\infty(\Gamma) \to \mathfrak X(\Gamma), f \mapsto X_f$ 
is a Lie \emph{anti}-homomorphism, that is, 
\begin{equation} \label{eq:HVF_anti_hom}
	X_{\{f,g\}} = -\,[X_f,X_g].
\end{equation}

By saying that a one-parameter group 
$\phi_s \colon \Gamma \to \Gamma$ of symplectomorphisms 
is \emph{generated} by a function $g \in C^\infty(\Gamma)$, 
we mean that $\phi_s$ is the flow of the Hamiltonian 
vector field to $g$, i.e.\ that
\begin{equation}
	\frac{\D}{\D s} \phi_s(\gamma) = X_g(\phi_s(\gamma))
\end{equation}
for $\gamma \in \Gamma$, or equivalently
\begin{align} \label{eq:sympl_gen_PB}
	\frac{\D}{\D s} (f \circ \phi_s) &= \Big(\D f (X_g)\Big) \circ \phi_s \nonumber\\
	&= \{f,g\} \circ \phi_s
\end{align}
for $f \in C^\infty(\Gamma)$. Here both sides of
\eqref{eq:sympl_gen_PB} are to be understood as 
evaluated pointwise.

\subsection{Poincaré-invariant Hamiltonian systems and their momentum maps}

A \emph{classical Poincaré-invariant system} will be described 
by a phase space $(\Gamma, \omega)$ --~i.e.\ a symplectic 
manifold~-- with a symplectic action
\begin{equation}
	\Phi \colon \mathcal P \times \Gamma \to \Gamma, \; ((\Lambda, a), \gamma) \mapsto \Phi_{(\Lambda, a)} (\gamma)
\end{equation}
of the Poincaré group (in fact, for most of our purposes 
an action of $\mathcal P_+^\uparrow$ is enough). We will 
take $\Phi$ to be a left action, i.e.\ to satisfy\footnote
	{We refer to \cite{Giulini:2015} for a detailed discussion 
	of left versus right actions and the corresponding sign 
	conventions that will also play an important role in 
	the following.}
\begin{equation}
	\Phi_{(\Lambda_1, a_1)} \circ \Phi_{(\Lambda_2, a_2)} = \Phi_{(\Lambda_1 \Lambda_2, a_1 + \Lambda_1 a_2)} \; .
\end{equation}
We will denote such systems as 
$(\Gamma, \omega, \Phi)$.

The left action $\Phi$ of $\mathcal{P}$ on $\Gamma$
induces vector fields $V_\xi$ on $\Gamma$ (the 
so-called \emph{fundamental vector fields}), one for 
each $\xi$ in the Lie algebra $\mathfrak{p}$ of 
$\mathcal{P}$. They are given by
\begin{equation}
	V_\xi(\gamma) := \left. \frac{\D}{\D s} \Phi_{\exp(s\xi)}(\gamma) \right\vert_{s=0} \; ,
\end{equation}
so that the map $\mathfrak{p} \to \mathfrak X(\Gamma), 
\xi\mapsto V_\xi$, given by the differential 
of $\Phi$ with respect to its first argument and 
evaluated at the group identity, is clearly linear. 
In fact, it is straightforward to show that it 
is an anti-homomorphism from the Lie algebra 
$\mathfrak{p}$ into the Lie algebra $\mathfrak X(M)$,\footnote
	{Had we chosen $\Phi$ to be a right action, we would 
	have obtained a proper Lie homomorphism; compare 
	\cite[appendix B]{Giulini:2015}.}
i.e.\
\begin{equation} \label{eq:fundVF_anti_hom}
	\bigl[V_{\xi_1}, V_{\xi_2}\bigr] = - V_{[\xi_1, \xi_2]}.
\end{equation}
Moreover, a similar calculation shows 
\cite[appendix B]{Giulini:2015}
\begin{equation} \label{eq:fundVF_equivariance}
	(\DD \Phi_{(\Lambda,a)}) \circ V_\xi = V_{\mathrm{Ad}_{(\Lambda,a)}(\xi)} \circ \Phi_{(\Lambda,a)} \; ,
\end{equation}
where $\DD \Phi_{(\Lambda,a)} \colon T\Gamma \to T\Gamma$ 
denotes the differential of 
$\Phi_{(\Lambda,a)} \colon \Gamma \to \Gamma$.

As $\mathcal{P}$ acts by symplectomorphisms, 
the fundamental vector fields $V_\xi$ are locally Hamiltonian,
i.e.\ locally (in a neighbourhood of each point), for 
each $\xi \in \mathfrak p$ there exists a local function 
$f_\xi$ such that $\D f_\xi = \iota_{V_\xi} \omega$. 
In fact, due to the Poincaré algebra being perfect (in 
spacetime dimension greater than $2$), the $f_\xi$ can be 
shown to exist globally, so that each $V_\xi$ is a globally 
defined Hamiltonian
vector field (i.e.\ each one-parameter 
group $\Phi_{\exp(s\xi)}\colon \Gamma \to \Gamma$ of
symplectomorphisms is generated, in the sense of 
\eqref{eq:sympl_gen_PB}, by the corresponding function 
$f_\xi$). Moreover, due to $\mathfrak p$ having vanishing 
second cohomology, the $f_\xi$ can be chosen 
in such a way that the map $\xi\mapsto f_\xi$ 
from the Lie algebra $\mathfrak{p}$ to the Lie 
algebra $C^\infty(\Gamma)$ (the Lie product 
of the latter being the Poisson bracket) is a 
Lie homomorphism, i.e.\ 
\begin{equation} \label{eq:Poinc_gen_hom}
	\left\{f_{\xi_1},f_{\xi_2}\right\} = f_{[\xi_1,\xi_2]}.
\end{equation}
I.e., for spacetime dimension greater than $2$, any 
symplectic action of the Poincaré group is always a 
\emph{Poisson action}. Details of these arguments may 
be found in \cite{Schwartz.Giulini:2020}.
Note that, according to \eqref{eq:fundVF_anti_hom} 
and \eqref{eq:HVF_anti_hom}, both maps 
$\xi\mapsto V_\xi$ and $V_\xi\mapsto f_\xi$ are 
Lie \emph{anti}-homomorphisms. Hence their combination 
$\xi\mapsto f_\xi$ is a proper Lie homomorphism 
(no minus sign on the right-hand side of 
\eqref{eq:Poinc_gen_hom}).

Now, we will deduce the transformation properties of the 
generators $f_\xi$ under the action of $\mathcal{P}$. Taking the 
pullback of the equation $\omega(V_\xi, \cdot) = \D f_\xi$ 
with $\Phi_{(\Lambda,a)^{-1}}$ and using the invariance of 
$\omega$ as well as \eqref{eq:fundVF_equivariance}, we 
immediately deduce
\begin{equation} \label{eq:Poinc_gen_Ad_equivariance}
	\Phi_{(\Lambda,a)^{-1}}^* f_\xi := f_\xi \circ \Phi_{(\Lambda,a)^{-1}} = f_{\mathrm{Ad}_{(\Lambda,a)}(\xi)} \; ,
\end{equation}
which may also be read as the invariance of the 
real-valued function \mbox{$f\colon \mathfrak{p} \times \Gamma \to \mathbb R$}, 
$(\xi,\gamma) \mapsto f_\xi(\gamma)$, under the 
combined left action of $\mathcal{P}$ 
on $\mathfrak{p}\times\Gamma$ given by 
$\mathrm{Ad}\times\Phi$. \mbox{Alternatively}, since 
$\xi\mapsto f_\xi$ is linear, we may regard 
$f$ as $\mathfrak{p}^*$-valued function on 
$\Gamma$, where $\mathfrak{p}^*$ denotes the 
vector space dual to $\mathfrak{p}$. 
This map is called the \emph{momentum map}\footnote
	{See \cite[chap.\,4.2]{Abraham.Marsden:1978} 
	for a general discussion on the notion of `momentum map'
	and also \cite{Giulini:2015} for an account of its 
	use and properties restricted to the case of 
	Poincaré-invariant systems.}
for the given system $(\Gamma,\omega,\Phi)$, 
which according to \eqref{eq:Poinc_gen_Ad_equivariance} is then
$\mathrm{Ad}^*$-equivariant:
\begin{equation} \label{eq:momentum_map_Ad-star_equivariance}
	f \circ \Phi_{(\Lambda,a)} = \mathrm{Ad}^*_{(\Lambda,a)} \circ f
	\iff
	\mathrm{Ad}^*_{(\Lambda,a)} \circ f \circ \Phi_{(\Lambda,a)^{-1}} = f
\end{equation}
The second expression is again meant to 
stress that the condition of equivariance is 
equivalent to the invariance of the function $f$ 
under the combined left actions in its domain 
and target spaces (invariance of the graph). 
Note that $\mathrm{Ad}^*$ denotes the co-adjoint 
representation of $\mathcal{P}$ on $\mathfrak{p}^*$, 
given by $\mathrm{Ad}^*_{(\Lambda,a)} := (\mathrm{Ad}_{(\Lambda,a)^{-1}})^\top$ with superscript $\top$
denoting the transposed map.\enlargethispage{\baselineskip}

Points in $\Gamma$ faithfully represent the state 
of the physical system whereas observables 
correspond to functions on $\Gamma$. In order to 
implement time evolution we shall employ a 
`classical Heisenberg picture', in which the 
phase space point remains the same at all 
times, whereas the evolution will correspond to 
the changes of observables according to their 
association to different spacelike 
hyperplanes in spacetime. Although this is 
different from the (`Schrödinger picture') 
approach usually taken in classical mechanics 
(where the state of the system is given by a 
phase space point changing in `time', which is 
an external parameter), this point of view is 
clearly better adapted to the 
Poincaré-relativistic framework, in which 
there simply is no absolute notion of time.
\enlargethispage{\baselineskip}

Choosing a set 
of ten basis vectors $(P_\mu,J_{\mu\nu})$ for 
$\mathfrak{p}$ obeying \eqref{eq:Poinc_gen} (compare 
appendix \ref{app:sign_convention_so}), we can contract the 
$\mathfrak{p}^*$-valued momentum map with each of 
these basis vectors in order to obtain the 
corresponding ten real-valued component functions 
of the momentum map. By some abuse of notation 
we shall call these component functions by the same 
letters $(P_\mu,J_{\mu\nu})$ as the Lie algebra 
elements themselves. 
\eqref{eq:Poinc_gen_hom} now says that the map 
that sends the Lie algebra elements 
$P_\mu$ and $J_{\mu\nu}$ in $\mathfrak{p}$ 
to the corresponding component functions 
of the momentum map is a Lie homomorphism from 
$\mathfrak{p}$ to the Lie algebra 
$C^\infty(\Gamma,\mathbb{R})$ (the latter with 
Poisson bracket as Lie multiplication):
\begin{subequations}
	\begin{align}
	\{P_\mu, P_\nu\} &= 0 \\
	\{J_{\mu\nu}, P_\rho\} &= \eta_{\mu\rho} P_\nu - \eta_{\nu\rho} P_\mu \\
	\{J_{\mu\nu}, J_{\rho\sigma}\} &= \eta_{\mu\rho} J_{\nu\sigma} + \text{(antisymm.)}
	\end{align}
\end{subequations}
The $\mathrm{Ad}^*$-equivariance of the momentum map can now be 
written down in component form if we first set 
$\xi = P_\mu$ and then $\xi = J_{\mu\nu}$. 
Indeed, considering
\eqref{eq:Poinc_gen_Ad_equivariance} and 
recalling our abuse of notation in 
denoting the real-valued phase space functions 
$f_{P_\mu}$ and $f_{J_{\mu\nu}}$
again with the letters $P_\mu$ and 
$J_{\mu\nu}$, we can immediately read from 
\eqref{eq:ad_rep} of appendix 
\ref{app:adj_rep}, in which we need to 
replace $\E_a$ with $P_\mu$ and $B_{ab}$ 
with $-J_{\mu\nu}$ according to 
\eqref{eq:sign_convention_so} of 
appendix \ref{app:sign_convention_so}, that
\begin{subequations} \label{eq:coadjoint_rep_components}
\begin{alignat}{2}
\label{eq:coadjoint_rep_components_a}
	& P_\mu \circ \Phi_{(\Lambda, a)}
	&&= (\Lambda^{-1})^\nu_{\phantom{\nu}\mu} \, P_\nu \; , \\
\label{eq:coadjoint_rep_components_b}
	& J_{\mu\nu} \circ \Phi_{(\Lambda, a)}
	&&= (\Lambda^{-1})^\rho_{\phantom{\rho}\mu} (\Lambda^{-1})^\sigma_{\phantom{\sigma}\nu} \,
	J_{\rho\sigma} 
	+ a_\mu (\Lambda^{-1})^\rho_{\phantom{\rho}\nu} \, P_\rho
	- a_\nu (\Lambda^{-1})^\rho_{\phantom{\rho}\mu} \, P_\rho \; .
\end{alignat}
\end{subequations}
Note that the left-hand sides of \eqref{eq:coadjoint_rep_components} are 
precisely what we need; that is, we need 
the composition with $\Phi_{(\Lambda, a)}$ rather than 
$\Phi_{(\Lambda, a)^{-1}}$ to evaluate the momenta 
$P_\mu$ and $J_{\mu\nu}$ on the actively 
Poincaré-displaced phase space points. Note also 
that if we had put the 
indices upstairs and had used, e.g., 
$P^\mu = \eta^{\mu\nu} P_{\nu}$ rather than $P_\mu$ 
then the right-hand side of \eqref{eq:coadjoint_rep_components_a} would read 
$\Lambda^\mu_{\phantom{\mu}\nu} \, P^\nu$, and 
correspondingly in 
\eqref{eq:coadjoint_rep_components_b}. 
Finally recall that the last term on the right-hand 
side of \eqref{eq:coadjoint_rep_components_b} just 
reflects the familiar transformation of angular 
momentum (the momentum associated to spatial 
rotations) under spatial translations, which is 
typical for the \emph{co}-adjoint representation, 
which here gets extended to the momentum 
associated to boost transformations\footnote
	{One easily checks that the signs are 
	right: translating a system whose momentum points 
	in $y$-direction by a positive amount into the 
	$x$-direction should enhance the angular momentum 
	in $z$-direction. This is just what 
	\eqref{eq:coadjoint_rep_components_b} implies.}.

\subsection{The Pauli--Lubański vector}
\label{sec:PauliLubanski}
Given a classical Poincaré-invariant system, the 
Pauli--Lubański vector $W$ is the \mbox{$V$-valued} phase 
space function defined in components by
\begin{equation}
	W_\mu = - \frac{1}{2} \varepsilon_{\mu\nu\rho\sigma} P^\nu J^{\rho\sigma}
\end{equation}
where $\varepsilon$ denotes the volume form of 
Minkowski space (whose components in a \mbox{positively} 
oriented orthonormal basis are just given by the 
usual totally \mbox{antisymmetric} symbol, with 
$\varepsilon_{0123} = +1$). The sign convention 
in this definition can be understood as follows. 
We imagine a situation in which $P$ is timelike and 
future-directed (\mbox{positive} energy, see above), 
and consider the spatial components of $W$ with 
respect to an \mbox{orthonormal} basis $\{\E_0, \dots, \E_3\}$ 
of $V$ with $(\E_0)^\mu = P^\mu / \sqrt{-P_\nu P^\nu}$ 
(`momentum rest frame'). For those, we obtain
\begin{equation}
	\frac{W_a}{\sqrt{-P_\mu P^\mu}} = - \frac{1}{2} \varepsilon_{a0\rho\sigma} J^{\rho\sigma} = \frac{1}{2} {^{(3)}\varepsilon}_{abc} J^{bc}
\end{equation}
where the ${^{(3)}\varepsilon}_{abc}$ is the 
three-dimensional antisymmetric symbol / the 
components of the spatial volume form. Thus, since 
$J^{bc} = J_{bc}$ generates rotations from $\E_b$ 
towards $\E_c$, we see that $W_a/\sqrt{-P_\mu P^\mu}$ 
generates rotations `along the $\E_a$ axis' in the usual, 
three-dimensional sense. Thus, $W/\sqrt{-P_\mu P^\mu}$ 
can be interpreted as the `spatial spin vector' in the 
momentum rest frame, which is the usual interpretation 
of the Pauli--Lubański vector.

Rewriting the definition of $W$ as
\begin{equation}
	W_\mu = - \frac{1}{2} \varepsilon_{\mu\nu\rho\sigma} P^\nu J^{\rho\sigma} = \frac{1}{2} \varepsilon_{\nu\rho\sigma\mu} P^\nu J^{\rho\sigma} = \frac{1}{3!} \varepsilon_{\nu\rho\sigma\mu} (P^\flat \wedge J)^{\nu\rho\sigma},
\end{equation}
we see that in the language of exterior algebra
\begin{equation}
	W = (*(P^\flat\wedge J))^\sharp
\end{equation}
where $*$ is the Hodge star operator. Here we 
use the standard sign conventions for the Hodge 
operator, i.e.\ the definition 
$\alpha \wedge *\beta = \eta(\alpha,\beta) \, \varepsilon$; 
see for example \cite{Straumann:2013} or 
\cite[appendix A]{Giulini:2015}.
\enlargethispage{\baselineskip}

\section{The Newton--Wigner position as a `centre of spin'}
\label{sec:NW_centre_of_spin}

In this section we will explain our understanding 
and present our geometric clarification of 
Fleming's statement in \cite{Fleming:1965a} that 
the Newton--Wigner position may be understood as 
a `centre of spin'. To this end, we introduce 
Fleming's geometric framework for special-relativistic 
position observables, and then discuss the definition 
of position observables by spin supplementary conditions 
(\acronym{SSC}s). Finally, we introduce the notion of a position 
observable being a `centre of spin', and prove that 
the Newton--Wigner position is the only continuous 
position observable defined by an \acronym{SSC} that represents 
a centre of spin in that sense.\enlargethispage{\baselineskip}

\subsection{Position observables on spacelike hyperplanes}
\label{sec:pos_obs}

We start by describing the general framework 
developed by Fleming in \cite{Fleming:1965a} 
and also \cite{Fleming:1966} for the description 
of special-relativistic position observables, 
translated to our case of classical systems from 
Fleming's quantum language. Consider a classical 
Poincaré-invariant system $(\Gamma, \omega, \Phi)$. 
By a \emph{position observable} $\chi$ for this 
system we understand a `procedure' which, given 
any spacelike hyperplane $\Sigma \in \mathsf{SpHP}$ 
in (affine) Minkowski spacetime, allows us to 
`localise' the system on $\Sigma$. More precisely, 
this means that for any $\Sigma \in \mathsf{SpHP}$, 
we have an $M$-valued phase space function
\begin{equation}
	\chi(\Sigma) \colon \Gamma \to M
\end{equation}
with image contained in $\Sigma$, whose value 
$\chi(\Sigma)(\gamma)$ for $\gamma \in \Gamma$ is 
to be interpreted as the `$\chi$-position' of our 
system in state $\gamma$ on the hyperplane $\Sigma$.

Any spacelike hyperplane $\Sigma \in \mathsf{SpHP}$ 
is uniquely characterised by its (timelike) 
future-directed unit normal $u \in V$ and its distance 
$\tau\in\mathbb R$ to the origin $o \in M$, measured 
along the straight line through $o$ in direction $u$. 
In terms of these, it has the form 
\begin{equation} \label{eq:hyperplane_params}
	\Sigma = \{x\in M : u_\mu x^\mu = -\tau\},
\end{equation}
where we identified $M$ with $V$. 
From now on, whenever convenient, we will identify 
$\Sigma$ with the tuple $(u, \tau)$. The condition 
that the image of $\chi(\Sigma)$ be contained 
in $\Sigma$ then takes the form
\begin{equation} \label{eq:pos_obs_image}
	u_\mu \chi^\mu(u, \tau)(\gamma) = - \tau.
\end{equation}
We can now also spell out explicitly the left action 
of $\mathcal{P}$ on $\mathsf{SpHP}$ that is induced 
from the left action of $\mathcal{P}$ on $M$
(as already mentioned below \eqref{eq:def_SpHp}):
\begin{equation} \label{eq:action_on_SpHP}
	(\Lambda,a) \cdot (u,\tau) = (\Lambda u, \tau - \Lambda u \cdot a)
\end{equation}

Fixing $u$ and varying $\tau$ in 
\eqref{eq:hyperplane_params}, we obtain the spacelike 
hyperplanes corresponding to different `instants of 
time' $\tau$ in the Lorentz frame corresponding to 
$u$. Thus, for a fixed state $\gamma\in\Gamma$ and fixed 
frame $u$, the set
\begin{equation}
	\{\chi(u, \tau) (\gamma) : \tau \in \mathbb R\} \subset M
\end{equation}
gives the `worldline' of the $\chi$-position of the 
system. Following Fleming \cite{Fleming:1965a}, who 
says that this is a requirement `easily agreed upon', 
we require that this worldline should be parallel to 
the four-momentum\footnote
	{This assumption is natural for closed systems as we 
	consider here. For non-closed systems, i.e.\ systems 
	without local energy--momentum conservation, the 
	four-velocity is in general not parallel to the 
	four-momentum; see, e.g., the discussion at the beginning 
	of section 2.6 in \cite{Giulini:2018}.},
i.e.\ 
$\frac{\partial \chi(u, \tau)}{\partial \tau} \propto P$. 
Together with \eqref{eq:pos_obs_image}, this implies 
condition \eqref{eq:pos_obs_time_der} in the 
definition below, which is meant to sum up all the 
preceding considerations.
\begin{defn} \label{defn:pos_obs}
	A \emph{position observable} for a classical 
	Poincaré-invariant system $(\Gamma, \omega, \Phi)$ 
	with causal four-momentum is a map
	\begin{equation} \label{eq:pos_obs}
		\chi \colon \mathsf{SpHP} \times \Gamma \to M, \; (\Sigma, \gamma) \mapsto \chi(\Sigma)(\gamma)
	\end{equation}
	satisfying
	\begin{equation}
		\chi(\Sigma)(\gamma) \in \Sigma
	\end{equation}
	for all $\Sigma \in \mathsf{SpHP}$ and all 
	$\gamma \in \Gamma$ (or, equivalently, 
	\eqref{eq:pos_obs_image}), as well as
	\begin{equation} \label{eq:pos_obs_time_der}
		\frac{\partial \chi_\mu(u, \tau)}{\partial \tau} = \frac{1}{(- u \cdot P)} P_\mu \; .
	\end{equation}
	For fixed $\Sigma\in\mathsf{SpHP}$, we will often 
	view $\chi(\Sigma) \colon \Gamma \to M$ as a 
	phase space function in its own right.
\end{defn}
Note that \eqref{eq:pos_obs_time_der} and 
\eqref{eq:pos_obs_image} imply that the 
four-momentum must be causal for such a position 
observable to exist.
\enlargethispage{\baselineskip}

In addition to the demands of the positions 
$\chi(\Sigma)$ being located on $\Sigma$ and 
of `worldlines' in direction of the four-momentum, 
Fleming also introduces the following \mbox{covariance} 
requirement (which we, different to Fleming, do not 
include in the definition of a position observable):
\begin{defn} \label{defn:pos_obs_cov}
	A position observable for a classical 
	Poincaré-invariant system $(\Gamma, \omega, \Phi)$ 
	is said to be \emph{covariant} if and only if
	\begin{equation} \label{eq:pos_obs_cov_1}
		\chi \Big((\Lambda, a) \cdot \Sigma\Big) \Big(\Phi_{(\Lambda, a)}(\gamma)\Big) = (\Lambda, a) \cdot \Big( \chi(\Sigma)(\gamma) \Big)
	\end{equation}
	for all $\Sigma\in\mathsf{SpHP}$, $\gamma \in \Gamma$ 
	and $(\Lambda, a) \in \mathcal P$. This can be 
	read concisely as saying that the map 
	\eqref{eq:pos_obs} is invariant under 
	the natural left action induced from those on 
	the domain and target spaces (invariance of 
	$\chi$'s graph):
	\begin{equation} \label{eq:pos_obs_cov_2}
		\chi = (\Lambda,a) \circ \chi \circ \left( (\Lambda,a)^{-1} \times \Phi_{(\Lambda,a)^{-1}} \right).
	\end{equation}
\end{defn}

This is indeed a sensible notion of covariance: 
it demands that, for any \mbox{Poincaré} \mbox{transformation} 
$(\Lambda, a)$, the $\chi$-position of the 
transformed system $\Phi_{(\Lambda, a)}(\gamma)$ 
on the \mbox{transformed} hyperplane 
$(\Lambda, a) \cdot \Sigma$ be the transform 
of the `original position' $\chi(\Sigma)(\gamma)$. 
In terms of components, \eqref{eq:pos_obs_cov_1} 
assumes the form
\begin{equation}
	\chi^\mu (\Lambda u, \tau - \Lambda u\cdot a) \circ \Phi_{(\Lambda, a)} = \Lambda^\mu_{\hphantom{\mu}\nu} \chi^\nu(u, \tau) + a^\mu\,,
\end{equation}
taking into account \eqref{eq:action_on_SpHP}.

\subsection{Spin supplementary conditions}

The most important and widely used procedure to 
define special-relativistic position observables 
is by so-called \emph{spin supplementary conditions}. 
Suppose we are given a causal, future-directed 
vector $P \in V$ and an antisymmetric 2-tensor 
$J \in \bigwedge^2 V^*$, describing the four-momentum 
and the angular momentum (with respect to the origin 
$o \in M$) of some physical system. For any 
future-directed timelike vector $f \in V$, 
we then consider the equation
\begin{equation} \label{eq:SSC}
	0 = S_{\mu\nu} f^\nu
\end{equation}
with $S_{\mu\nu} := J_{\mu\nu} - x_\mu P_\nu + x_\nu P_\mu$, 
which we view as an equation for $x \in M$. Since 
$S$ is the angular momentum tensor with respect 
to the reference point $x$ (instead of the origin 
$o$ as for $J$), or the \emph{spin tensor} with 
respect to $x$, \eqref{eq:SSC} is called the 
\emph{spin supplementary condition} (\acronym{SSC}) with 
respect to $f$. As is well-known (and easily 
verified), the set of its solutions $x$ is a 
line in $M$ with tangent $P$, namely
\begin{equation} \label{eq:SSC_worldline}
	\{x \in M : 0 = S_{\mu\nu} f^\nu\} = \left\{x \in M : x_\mu = \frac{J_{\mu\rho} f^\rho}{f\cdot P} + \lambda P_\mu \; \text{with} \; \lambda \in \mathbb R\right\}.
\end{equation}
This line can be given the interpretation of the 
`centre of energy' worldline of our system with 
respect to the Lorentz frame defined by $f$. 
See \cite{Costa.EtAl:2018} and references therein 
for further discussion on the interpretation and 
impact of various \acronym{SSC}s as regards equations of motion 
in general relativity.
\enlargethispage{\baselineskip}

The idea is now to explicitly combine the 
\acronym{SSC}-based approach with Fleming's \mbox{geometric} ideas, 
thereby introducing the two independent parameters 
$f$ from \eqref{eq:SSC_worldline} and $u$ from 
\eqref{eq:hyperplane_params}. We define a position
observable in the sense of definition \ref{defn:pos_obs} 
in the following way: given a classical 
Poincaré-invariant system $(\Gamma, \omega, \Phi)$ 
with causal four-momentum and a state 
$\gamma \in \Gamma$, we consider the \acronym{SSC} worldline
defined by \eqref{eq:SSC} where we now take
$P_\mu(\gamma)$ for the four-momentum and 
$J_{\mu\nu}(\gamma)$ for the angular momentum tensor. 
We then simply define $\chi(\Sigma)(\gamma)$ to be 
the intersection of this worldline with the 
hyperplane $\Sigma = (u,\tau)$. This means that we
take the $x(\lambda)$ from \eqref{eq:SSC_worldline}
and determine the parameter $\lambda$
from \eqref{eq:pos_obs_image}, i.e.\ from 
$x(\lambda) \cdot u + \tau = 0$. Inserting the 
$\lambda = \lambda(u,\tau)$ so determined leads to
\begin{defn}
	The \emph{\acronym{SSC} position observable} with 
	respect to $f$ is given by
	\begin{equation} \label{eq:SSC_pos_obs}
		\chi_\mu(u, \tau) = \frac{J_{\mu\rho} f^\rho}{f\cdot P} 
		+ \frac{\tau P_\mu}{(- u \cdot P)} 
		- \frac{J_{\lambda\rho} u^\lambda f^\rho}{(- f \cdot P)}
		\, \frac{P_\mu}{(- u \cdot P)} \; .
	\end{equation}
	Let us again stress the interpretation of this 
	expression: it is the \acronym{SSC} position with respect to 
	$f$ (i.e.\ a point on the `centre of energy' worldline 
	with respect to $f$) as localised on the hyperplane 
	characterised by unit normal $u$ and distance $\tau$ 
	to the origin, i.e.\ as seen in the Lorentz frame 
	with respect to $u$ at `time' $\tau$.
\end{defn}

Note that for this definition to make sense, $f$ 
does not have to be a fixed timelike future-directed 
vector: it can depend on the normal $u$ (and could 
even depend on $\tau$), and it can also depend on 
phase space\footnote
	{Various choices for $f$ were given 
	distinguished names in the literature. 
	The main ones, different from the Newton--Wigner 
	condition to be discussed here, are as follows. 
	If $f$ is meant to just characterise a fixed 
	`laboratory frame', which may be preferred for 
	any reason, like rotational symmetries in that 
	frame, the \acronym{SSC} is named after 
	Corinaldesi \& Papapetrou 
	\cite{Corinaldesi.Papapetrou:1951}. If $f$ is 
	proportional to the total linear momentum of the 
	system, the \acronym{SSC} is named after Tulczyjew \cite{Tulczyjew:1959} and Dixon
	\cite{Dixon:1970}. If $f$ is chosen in a somewhat 
	self-referential way to be the four-velocity of 
	the worldline that is to be determined by the 
	very \acronym{SSC} containing that $f$, the condition is 
	named after Frenkel \cite{Frenkel:1926}, 
	Mathisson \cite{Mathisson:1937a,Mathisson:1937b}, 
	and Pirani \cite{Pirani:1956a,Pirani:1956b}.}.
Of course this means that according to this 
dependence of $f$, we will possibly be considering 
different worldlines for different choices of $u$.
\begin{exmp}
	\begin{enumerate}[label=(\roman*)]
		\item Choosing $f = u$, we are considering, for 
			each $u$, the \acronym{SSC} worldline with respect to $u$, 
			i.e.\ the \emph{centre of energy} worldline\footnote
				{Note that it was called `centre of mass' by Fleming \cite{Fleming:1965a}.}
			with respect to $u$. Using \eqref{eq:SSC_pos_obs}, 
			the centre of energy position observable has the form
			\begin{equation}
				\chi^\mathrm{CE}_\mu(u, \tau) = \frac{J_{\mu\rho} u^\rho}{u\cdot P} + \frac{\tau P_\mu}{(- u \cdot P)} \; .
			\end{equation}

		\item In the case of timelike four-momentum, we 
			can choose $f = P$ the four-momentum (the 
			Tulczyjew--Dixon \acronym{SSC}), such that the corresponding 
			\acronym{SSC} worldline is the centre of energy worldline in 
			the momentum rest frame of the system. This worldline, 
			which is obviously independent of $u$, was called the 
			\emph{centre of inertia} worldline by Fleming 
			\cite{Fleming:1965a}. The centre of inertia has the form
			\begin{equation}
				\chi^\mathrm{CI}_\mu(u, \tau) = -\frac{J_{\mu\rho} P^\rho}{m^2 c^2} + \frac{\tau P_\mu}{(- u \cdot P)} - \frac{J_{\lambda\rho} u^\lambda P^\rho}{m^2 c^2}\; \frac{P_\mu}{(- u \cdot P)} \; ,
			\end{equation}
			where $m = \sqrt{-P^2}/c$ is the mass of the system.
		\item Choosing $f = u + \frac{P}{mc}$ where 
			$m = \sqrt{-P^2}/c$ is the mass of the system (again 
			only possible in the case of timelike four-momentum), 
			we obtain the \emph{Newton--Wigner position observable}. 
			Evaluating \eqref{eq:SSC_pos_obs}, it has the form
			\begin{equation} \label{eq:NW_pos}
				\chi_\mu^\mathrm{NW}(u, \tau) = -\frac{J_{\mu\rho} \left(u^\rho + \frac{P^\rho}{mc}\right)}{mc - u\cdot P} + \frac{\tau P_\mu}{(- u \cdot P)} - \frac{J_{\lambda\rho} u^\lambda P^\rho}{mc (mc - u \cdot P)}\, \frac{P_\mu}{(- u \cdot P)} \; .
			\end{equation}
	\end{enumerate}
\end{exmp}

Of course, the \acronym{SSC} position observable 
\eqref{eq:SSC_pos_obs} will generally not be 
covariant in the sense of definition 
\ref{defn:pos_obs_cov} unless $f$ is also assumed 
to transform appropriately. If $f$ depends 
on $\Sigma\in\mathsf{SpHP}$ and $\gamma\in\Gamma$ 
and takes values in $V$ it seems obvious that for the 
resulting position to be covariant $f$ itself 
must be a covariant function under the 
combined actions on its domain and target spaces.
Indeed, we have
\begin{prop} \label{prop:SSC_pos_obs_cov}
	If the vector $f$ defining the \acronym{SSC} position 
	observable $\chi$ is a function
	\begin{equation}
		f \colon \mathsf{SpHP} \times \Gamma \to V, \quad (\Sigma,\gamma) \mapsto f(\Sigma)(\gamma),
	\end{equation}
	such that
	\begin{equation} \label{eq:SSC_vec_cov}
		f \Big((\Lambda, a) \cdot \Sigma\Big) \Big(\Phi_{(\Lambda, a)}(\gamma)\Big) = \Lambda \cdot \Big( f(\Sigma)(\gamma) \Big)
	\end{equation}
	for all $\Sigma \in \mathsf{SpHP}$, 
	$\gamma \in \Gamma$, and $(\Lambda, a) \in \mathcal P$,
	then $\chi$ is a covariant position observable. 
	Again we note that, just like in the transition 
	from \eqref{eq:pos_obs_cov_1} to \eqref{eq:pos_obs_cov_2}, 
	we may rewrite \eqref{eq:SSC_vec_cov} equivalently 
	as expressing the invariance of $f$ (i.e.\ its graph) 
	under simultaneous actions on its domain and target 
	spaces (using that translations act trivially on the 
	target space $V$):
	\begin{equation}
		f = \Lambda \circ f \circ \left( (\Lambda,a)^{-1} \times \Phi_{(\Lambda,a)^{-1}} \right)
	\end{equation}

	\begin{proof}
		At first, suppose we are given a future-directed 
		timelike four-momentum $P \in V$ and an angular 
		momentum tensor $J \in \bigwedge^2 V^*$, as well 
		as a future-directed timelike vector $f$ for the 
		definition of an \acronym{SSC}. In addition, fix a Poincaré 
		transformation $(\Lambda, a) \in \mathcal P$. 
		If we now consider (a) the \acronym{SSC} worldline for $P$ 
		and $J$ with respect to $f$, and (b) the \acronym{SSC} 
		worldline for the transformed four-momentum 
		$P'=\Lambda P$ and angular momentum 
		$J' = ((\Lambda^{-1})^\top \otimes (\Lambda^{-1})^\top)J + a^\flat \wedge (\Lambda^{-1})^\top P^\flat$ 
		(compare \eqref{eq:coadjoint_rep_components_b}) 
		with respect to the transformed vector $\Lambda f$, 
		it is easy to check that the second worldline is 
		the Poincaré transform by $(\Lambda, a)$ of the 
		first. That is, by Poincaré transforming the 
		four-momentum and angular momentum of the system 
		as well as the `direction vector' for the \acronym{SSC}, we 
		Poincaré transform the \acronym{SSC} worldline.

		Now, the \acronym{SSC} position $\chi(\Sigma)(\gamma)$ is 
		defined to be the intersection of the hyperplane 
		$\Sigma$ with the \acronym{SSC} worldline of $\gamma$ with 
		respect to $f(\Sigma)(\gamma)$. Thus, the `new position'
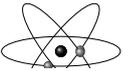
		\begin{equation}
			\chi \Big((\Lambda, a) \cdot \Sigma\Big) \Big(\Phi_{(\Lambda, a)}(\gamma)\Big)
		\end{equation}
		is the intersection of the transformed hyperplane 
		$(\Lambda, a) \cdot \Sigma$ with the \acronym{SSC} worldline 
		of the transformed system $\Phi_{(\Lambda, a)}(\gamma)$ 
		with respect to the transformed vector 
		$\Lambda \cdot \Big( f(\Sigma)(\gamma) \Big)$, 
		where we used the covariance requirement 
		\eqref{eq:SSC_vec_cov}. But according to our earlier 
		considerations, this means that the `new position' 
		is the intersection of the transformed hyperplane 
		with the transform of the original \acronym{SSC} worldline~-- 
		i.e.\ the transform of the original position 
		$\chi(\Sigma)(\gamma)$. This means that the position 
		observable is covariant.
	\end{proof}
\end{prop}

Since the vectors defining the centre of energy, 
the centre of inertia and the Newton--Wigner position 
satisfy \eqref{eq:SSC_vec_cov}, all of these are 
covariant position observables. We stress once 
more that for this to be true we need to take into 
account the action of the Poincaré group on 
$\mathsf{SpHP}$. This remark is particularly 
relevant in the Newton--Wigner case, in which $f$ 
is the sum of two vectors, $u$ and $P/(mc)$, the 
first being associated to an element of 
$\mathsf{SpHP}$ and the second to an element of 
$\Gamma$. Covariance cannot be expected to hold 
for non-trivial actions on $\Gamma$ alone. In the 
next section we will offer an insight as to why 
this somewhat `hybrid' combination for $f$ in 
terms of an `external' vector $u$ and an 
`internal' vector $P/(mc)$ appears. The latter is 
internal, or dynamical, in the sense that it is defined 
entirely by the physical state of the system, i.e.\ 
a point in $\Gamma$, while the former is external, or 
kinematical, in the sense that it refers to the choice 
of $\Sigma\in\mathsf{SpHP}$, which is entirely 
independent of the physical system and its state.

Finally, we will need the following well-known 
result for \acronym{SSC}s with respect to different vectors 
$f$, which was first shown by M{\o}ller in 1949 in 
\cite{Moeller:1949}; see also 
\cite[theorem 17]{Giulini:2015} for a recent 
and more geometric discussion:
\begin{thm}[Møller disc and radius] \label{thm:Moller_disc}
	Suppose we are given the future-directed timelike 
	four-momentum vector $P \in V$ and the angular momentum 
	tensor $J \in \bigwedge^2 V^*$ of some physical system. 
	Consider the bundle of all possible \acronym{SSC} worldlines 
	\eqref{eq:SSC_worldline} for this system, defined by 
	considering all future-directed timelike vectors $f$. 
	The intersection of this bundle with any hyperplane 
	$\Sigma \in \mathsf{SpHP}$ orthogonal to $P$ is a 
	two-dimensional disc (the so-called \emph{M\o ller disc}) 
	in the plane orthogonal to the Pauli--Lubański vector 
	$W = (*(P^\flat \wedge J))^\sharp$, whose centre is the 
	centre of inertia on $\Sigma$ and whose radius is the 
	\emph{Møller radius}
	\begin{equation}
		R_M = \frac{S}{mc} \; ,
	\end{equation}
	where $S = \sqrt{W^2}/(mc)$ is the spin of the system 
	and $m = \sqrt{-P^2}/c$ its mass.
\end{thm}

\subsection{The centre of spin condition}

For a system with timelike four-momentum, the 
Pauli--Lubański vector $W$ has the interpretation 
of being ($mc$ times) the spin vector in the momentum 
rest frame. We now define the spin vector in an 
arbitrary Lorentz frame by boosting $W/(mc)$ to the 
new frame:
\begin{defn}
	Given the timelike four-momentum $P \in V$ and the 
	Pauli--Lubański vector $W \in P^\perp$ of a 
	physical system, its \emph{spin vector} in the 
	Lorentz frame given by the future-directed unit 
	timelike vector $u$ is
	\begin{equation}
		s(u) := B(u) \cdot \frac{W}{mc} \; ,
	\end{equation}
	where $B(u) \in \mathcal L_+^\uparrow$ is the 
	unique Lorentz boost with respect to $\frac{P}{mc}$ 
	(i.e.\ containing $\frac{P}{mc}$ in its timelike 
	2-plane of action) that maps $\frac{P}{mc}$ to 
	$u$, with $m = \sqrt{-P^2}/c$ being the mass. In terms 
	of components, this boost is given by\footnote
		{Generally, given two unit 
		timelike future-pointing vectors $n_1$ and $n_2$,
		then the boost that maps $n_1$ onto $n_2$ and 
		fixes the spacelike plane orthogonal to 
		$\mathrm{span}\{n_1,n_2\}$ is given by the 
		combination $\rho_{n_1+n_2}\circ\rho_{n_1}$ of two 
		hyperplane-reflections, where 
		$\rho_n := \mathrm{id}_V - 2\frac{n\otimes n^\flat}{n^2}$ 
		is the reflection at the hyperplane orthogonal to 
		$n$. Setting $n_1 = P/(mc)$ and $n_2 = u$ gives 
		\eqref{eq:linking_boost_explicit}.}
	\begin{equation} \label{eq:linking_boost_explicit}
		B^\mu_{\hphantom{\mu}\nu}(u)
		= \delta^\mu_\nu + \frac{\left(\frac{P^\mu}{mc} 
		+ u^\mu\right) \left(\frac{P_\nu}{mc} 
		+ u_\nu\right)}{1 - u \cdot \frac{P}{mc}} 
		- 2 \frac{u^\mu P_\nu}{mc} \; .
	\end{equation}
\end{defn}

\begin{defn}
	A \emph{centre of spin} position observable for 
	a 	classical Poincaré-invariant system 
	$(\Gamma, \omega, \Phi)$ with timelike four-momentum 
	is a position observable $\chi$ satisfying
	\begin{equation}
		s_\mu(u) = - \frac{1}{2} \varepsilon_{\mu\nu\rho\sigma} u^\nu S^{\rho\sigma}(u),
	\end{equation}
	where $S_{\rho\sigma}(u) := J_{\mu\nu} - \chi_\mu(u, \tau) P_\nu + \chi_\nu(u, \tau) P_\mu$ 
	is the spin tensor\footnote
		{Since $\frac{\partial\chi(u,\tau)}{\partial\tau}$ is 
		proportional to $P$, the spin tensor is independent 
		of $\tau$.}
	with respect to $\chi$. Expressed in terms of the 
	Hodge operator, this condition reads
	\begin{equation}
		s(u) = (*(u^\flat \wedge S(u)))^\sharp.
	\end{equation}
\end{defn}
With respect to an orthonormal basis 
$\{u = \E_0, \dots, \E_3\}$ adapted to $u$, the 
centre of spin condition takes the form
\begin{equation}
	s_0(u) = 0, \quad
	s_a(u) = - \frac{1}{2} \varepsilon_{a0\rho\sigma} S^{\rho\sigma}(u) 
		= \frac{1}{2} {^{(3)}\varepsilon}_{abc} S^{bc}(u),
\end{equation}
through which it acquires an immediate interpretation: 
a position observable is a centre of spin if and only 
if, for any Lorentz frame $u$, the spin vector 
defined by boosting the Pauli--Lubański vector 
to $u$ really generates spatial rotations around 
the point given by the position observable.

We will now rewrite the centre of spin condition. 
Since $S(u) = J - (\chi(u,\tau))^\flat \wedge P^\flat$, 
we can rewrite the Pauli--Lubański vector as 
$W = \left[* \left(\frac{P^\flat}{mc} \wedge J \right) \right]^\sharp = \left[* \left(\frac{P^\flat}{mc} \wedge S(u) \right) \right]^\sharp$. 
Thus, the centre of spin condition takes the form
\begin{equation}
	(B(u)^{-1})^\top \left[*\left(\frac{P^\flat}{mc} \wedge S(u)\right)\right] = *(u^\flat \wedge S(u)).
\end{equation}
Since $B(u)$ is a Lorentz transformation, i.e.\ an 
isometry of $(V,\eta)$, and it maps $P/(mc)$ to $u$, 
this is equivalent to
\begin{equation}
	u^\flat \wedge \left((B(u)^{-1})^\top \otimes (B(u)^{-1})^\top\right) (S(u)) = u^\flat \wedge S(u).
\end{equation}
Using the explicit form \eqref{eq:linking_boost_explicit} 
of $B(u)$, we see that
\begin{equation}
	\left((B(u)^{-1})^\top \otimes (B(u)^{-1})^\top\right) (S(u)) = S(u) + \frac{\frac{P^\flat}{mc} \wedge \left(\iota_{u + \frac{P}{mc}} S(u)\right)}{1 - u \cdot \frac{P}{mc}} + u^\flat \wedge (\ldots).
\end{equation}
Thus, we have the following:
\begin{lemqed}
	The centre of spin condition is equivalent to
	\begin{equation} \label{eq:centre_of_spin_equiv_1}
		u^\flat \wedge P^\flat \wedge \left(\iota_{u + \frac{P}{mc}} S(u)\right) = 0. \qedhere
	\end{equation}
\end{lemqed}
Since the Newton--Wigner position observable is defined 
by the \acronym{SSC} $\iota_{u + \frac{P}{mc}} S(u) = 0$, the 
preceding result immediately implies
\begin{thmqed}
	The Newton--Wigner position observable $\chi^\mathrm{NW}$ 
	is a centre of spin.
\end{thmqed}

Further rewriting the centre of spin condition, we see 
that \eqref{eq:centre_of_spin_equiv_1} is equivalent to
\begin{equation}
	\iota_{u + \frac{P}{mc}} S(u) \in \mathrm{span} \{u^\flat, P^\flat\}.
\end{equation}
Due to the antisymmetry of $S(u)$, this is equivalent to
\begin{equation} \label{eq:centre_of_spin_equiv_2}
	\iota_{u + \frac{P}{mc}} S(u) \in \mathrm{span} \left\{u^\flat - \frac{P^\flat}{mc} \right\}.
\end{equation}
Using this, we can show:
\begin{lem} \label{lem:centre_of_spin_diff_NW}
	$\chi$ is a centre of spin $\iff$ 
	$\chi(u, \tau) - \chi^\mathrm{NW}(u, \tau) \in \mathrm{span}\{u, P\}$.

	\begin{proof}
		Writing $D := \chi(u, \tau) - \chi^\mathrm{NW}(u,\tau)$, 
		the spin tensor of $\chi$ may be expressed as $S(u) = S^\mathrm{NW}(u) - D^\flat \wedge P^\flat$. 
		Thus, \eqref{eq:centre_of_spin_equiv_2} is equivalent to
		\begin{equation} \label{eq:centre_of_spin_equiv_3}
			\iota_{u + \frac{P}{mc}} (D^\flat \wedge P^\flat) \in \mathrm{span} \left\{u^\flat - \frac{P^\flat}{mc} \right\}.
		\end{equation}
		We have $\iota_{u + \frac{P}{mc}} (D^\flat \wedge P^\flat) = (D\cdot u + \frac{D\cdot P}{mc}) P^\flat - (P\cdot u - mc) D^\flat$, 
		and thus \eqref{eq:centre_of_spin_equiv_3} implies that 
		for all $v \in u^\perp \cap P^\perp$, we have
		\begin{equation}
			v\cdot D = 0.
		\end{equation}
		But this means $D \in (u^\perp \cap P^\perp)^\perp = \mathrm{span}\{u, P\}$.

		Conversely, if $D \in \mathrm{span}\{u, P\}$, we 
		have $\iota_{u + \frac{P}{mc}} (D^\flat \wedge P^\flat) \in \mathrm{span} \left\{ \iota_{u + \frac{P}{mc}} (u^\flat \wedge P^\flat) \right\}$. 
		But now
		\begin{equation}
			\iota_{u + \frac{P}{mc}} (u^\flat \wedge P^\flat) = \left(-1 + u\cdot \frac{P}{mc}\right) P^\flat - (u\cdot P - mc) u^\flat = (mc - u\cdot P) \left(u^\flat - \frac{P^\flat}{mc}\right),
		\end{equation}
		and thus we have \eqref{eq:centre_of_spin_equiv_3}, 
		i.e.\ $\chi$ is a centre of spin.
	\end{proof}
\end{lem}

We can now prove the main result of this section.
\begin{thm} \label{thm:NW_SSC_centre_of_spin}
	The Newton--Wigner position observable $\chi^\mathrm{NW}$ 
	is the only centre of spin position observable that is 
	continuous and defined by an \acronym{SSC}.
	
	\begin{proof}
		Let $\chi$ be an \acronym{SSC} position observable. Writing 
		$D(u,\tau) := \chi(u, \tau) - \chi^\mathrm{NW}(u,\tau)$, 
		we know by the M\o ller disc theorem (theorem 
		\ref{thm:Moller_disc}) that the projection of 
		$D(u, \tau)$ orthogonal to $P$ is orthogonal to the 
		Pauli--Lubański vector $W$. Thus, since $P$ itself 
		is orthogonal to $W$, we have
		\begin{equation} \label{eq:NW_SSC_unique_proof_orth}
			D(u, \tau) \perp W
		\end{equation}
		for any $(u, \tau) \in \mathsf{SpHP}$. In addition, 
		we know that $D(u, \tau) \perp u$; in particular, 
		$D(u, \tau)$ is spacelike for any 
		$(u, \tau) \in \mathsf{SpHP}$.

		Now suppose that $\chi$ is a centre of spin. By lemma 
		\ref{lem:centre_of_spin_diff_NW} this means that
		\begin{equation}
			D(u, \tau) \in \mathrm{span}\{u, P\}
		\end{equation}
		for all $(u, \tau) \in \mathsf{SpHP}$. Using 
		\eqref{eq:NW_SSC_unique_proof_orth} and $P \perp W$, 
		we conclude that
		\begin{equation}
			\text{for all} \; u \; \text{with} \; u\cdot W \ne 0 : D(u, \tau) \in \mathrm{span}\{P\}.
		\end{equation}
		Since $D(u, \tau)$ has to be spacelike, we thus have shown
		\begin{equation}
			D(u, \tau) = 0 \; \text{for all} \; u \; \text{with} \; u\cdot W \ne 0.
		\end{equation}
		If $W \ne 0$, the set of future-directed unit timelike 
		$u$ satisfying $u \cdot W \ne 0$ is dense in the 
		hyperboloid of all possible $u$, and thus assuming 
		continuity of $\chi$, we conclude that $D(u, \tau) = 0$ 
		for all $u$, finishing the proof.

		If $W = 0$, then by the M\o ller disc theorem all 
		\acronym{SSC} worldlines coincide, and thus we also have 
		$\chi = \chi^\mathrm{NW}$.
	\end{proof}
\end{thm}

Looking back into the various steps of the proofs 
it is interesting to note how the 
`extrinsic--intrinsic' combination $u+P/(mc)$ 
for $f$ came about. It entered through 
the unique boost transformation 
\eqref{eq:linking_boost_explicit} that 
was needed in order to transform an intrinsic 
quantity to an externally specified rest frame. 
The intrinsic quantity is the spin vector 
in the momentum rest frame, i.e.\ the
Pauli--Lubański vector, which is a function 
of $\Gamma$ only, and the externally specified 
frame is defined by $u$, which is independent 
of $\Gamma$ and determined through the choice of 
$\Sigma\in\mathsf{SpHP}$.

\section[A Newton--Wigner theorem for classical elementary systems]{A Newton--Wigner theorem for classical\newline elementary systems}
\label{sec:NW_theorem}

For elementary Poincaré-invariant quantum systems 
--~i.e.\ quantum systems with an irreducible unitary 
action of the Poincaré group~-- the Newton--Wigner 
position operator is uniquely characterised by 
transforming `as a position should' under translations, 
rotations and time reversal, having commuting components 
and satisfying a regularity condition. This has been 
well-known since the original publication by Newton 
and Wigner \cite{Newton.Wigner:1949}. As advertised in 
the introduction, we shall now prove an analogous 
statement for classical systems.

For the whole of this section, we fix a future-directed 
unit timelike vector $u$ defining a Lorentz frame, 
and an adapted positively oriented orthonormal basis 
$\{u = \E_0, \dots, \E_3\}$. Unless otherwise stated, 
phrases such as `temporal', `spatial' and the like 
refer to the preferred time direction given by $u$. 
We will raise and lower spatial indices by the Euclidean 
metric $\delta$ induced by the Minkowski metric $\eta$ 
on the orthogonal complement of $u$; the components of 
$\delta$ in the adapted basis are simply given by the 
usual Kronecker delta. We denote the spatial volume form 
by $^{(3)}\varepsilon = \iota_u\varepsilon$.

Similar to the notation introduced in chapter \ref{chap:geometric_structures}, 
we will employ a `three-vector' notation for spatial 
vectors, for example writing $\vect A = (A^a)$. We then use 
the usual three-vector notations for the Euclidean scalar 
product $\vect A \cdot \vect B = A_a B^a$, the Euclidean 
norm $|\vect A| := \sqrt{\vect A^2}$ and the vector product 
$(\vect A \times \vect B)_a = {^{(3)}\varepsilon_{abc}} A^b B^c$.

\subsection{Classical elementary systems}

In the quantum case, an elementary system is given by a 
Hilbert space with an \emph{irreducible} unitary action 
of the Poincaré group~-- i.e.\ each state of the system 
is connected to any other by a Poincaré transformation. 
In direct analogy, we define the notion of a classical 
elementary system:
\begin{defn}
	A \emph{classical elementary system} is a classical 
	Poincaré-invariant system $(\Gamma, \omega, \Phi)$, 
	where $\Phi$ is a transitive action of the proper 
	orthochronous Poincaré group $\mathcal P_+^\uparrow$.
\end{defn}
Note the we only assumed an action of the identity 
connected component of the Poincaré group, whereas 
Arens in \cite{Arens:1971b} considered the whole 
Poincaré group. In the classical context, simple 
transitivity replaces irreducibility in the quantum 
case.

Arens classified the classical elementary systems\footnote
	{In fact, Arens classified what he called 
	\emph{one-particle} elementary systems (systems 
	that admit a map from $\Gamma$ to the set of lines 
	in Minkowski space which is equivariant with respect 
	to a certain subgroup of $\mathcal P_+^\uparrow$). 
	However, he also proved that this `one-particle' 
	condition is fulfilled for an elementary system 
	if and only if the four-momentum is not zero.}
in \cite{Arens:1971b}; the classification proceeds 
in terms of the system's four-momentum and 
Pauli--Lubański vector (similar to the Wigner
classification in the quantum case \cite{Wigner:1939}). We are only 
interested in the case of timelike four-momentum. 
For this case, the phase space can be explicitly 
constructed as follows:

\begin{thm}[Phase space of a classical elementary system] \label{thm:phase_space}
	Any classical elementary system with timelike 
	four-momentum is equivalent (in the sense of a 
	symplectic isomorphism respecting the action of 
	$\mathcal P_+^\uparrow$) to precisely one of the 
	following two cases:
	\begin{enumerate}[label=(\roman*)]
		\item (Spin zero, one parameter $m \in \mathbb R_+$)
			\begin{itemize}
				\item Phase space $\Gamma = T^*\mathbb R^3$ with 
					coordinates $(\vect x, \vect p)$, symplectic 
					form $\omega = \D x^a \wedge \D p_a$
				\item Poincaré generators (i.e.\ component functions of
					the momentum map):
					\begin{subequations}
					\begin{align}
						\text{spatial translations} \quad P_a &= p_a \\
						\text{time translation} \quad P_0 &= -\sqrt{m^2 c^2 + \vect p^2} \\
						\text{rotations} \quad J_{ab} &= x_a p_b - x_b p_a \\
						\text{boosts} \quad J_{a0} &= P_0 x_a
					\end{align}
					\end{subequations}
			\end{itemize}

		\item (Spin non-zero, two parameters $m, S \in \mathbb R_+$)
			\begin{itemize}
				\item Phase space $\Gamma = T^*\mathbb R^3 \times \mathsf S^2$ 
					with coordinates $(\vect x, \vect p)$ for $T^*\mathbb R^3$, 
					symplectic form $\omega = \D x^a \wedge \D p_a + S \cdot \D\Omega^2$ 
					where $\D\Omega^2$ is the standard volume form 
					on $\mathsf S^2$. We denote the phase space function 
					projecting onto the second factor $\mathsf S^2$ by 
					$\hatvect s \colon \Gamma \to \mathsf S^2 \subset\mathbb R^3$. 
					The \emph{spin vector} observable is the 
					$\mathsf S^2_S$-valued phase space function 
					$\vect s := S \cdot \hatvect s$; its components satisfy 
					the Poisson bracket relations
					\begin{equation}
						\{s_a, s_b\} = {^{(3)}\varepsilon_{abc}} s^c.
					\end{equation}
					Here $\mathsf S^2_S \subset \mathbb R^3$ denotes 
					the 2-sphere of radius $S$ in $\mathbb R^3$.
				\item Poincaré generators (i.e.\ component functions of
					the momentum map):
					\begin{subequations}
					\begin{align}
						\text{spatial translations} \quad P_a &= p_a \\
						\text{time translation} \quad P_0 &= -\sqrt{m^2 c^2 + \vect p^2} \\
						\text{rotations} \quad J_{ab} &= x_a p_b - x_b p_a + {^{(3)}\varepsilon_{abc}} s^c \\
						\text{boosts} \quad J_{a0} &= P_0 x_a - \frac{(\vect p \times \vect s)_a}{mc - P_0}
					\end{align}
					\end{subequations}
			\end{itemize}
	\end{enumerate}
\end{thm}

Note that in fact the explicit construction of the systems 
in \cite{Arens:1971b} as co-adjoint orbits of $\mathcal P_+^\uparrow$ 
is quite different in appearance to the forms given above. 
However, one can show that the above systems are indeed 
elementary systems (i.e.\ that the action of $\mathcal P_+^\uparrow$ 
is transitive), and thus due to Arens' uniqueness result 
they are possible representatives of their respective 
classes. We will use the forms given above, which were 
anticipated by Bacry in \cite{Bacry:1967}, since they will 
be easier to explicitly work with. To unify notation, we 
let $S = 0, \vect s := 0$ in the case of zero-spin systems. 
Furthermore, we introduce the open subset of phase space 
$\Gamma^* := \Gamma \setminus \{|\vect P| = 0\}$ and the 
$\mathsf S^2$-valued function $\hatvect P := \frac{\vect P}{|\vect P|}$ 
on $\Gamma^*$.\looseness-1
\enlargethispage{\baselineskip}

Using the explicit form of the systems given in theorem 
\ref{thm:phase_space}, one directly checks:
\begin{lemqed} \label{lem:mom_spin_CIV}
	For a classical elementary system with timelike 
	four-momentum, the functions $P_a, \hatvect P \cdot \vect s$ 
	(or just the $P_a$ in the case of zero spin) form a 
	complete involutive set on $\Gamma^*$ (or the whole of 
	$\Gamma$ in the case of zero spin).
\end{lemqed}

The behaviour of the momentum and spin vectors under translations 
and rotations is also easily obtained:
\begin{lem} \label{lem:mom_spin_trafo}
	For a classical elementary system with timelike 
	four-momentum, $\vect P$ and $\vect s$ are invariant under 
	translations and `transform as vectors' under spatial 
	rotations, i.e.\ we have
	\begin{equation}
		\{P_a, V_b\} = 0, \quad 
		\{J_{ab}, V_c\} = \delta_{ac} V_b - \delta_{bc} V_a 
		\quad \text{for}\quad 
		\vect V = \vect P, \vect s.
	\end{equation}
	\begin{proof}
		For $\vect P$, these are part of the Poincaré algebra relations and thus true by definition. For $\vect s$, they are easily confirmed using the explicit form of the Poincaré generators.
	\end{proof}
\end{lem}

For our considerations, we will need to know how the 
\emph{time reversal} operation with respect to the 
hyperplane in $M$ through the origin $o \in M$ and 
orthogonal to $u = \E_0$ is implemented on phase 
space. In order to get this right, we recall that 
the incorporation of time reversal in the context 
of special relativity corresponds, by its very 
definition, to a particular upward $\mathbb{Z}_2$ 
extension\footnote
	{Here we are using the terminology of 
	\cite[p.\,xx]{Conway.EtAl:AOFG}, according to which 
	a group $G$ with normal subgroup $A$ and quotient 
	$G/A \cong B$ is either called an 
	\emph{upward extension} of $A$ by $B$ or a 
	\emph{downward extension} of $B$ by $A$.}
of $\mathcal{P}_+^\uparrow$, i.e.\ the formation 
of a new group called 
$\mathcal{P}_+^\uparrow\cup\mathcal{P}_-^\downarrow$ 
of which $\mathcal{P}_+^\uparrow$ is a normal 
subgroup with $(\mathcal{P}_+^\uparrow\cup\mathcal{P}_-^\downarrow)/\mathcal{P}_+^\uparrow\cong\mathbb{Z}_2$. 
It is the particular nature of this extension that 
eventually defines what is meant by 
\emph{time reversal}: it consists in the requirement 
that the outer automorphism induced by the only 
non-trivial element of $\mathbb{Z}_2$ on the Lie 
algebra $\mathfrak{p}$ of $\mathcal{P}_+^\uparrow$ 
shall be the one which reverses the sign of spatial 
translations and rotations and leaves invariant 
boosts and time translations; see, e.g., 
\cite{Bacry.LevyLeblond:1968}. Implementing time 
reversal on phase space then means to extend the 
action of $\mathcal{P}_+^\uparrow$ to an action 
of $\mathcal{P}_+^\uparrow\cup\mathcal{P}_-^\downarrow$.

Now, according to this scheme, we can immediately 
write down how our particular time reversal transformation 
on phase space, $T_u\colon\Gamma\to\Gamma$, acts on the 
Poincaré generators, i.e.\ the component functions of the 
momentum map: 
\begin{equation}
	P_a \circ T_u = - P_a \; , \quad 
	J_{ab} \circ T_u = - J_{ab} \; , \quad 
	J_{a0} \circ T_u = J_{a0} \; , \quad 
	P_0 \circ T_u = P_0
\end{equation}
From this the well-known result follows 
that time reversal (as defined above) necessarily 
corresponds to an \emph{anti}-symplectomorphism 
(inverting the sign of the symplectic form). 
Hence, in the process of extending our symplectic 
action of $\mathcal{P}_+^\uparrow$ on $\Gamma$ to an action 
of $\mathcal{P}_+^\uparrow\cup\mathcal{P}_-^\downarrow$ 
satisfying the time reversal criterion above, we had to 
generalise to possibly \emph{anti}-symplectomorphic 
actions. This is akin to the situation in quantum 
mechanics, where, as is well-known, time reversal 
necessarily corresponds to an \emph{anti}-unitary 
transformation.

It is now clear how time reversal is implemented 
in the case at hand:
\begin{lemqed} \label{lem:time_rev}
	For an elementary system as in theorem
	\ref{thm:phase_space}, time reversal with respect 
	to the hyperplane through the origin and orthogonal 
	to $u = \E_0$ is given by
	\begin{equation}
			T_u \colon (\vect x, \vect p, \hatvect s) \mapsto (\vect x, - \vect p, - \hatvect s). \qedhere
		\end{equation}
\end{lemqed}
Unless otherwise stated, in the following we will 
always mean time reversal with respect to the 
hyperplane through the origin and orthogonal to 
$u = \E_0$ when saying `time reversal'.
\enlargethispage{\baselineskip}

\subsection{Statement and interpretation of the Newton--Wigner theorem}

The classical Newton--Wigner theorem we are going to 
prove can be formulated very similar to the quantum case:
\begin{thm}[Classical Newton--Wigner theorem] \label{thm:NW}
	For a classical elementary system with timelike 
	four-momentum, there is a unique $\mathbb R^3$-valued 
	phase space function $\vect X$ that
	\begin{enumerate}[label=(\roman*)]
		\item is $C^1$,
		\item has Poisson-commuting components,
		\item \label{ass:NW_transl} satisfies the canonical 
			Poisson relations $\{X^a, P_b\} = \delta^a_b$ 
			with the generators of spatial translations with 
			respect to $u = \E_0$,
		\item \label{ass:NW_rot} transforms `as a (position) 
			vector' under spatial rotations with respect to 
			$u = \E_0$, i.e.\ satisfies $\{J_{ab}, X^c\} = \delta_a^c X_b - \delta_b^c X_a$, and
		\item is invariant under time reversal with respect to 
			the hyperplane through the origin and orthogonal to 
			$u = \E_0$, i.e.\ satisfies $\vect X \circ T_u = \vect X$.
	\end{enumerate}

	In terms of the Poincaré generators, it is given by
	\begin{equation} \label{eq:NW_pos_thm}
		X_a = -\frac{J_{a0}}{mc} - \frac{J_{ab} P^b}{mc (mc - P_0)} - \frac{J_{b0} P^b}{P_0 mc (mc - P_0)} P_a \; ,
	\end{equation}
	where $m = \sqrt{P_0^2 - \vect P^2}/c$ is the mass of the system.
\end{thm}
Before proving the theorem in the next section, 
we will now discuss the interpretation of the 
`position' $\vect X$ it characterises. We want to 
interpret the value of $\vect X$ (in some state 
$\gamma \in \Gamma$) as the spatial components 
of a point in Minkowski spacetime $M$. Since $\vect X$ 
is invariant under time reversal with respect to 
the hyperplane through the origin and orthogonal 
to $u = \E_0$, it can be interpreted as defining a 
point on this hyperplane. Thus, if we want to use 
the phase space function from the Newton--Wigner 
theorem to define a position observable $\chi$ in 
the sense of section \ref{sec:pos_obs}, we should 
set (in our basis adapted to $u$)
\begin{equation} \label{eq:pos_obs_def_by_NW_thm}
	\chi^a(u, \tau = 0) := X^a \; , \quad
	\chi^0(u, \tau = 0) := 0.
\end{equation}
The transformation behaviour of $\vect X$ under spatial 
translations and rotations (i.e.\ assumptions 
\ref{ass:NW_transl} and \ref{ass:NW_rot} of theorem 
\ref{thm:NW}) will then ensure that the position 
observable $\chi$ be covariant (in the sense of definition 
\ref{defn:pos_obs_cov}) regarding these transformations.

In fact, comparing \eqref{eq:NW_pos_thm} to the expression 
\eqref{eq:NW_pos} for the Newton--Wigner position 
observable $\chi^\mathrm{NW}$, we see that we have 
(in our adapted basis)
\begin{equation}
	\chi^{\mathrm{NW},a}(u, \tau = 0) = X^a \; , \quad
	\chi^{\mathrm{NW},0}(u, \tau = 0) = 0\colon
\end{equation}
the position $\vect X$ characterised by theorem 
\ref{thm:NW} is the one given by the Newton--Wigner 
position observable $\chi^\mathrm{NW}$ on the 
hyperplane $(u, 0) \in \mathsf{SpHP}$ (which is a 
covariant position observable due to proposition 
\ref{prop:SSC_pos_obs_cov}). Let us also remark that 
since any position observable's dependence on $\tau$ 
is fixed by \eqref{eq:pos_obs_time_der}, a position 
observable satisfying \eqref{eq:pos_obs_def_by_NW_thm} 
is equal to the Newton--Wigner observable 
$\chi^\mathrm{NW}$ on the whole family of hyperplanes 
$\Sigma \in \mathsf{SpHP}$ with normal vector $u$.

Combining this identification with the observation that 
we can freely choose the origin $o\in M$, we can restate 
the Newton--Wigner theorem in the following form:
\begin{thmqed}[Classical Newton--Wigner theorem, version 2] \label{thm:NW2}
	For a classical elementary system with timelike 
	four-momentum, given any hyperplane $\Sigma = (u, \tau) \in \mathsf{SpHP}$, 
	there is a unique $\Sigma$-valued phase space function 
	$\chi^\mathrm{NW}(\Sigma)$ that
	\begin{subequations}
	\begin{enumerate}[label=(\roman*)]
		\item is $C^1$,
		\item has Poisson-commuting components, i.e.
			\begin{equation}
				\left\{ \chi^{\mathrm{NW},\mu}(\Sigma), \chi^{\mathrm{NW},\nu}(\Sigma) \right\} = 0,
			\end{equation}
		\item satisfies the canonical Poisson relations with 
			the generators of spatial translations with respect 
			to $u$, i.e.
			\begin{equation}
				v_\mu w^\nu \left\{ \chi^{\mathrm{NW},\mu}(\Sigma), P_\nu \right\} = v\cdot w \; \text{for} \; v, w \in u^\perp,
			\end{equation}
		\item transforms `as a position' under spatial rotations 
			with respect to $u$, i.e.\ satisfies
			\begin{equation}
				v^\mu \tilde v^\nu w_\rho \left\{ J_{\mu\nu}, \chi^{\mathrm{NW},\rho}(\Sigma) \right\} = v^\mu \tilde v^\nu w_\rho \left[\delta_\mu^\rho \chi^\mathrm{NW}_\nu(\Sigma) - \delta_\nu^\rho \chi^\mathrm{NW}_\mu(\Sigma)\right] \; \text{for} \; v, \tilde v, w \in u^\perp,
			\end{equation}
			and
		\item is invariant under time reversal with respect 
			to $\Sigma$.
	\end{enumerate}
	\end{subequations}

	These $\chi^\mathrm{NW}(\Sigma)$ together form the 
	Newton--Wigner observable as given by \eqref{eq:NW_pos}.
\end{thmqed}

\subsection{Proof of the Newton--Wigner theorem}

\begin{proof}[Proof of theorem \ref{thm:NW}]
For the whole of the proof, we will work with the 
explicit form of the phase space of our elementary 
system given in theorem \ref{thm:phase_space}. It is 
easily verified that in this explicit form, $\vect x$ 
(i.e.\ the coordinate of the base point in 
$T^*\mathbb R^3$) is a phase space function with the 
properties demanded for $\vect X$. Thus we need to 
prove uniqueness. Our proof will follow the proof of 
the quantum-mechanical Newton--Wigner theorem given 
by Jordan in \cite{Jordan:1980}, some parts of which 
can be applied literally to the classical case.

We will several times need the following.
\begin{lem} \label{lem:invar_function}
	Consider a classical elementary system with timelike 
	four-momentum, with phase space $\Gamma$, and some 
	open subset $\tilde\Gamma$ of 
	$\Gamma^* = \Gamma \setminus \{|\vect P| = 0\}$. 
	Let $f$ be an $\mathbb R$-valued $C^1$ function 
	defined on $\tilde\Gamma$ that is invariant under 
	spatial translations and rotations, i.e.\ 
	$\{P_a, f\} = 0 = \{J_{ab}, f\}$. Then $f$ is a function of 
	$|\vect P|, \hatvect P \cdot \vect s$. \footnote
		{By `$f$ is a function of $|\vect P|, \hatvect P \cdot \vect s$ ' 
		we mean that $f$ depends on phase space only via 
		$|\vect P|, \hatvect P \cdot \vect s$, i.e.\ that there 
		is a $C^1$ function $F\colon U \to \mathbb R$, 
		$U = \left\{(|\vect P|(\gamma), (\hatvect P \cdot \vect s)(\gamma)) : \gamma \in \tilde\Gamma\right\} \subset \mathbb R_+ \times [-S,S]$ 
		satisfying
		\[f(\gamma) = F(|\vect P|(\gamma), (\hatvect P \cdot \vect s)(\gamma)) \; \text{for all} \; \gamma \in \tilde\Gamma.\]}

	\begin{proof}
		$f$ Poisson-commutes with $\vect P$ and $J_{ab}$. 
		Therefore it also Poisson-commutes with $\vect P$ 
		and $\frac{1}{2} {^{(3)}\varepsilon^{abc}} \hat P_a J_{bc} = \hatvect P \cdot \vect s$. 
		Now $\vect P, \hatvect P \cdot \vect s$ form a complete 
		involutive set on $\Gamma^*$ (lemma 
		\ref{lem:mom_spin_CIV}), so since $f$ Poisson-commutes 
		with them, it must be a function of 
		$\vect P, \hatvect P \cdot \vect s$. Since $f$ and 
		$\hatvect P \cdot \vect s$ are rotation invariant (by 
		lemma \ref{lem:mom_spin_trafo}), $f$ must be a 
		function of $|\vect P|, \hatvect P \cdot \vect s$.
	\end{proof}
\end{lem}

Let now $\vect X$ be an observable as in the 
statement of theorem \ref{thm:NW}, and consider 
the difference $\vect d := \vect X - \vect x$. Due to 
the assumptions of theorem \ref{thm:NW}, $\vect d$ 
is $C^1$, is invariant under translations 
(i.e.\ $\{d^a, P_b\} = 0$), transforms as a vector under 
spatial rotations (i.e.\ $\{J_{ab}, d^c\} = \delta_a^c d_b - \delta_b^c d_a$) 
and is invariant under time reversal with respect to 
the hyperplane through the origin and orthogonal to 
$u$ (i.e.\ $\vect d \circ T_u = \vect d$).

\begin{lem} \label{lem:invar_vector_no_mom_comp}
	Let $\vect A$ be a $\mathbb R^3$-valued $C^1$ phase 
	space function on a classical elementary system with 
	timelike four-momentum that is invariant under 
	translations, transforms as a vector under spatial 
	rotations and is invariant under time reversal. Then 
	$\vect A \cdot \vect P = 0$.

	\begin{proof}
		Since $\vect P$ is invariant under translations and 
		a vector under rotations, $\vect A \cdot \vect P$ is 
		invariant under translations and rotations. By lemma 
		\ref{lem:invar_function}, $\left. \vect A \cdot \vect P \right|_{\Gamma^*}$ 
		is a function of $|\vect P|, \hatvect P \cdot \vect s$. 
		This means we have
		\begin{equation} \label{eq:pf_lem_scalar_product_mom_1}
			\left. \vect A \cdot \vect P \right|_{\Gamma^*} = F(|\vect P|, \hatvect P \cdot \vect s)
		\end{equation}
		for some function $F\colon \mathbb R_+ \times [-S,S] \to \mathbb R$.

		Now considering time reversal $T_u$, on the one hand 
		we have (using lemma \ref{lem:time_rev})
		\begin{subequations}
		\begin{equation}
			|\vect P| \circ T_u = |\vect P \circ T_u| = |-\vect P| = |\vect P|
		\end{equation}
		and
		\begin{align}
			(\hatvect P \cdot \vect s) \circ T_u &= \left(\frac{1}{2} {^{(3)}\varepsilon^{abc}} \hat P_a J_{bc}\right) \circ T_u \nonumber\\
			&= \frac{1}{2} {^{(3)}\varepsilon^{abc}} (\hat P_a \circ T_u) (J_{bc} \circ T_u) \nonumber\\
			&= \frac{1}{2} {^{(3)}\varepsilon^{abc}} (-\hat P_a) (-J_{bc}) \nonumber\\
			&= \frac{1}{2} {^{(3)}\varepsilon^{abc}} \hat P_a J_{bc} \nonumber\\
			&= \hatvect P \cdot \vect s,
		\end{align}
		\end{subequations}
		implying
		\begin{equation} \label{eq:pf_lem_scalar_product_mom_2}
			F(|\vect P|, \hatvect P \cdot \vect s) \circ T_u = F(|\vect P| \circ T_u, (\hatvect P \cdot \vect s) \circ T_u) = F(|\vect P|, \hatvect P \cdot \vect s).
		\end{equation}
		On the other hand, $\vect A$ is invariant under time 
		reversal while $\vect P$ changes its sign, implying that 
		$(\vect A \cdot \vect P) \circ T_u = - \vect A \cdot \vect P$. 
		Combining this with \eqref{eq:pf_lem_scalar_product_mom_1} 
		and \eqref{eq:pf_lem_scalar_product_mom_2}, we obtain 
		$\left. \vect A \cdot \vect P \right|_{\Gamma^*} = 0$, 
		and continuity implies $\vect A \cdot \vect P = 0$.
	\end{proof}
\end{lem}

For zero spin, we can easily complete the proof of the 
Newton--Wigner theorem. Since the difference vector 
$\vect d$ is translation invariant and the $P_a$ form a 
complete involutive set on $\Gamma$, $\vect d$ must be a 
function of $\vect P$. Then since it is a vector under 
rotations, it must be of the form
\begin{equation}
	\vect d(\vect P) = F(|\vect P|) \vect P
\end{equation}
for some function $F$ of $|\vect P|$. Then, since 
according to lemma \ref{lem:invar_vector_no_mom_comp} 
$\vect d \cdot \vect P$ is zero, $\vect d$ is zero. Thus, 
for the spin-zero case, we have proved the Newton--Wigner 
theorem without any use of the condition of 
Poisson-commuting components of the position observable.

For the non-zero spin case, we continue as follows.

\begin{lem} \label{lem:invar_vector_in_basis}
	Let $\vect A$ be a $\mathbb R^3$-valued $C^1$ phase space 
	function on a classical elementary system with timelike 
	four-momentum and non-zero spin that is invariant under 
	translations, transforms as a vector under spatial 
	rotations and satisfies $\vect A \cdot \vect P = 0$. 
	Then it is of the form
	\begin{equation} \label{eq:invar_vector_in_basis}
		\vect A = B \hatvect P \times \vect s + C \hatvect P \times (\hatvect P \times \vect s)
	\end{equation}
	on $\Gamma^* \setminus \{\vect s \parallel \hatvect P\}$, 
	where $B$ and $C$ are $C^1$ functions of $|\vect P|$ 
	and $\hatvect P \cdot \vect s$, i.e.\ $C^1$ functions
	\[B,C \colon \mathbb R_+ \times (-S,S) \to \mathbb R.\]

	\begin{proof}
		For the whole of this proof, we will work on 
		$\tilde\Gamma := \Gamma^* \setminus \{\vect s \parallel \hatvect P\}$. 
		Since evaluated at each point of $\tilde\Gamma$, the 
		$\mathbb R^3$-valued functions 
		$\hatvect P, \hatvect P \times \vect s, \hatvect P \times (\hatvect P \times \vect s)$ 
		form an orthogonal basis of $\mathbb R^3$, and since 
		we have $\vect A \cdot \vect P = 0$, we can write 
		$\vect A$ in the form \eqref{eq:invar_vector_in_basis} with 
		coefficients $B,C$ given by
		\begin{align}
			B &= \frac{\vect A \cdot (\hatvect P \times \vect s)}{|\hatvect P \times \vect s|} \; , \\
			C &= \frac{\vect A \cdot (\hatvect P \times (\hatvect P \times \vect s))}{|\hatvect P \times (\hatvect P \times \vect s)|} \; .
		\end{align}
		Since $\vect A$, $\vect P$ and $\vect s$ are invariant 
		under translations and vectors under rotations, these 
		equations imply that $B,C$ are invariant under 
		translations and rotations. The result follows with 
		lemma \ref{lem:invar_function}.
	\end{proof}
\end{lem}

Now we consider again the difference vector 
$\vect d = \vect X - \vect x$. It satisfies 
$\vect d \cdot \vect P = 0$ by lemma \ref{lem:invar_vector_no_mom_comp}, 
and thus we have
\begin{equation}
	\vect X \cdot \vect P = \vect x \cdot \vect P.
\end{equation}
Since we assume that the components of $\vect X$ 
Poisson-commute with each other and that 
$\{X^a, P_b\} = \delta^a_b$, this implies
\begin{equation}
	\{X^a, \vect x \cdot \vect P\} = \{X^a, \vect X \cdot \vect P\} = X^a.
\end{equation}
Combining this with $\{x^a, \vect x \cdot \vect P\} = x^a$, 
we obtain
\begin{equation}
	\{d^a, \vect x \cdot \vect P\} = d^a.
\end{equation}
On the other hand, for any function $F$ of $\vect P$ and 
$\vect s$, we have
\begin{equation}
	\{F(\vect P, \vect s), \vect x \cdot \vect P\} = \{F(\vect P, \vect s), x^a\} P_a = - \frac{\partial F(\vect P, \vect s)}{\partial P_a} P_a = - |\vect P| \left. \frac{\partial F}{\partial |\vect P|} \right|_{\hatvect P = \mathrm{const.}, \vect s = \mathrm{const.}}.
\end{equation}
This implies
\begin{equation} \label{eq:der_diff}
	\vect d = - |\vect P| \left. \frac{\partial \vect d}{\partial |\vect P|} \right|_{\hatvect P = \mathrm{const.}, \vect s = \mathrm{const.}}.
\end{equation}
Combining lemmas \ref{lem:invar_vector_no_mom_comp} and 
\ref{lem:invar_vector_in_basis}, we know that $\vect d$ 
has the form \eqref{eq:invar_vector_in_basis} on 
$\Gamma^* \setminus \{\vect s \parallel \hatvect P\}$ for two 
functions $B, C \colon \mathbb R_+ \times (-S,S) \to \mathbb R$. 
Thus \eqref{eq:der_diff} implies the two equations
\begin{equation}
	B(|\vect P|, \hatvect P \cdot \vect s) = - |\vect P| \frac{\partial B(|\vect P|, \hatvect P \cdot \vect s)}{\partial |\vect P|}, \; C(|\vect P|, \hatvect P \cdot \vect s) = - |\vect P| \frac{\partial C(|\vect P|, \hatvect P \cdot \vect s)}{\partial |\vect P|}
\end{equation}
on $\mathbb R_+ \times (-S,S)$. These equations 
determine the $|\vect P|$ dependence of $B$ and $C$; 
they must be proportional to $|\vect P|^{-1}$. However, 
for $\vect d$ to be $C^1$ on the whole of $\Gamma$, in 
fact for \eqref{eq:invar_vector_in_basis} not to diverge 
as $|\vect P|\to 0$ even when coming from a \emph{single} 
direction $\hatvect P$, we then need $B$ and $C$ to vanish. 
Continuity implies $\vect d = 0$ on all of $\Gamma$. This 
finishes the proof of the Newton--Wigner theorem.
\end{proof}

\section{Conclusion}

In this chapter we have studied the localisation 
problem for classical system whose phase space 
is a symplectic manifold. We focussed on the 
Newton--Wigner position observable and asked for precise 
characterisations of it in order to gain additional 
understanding, over and above that already known from 
its practical use for the solution of concrete problems 
of motion, e.g., in general-relativistic astrophysics 
\cite{Steinhoff:2011,Schaefer.Jaranowski:2018}. 
We proved two theorems that we believe advance our 
understanding in the desired direction: 
first we showed how Fleming's geometric scheme 
\cite{Fleming:1965a} in combination with the 
characterisation of worldlines through \acronym{SSC}s (Spin 
Supplementary Conditions) allows to give a 
precise meaning to, and proof of, the fact 
that the Newton--Wigner position is the unique 
centre of spin. Given that interpretation, it 
also offers an insight as to why the Newton--Wigner 
\acronym{SSC} uses a somewhat unnatural looking `hybrid' 
combination $f = u + \frac{P}{mc}$, where $u$ is 
`external' or `kinematical', and $P$ is `internal' 
or `dynamical'. Then, restricting to elementary 
systems, i.e.\ systems whose phase space admits a 
transitive action of the proper orthochronous 
Poincaré group, we proved again a uniqueness result 
to the effect that the Newton--Wigner observable 
is the unique phase space function whose 
components satisfy the `familiar' Poisson 
relations, provided it is continuously 
differentiable, time-reversal invariant, and transforms 
as a vector under spatial rotations. These properties 
seem to be the underlying reason for the distinguished 
rôle it plays in solution strategies like those of 
\cite{Steinhoff:2011,Schaefer.Jaranowski:2018}, 
despite the fact that on a more general level of 
theorisation other choices (characterised by other 
\acronym{SSC}s) are often considered more appropriate; 
see, e.g., \cite{puetzfeld15}. We believe 
that our results add a conceptually clear and 
mathematically precise Hamiltonian underpinning 
of what the choice of the Newton--Wigner observable 
entails, at least in a special-relativistic 
context or, more generally, in general-relativistic 
perturbation theory around Minkowski space.


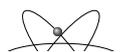
\chapter{Conclusion}
\label{chap:conclusion}

In this thesis, we have developed and analysed systematic 
methods for the description of quantum-mechanical systems 
to post-Newtonian gravitational fields.
As explained in the introduction, we see the virtue 
of our systematic calculations in their firm rooting
in explicitly spelled out principles, that leave no 
doubt concerning the questions of consistency and 
completeness of the obtained `relativistic corrections'. 
This, in our opinion, distinguishes our work from previous 
ones by other authors, who were also concerned with the 
coupling of composite particle quantum systems --~like 
atoms or molecules~-- to external gravitational fields,
who phrase their account of `relativistic corrections' 
in terms of semi-classical notions, like smooth 
worldlines and comparisons of their associated lengths 
(i.e.\ `proper time' and `redshift'); e.g.\ 
\cite{dimopoulos08,zych11,pikovski15,roura18,giese19,loriani19,zych19}. 
In our opinion, answers to the fundamental question of 
gravity--matter coupling in quantum mechanics should 
not be based on \emph{a priori} restricted states that 
imply a semi-classical behaviour of some of the 
(factorising) degrees of freedom. Rather, they should 
apply to all states in an equally valid fashion.
\enlargethispage{2\baselineskip}

In chapter \ref{chap:PNSE}, we have shown how to 
systematically derive a Schrödinger equation with 
post-Newtonian correction terms describing a single 
quantum particle in a general post-Newtonian curved 
background spacetime by means of a \acronym{WKB}-like 
formal expansion of the minimally coupled \KGe. We 
extended this method to account for, in principle, 
post-Newtonian terms of arbitrary orders in $c^{-1}$, 
although it gets recursive at higher orders, making it 
computationally more difficult to handle than methods 
based on formal quantisation of the classical description 
of the particle. Nevertheless, we believe this scheme to 
be better suited for concrete predictions, since it is 
more firmly based on first principles and also more 
systematic than \emph{ad hoc} canonical quantisation or 
path integral procedures as employed widely in the literature. 
For example, no operator ordering ambiguities arise; instead, 
the \acronym{WKB} method can be seen as \emph{predicting} 
the ordering.

Comparing the Klein--Gordon expansion method to canonical 
quantisation, we have found that in the case of a general 
metric, even at lowest post-Newtonian order, the two 
procedures lead to slightly different quantum Hamiltonians, 
\emph{independent of ordering ambiguities}\footnote
	{At least if only simple symmetrising procedures 
	are allowed for as ordering schemes in canonical 
	quantisation, see the discussion at the end of section 
	\ref{sec:WKB_trafo_flat}.}.
For the concrete case of the metric of the 
Eddington--Robertson \acronym{PPN} test theory, the 
Hamiltonians obtained from the two methods differ in a 
term including the Eddington--Robertson parameter $\gamma$, 
depending on the ordering scheme employed in canonical 
quantisation. Although the relevant term is proportional 
to the Laplacian of the Newtonian potential, i.e. (in 
lowest order) to the mass density generating the 
gravitational field, and thus is irrelevant in physical 
situations concerning the \emph{outside} of the generating 
matter distribution, this example shows that for the 
interpretation of tests of general relativity with quantum 
systems, the method used to derive the quantum Hamiltonian 
plays a decisive rôle.
\enlargethispage{\baselineskip}

For the case of stationary background metrics, without 
employing any expansion of the metric, we 
showed that up to linear order in spatial momenta, the 
Hamiltonians obtained from canonical quantisation 
and from the Klein--Gordon equation agree. 
In particular, this means that the lowest-order coupling 
to the `gravitomagnetic' field components $g^{0a}$ is 
independent of the gravity--quantum matter coupling method.

Concerning the applicability of the \acronym{WKB}-like expansion
method for concrete calculations, it could be an interesting 
question for future research if and how the transformation of 
the Hamiltonian from the Klein--Gordon inner product to an 
$\mathrm L^2$-scalar product --~be it flat or with respect 
to the induced metric measure~-- 
can be implemented more systematically, not relying on 
direct calculations with the already-computed Hamiltonian.

Turning to the description of composite systems, in chapter 
\ref{chap:atom_in_gravity} we extended the calculation of 
\cite{sonnleitner18} of a Hamiltonian describing 
an electromagnetically interacting 
two-particle system so as to include post-Newtonian gravity 
as described by the Eddington--Robertson \acronym{PPN} metric. 
Starting from first principles, we 
performed a post-Newtonian expansion in terms 
of the inverse velocity of light that led to 
leading-order corrections comprising special- 
and general-relativistic effects. The former were 
fully encoded in \cite{sonnleitner18}, but the 
latter are new. As in \cite{sonnleitner18} we neglected 
all terms of third and higher order in $c^{-1}$, which 
physically means that we neglected radiation-reaction 
and also that we avoided obstructions on the applicability 
of the Hamiltonian formalism that result from the 
infamous `no-interaction theorem' \cite{currie63,sudarshan16},
whose impact only starts at the 6th order in a 
$c^{-1}$ expansion \cite{martin78}.

Similar to the gravity-free case, we now \emph{derived}
the result that the centre of mass motion of the system 
can be viewed as that of a `composite point particle', including 
in its mass the internal energy of the system. This result may 
be anticipated in a heuristic fashion on semi-classical 
grounds, but, as seen, its proper derivation requires 
some efforts. We stress once more 
that for this interpretation it was crucial to express the 
Hamiltonian in terms of the physical space-time metric. 
As a result, our work lends some justification to current 
experimental proposals in atom interferometry that so far were 
based on these heuristic ideas, on the basis of which completeness 
of the relativistic effects could not be reliably judged; 
e.g.\ \cite{zych11,pikovski15,roura18,giese19,loriani19,zych19}. 

However, in order to obtain a fully solid framework for 
the discussion of atom interferometry in post-Newtonian 
gravity, stopping at the Hamiltonian is not enough: 
one has to describe the whole experimental situation 
solely in terms of operationally defined quantities. Such 
a systematic operational analysis of atom interferometers 
under gravity, which is now possible based on the Hamiltonian 
we obtained, we see as the most important future application 
of the results of this thesis. This may lead to interesting 
new possibilities of testing gravitational effects with 
quantum systems: in particular it might enable the measuring 
of post-Newtonian parameters, i.e.\ proper tests of general 
relativity, with laboratory experiments.

\appendix
\numberwithin{equation}{chapter}

\chapter{Calculation of the classical Hamiltonian of a free particle} \label{app:details_class_Ham}

Here, we will give a full exposition of the calculation of 
the classical Hamiltonian of a free particle in a curved 
spacetime in $3+1$ decomposition.
\enlargethispage*{2\baselineskip}

In $3+1$ decomposition, spacetime is foliated into 
three-dimensional spacelike Cauchy surfaces that are 
labelled by a `foliation parameter' $t$. We employ 
\mbox{adapted} \mbox{coordinates} $(x^0 = ct, x^a)$ where $x^a$ are 
coordinates on these Cauchy surfaces. This gives a 
\mbox{decomposition} of the spacetime metric as
\begin{equation}
	g_{ab} = {^{(3)}g_{ab}}, \quad 
	g_{0a} = {^{(3)}g_{ab}} \beta^b =: \beta_a, \quad
	g_{00} = -\alpha^2 + {^{(3)}g_{ab}} \beta^a \beta^b \; ,
\end{equation}
where $^{(3)}g$ is the induced metric on the Cauchy 
surfaces, $\beta$ is the shift vector field and $\alpha$ 
is the lapse function. Geometrically speaking, lapse and 
shift arise from decomposing the `time evolution' vector 
field\footnote
	{Denoting the embeddings defining the foliation as 
	$\mathcal E_t\colon \Sigma \to M, \; t \in \mathbb R$ where 
	$\Sigma$ is the abstract Cauchy surface, the time evolution 
	vector field is given as the derivation 
	\[\left.\frac{\partial}{\partial t}\right|_{\mathcal E_s(q)} f 
		:= \left.\frac{\D}{\D s'} f(\mathcal E_{s'}(q))\right|_{s' = s}\]
	for $q\in \Sigma$, $s\in\mathbb R$ and $f \in C^\infty(M)$; 
	i.e.\ this vector field is independent of the choice of 
	coordinates and depends just on the foliation, even if it 
	was expressed above as a coordinate vector field. 
	\cite[(17.43)]{giulini14}\looseness-1}
$\partial_0 = c^{-1} \partial/\partial t$ into its 
components tangential and normal to the Cauchy surfaces as
\begin{equation}
	\frac{1}{c} \frac{\partial}{\partial t} = \alpha n + \beta,
\end{equation}
where $n$ is the future-directed unit normal to the Cauchy 
surfaces and $\beta$ is the tangential component 
\cite[(17.44)]{giulini14}.

Parametrising the worldline of a free particle by 
$t$, its Lagrangian (compare the classical action 
\eqref{eq:class_action}) in these coordinates is
\begin{equation} \label{eq:class_Lagr}
	L = -mc \sqrt{-g_{\mu\nu} \dot x^\mu \dot x^\nu} 
	= -mc \left(\alpha^2 c^2 - {^{(3)}g_{ab}} \beta^a \beta^b c^2 
		- 2 c \, {^{(3)}g_{ab}} \dot x^a \beta^b 
		- {^{(3)}g_{ab}} \dot x^a \dot x^b\right)^{1/2},
\end{equation}
where a dot denotes differentiation with respect to $t$.
 
From this, we compute the momentum $p_a$ conjugate to $x^a$ 
to be
\begin{equation}
	p_a = \frac{\partial L}{\partial \dot x^a} 
	= \frac{mc}{(\ldots)^{1/2}}(c \beta_a + {^{(3)}g_{ab}} \dot x^b).
\end{equation}
Contracting with the inverse $^{(3)}g^{ab}$ of $^{(3)}g_{ab}$, 
we obtain
\begin{equation} \label{eq:dotx}
	\dot x^a = \frac{(\ldots)^{1/2}}{mc} \, {^{(3)}g^{ab}} p_b - c \beta^a \; .
\end{equation}

To fully express the velocity $\dot x^a$ in terms of 
the momentum $p_a$, we have to express $(\ldots)^{1/2} 
= \left(\alpha^2 c^2 - {^{(3)}g}(\dot{\vect x} + c \vect \beta, 
\dot{\vect x} + c \vect \beta)\right)^{1/2}$ in terms of 
$p_a$. Using \eqref{eq:dotx}, we have
\begin{align}
	^{(3)}g(\dot{\vect x} + c \vect \beta, \dot{\vect x} + c \vect \beta) 
	&= \frac{(\ldots)}{m^2c^2} \, {^{(3)}g^{ab}} p_a p_b \nonumber\\
	&= \frac{\alpha^2 c^2 - {^{(3)}g}(\dot{\vect x} + c \vect\beta, 
		\dot{\vect x} + c \vect\beta)}{m^2 c^2} \, {^{(3)}g^{ab}} p_a p_b \; .
\end{align}
Writing $^{(3)}g^{-1}(\vect p,\vect p) 
:= {^{(3)}g^{ab}} p_a p_b$, this is equivalent to
\begin{align}
	^{(3)}g(\dot{\vect x} + c \vect\beta, \dot{\vect x} + c \vect\beta) 
	&= \frac{\alpha^2 \, {^{(3)}g^{-1}}(\vect p,\vect p)}{m^2} 
		\, \frac{1}{1 + {^{(3)}g^{-1}}(\vect p,\vect p) / (m^2 c^2)} 
		\nonumber\\
	&= \frac{c^2 \alpha^2 \, {^{(3)}g^{-1}}(\vect p,\vect p)}{m^2c^2 
		+ {^{(3)}g^{-1}}(\vect p,\vect p)}
\end{align}
Using this, we get
\begin{equation} \label{eq:sqrt_expression}
	(\ldots)^{1/2} 
	= \left(\alpha^2 c^2 - \frac{c^2 \alpha^2 \, 
		{^{(3)}g^{-1}}(\vect p,\vect p)}{m^2 c^2 
			+ {^{(3)}g^{-1}}(\vect p,\vect p)}\right)^{1/2} 
	\kern-1em= \frac{m c^2 \alpha}{\left[m^2 c^2 
		+ {^{(3)}g^{-1}}(\vect p,\vect p)\right]^{1/2}} \; .
\end{equation}
Inserting \eqref{eq:sqrt_expression} into \eqref{eq:dotx}, 
we can express the velocities in terms of the momenta as
\begin{equation} \label{eq:dotx_momentum}
	\dot x^a 
	= \frac{\alpha c}{\left[m^2 c^2 
			+ {^{(3)}g^{-1}}(\vect p,\vect p)\right]^{1/2}} \, 
		{^{(3)}g^{ab}} p_b - c \beta^a \; .
\end{equation}
\enlargethispage{2\baselineskip}

Using \eqref{eq:sqrt_expression} and \eqref{eq:dotx_momentum}, 
the Hamiltonian corresponding to the Lagrangian 
\eqref{eq:class_Lagr} is
\begin{align}
	H &= p_a \dot x^a - L \nonumber\\
	&= \frac{\alpha c}{\left[m^2 c^2 
				+ {^{(3)}g^{-1}}(\vect p,\vect p)\right]^{1/2}} \, 
			{^{(3)}g^{-1}}(\vect p,\vect p) 
		- c \beta^a p_a 
		+ \frac{m^2 c^3 \alpha}{\left[m^2 c^2 
			+ {^{(3)}g^{-1}}(\vect p,\vect p)\right]^{1/2}} \nonumber\\
	&= \alpha c \left[m^2 c^2 + {^{(3)}g^{-1}}(\vect p,\vect p)\right]^{1/2} 
		\kern-.8em- c \beta^a p_a \; .
\end{align}
Rewriting this in terms of the components of the spacetime 
metric using the relations $g^{00} = -\alpha^{-2}$, 
$g^{0a} = \alpha^{-2} \beta^a$, $g^{ab} = {^{(3)}g^{ab}} 
- \alpha^{-2} \beta^a \beta^b$, the Hamiltonian reads
\begin{equation}
	H = \frac{1}{\sqrt{-g^{00}}} \, m c^2 \left[1 
			+ \left(g^{ab} - \frac{1}{g^{00}} g^{0a} g^{0b}\right) 
				\frac{p_a p_b}{m^2 c^2}\right]^{1/2} 
		\kern-.8em+ \frac{c}{g^{00}} g^{0a} p_a \; .
\end{equation}


\chapter{Christoffel symbols of the Eddington--Robertson \acronym{PPN} metric}
\label{app:Christoffel}

Here, we compute the Christoffel symbols
\begin{equation}
	\Gamma^\rho_{\mu\nu} = \frac{1}{2} g^{\rho\sigma} (\partial_\mu g_{\nu\sigma} + \partial_\nu g_{\mu\sigma} - \partial_\sigma g_{\mu\nu})
\end{equation}
of the Eddington--Robertson \acronym{PPN} metric as given 
by \eqref{eq:ER_PPN_metric}, \eqref{eq:ER_PPN_metric_inv}, 
keeping full track of all details of the $c^{-1}$ expansion.
\enlargethispage{2\baselineskip}

\begin{align}
	\Gamma^0_{00} &= \frac{1}{2} g^{00} \partial_0 g_{00} + \Or(c^{-7}) \nonumber\\
	&= \frac{1}{2} \left(\frac{1}{g_{00}} + \Or(c^{-6})\right) \partial_0 g_{00} + \Or(c^{-7}) \nonumber\\
	&= \frac{1}{2c} \partial_t {\underbrace{\ln(-g_{00})}_{\mathrlap{= 2\frac{\phi}{c^2} + 2(\beta-1) \frac{\phi^2}{c^4} + \Or(c^{-6})}}} + \Or(c^{-7}) \nonumber\\
	&= \frac{\partial_t \phi}{c^3} + 2(\beta-1) \frac{\phi \partial_t \phi}{c^5} + \Or(c^{-7})
\end{align}

\begin{align}
	\Gamma^0_{0a} &= \frac{1}{2} g^{00} (\cancel{\partial_0 g_{0a}} + \partial_a g_{00} - \cancel{\partial_0 g_{0a}}) + \Or(c^{-7}) \nonumber\\
	&= \frac{1}{2} \partial_a \ln(-g_{00}) + \Or(c^{-6}) \nonumber\\
	&= \frac{\partial_a \phi}{c^2} + 2(\beta-1) \frac{\phi \partial_a \phi}{c^4} + \Or(c^{-6})
\end{align}

\begin{align}
	\Gamma^0_{ab} &= \frac{1}{2} g^{00} ({\underbrace{\partial_a g_{0b} + \partial_b g_{0a}}_{= \Or(c^{-5})}} - \partial_0 g_{ab}) + \Or(c^{-7}) \nonumber\\
	&= \frac{1}{2} \left(-1 + \Or(c^{-2})\right) \frac{-1}{c} \partial_t \left(\left(1 - 2\gamma \frac{\phi}{c^2}\right) \delta_{ab} + \Or(c^{-4})\right) + \Or(c^{-5}) \nonumber\\
	&= -\gamma \frac{\partial_t \phi}{c^3} \delta_{ab} + \Or(c^{-5})
\end{align}

\begin{align}
	\Gamma^a_{00} &= \frac{1}{2} g^{ab} (2{\underbrace{\partial_0 g_{0b}}_{\mathclap{= \Or(c^{-6})}}} - \partial_b g_{00}) + \Or(c^{-7}) \nonumber\\
	&= \frac{1}{2} \left(\left(1 + 2\gamma \frac{\phi}{c^2}\right) \delta^{ab} + \Or(c^{-4})\right) \partial_b \left(1 + 2 \frac{\phi}{c^2} + 2\beta \frac{\phi^2}{c^4} + \Or(c^{-6})\right) + \Or(c^{-6}) \nonumber\\
	&= \delta^{ab}\left(\frac{\partial_b \phi}{c^2} + 2(\beta + \gamma) \frac{\phi\partial_b \phi}{c^4}\right) + \Or(c^{-6})
\end{align}

\begin{align}
	\Gamma^a_{0b} &= \frac{1}{2} g^{ac} (\partial_0 g_{bc} + {\underbrace{\partial_b g_{0c} - \partial_c g_{0b}}_{\mathclap{= \Or(c^{-5})}}}) + \Or(c^{-7}) \nonumber\\
	&= \frac{1}{2} \left(\delta^{ac} + \Or(c^{-2})\right) \frac{1}{c} \partial_t \left(\left(1 - 2\gamma \frac{\phi}{c^2}\right) \delta_{bc} + \Or(c^{-4})\right) + \Or(c^{-5}) \nonumber\\
	&= -\gamma \delta^a_b \frac{\partial_t \phi}{c^3} + \Or(c^{-5})
\end{align}

\begin{align}
	\Gamma^a_{bc} &= \frac{1}{2} g^{ad} (\partial_b g_{cd} + \partial_c g_{bd} - \partial_d g_{bc}) + \Or(c^{-7}) \nonumber\\
	&= \frac{1}{2} \left(\delta^{ad} + \Or(c^{-2})\right) (-2\gamma) \left(\delta_{cd} \frac{\partial_b \phi}{c^2} + \delta_{bd} \frac{\partial_c \phi}{c^2} - \delta_{bc} \frac{\partial_d \phi}{c^2} + \Or(c^{-4})\right) + \Or(c^{-7}) \nonumber\\
	&= -\gamma \frac{\delta^a_c \partial_b \phi + \delta^a_b \partial_c \phi - \delta_{bc} \delta^{ad} \partial_d \phi}{c^2} + \Or(c^{-4})
\end{align}
The last result implies $g^{bc} \Gamma^a_{bc} = \delta^{ab} \gamma\frac{\partial_b \phi}{c^2} + \Or(c^{-4})$, in turn implying
\begin{align} \label{eq:trace_Christoffel_spatial}
	g^{\mu\nu} \Gamma^a_{\mu\nu} &= g^{00} \Gamma^a_{00} + g^{bc} \Gamma^a_{bc} + \Or(c^{-8}) \nonumber\\
	&= (-1 + \Or(c^{-2})) \left(\delta^{ab} \frac{\partial_b \phi}{c^2} + \Or(c^{-4})\right) + \delta^{ab} \gamma\frac{\partial_b \phi}{c^2} + \Or(c^{-4}) \nonumber\\
	&= (\gamma-1) \delta^{ab} \frac{\partial_b \phi}{c^2} + \Or(c^{-4}).
\end{align}


\chapter{Sign conventions for generators of special orthogonal groups}
\label{app:sign_convention_so}

Here we discuss our choice of sign convention for the 
generators of special orthogonal groups, in particular 
the Lorentz group.\enlargethispage*{3\baselineskip}

Let $V$ be a finite-dimensional real vector space 
with a non-degenerate, symmetric bilinear form 
$g\colon V \times V \to \mathbb R$. Note that we do 
not assume anything about the signature of $g$. 
We introduce the `musical isomorphism'
\begin{equation} \label{eq:musical_isom}
	V \to V^*, v \mapsto v^\flat := g(v, \cdot)
\end{equation}
induced by $g$.

We fix a basis $\{\E_a\}_a$ of $V$. As bases for its 
dual vector space $V^*$ we distinguish its natural 
dual basis $\{\theta^a\}_a$, where $\theta^a(\E_b)=\delta^a_b$,
and the ($g$-dependent) image of $\{\E_a\}_a$ under 
\eqref{eq:musical_isom}, which is just $\{e^\flat_a\}_a$, 
where $e^\flat_a=g_{ab}\theta^b$, so that $e^\flat_a(\E_b)=g_{ab}$.
The reason for this will become clear now. 

For each $a,b \in \{1, \dots, \dim V\}$ we introduce 
the endomorphism
\begin{equation} \label{eq:lie_basis-1}
	B_{ab} := \E_a \otimes \E_b^\flat - \E_b \otimes \E_a^\flat \in \mathrm{End}(V)
\end{equation}
which satisfies 
\begin{equation}
	g(v, B_{ab}(w)) = g(v,\E_a) g(\E_b,w) - g(v,\E_b) g(\E_a,w) = -g(B_{ab}(v), w).
\end{equation}
This means that $B_{ab}$ is anti-self-adjoint 
with respect to $g$ and hence that it is an 
element of the Lie algebra $\mathfrak{so}(V,g)$ 
of the Lie group $\mathsf{SO}(V,g)$ of special 
orthogonal transformations of $(V,g)$:
\begin{equation}
	B_{ab} \in \mathfrak{so}(V,g).
\end{equation}
As $B_{ab} = -B_{ba}$, it is the set 
$\{B_{ab}: 1\le a < b\le \dim V\}$ which is linearly 
independent and of the same dimension as 
$\mathfrak{so}(V,g)$. Hence this set forms a basis of 
$\mathfrak{so}(V,g)$ so that any 
$\omega \in \mathfrak{so}(V,g)$ can be uniquely 
written in the form
\begin{equation}
	\omega = \sum_{1\le a < b \le \dim V} \omega^{ab} B_{ab} = \frac{1}{2} \omega^{ab} B_{ab} \; ,
\end{equation}
where 
\begin{equation} \label{eq:lie_basis_anti_self_adj}
	\omega^{ab} = -\omega^{ba} \; .
\end{equation}
This representation can easily be compared 
to the usual one in terms of the metric-independent basis 
$\{\E_a\otimes\:\theta^b : 1\le a,b \le \dim V\}$ of 
$\mathrm{End}(V)$ in the following way: for 
$\omega=\omega^a_{\hphantom{a}c}\, \E_a \otimes \: \theta^c$,
we have $\omega\in \mathfrak{so}(V,g)$ if and only if 
\begin{equation} \label{eq:end_basis_anti_self_adj}
	\omega^a_{\hphantom{a}c} \, g^{cb} = - \omega^b_{\hphantom{a}c} \, g^{ca} \; .
\end{equation}
It is the obvious simplicity of \eqref{eq:lie_basis_anti_self_adj}
as opposed to \eqref{eq:end_basis_anti_self_adj}
as conditions for $\omega\in\mathrm{End}(V)$
being contained in  
$\mathfrak{so}(V,g)\subset\mathrm{End}(V)$ that 
makes it easier to work with the basis 
$\E_a\otimes \E_b^\flat$ of $\mathrm{End}(V)$ rather 
than $\E_a\otimes\:\theta^b$. Note that the components 
of $\omega$ with respect to the two bases considered 
above are connected by the equation
\begin{equation}
	\omega^{ab} = \omega^a_{\hphantom{a}c} \, g^{cb} \; .
\end{equation}

The basis elements $B_{ab}$ satisfy 
the commutation relations
\begin{align}
	[B_{ab}, B_{cd}] &= g_{bc} B_{ad} + g_{ad} B_{bc} - g_{ac} B_{bd} - g_{bd} B_{ac} \nonumber\\
	&= g_{bc} B_{ad} + \text{(antisymm.)},
\end{align}
where `antisymm.' denotes antisymmetrisation as shown 
in the first line of the equation.

From now on, we will assume the basis $\{\E_a\}_a$ to be 
orthonormal. For notational convenience, 
for $a,b\in \{1, \dots, \dim V\}$ we define
\begin{equation}
	\varepsilon_{ab} := g_{aa} g_{bb} = \pm 1
\end{equation}
which has the value $+1$ if $g_{aa} = g(\E_a, \E_a)$ 
and $g_{bb} = g(\E_b, \E_b)$ have the same sign, 
and $-1$ if they have opposite signs\footnote
	{Note that repeated indices on the same 
	level, i.e.\ both up or both down, are not to be 
	summed over.}.

We now want to compute the exponential  
$\exp(\alpha B_{ab}) \in \mathsf{SO}(V,g)$. 
At first, we note that
\begin{align}
	(B_{ab})^2 
	&= - g_{bb} \E_a \otimes \E_a^\flat 
	 - g_{aa} \E_b \otimes \E_b^\flat \nonumber\\
	&= - \varepsilon_{ab} \; \mathrm{Pr}_{ab} \; ,
\end{align}
where $\mathrm{Pr}_{ab} := \mathrm{Pr}_{\mathrm{span}\{\E_a, \E_b\}}$ 
denotes the $g$-orthogonal projector onto the plane 
$\mathrm{span}\{\E_a, \E_b\}$ in $V$.\footnote
	{In the general case of two linearly independent 
	vectors $v,w \in V$, not necessarily orthonormal, the 
	orthogonal projector is given by
		\begin{equation*}
			\mathrm{Pr}_{\mathrm{span}(v,w) = \frac{1}{g(v,v) g(w,w) - (g(v,w))^2}} \left[ g(w,w) \; v \otimes v^\flat + g(v,v) \; w \otimes w^\flat - g(v,w) \; (v \otimes w^\flat + w \otimes v^\flat) \right],
		\end{equation*}
		implying
		\begin{align*}
			(v \otimes w^\flat - w \otimes v^\flat)^2 &= - g(w,w) \; v \otimes v^\flat - g(v,v) \; w \otimes w^\flat + g(v,w) \; (v \otimes w^\flat + w \otimes v^\flat) \nonumber\\
			&= - \left[ g(v,v) g(w,w) - (g(v,w))^2 \right] \mathrm{Pr}_{\mathrm{span}(v,w)}.
		\end{align*}
	}
Using this and $B_{ab} \circ \mathrm{Pr}_{ab} = B_{ab}$, 
the exponential series evaluates to
\begin{align}
	\exp(\alpha B_{ab}) 
	&= (\mathrm{id}_V - \mathrm{Pr}_{ab}) 
	 + \sum_{k=0}^\infty \frac{1}{(2k)!} \, 
	   \alpha^{2k} (-\varepsilon_{ab})^k \, 
	   \mathrm{Pr}_{ab} \nonumber\\ &\quad+ 
	   \sum_{k=0}^\infty \frac{1}{(2k+1)!} \, 
	   \alpha^{2k+1} (-\varepsilon_{ab})^k \, 
	    B_{ab} \circ \mathrm{Pr}_{ab} \nonumber\\
	&= (\mathrm{id}_V - \mathrm{Pr}_{ab})
	 + \left\{ \!\! \begin{aligned}
		&\cos(\alpha) \, \mathrm{id}_V + \sin(\alpha) \, 
		 B_{ab} \, , & \varepsilon_{ab} = +1 \\
		&\cosh(\alpha) \, \mathrm{id}_V + \sinh(\alpha) \, 
		 B_{ab} \, , & \varepsilon_{ab} = -1
	\end{aligned} \right\} \circ \mathrm{Pr}_{ab} \; .
\end{align}
Geometrically, this transformation is either a 
rotation by angle $\alpha$ (for $\varepsilon_{ab} = +1$) 
or a boost by rapidity $\alpha$ (for 
$\varepsilon_{ab} = -1$) in the plane 
$\mathrm{span}\{\E_a, \E_b\}$. The direction of the 
transformation depends on the signs of $g_{aa}, g_{bb}$:
\begin{itemize}
	\item $\varepsilon_{ab} = +1$:
		\begin{enumerate}[label=(\roman*)]
			\item $g_{aa} = g_{bb} = +1$: We have 
				$B_{ab}(\E_a) = -\E_b, B_{ab}(\E_b) = \E_a$. 
				Thus, $\exp(\alpha B_{ab})$ is a rotation 
				by $\alpha$ \emph{from $\E_b$ towards $\E_a$}.
			\item $g_{aa} = g_{bb} = -1$: We have 
				$B_{ab}(\E_a) = \E_b, B_{ab}(\E_b) = -\E_a$. 
				Thus, $\exp(\alpha B_{ab})$ is a rotation 
				by $\alpha$ \emph{from $\E_a$ towards $\E_b$}.
		\end{enumerate}
	\item $\varepsilon_{ab} = -1$:
		\begin{enumerate}[label=(\roman*)]
			\item $g_{aa} = +1, g_{bb} = -1$: We have 
				$B_{ab}(\E_a) = -\E_b, B_{ab}(\E_b) = -\E_a$. 
				Thus, $\exp(\alpha B_{ab})$ is a boost 
				by $\alpha$ \emph{`away' from $\E_a + \E_b$}.
			\item $g_{aa} = -1, g_{bb} = +1$: We have 
				$B_{ab}(\E_a) = \E_b, B_{ab}(\E_b) = \E_a$. 
				Thus, $\exp(\alpha B_{ab})$ is a boost 
				by $\alpha$ \emph{`towards' $\E_a + \E_b$}.
		\end{enumerate}
\end{itemize}
\enlargethispage{\baselineskip}

Now we will apply the preceding considerations to the case 
of (the `difference' vector space of) Minkowski spacetime, 
where for now \emph{we leave open the signature convention 
for the metric} (either $(+{--}-)$ or $(-{++}+)$). We work 
with respect to a positively oriented orthonormal basis 
$\{\E_\mu\}_{\mu = 0, \dots, 3}$ where $\E_0$ is timelike. 
Latin indices will denote spacelike directions.

In the case of `mostly minus' signature $(+{--}-)$, 
$B_{ab}$ generates rotations from $\E_a$ towards $\E_b$ and 
$B_{a0}$ generates boosts (with respect to $\E_0$) in 
direction of $\E_a$. In the case of `mostly plus' signature 
$(-{++}+)$, $B_{ba} = - B_{ab}$ generates rotations from 
$\E_a$ towards $\E_b$ and $B_{0a} = - B_{a0}$ generates 
boosts (with respect to $\E_0$) in direction of $\E_a$.

Thus, since we want to use the notation $J_{ab}$ for the 
spacelike rotational generator generating rotations from 
$\E_a$ towards $\E_b$, we have to set
\begin{equation} \label{eq:sign_convention_so}
	J_{\mu\nu} = \begin{cases}
		B_{\mu\nu} & \text{for $(+{--}-)$ signature},\\
		-B_{\mu\nu} & \text{for $(-{++}+)$ signature}
	\end{cases}
\end{equation}
for the Lorentz generators. Adopting this convention, 
boosts in direction of $\E_a$ are then generated 
by $J_{a0}$. The commutation relations for the 
$J_{\mu\nu}$ are
\begin{equation}
	[J_{\mu\nu}, J_{\rho\sigma}] = \begin{cases}
		\eta_{\mu\sigma} J_{\nu\rho} + \text{(antisymm.)} & \text{for $(+{--}-)$ signature},\\
		\eta_{\mu\rho} J_{\nu\sigma} + \text{(antisymm.)} & \text{for $(-{++}+)$ signature},
	\end{cases}
\end{equation}
and general Lorentz algebra elements 
$\omega \in \mathrm{Lie}(\mathcal L)$ can be written as
\begin{equation}
	\omega = \pm \frac{1}{2} \omega^{\mu\nu} J_{\mu\nu} \; \text{with} \; \omega^{\mu\nu} = \omega^\mu_{\hphantom{\mu}\rho} \; \eta^{\rho\nu}
\end{equation}
in terms of their components $\omega^\mu_{\hphantom{\mu}\rho}$ 
as endomorphisms, where the upper/lower sign holds for 
$(+{--}-)$/$(-{++}+)$ signature.


\chapter{Notes on the adjoint representation}
\label{app:adj_rep}

Here we wish to make a few remarks and collect a 
few formulae concerning the adjoint and co-adjoint 
representation of the general linear group of a vector 
space $V$.

In the defining representation on $V$, an element 
$\Lambda \in \mathsf{GL}(V)$ is given in terms of the 
basis $\{\E_a\}_a$ by the coefficients 
$\Lambda^a_{\phantom{a}_b}$, where
\begin{equation} \label{eq:coeff_def_rep}
	\Lambda \E_a = \Lambda^b_{\phantom{b}a} \E_b \; .
\end{equation}  
This defines a left action of $\mathsf{GL}(V)$ on 
$V$. The corresponding left action of 
$\mathsf{GL}(V)$ on the dual space $V^*$ is given 
by the inverse-transposed, i.e.\  
$\mathsf{GL}(V) \times V^* \to V^*$, 
$(\Lambda,\alpha) \mapsto (\Lambda^{-1})^\top
\alpha := \alpha \circ \Lambda^{-1}$. For the  
basis $\{\theta^a\}_a$ of $V^*$ dual to the 
basis $\{\E_a\}_a$ this means
\begin{equation}
	\theta^a\circ \Lambda^{-1} = (\Lambda^{-1})^a_{\phantom{a}b} \, \theta^b \; .
\end{equation}  
In contrast, for the basis $\{\E_a^\flat\}_a$
of $V^*$, this reads in general
\begin{equation}
	{\E_b^\flat} \circ \Lambda^{-1} = g^{ac} g_{bd} (\Lambda^{-1})^d_{\phantom{d}c} \E_a^\flat \; ,
\end{equation}  
which for isometries $\Lambda \in \mathsf{O}(V,g)$ simply becomes
\begin{equation} \label{eq:coeff_def_rep_dual_isom}
	{\E_b^\flat} \circ \Lambda^{-1} = \Lambda^a_{\phantom{a}b} \E_a^\flat \; .
\end{equation}  

The adjoint representation of $\mathsf{GL}(V)$ on 
$\mathrm{End}(V)\cong V\otimes V^*$ or any Lie subalgebra
of $\mathrm{End}(V)$ is by conjugation, which for 
our basis \eqref{eq:lie_basis-1} implies, using \eqref{eq:coeff_def_rep} and \eqref{eq:coeff_def_rep_dual_isom},
\begin{equation}
	\mathrm{Ad}_\Lambda B_{ab}
	= \Lambda \circ B_{ab} \circ \Lambda^{-1}
	= \Lambda^c_{\phantom{c}a} \Lambda^d_{\phantom{d}b} \, B_{cd}
	\quad \text{for} \; \Lambda \in \mathsf{O}(V,g).
\end{equation}

The adjoint 
representation of the inhomogeneous 
group $\mathsf{GL}(V)\ltimes V$ on its Lie algebra 
$\mathrm{End}(V)\oplus V$ is given by, for any 
$X \in \mathrm{End}(V)$ and $y \in V$,
\begin{equation}
	\mathrm{Ad}_{(\Lambda,a)}(X,y) =
	\left( \Lambda\circ X\circ\Lambda^{-1} ,
	\Lambda y - (\Lambda \circ X \circ \Lambda^{-1}) a \right).
\end{equation}
In the main text we will use this formula for 
$(\Lambda,a)$ being replaced by its inverse 
$(\Lambda,a)^{-1} = (\Lambda^{-1}, -\Lambda^{-1}a)$:
\begin{equation}
	\mathrm{Ad}_{(\Lambda,a)^{-1}}(X,y) =
	\left( \Lambda^{-1} \circ X \circ \Lambda ,
	\Lambda^{-1} y + (\Lambda^{-1} \circ X) a \right)
\end{equation}  
Applied to the basis vectors separately, 
i.e.\ to $(X,y) = (0,\E_b)$ and $(X,y) = (B_{bc},0)$,
for $\Lambda \in \mathsf{O}(V,g)$ we get 
\begin{subequations} \label{eq:ad_rep}
\begin{alignat}{2}
	&\mathrm{Ad}_{(\Lambda,a)^{-1}}(0,\E_b)
	&&= \left( 0, (\Lambda^{-1})^c_{\phantom{c}b} \E_c \right), \\
	&\mathrm{Ad}_{(\Lambda,a)^{-1}}(B_{bc},0)
	&&= \left( (\Lambda^{-1})^d_{\phantom{d}b} (\Lambda^{-1})^e_{\phantom{e}c} \, B_{de} ,
	- a_b (\Lambda^{-1})^d_{\phantom{d}c} \E_d
	+ a_c (\Lambda^{-1})^d_{\phantom{d}b} \E_d \right)
\end{alignat}  
\end{subequations}
where $a_b := \E_b^\flat(a) = g_{bc} a^c$ in the second 
equation. From these equations we immediately 
deduce \eqref{eq:coadjoint_rep_components} in the 
case of four spacetime dimensions (greek indices) 
and signature mostly plus, in which case 
$J_{\mu\nu} = -B_{\mu\nu}$ according to
\eqref{eq:sign_convention_so}.

\printbibliography


\chapter*{Curriculum Vitae}

Philip Klaus Schwartz, born 12.07.1994 in Langenhagen

\section*{Education and professional experience}
\begin{tabular}{r|p{12.2cm}}
	2000--2004 & Freie Evangelische Schule Hannover\\
	2004--2011 & Leibnizschule Hannover, Abitur 2011\\
	2011--2015 & Studies of physics and mathematics at the Leibniz University Hannover, 2014 Bachelor of Science in Physics\\
	2015--2016 & `Part \acronym{III} of the Mathematical Tripos' at the University of Cambridge, 2016 Master of Advanced Study in Applied Mathematics\\
	since 2016 & \emph{Wissenschaftlicher Mitarbeiter} (scientific employee) / PhD student at the Institute for Theoretical Physics of the Leibniz University Hannover
\end{tabular}

\section*{Scholarships}
\begin{tabular}{r|p{12.2cm}}
	2014--2016 & \emph{Studienstiftung des Deutschen Volkes}\\
	2015--2016 & \emph{Trinity Studentship in Mathematics} from Trinity College, Cambridge
\end{tabular}

\end{document}